\documentclass{lmcs}
\pdfoutput=1

\usepackage{lastpage}
\lmcsdoi{22}{3}{11}
\lmcsheading{}{\pageref{LastPage}}{}{}%
{Apr.~04,~2025}{Sep.~22,~2026}{}

\usepackage{hyperref}

\usepackage[utf8]{inputenc}
\usepackage[ruled,lined]{algorithm2e}

\usepackage{booktabs} 

\usepackage{IEEEtrantools}

\usepackage{listings}
\usepackage[colorinlistoftodos,prependcaption]{todonotes}

\usepackage{textcomp}
\usepackage{cancel}

\usepackage{tikz}
\usepackage{tikz-cd}
\usetikzlibrary{patterns}

\usepackage{graphicx}

\usepackage{amsmath,amssymb,latexsym,fancyhdr}

\usepackage{etex}
\usetikzlibrary{calc} 						
\usetikzlibrary{shapes,arrows}				
\usetikzlibrary{shapes.multipart}			
\usetikzlibrary{positioning}
\usepackage[all,cmtip]{xy}					
\usepackage{tikz-qtree,tikz-qtree-compat}	
\usepackage{mathtools}
\usepackage{url} 
\usepackage{graphics}

\usepackage{wrapfig}
\usepackage{accents} 

\newcommand{\cfrozen}[1]{\red{#1}}

\newcommand{\red}[1]{{\color{red} #1}}

\newcommand{\naturals}{\mathbb{N}}

\def\defemb#1#2{\expandafter\def\csname #1\endcsname
{\relax\ifmmode #2\else\hbox{$#2$}\fi}}
\defemb{cA}{\mathcal{A}}
\defemb{cB}{\mathcal{B}}
\defemb{cC}{\mathcal{C}}
\defemb{cD}{\mathcal{D}}
\defemb{cE}{\mathcal{E}}
\defemb{cF}{\mathcal{F}}
\defemb{cG}{\mathcal{G}}
\defemb{cH}{\mathcal{H}}
\defemb{cI}{\mathcal{I}}
\defemb{cJ}{\mathcal{J}}
\defemb{cK}{\mathcal{K}}
\defemb{cL}{\mathcal{L}}
\defemb{cM}{\mathcal{M}}
\defemb{cN}{\mathcal{N}}
\defemb{cO}{\mathcal{O}}
\defemb{cP}{\mathcal{P}}
\defemb{cQ}{\mathcal{Q}}
\defemb{cR}{\mathcal{R}}
\defemb{cS}{\mathcal{S}}
\defemb{cT}{\mathcal{T}}
\defemb{cU}{\mathcal{U}}
\defemb{cV}{\mathcal{V}}
\defemb{cX}{\mathcal{X}}
\defemb{cZ}{\mathcal{Z}}

\newcommand{\ul}[1]{\underline{#1}}

\newcommand{\SInterpretation}{\cI}

\newcommand{\SemDomain}{\cA}

\newcommand{\Maude}{{\sf Maude}}

\newcommand{\CONFident}{\textsf{CONFident\/}}

\newcommand{\infChecker}{\mbox{\sf infChecker}}

\newcommand{\AGES}{\mbox{\sf AGES}}
\newcommand{\MaceFour}{\mbox{\sf Mace4}}

\newcommand{\TTTTwo}{\textsf{T\kern-0.15em\raisebox{-0.55ex}T\kern-0.15emT\kern-0.15em\raisebox{-0.55ex}2}}

\newcommand{\Var}{{\cV}ar} 

\newcommand{\bigfracLabel}[2]{
\begin{array}{c}
\displaystyle #1\\\hline\displaystyle #2
\end{array}}
\newcommand{\bigfracn}[3]{
\begin{array}[b]{c}
\displaystyle #1 \\\hline\displaystyle #2
\end{array}
\hbox to 0pt{\raisebox{0.7em}{{\tiny (#3)}}}
}

\newcommand{\toppos}{{\Lambda}}
\newcommand{\ol}[1]{\overline{#1}}  
\newcommand{\Pos}{{\mathcal{P}os}}

\newcommand{\arityOf}[2]{{#1_{/#2}}}

\newcommand{\cpos}{{\cdot}} 

\newcommand{\rootTerm}{\mathit{root}}

\newcommand{\Symbols}{{\cF}}

\newcommand{\DSymbols}{{\cD}}
\newcommand{\Variables}{{\cX}}

\newcommand{\TermsOn}[2]{{\cT(#1,#2)}}

\newcommand{\Terms}{{\TermsOn{\Symbols}{\Variables}}}

\newcommand{\trs}{\text{TRS}}
\newcommand{\etrs}{\text{ETRS}}

\newcommand{\cstrs}{\text{CS-TRS}}
\newcommand{\csctrs}{\text{CS-CTRS}}
\newcommand{\ctrs}{\text{CTRS}}
\newcommand{\crms}{\text{CRMS}}

\newcommand{\egtrs}{\text{EGTRS}}

\newcommand{\gtrs}{\text{GTRS}}

\newcommand{\mel}{\text{MEL}}

\newcommand{\Rmaps}[1]{{M_{#1}}}

\newcommand{\muTop}{{\mu_\top}}
\newcommand{\muBot}{{\mu_\bot}}

\newcommand{\NVar}[1]{\mathcal{V}ar^{\cancel{#1}}} 

\newcommand{\SSorts}{{\cS}}
\newcommand{\SKinds}{{\cK}}
\newcommand{\SSymbols}{{\Sigma}}
\newcommand{\SPredicates}{{\Pi}}

\newcommand{\SPSignature}{{\Omega}}

\newcommand{\RuleSymmetry}{\text{Sy}}

\newcommand{\RuleHornClause}{\text{HC}}

\newcommand{\RuleRlEq}{\text{R,E}} 
\newcommand{\RuleRewMEq}{\text{R/E}} 

\newcommand{\RuleReflexivity}{\text{Rf}}

\newcommand{\RuleCompatibility}{\text{Co}} 
\newcommand{\RuleTransitivity}{\text{Tr}} 
\newcommand{\RulePropagation}{\text{Pr}} 
\newcommand{\RuleInnerRed}{\text{InR}} 

\newcommand{\pUN}{\fS{UN}} 
\newcommand{\pUNred}{\fS{UN}^\to} 

\newcommand{\GLatom}{A}

\newcommand{\GLinference}{\cI}
\newcommand{\GLinferenceOf}[1]{\GLinference(#1)}

\newcommand{\ElementaryIS}{\text{EIS}}
\newcommand{\eis}{\text{EIS}}

\newcommand{\proofInISof}[2]{{\vdash_{#1}#2}}

\newcommand{\SigParameter}{\mathbb{S}}

\newcommand{\RMaParameter}{\mathbb{M}}

\newcommand{\EqParameter}{\mathbb{E}}
\newcommand{\TRSParameter}{\mathbb{R}}

\newcommand{\GLtheory}{\mathsf{Th}}

\newcommand{\rewtheoryOf}[1]{{\ol{#1}}}
\newcommand{\crtheoryOf}[1]{{\ol{#1^\CRabbr}}}

\newcommand{\gbop}{{\mathrel{\bowtie}}}

\newcommand{\deductionInThOf}[2]{{#1\vdash #2}}

\newcommand{\rootTree}{\mathit{root}}

\newcommand{\eqoriented}[1]{\stackrel{\diamond}{#1}}
\newcommand{\lhsr}{\ell}
\newcommand{\rhsr}{r}
\newcommand{\gencond}{c}
\newcommand{\gencondbis}{d}

\newcommand{\elhsr}{\lambda} 
\newcommand{\erhsr}{\rho}

\newcommand{\IF}{\Leftarrow}

\newcommand{\ceq}{=}
\newcommand{\cto}{\approx}

\newcommand{\CP}{{\sf CP}}

\newcommand{\ECP}{{\sf ECP}}

\newcommand{\SetOf}[1]{{\textit{#1}}}

\newcommand{\CCP}{{\sf CCP}}

\newcommand{\CVP}{{\sf CVP}}
\newcommand{\CVPof}[1]{{\sf CVP}^{#1}}
\newcommand{\CVPofR}{\CVPof{\rew{}}}
\newcommand{\CVPofPStickel}{\CVPof{\rewpstickel{}}}
\newcommand{\CVPofEqOne}{\CVPof{\equone{}}}

\newcommand{\LCCP}{{\sf LCCP}}

\newcommand{\GLCCP}{\text{GLCCP}}

\newcommand{\DCP}{{\sf DCP}}

\newcommand{\ccpOf}[1]{\pi_{#1}}
\newcommand{\ecpOf}[1]{\pi^{\textsc{ecp}}_{#1}}
\newcommand{\eccpOf}[1]{\pi^{\textsc{eccp}}_{#1}}
\newcommand{\iccpOf}[1]{\pi^{\textsc{i}}_{#1}}
\newcommand{\lccpOf}[1]{\pi^{\textsc{lccp}}_{#1}}
\newcommand{\ilccpOf}[1]{\pi^{i\textsc{lccp}}_{#1}}
\newcommand{\cvpOf}[2]{\pi^{#1}_{#2}}
\newcommand{\cvpROf}[1]{\cvpOf{\rew{}}{#1}}
\newcommand{\cvpEqOneOf}[1]{\cvpOf{\equone{}}{#1}}
\newcommand{\cvpPStickelOf}[1]{\cvpOf{\rewpstickel{}}{#1}}
\newcommand{\dcpOf}[1]{\pi^{\textsc{dcp}}_{#1}}

\newcommand{\genrelation}{\mathsf{R}}
\newcommand{\genrelationbis}{\mathsf{S}}
\newcommand{\invgenrelation}{\mathsf{R}^{-1}}
\newcommand{\invgenrelationbis}{\mathsf{S}^{-1}}
\newcommand{\genequivalence}{\mathsf{E}}
\newcommand{\genrelationUpE}{\genrelation^{\genequivalence}}
\newcommand{\composeRel}[2]{#1\circ #2}
\newcommand{\spcrel}[1]{\:#1\:}

\newcommand{\genrewmoduloOf}[2]{#1/#2}

\newcommand{\genequality}{\sim}

\newcommand{\rew}[1]{\to_{#1}}
\newcommand{\rewAtPos}[2]{\stackrel{#1}{\longrightarrow}_{#2}}

\newcommand{\rewp}[1]{\to^+_{#1}}
\newcommand{\rews}[1]{\to^*_{#1}}
\newcommand{\rewnorm}[1]{\to^!_{#1}}

\newcommand{\tos}[1]{\to^*_{#1}}

\newcommand{\equ}[1]{\genequality_{#1}}
\newcommand{\equone}[1]{\vdash\!\!\dashv_{#1}}
\newcommand{\equoneStar}[1]{{\vdash\!\!\stackrel{{}^*}{~}\!\!\dashv_{#1}}}

\newcommand{\oneCR}{\mid\!=\!\mid}

\newcommand{\oneCRs}{\mid\!\stackrel{*}{=}\!\mid}
\newcommand{\equequ}[1]{=_{#1}}

\newcommand{\innerrew}[1]{\stackrel{>\toppos}{\longrightarrow}_{#1}}

\newcommand{\leftrew}[1]{{{\:}_{#1}\!\!\leftarrow\:}}

\newcommand{\leftrewAtPos}[2]{{{\:}_{#2}\!\!\stackrel{#1}{\longleftarrow}\:}}

\newcommand{\leftrews}[1]{{{\:}^{\:*}_{#1}\!\!\leftarrow\:}}

\newcommand{\nfOf}[1]{{#1\!\!\downarrow}}

\newcommand{\theNfOf}[1]{{\widehat{#1}}}

\newcommand{\leftarrown}[1]{{{\:}^n\!\!\leftarrow}}

\newcommand{\leftarrowp}[1]{{{\:}^+\!\!\leftarrow}}
\newcommand{\leftarrows}[1]{{{\:}^*\!\!\leftarrow}}
\newcommand{\leftarrowWith}[1]{{{\:}_{#1}\!\!\leftarrow}}

\newcommand{\oneconversion}[1]{\leftrightarrow_{#1}}

\newcommand{\conversion}[1]{\leftrightarrow^*_{#1}}

\newcommand{\peakOf}[2]{{{}_{#1}\!\!\Uparrow_{#2}}}
\newcommand{\joinability}[1]{\downarrow_{#1}}
\newcommand{\rsjoinability}[1]{\downarrow^{\!\text{rs}}_{#1}}
\newcommand{\lsjoinability}[1]{\downarrow^{\!\text{ls}}_{#1}}
\newcommand{\sjoinability}[1]{\downarrow^{\!\text{s}}_{#1}}

\newcommand{\ejoinability}[1]{\:\widetilde{\downarrow}_{#1}\:}
\newcommand{\rsejoinability}[1]{\widetilde{\downarrow}^{\:\text{rs}}_{#1}}
\newcommand{\lsejoinability}[1]{\widetilde{\downarrow}^{\:\text{ls}}_{#1}}
\newcommand{\sejoinability}[1]{\widetilde{\downarrow}^{\:\text{s}}_{#1}}

\newcommand{\ejoin}[1]{\widetilde{\downarrow}_{#1}}

\newcommand{\joinOf}[2]{{{}_{#1}\!\!\Downarrow_{#2}}}
\newcommand{\ejoinOf}[2]{{{}_{#1}\widetilde{\Downarrow}{}_{#2}}}

\newcommand{\RmoduloEof}[2]{{#1/#2}}

\newcommand{\CRabbr}{\text{\sc cr}}
\newcommand{\RMabbr}{\mathit{rm}}
\newcommand{\showRMabbr}{\mathit{rm}}
\newcommand{\rmodulo}{\stackrel{\RMabbr}{\to}}

\newcommand{\rmStar}{{\;\mbox{$\stackrel{\RMabbr}{\longrightarrow}\hspace{.1cm}\hspace{-.2cm}^*\,$}}}
\newcommand{\rewmodulo}[1]{\stackrel{\RMabbr}{\to}_{#1}}
\newcommand{\rewsmodulo}[1]{\rmStar_{#1}}
\newcommand{\rewmodulos}[1]{\rmStar_{#1}}

\newcommand{\PSabbr}{\mathit{ps}}

\newcommand{\rpstickel}{\stackrel{\PSabbr}{\to}}

\newcommand{\rewpstickel}[1]{\stackrel{\PSabbr}{\to}_{#1}}
\newcommand{\rewspstickel}[1]{\;\mbox{$\stackrel{\PSabbr}{\longrightarrow}\hspace{.1cm}\hspace{-.2cm}^*_{#1}\,$}}

\pgfdeclarepatternformonly{VGrid}
{
  \pgfpointorigin
}
{
  \pgfqpoint{0.20cm}{0.20cm}
}
{  \pgfqpoint{0.20cm}{0.20cm}
}
{
  \pgfsetlinewidth{0.3pt}
  \pgfpathmoveto{\pgfqpoint{0cm}{0cm}}
  \pgfpathlineto{\pgfqpoint{0cm}{0.20cm}}
  \pgfusepath{stroke}
}

\pgfdeclarepatternformonly{HVGrid}
{
  \pgfpointorigin
}
{
  \pgfqpoint{0.20cm}{0.20cm}
}
{  \pgfqpoint{0.20cm}{0.20cm}
}
{
  \pgfsetlinewidth{0.3pt}
  \pgfpathmoveto{\pgfqpoint{0cm}{0cm}}
  \pgfpathlineto{\pgfqpoint{0cm}{0.20cm}}
  \pgfpathmoveto{\pgfqpoint{0cm}{0cm}}
  \pgfpathlineto{\pgfqpoint{0.20cm}{0cm}}
  \pgfusepath{stroke}
}

\newcommand{\ignore}[1]{}
\newcommand{\ignoreproofs}[1]{}

\newcommand{\seqnat}{\:{+\!+}\:}
\newcommand{\fS}[1]{\mathsf{#1}}

\tikzstyle{decision} = [diamond, draw, fill=yellow!20, text width=5em, text badly centered, minimum height=4em, inner sep=0pt, aspect=2]
\tikzstyle{block} = [rectangle, draw,fill=blue!20, text width=5em, text centered, minimum height=4em, rounded corners]
\tikzstyle{cloud} = [ellipse, draw,fill=red!20, text width=5em, text centered, minimum height=4em]
\tikzstyle{line} = [draw, -latex']
\tikzstyle{blockR} = [rectangle, draw, fill=red!20, text centered, minimum height=4em, rounded corners, minimum height=0.75cm]
\tikzstyle{blockB} = [rectangle, draw, fill=blue!20, text centered, minimum height=4em, rounded corners, minimum height=0.75cm]
\tikzstyle{blockG} = [rectangle, draw, fill=green!20, text centered, minimum height=4em, rounded corners, minimum height=0.75cm]
\tikzstyle{blockY} = [rectangle, draw, fill=yellow!20, text centered, minimum height=4em, rounded corners, minimum height=0.75cm]
\tikzstyle{blockW} = [rectangle, draw, fill=white!20, text centered, minimum height=4em, rounded corners, minimum height=0.75cm]
\tikzstyle{blockK} = [rectangle, draw, fill=gray!20, text centered, minimum height=4em, rounded corners, minimum height=0.75cm]
\tikzstyle{mblockB} = [rectangle, draw, fill=blue!20, text centered, minimum height=4em, double,rounded corners, minimum height=0.75cm]
\tikzstyle{mblockW} = [rectangle, draw, fill=white!20, text centered, minimum height=4em, double,rounded corners, minimum height=0.75cm]
\tikzstyle{mblockR} = [rectangle, draw, fill=red!20, text centered, minimum height=4em, double,rounded corners, minimum height=0.75cm]
\tikzstyle{mblockG} = [rectangle, draw, fill=green!20, text centered, minimum height=4em, double,rounded corners, minimum height=0.75cm]

\tikzstyle{circleY} = [circle, draw, fill=yellow!20, text centered, minimum height=4em, rounded corners, minimum height=0.75cm]
\tikzstyle{circleR} = [circle, draw, fill=red!20, text centered, minimum height=4em, rounded corners, minimum height=0.75cm]
\tikzstyle{circleG} = [circle, draw, fill=green!20, text centered, minimum height=4em, rounded corners, minimum height=0.75cm]
\tikzstyle{circleW} = [circle, draw, fill=white!20, text centered, minimum height=4em, rounded corners, minimum height=0.75cm]
\tikzstyle{triangleW} = [isosceles triangle, draw, fill=white!20, text centered, shape border rotate = 90, isosceles triangle stretches]
\tikzstyle{triangleG} = [isosceles triangle, draw, fill=green!20, text centered, shape border rotate = 90, isosceles triangle stretches]
\tikzstyle{triangleR} = [isosceles triangle, draw, fill=red!20, text centered, shape border rotate = 90, isosceles triangle stretches]

\newcommand{\pholder}{\_\!\_}

\newtheorem{notation}[thm]{Notation}

\begin{document}

\title{Confluence of Conditional Rewriting Modulo
}%
\thanks{Partially supported 
by grants PID2024-162030OB-100 funded by MCIN/AEI/10.13039/501100011033 and ERDF A way of making Europe,
and CIPROM/2022/6 funded by Generalitat Valenciana}

\author[S.~Lucas]{Salvador Lucas\lmcsorcid{0000-0001-9923-2108}}[]

\address{DSIC \& VRAIN, Universitat Polit\`ecnica de Val\`encia, Spain} 
\email{slucas@dsic.upv.es}


\begin{abstract}
Sets of equations $E$ play an important computational role in rewriting-based systems $\cR$. The equivalence relation $\equequ{E}$ induced by $E$ 
introduces a partition of terms into $E$-equivalence classes on which 
rewriting computations, denoted $\rew{\cR/E}$ and called \emph{rewriting modulo $E$}, are issued.
This paper investigates \emph{confluence of $\rew{\cR/E}$}, usually called 
\emph{$E$-confluence}, 
for \emph{conditional} rewriting-based systems, where rewriting steps are determined by \emph{conditional} rules.
We rely on Jouannaud and Kirchner's framework to investigate 
confluence of an abstract relation $\genrelation$ modulo an abstract equivalence relation $\genequivalence$ on a set $A$.
We show how to particularize such a framework to be used with conditional systems.
Then, we show how to define appropriate finite 
sets of \emph{conditional pairs} to prove and disprove $E$-confluence.
We introduce (i) \emph{Logic-based Conditional Critical Pairs}, which do not require the use of (often infinitely many) $E$-unifiers to provide a finite representation of the \emph{local peaks} considered in the abstract framework.
We also introduce (ii) \emph{parametric Conditional Variable Pairs} which are essential to deal with conditional rules in the analysis of $E$-confluence.
Finally, we introduce (iii) \emph{Down Conditional Pairs} which are often necessary to \emph{disprove} $E$-confluence.
Our results apply to well-known classes of rewriting-based systems, improving on previous results.
As for unconditional systems, our results apply to \emph{Equational Term Rewriting Systems}, first investigated by Huet and then by Jouannaud, and Jouannaud and Kirchner, among others.
As for conditional systems, our results also apply to conditional rewrite theories and \Maude.
\end{abstract}

\maketitle

\section{Introduction}
\label{SecIntroduction}

A sequence $\fS{0}$, $\fS{s}(\fS{0})$, $\fS{s}(\fS{s}(\fS{0}))$
of numbers in Peano's notation is usually written as a \emph{term} 
by using a `pairing' (binary) operator $\fS{\seqnat}$ as in
$t_1=(\fS{0}\fS{\seqnat}\fS{s}(\fS{0}))\fS{\seqnat} \fS{s}(\fS{s}(\fS{0}))$ or 
$t_2=\fS{0}\fS{\seqnat}(\fS{s}(\fS{0})\fS{\seqnat}\fS{s}(\fS{s}(\fS{0})))$.
This is necessary when computing with (variants of) 
Term Rewriting Systems (TRSs \cite{BaaNip_TermRewritingAndAllThat_1998}), as function symbols have a fixed arity.
However, \emph{multiple} presentations of the sequence are possible, as in $t_1$ and $t_2$ above.
We can overcome this if $\fS{\seqnat}$ is made \emph{associative}, i.e., the equation $xs\fS{\seqnat}(ys\fS{\seqnat}zs)  =  (xs\fS{\seqnat}ys)\fS{\seqnat}zs$ is satisfied for all terms $xs$, $ys$, and $zs$.
Then, $t_1$ and $t_2$ are made \emph{equivalent modulo associativity} and become 
members of an \emph{equivalence class} $[t]$, consisting of all terms which are equivalent to $t$ modulo associativity. 
Here, $t$ can be $t_1$ or $t_2$, the specific choice being immaterial.

In general, if $\Terms$ is the set of terms built from a signature $\Symbols$
and variables in $\Variables$,
a set of equations $E$ on terms defines an equivalence $=_E$ and a partition $\genrewmoduloOf{\Terms}{\!\!=_E}$ of $\Terms$ into 
equivalence classes. When additionally considering a set of rules $\cR$, it is natural
\begin{figure}[t]
\begin{center}
\begin{tabular}{c}
\xymatrix{
[s]_E  \ar@{->}[dd]_{*}^{\RmoduloEof{\cR}{E}}\ar@{->}[rr]^{*}_{\cR/E} & & [t]_E\ar@{.>}[dd]^{*}_{\cR/E}\\
\\
[t']_E \ar@{.>}[rr]^{*}_{\cR/E} & & [u]_E
}
\end{tabular}
\end{center}
\caption{$E$-confluence}
\label{FigEConfluence}
\end{figure}
 to view rewriting computations as transformations 
$[s]_E\rew{\cR/E}[t]_E$ of 
\emph{equivalence classes}. 
Here, $[s]_E\rew{\cR/E}[t]_E$ (i.e., \emph{rewriting modulo}) 
means that $s'\rew{\cR}t'$ for some $s'\in[s]_E$ and $t'\in[t]_E$.
We often just write $s\rew{\cR/E}t$ instead of $s\equequ{E}s'\rew{\cR}t'\equequ{E}t$.

In this paper we are interested in (proving and disproving) \emph{$E$-confluence} of rewriting modulo,
i.e., 
the commutation of the diagram in Figure \ref{FigEConfluence}, where for all 
equivalence classes $[s]_E$, $[t]_E$, and $[t']_E$ of $E$
such that $[s]_E\rews{\cR/E}[t]_E$ and $[s]_E\rews{\cR/E}[t']_E$ holds, 
there is an equivalence class $[u]_E$ such that 
both $[t]_E\rews{\cR/E}[u]_E$ and $[t']_E\rews{\cR/E}[u]_E$ hold. As usual, 
solid arrows in the diagram represent given reductions and dotted arrows represent \emph{possible} alternatives whose existence must be proved, see, e.g., \cite[page 473]{ChuRos_SomePropertiesOfConversion_TAMS36}.

\subsection{Rewriting with conditional systems}

Huet investigated confluence of rewriting with \emph{Term Rewriting Systems} $\cR$ (\trs{s} \cite{Huet_ConfluentReductionsAbstractPropertiesAndApplicationsToTermRewritingSystems_FOCS77,DerJou_RewriteSystems_HTCS90,BaaNip_TermRewritingAndAllThat_1998,Ohlebusch_AdvancedTopicsInTermRewriting_2002,Terese_TermRewritingSystems_2003}) 
modulo an equivalence defined by means of a set of equations $E$
\cite[Section 3.4]{Huet_ConfluentReductionsAbstractPropertiesAndApplicationsToTermRewritingSystems_JACM80}.
Jouannaud used the term \emph{Equational Term Rewriting System} (\etrs{})
for the combination of a \trs{} $\cR$ and a set of equations $E$
\cite{Jouannaud_ConfluentAndCoherentEquationalTermRewritingSystemsApplicationToProofsInAbstractDataTypes_CAAP83}.
Term Rewriting Systems $\cR$ in \etrs{s} consist of unconditional rules $\lhsr\to\rhsr$ only; 
and the considered sets of equations $E$ consist of unconditional equations $s\equequ{}t$.

In this paper, as in \emph{Generalized Term Rewriting Systems} (\gtrs{s} \cite{Lucas_LocalConfluenceOfConditionalAndGeneralizedTermRewritingSystems_JLAMP24}), we consider \emph{conditional rules}  $\lhsr\to\rhsr\IF\gencond$ 
where $\gencond$ 
is a sequence of \emph{atoms} (that is, expressions $P(t_1,\ldots,t_n)$ where $P$ is a \emph{predicate symbol} and $t_1,\ldots,t_n$ are terms), possibly defined by Horn clauses.
A \gtrs{}  
is a tuple 
$\cR=(\Symbols,\SPredicates,\mu,H,R)$, where 
$\Symbols$ (resp.\ $\SPredicates$) is a signature of \emph{function (resp.\ predicate) symbols},
with 
$\rew{},\rews{}\spcrel{\in}\SPredicates$;
$\mu$ is a \emph{replacement map} 
selecting \emph{active} arguments of function symbols (i.e., they can be rewritten \cite{Lucas_ContextSensitiveRewriting_CSUR20});\footnote{In the following, we mark in \cfrozen{red} the non-active, or \emph{frozen} subterms in rules. Although only the pdf version of the document display them, 
the distinction can still be helpful, avoiding more `invasive' notations such as $\ul{underlining}$ or $\ol{overlining}$.}
 $H$ is a set of definite Horn clauses $A\IF\gencond$, 
 where the predicate symbol of $A$ is not $\rew{}$ or $\rews{}$; and 
$R$ is a set of rewrite rules $\lhsr\to\rhsr\IF\gencond$ such that $\lhsr\notin\Variables$ 
\cite[Definition 51]{Lucas_LocalConfluenceOfConditionalAndGeneralizedTermRewritingSystems_JLAMP24}.
In both cases, $\gencond$ is a sequence of atoms.

Since \emph{conditional equations} $s\ceq t\IF\gencond$
in a set of equations $E$ 
can often be seen as Horn clauses, 
we could consider \gtrs{s} $(\Symbols,\SPredicates,\mu,E\cup H,R)$ including such equational components as Horn clauses.
However,
the first-order theory associated with a \gtrs{} pays no attention to equational components included as Horn clauses.
For instance, equality is a reflexive, symmetric and transitive relation, but the theory associated with a \gtrs{} does not give support to this fact. Thus, in order to properly treat equalities it is convenient to consider $E$ as a \emph{special} subset of Horn clauses to appropriately model equality in the associated first-order theory. Thus, we introduce 
\emph{Equational Generalized Term Rewriting Systems} (\egtrs{s}) 
$\cR=(\Symbols,\SPredicates,\mu,E,H,R)$ where such a modeling is explicit in the theory associated with $\cR$. 

\begin{exa}\label{ExSumListsA}
The following \egtrs{} $\cR$ can be used to compute the sum of a nonempty 
sequence 
$n_1\fS{\seqnat}n_2\fS{\seqnat}\cdots\fS{\seqnat}n_p$ of $p\geq 1$
natural numbers $n_i$ expressed in Peano's notation.
A single number is considered a sequence as well.
A signature 
$\Symbols=\{\fS{0},\fS{s},\fS{\seqnat}\}$ of function symbols  
is used to represent them.
\begin{IEEEeqnarray}{r'C'l}
xs\fS{\seqnat}\cfrozen{(ys\fS{\seqnat}zs)} & = & (xs\fS{\seqnat}\cfrozen{ys})\fS{\seqnat}\cfrozen{zs}\label{ExSumListsA_eq1}\\
\fS{Nat}(\fS{0})\label{ExSumListsA_clause1}\\
\fS{Nat}(\fS{s}(n)) & \IF & \fS{Nat}(n)\label{ExSumListsA_clause2}\\
x\cto y & \IF & x\rews{}y\label{ExSumListsA_clause3}\\
\fS{0}+n & \to & n\label{ExSumListsA_rule1}\\
\fS{s}(m)+n & \to & \fS{s}(m+n)\label{ExSumListsA_rule2}\\
\fS{sum}(\cfrozen{m}) & \to & n\IF m\cto n,\fS{Nat}(n)\label{ExSumListsA_rule3}\\
\fS{sum}(\cfrozen{ms}) & \to & m+n \IF  ms\cto m \fS{\seqnat} \cfrozen{ns},  \fS{Nat}(m), \fS{sum}(\cfrozen{ns})\cto n\label{ExSumListsA_rule4}
\end{IEEEeqnarray}

\noindent
With $\mu(\fS{\seqnat})=\{1\}$, we specify that only the first argument of $\fS{\seqnat}$ can be rewritten.
With $\mu(\fS{sum})=\emptyset$, we forbid rewritings on the argument of $\fS{sum}$.
Predicate $\fS{Nat}$ defined by clauses (\ref{ExSumListsA_clause1}) and (\ref{ExSumListsA_clause2}) identifies an expression 
as representing a \emph{natural number} as $\fS{s}^n(\fS{0})$ for some $n\geq 0$;
clause (\ref{ExSumListsA_clause3}) describes the interpretation of conditions $s\cto t$ as 
\emph{reachability}
in conditional rules like (\ref{ExSumListsA_rule4}).
The application of a rule like (\ref{ExSumListsA_rule4}) 
to a term $\fS{sum}(t)$  is as follows: for each substitution $\sigma$ such that $\sigma(ms)=t$, 
if 
(i) $t$ rewrites (modulo $E$) to $\sigma(m \fS{\seqnat}\cfrozen{ns})$,
(ii) $\fS{Nat}(\sigma(m))$ holds, and 
(iii) 
$\fS{sum}(\cfrozen{\sigma(ns)})$ rewrites to $\sigma(n)$,
then we obtain $\sigma(m)+\sigma(n)$.
\end{exa}
%
\noindent
A number of rewriting-based systems like  
\trs{s},
\cstrs{s} \cite{Lucas_ContextSensitiveRewriting_CSUR20}, 
\ctrs{s}  \cite{Kaplan_ConditionalRewriteRules_TCS84,BerKlo_ConditionalRewriteRulesConfluenceAndTermination_JCSS86},
and \csctrs{s} \cite[Section 8.1]{Lucas_ApplicationsAndExtensionsOfContextSensitiveRewriting_JLAMP21}
can be seen as particular \gtrs{s} \cite[Section 7.3]{Lucas_LocalConfluenceOfConditionalAndGeneralizedTermRewritingSystems_JLAMP24}.
Similarly, \egtrs{s} capture the corresponding equational versions of these systems, like \etrs{s}, 
\emph{Membership Equational Specifications} \cite{Meseguer_MembershipAlgebraAsALogicalFrameworkForEquationalSpecification_WADT97}, 
and \emph{Generalized Rewrite Theories} \cite{BruMes_SemanticFoundationsForGeneralizedRewriteTheories_TCS06}.
As for \gtrs{s}, see, e.g., \cite[Section 8]{Lucas_TerminationOfGeneralizedTermRewritingSystems_FSCD24},
\egtrs{s} are intended to be useful to model and analyze properties of computations with sophisticated languages like \Maude{}
\cite{ClavelEtAl_MaudeBook_2007}. In particular, to investigate confluence of \Maude{} programs.

\subsection{Critical pairs in the analysis of rewriting modulo}

Huet's starting point was a set 
of \emph{equational axioms} (i.e., a set of pairs of terms),
which is partitioned into
a \trs{} $\cR$  
and a set $E$  of \emph{variable preserving} equations (i.e., $\Var(s)=\Var(t)$ for all $s\equequ{}t\in E$, c.f.\ \cite[Definition 3.2.4(4)]{Ohlebusch_AdvancedTopicsInTermRewriting_2002}).
The pair $(\cR,E)$ is called an \emph{equational theory}.
For the analysis of $E$-confluence of 
$\cR$ (although he did not use such a denomination, see Section~\ref{SecHuet80} for details), 
Huet considered 
critical pairs  
(see \cite[page 809]{Huet_ConfluentReductionsAbstractPropertiesAndApplicationsToTermRewritingSystems_JACM80}),
as in his characterization of local confluence of rewriting with \trs{s}.
They are obtained not only from rules in $\cR$, but also from
rules $s\rew{}t$ and $t\rew{}s$ obtained by the two possible \emph{orientations} of equations $s\equequ{}t$ in $E$ \cite[page 817]{Huet_ConfluentReductionsAbstractPropertiesAndApplicationsToTermRewritingSystems_JACM80}.
He proved that equational theories $(\cR,E)$ are confluent modulo $E$ if 
$\cR$ is left-linear, 
all aforementioned critical pairs $\langle u,v\rangle$, are joinable modulo $E$, i.e., $u\rews{\cR}u'\equequ{E}v'\leftrews{\cR}v$ holds,
and 
$\rew{\cR}\circ\equequ{E}$ is terminating
\cite[Theorem 3.3]{Huet_ConfluentReductionsAbstractPropertiesAndApplicationsToTermRewritingSystems_JACM80}.

A pair $(\cR,E)$ where $\cR$ is a \trs{} 
and $E$ a set of equations (no restriction is imposed) 
is called an \emph{Equational \trs} (\etrs) in \cite{Jouannaud_ConfluentAndCoherentEquationalTermRewritingSystemsApplicationToProofsInAbstractDataTypes_CAAP83}.
Jouannaud's result about confluence of $\cR$ modulo $E$ does \emph{not} require 
left-linearity of $\cR$ \cite[Theorem 4]{Jouannaud_ConfluentAndCoherentEquationalTermRewritingSystemsApplicationToProofsInAbstractDataTypes_CAAP83}.
As a counterpart, \emph{E-critical pairs} are used \cite[Definition 10]{Jouannaud_ConfluentAndCoherentEquationalTermRewritingSystemsApplicationToProofsInAbstractDataTypes_CAAP83}.
In contrast to critical pairs,  
$E$-unification, i.e., the existence of 
$\theta$ such that  $\theta(\lhsr|_p)=_E\theta(\lhsr')$ (then $\theta$ is called an \emph{$E$-unifier}), is used
to obtain  
an $E$-critical pair $\langle\theta(\lhsr)[\theta(\rhsr')]_p,\theta(\rhsr)\rangle$.
In sharp contrast to critical pairs, 
(i) there is no general \emph{$E$-unification algorithm} and 
\emph{``for each equational theory one must invent a special algorithm''} \cite[page 74]{Plotkin_BuildingInEquationalTheories_MI72}, which may fail to exist.
Furthermore, even for 
$E$-unifying terms, 
(ii) there can be several, even \emph{infinitely many} $E$-unifiers $\theta$
which must be considered to obtain a \emph{complete set of $E$-critical pairs} which can be
used to check $E$-confluence of $\cR$ 
\cite{BaaSny_UnificationTheory_HAR01,PetSti_CompleteSetsOfReductionsForSomeEquationalTheories_JACM81,Siekmann_UnificationTheory_JSC89}.

\subsection{Plan and contributions of the paper}
The paper is organized as follows.

\subsubsection{Preliminaries}

Some preliminary notions and notations are given in Section~\ref{SecPreliminaries}.
Then, 
Section~\ref{SecPeaksAndJoinsInAbstractReduction} introduces some specific notations and definitions which are useful to deal with the abstract framework by Jouannaud and Kirchner.

\subsubsection{Jouannaud and Kirchner's abstract framework}

Section~\ref{SecAbstractFrameworkJK86} 
introduces Jouannaud and Kirchner's framework \cite[Section 2]{JouKir_CompletionOfASetOfRulesModuloASetOfEquations_SIAMJC86} (extending Huet's one \cite[Section 2]{Huet_ConfluentReductionsAbstractPropertiesAndApplicationsToTermRewritingSystems_JACM80}) as it provides an analysis of confluence of an abstract relation 
$\genrelation$  on a set $A$
modulo an equivalence relation $\genequivalence$ (i.e., $\genequivalence$-confluence of $\genrelation$)
 as the joinability (modulo) of 
\begin{IEEEeqnarray}{r'C'l}
\text{\emph{local confluence} peaks } t\leftrew{\genrelationUpE}s\rew{\genrelation}t'
&\text{and} &
\text{\emph{local coherence} peaks } t\leftrew{\genrelationUpE}s\equone{\genequivalence}t',
\label{LblAbstractLocalConfluenceAndCoherencePeaks}
\end{IEEEeqnarray}
where 
$\genrelationUpE$ is an auxiliary relation including $\genrelation$ and included in 
$\genrelation/\genequivalence$, 
and 
$\equone{\genequivalence}$ is a symmetric relation whose reflexive and transitive closure coincides with 
the equivalence $\genequivalence$.
In this section, we also discuss the use of the abstract framework to \emph{disprove} $\genequivalence$-confluence.

\subsubsection{Equational Generalized Term Rewriting Systems}

Section~\ref{SecEGTRSs} introduces \emph{Equational Generalized Term Rewriting Systems (\egtrs{s})} $\cR=(\Symbols,\SPredicates,\mu,E,H,R)$ 
and shows how to use the aforementioned abstract results to prove confluence modulo of \egtrs{s}.
In particular, we show that a modification of the treatment of the conditional part of clauses in $H$ and $R$ (with respect to what is described in \cite[Section 7.3]{Lucas_LocalConfluenceOfConditionalAndGeneralizedTermRewritingSystems_JLAMP24} for \gtrs{s})
can be used for \egtrs{s} to fit Jouannaud and Kirchner's framework.

\subsubsection{Local confluence and local coherence peaks}

Section~\ref{SecReductionAndCoherencePeaks} shows how local confluence and local coherence peaks (\ref{LblAbstractLocalConfluenceAndCoherencePeaks}) 
look like when the parameter $\genrelationUpE$ in Jouannaud and Kirchner's framework is particularized to either 
\begin{enumerate}
\item\label{LblHuetApproach}
$\rew{\cR}$, the rewrite relation of $\cR$, which fits Huet's approach \cite{Huet_ConfluentReductionsAbstractPropertiesAndApplicationsToTermRewritingSystems_FOCS77,Huet_ConfluentReductionsAbstractPropertiesAndApplicationsToTermRewritingSystems_JACM80};
or to 
\item\label{LblJKApproach} $\rew{\cR,E}$,  
where $E$-matching substitutions, rather than matching substitutions, are used in rewriting steps, i.e., $t|_p\equequ{E}\sigma(\lhsr)$, instead of $t|_p=\sigma(\lhsr)$, is required to apply a rule $\lhsr\to\rhsr$.
This idea was considered by Peterson \& Stickel's \cite{PetSti_CompleteSetsOfReductionsForSomeEquationalTheories_JACM81} to 
obtain their results on $E$-confluence (or $E$-completeness in their terminology).
\end{enumerate}
The different joinability conditions for each kind of peak are also discussed.
If (\ref{LblHuetApproach}) is used, we just call them local confluence and coherence peaks.
If (\ref{LblJKApproach}) is used, we call them local $E$-confluence and $E$-coherence peaks.

\subsubsection{$E$-confluence as joinability of local confluence and coherence peaks}

Section~\ref{SecConditionalPairsAlphaAndGamma} describes the 
\emph{conditional pairs}
we use to obtain \emph{finite representations} of the aforementioned peaks
when $\genrelation^\genequivalence$ is $\rew{\cR}$ (as in Huet's approach).
For this purpose, only \emph{Conditional Critical Pairs} (CCPs) (obtained from the rules in $R$ and also from the combination of rules and equations in $E$) and (different classes of) 
\emph{Conditional Variable Pairs} (CVPs) are necessary.
$E$-unification is \emph{not} used.
Although CCPs are already well-known for their use in proofs of confluence of \ctrs{s} \cite{Kaplan_SimplifyingConditionalTermRewriteSystems_JSC87}
and \gtrs{s};
and CVPs are also known from their use to prove confluence of \gtrs{s} \cite[Section 7.5]{Lucas_LocalConfluenceOfConditionalAndGeneralizedTermRewritingSystems_JLAMP24}, we use them here in novel ways.

Section~\ref{SecEConfluenceWithRandCpeaks} shows that joinability (using $\rews{\cR}$) 
modulo $E$ of the aforementioned pairs
characterizes 
properties $\alpha$ and $\gamma$ of $\rew{\cR}$ for \egtrs{s} $\cR$. 
This provides a sufficient condition for $E$-confluence of \egtrs{s}.
For instance, $E$-confluence of $\cR$ in Example \ref{ExSumListsA} is proved using this result (see Example \ref{ExSumListsA_EConfluent}).
We also provide a sufficient criterion for \emph{disproving} $E$-confluence of 
\egtrs{s} by using CCPs and CVPs.
Table \ref{TableConditionalPairsForPeaksHuet} 
summarizes the use of the considered conditional pairs to prove joinability of local confluence and coherence peaks.

\subsubsection{$E$-confluence as joinability of local $E$-confluence and $E$-coherence peaks}

Section~\ref{SecConditionalPairsForNonDisjointPeaksJK} investigates the 
\emph{conditional pairs}
we use to obtain \emph{finite representations} of the aforementioned peaks
when $\genrelation^\genequivalence$ is $\rew{\cR,E}$.
We introduce \emph{Logic-based Conditional Critical Pairs} (LCCPs) which are intended to provide a \emph{logic-based} definition of (conditional) critical pair which \emph{avoids} the explicit computation of $E$-unifiers.
For this purpose, the $E$-unification condition $\lhsr|_p=^?\lhsr'$ for rules $\lhsr\to\rhsr\IF\gencond$ and $\lhsr'\to\rhsr'\IF\gencond'$ 
from which $E$-unifiers $\sigma$ satisfying $\sigma(\lhsr|_p)\equequ{E}\sigma(\lhsr')$ 
would be obtained is placed in the conditional part of the conditional pair associated with them as an equational condition $\lhsr|_p=\lhsr'$.
For this purpose, only LCCPs (also obtained from the rules in $R$ and 
from the equations in $E$), different classes of CVPs, 
and a new kind of conditional pairs specifically introduced here (we call \emph{Down Conditional Pairs, DCPs}) are used.
Again, $E$-unifiers are \emph{not} necessary.
We prove that their joinability (using $\rews{\cR,E}$)  
modulo $E$ characterizes Jouannaud and Kirchner's  
properties of 
\emph{local confluence of $\rew{\cR,E}$ modulo $E$ with $\rew{\cR}$}
and \emph{local coherence of $\rew{\cR,E}$ modulo $E$}.
This provides another sufficient condition for $E$-confluence of \egtrs{s}, which is given in Section~\ref{SecEConfluenceWithPSRandCpeaks}.
As we show by means of examples, this new criterion is complementary to the aforementioned one.
We also provide a sufficient criterion for \emph{disproving} $E$-confluence of 
\egtrs{s} by using LCCPs, CVPs, and DCPs, which are proved necessary for this purpose by means of an example (Example \ref{ExPeakNoCPs_notEconfluent}).
Table \ref{TableConditionalPairsForEcriticalAndCoherencePeaks}
summarizes the use of the considered conditional pairs  to prove joinability of local $E$-confluence and $E$-coherence peaks.

\subsubsection{Application to rewriting-based systems}
Section~\ref{SecApplicationToETRSs} discusses the application of our results to \etrs{s}.
When applied to \emph{left-linear \etrs{s}}, as a particular case of \egtrs{s},
our  results strictly subsume Huet's result for proving $E$-confluence
(\cite[Theorem 3.3]{Huet_ConfluentReductionsAbstractPropertiesAndApplicationsToTermRewritingSystems_JACM80}).
Regarding Jouannaud and Kirchner's main result for $E$-confluence of \etrs{s}
(\cite[Theorem 16]{JouKir_CompletionOfASetOfRulesModuloASetOfEquations_SIAMJC86}), our results for proving $E$-confluence of \etrs{s} correspond to \cite[Theorem 16]{JouKir_CompletionOfASetOfRulesModuloASetOfEquations_SIAMJC86}, although we do not use the partition of rules into (a subset of) left-linear rules and the remaining ones.
However, our results for \emph{disproving} $E$-confluence of \etrs{s} are new.

Section~\ref{SecApplicationToConditionalETRSs} discusses the application of our results to \ctrs{s} endowed with a set of equations.
In particular, we show that for \egtrs{s} with an empty set of equations, 
our $E$-confluence results boil down into the results  for confluence of \gtrs{s} in \cite{Lucas_LocalConfluenceOfConditionalAndGeneralizedTermRewritingSystems_JLAMP24}.
We compare our results with Bouhoula, Jouannaud, and Meseguer's results on confluence of computations with Membership Equational Specifications \cite{BouJouMes_SpecificationAndProofInMembershipEquationalLogic_TAPSOFT97,BouJouMes_SpecificationAndProofInMembershipEquationalLogic_TCS00} on one side, and also with Dur\'an and Meseguer's results on $E$-confluence of conditional rewrite theories \cite{DurMes_OnTheChurchRosserAndCoherencePropertiesOfConditionalOrderSortedRewriteTheories_JLAP12}, on the other side.
 
 \subsubsection{Related work and conclusions}
Section~\ref{SecRelatedWork} discusses related work.
Section~\ref{SecConclusionAndFutureWork} concludes and points to some future work.

\subsubsection{Improvements}
This paper is an extended and revised version of \cite{Lucas_ConfluenceOfConditionalRewritingModulo_CSL24}.
Main improvements are:
\begin{enumerate}
\item Section~\ref{SecPeaksAndJoinsInAbstractReduction} is new.
\item In Section~\ref{SecAbstractFrameworkJK86}, new results about the role of Jouannaud and Kirchner's local confluence and coherence properties modulo $\genequivalence$ to \emph{disprove} $\genequivalence$-confluence of $\genrelation$ are provided. 
The abstract treatment of confluence of rewriting modulo introduced by Huet
is developed in detail  and compared to Jouannaud and Kirchner's.

\item The definition of \egtrs{} in Section~\ref{SecEGTRSs} differs from \cite{Lucas_ConfluenceOfConditionalRewritingModulo_CSL24} in the 
inclusion of a \emph{replacement map} $\mu$.
From a syntactic point of view,  
\egtrs{s} can be seen now as particular \gtrs{s} where $E$ is treated as a distinguished subset of Horn clauses.
All technical results stemming from \cite{Lucas_ConfluenceOfConditionalRewritingModulo_CSL24} have been revised to accomodate the use of $\mu$.
\item The new approach in Sections \ref{SecConditionalPairsAlphaAndGamma} and \ref{SecEConfluenceWithRandCpeaks} that considers the case when 
$\genrelationUpE$ above is $\rew{\cR}$ is missing in \cite{Lucas_ConfluenceOfConditionalRewritingModulo_CSL24}.
This leads to new proofs of $E$-confluence that only use ``ordinary'' CCPs and different kinds of CVPs.
\item The use of $\rew{\cR,E}$ as $\genrelation^\genequivalence$ is developed in Sections \ref{SecConditionalPairsForNonDisjointPeaksJK} and \ref{SecEConfluenceWithPSRandCpeaks}.
The new notion of LCCP is compared with (conditional) $E$-critical pairs (ECPs).
\item A new Section~\ref{SecApplicationToETRSs} discusses the application of our results to \etrs{s}, showing that we strictly improve the results about $E$-confluence in \cite{Huet_ConfluentReductionsAbstractPropertiesAndApplicationsToTermRewritingSystems_JACM80} and also improve on \cite{JouKir_CompletionOfASetOfRulesModuloASetOfEquations_SIAMJC86}.
\item A new Section~\ref{SecApplicationToConditionalETRSs} discusses the application of our results to \ctrs{s} endowed with a set of equations, as in Membership Equational Specifications or Conditional Rewrite Theories.
\item The related work Section~\ref{SecRelatedWork} has been revised and extended to include additional work by researchers not considered in \cite{Lucas_ConfluenceOfConditionalRewritingModulo_CSL24}.
\end{enumerate}
Besides,
\begin{enumerate}
\item The terminology has been revised to closely follow the existing one in \cite{JouKir_CompletionOfASetOfRulesModuloASetOfEquations_SIAMJC86,DerOkaSiv_ConfluenceOfConditionalRewriteSystems_CTRS87}.
\item New diagrams and tables illustrating the structure and use of different kinds of confluence and coherence peaks are provided.
\item Several examples have been introduced to illustrate the notions and results of the paper.
\end{enumerate}

\section{Preliminaries}\label{SecPreliminaries}

In the following, \emph{s.t.} means \emph{such that} 
and \emph{iff} means \emph{if and only if}.
As in \cite[Section 10.1.2]{DraSin_IntermediateSetTheory_1996}, we use $A-B$, rather than $A\backslash B$ to denote the difference of sets $A$ and $B$, i.e., the set of elements of $A$ except those which also are in $B$ (if any).
We assume some familiarity with the basic notions of term rewriting \cite{BaaNip_TermRewritingAndAllThat_1998,Ohlebusch_AdvancedTopicsInTermRewriting_2002}
and first-order logic \cite{Fitting_FirstOrderLogicAndAutomatedTheoremProving_1997,Mendelson_IntroductionToMathematicalLogicFourtEd_1997}. 
For the sake of readability, though, here we summarize the main notions and notations we use.

\paragraph{Abstract Reduction Relations}
Given a binary relation $\genrelation\:\subseteq A\times A$ on a set $A$, 
we often write $a\spcrel{\genrelation} b$ or $b\spcrel{\invgenrelation}a$
instead of $(a,b)\in\genrelation$ (and call $\invgenrelation$ the \emph{inverse} relation of $\genrelation$).
The \emph{composition} 
of two relations $\genrelation$ and $\genrelation'$ is written $\composeRel{\genrelation}{\genrelation'}$ and defined as follows: for all $a,b\in A$, $a\:\composeRel{\genrelation}{\genrelation'}\:b$ iff there is $c\in A$ such that 
$a\:\genrelation\:c$ and
$c\:\genrelation'\:b$.
The \emph{reflexive} closure of $\genrelation$ is denoted by $\genrelation^=$;
the \emph{transitive} closure of $\genrelation$ is denoted by
$\genrelation^+$; 
and the \emph{reflexive and transitive} closure by $\genrelation^*$. 
An element $a\in A$ is $\genrelation$-\emph{irreducible} 
or an \emph{$\genrelation$-normal form}, 
if there is no $b$ such that $a\spcrel{\genrelation} b$;
we often drop ``$\genrelation$-'' if no confusion arises.
We say that $b$ is an $\genrelation$-normal form of $a$ (written $a\spcrel{\genrelation^!}b$), 
if  $a\spcrel{\genrelation^*}b$ and $b$ is an $\genrelation$-normal form; 
we also say that $a$ is $\genrelation$-normalizing, i.e., $a$ has an $\genrelation$-normal form.
Furthermore, $\genrelation$ is \emph{normalizing} if every $a\in A$ is $\genrelation$-normalizing.
We say that $b\in A$ is $\genrelation$-reachable from $a\in A$ if $a\genrelation^*b$.
We say that $a,b\in\:\genrelation$ are \emph{$\genrelation$-joinable}
if there is $c\in A$ such that  $a~\genrelation^*c$ and $b~\genrelation^*c$.
Also, $a,b\in\genrelation$ are  \emph{$\genrelation$-convertible}
if $a\:(\genrelation\cup\genrelation^{-1})^*\:b$.
Given $a\in A$, if there is no infinite sequence $a=a_1~\genrelation~a_2~\genrelation~\cdots~\genrelation~a_n~\genrelation\cdots$, then $a$ is $\genrelation$-\emph{terminating}; 
$\genrelation$ is \emph{terminating} if $a$ is $\genrelation$-terminating for all $a\in A$.
We say that $\genrelation$ is (locally) {\em confluent} if, for all $a,b,c\in A$, if $a~\genrelation^*b$ and $a~\genrelation^*c$ (resp.\ $a~\genrelation\: b$ and
$a~\genrelation\: c$), then $b$ and $c$ are $\genrelation$-joinable.

\paragraph{Multiset ordering}
A multiset $M$ is a collection in which elements (taken from a set $A$) 
are allowed to occur \emph{more than once} and the specific sequence of such occurrences does not matter, see \cite{DerMan_ProvingTerminationWithMultisetOrderings_ICALP79} and also \cite[Definition 2.3.8]{Ohlebusch_AdvancedTopicsInTermRewriting_2002} for a formal definition. 
Given multisets $M$ and $M'$ of elements of $A$, we write $M>M'$ if $M'$ is obtained from $M$ by
\emph{removing} elements from $M$ or by 
\emph{replacing} elements $a\in M$ by finitely many other elements
$a_1,\ldots\in A$ such that $a>a_i$ for all $i\geq 1$
\cite[Section II]{DerMan_ProvingTerminationWithMultisetOrderings_ICALP79}
(see also 
\cite[Definition 2.3.10]{Ohlebusch_AdvancedTopicsInTermRewriting_2002}).

\paragraph{Signatures, Terms, Positions}
In this paper, $\Variables$ denotes a
countable set of \emph{variables}.
A \emph{signature of symbols} is a set of \emph{symbols}
each with a fixed \emph{arity}. 
We often write $\arityOf{f}{k}$ to make explicit that $f$ is a $k$-ary symbol as often done in logic programming.
When dealing with symbols not using the standard prefix notation $f(\pholder,\ldots,\pholder)$, we often use placeholders to make explicit where the arguments occur, as in $\pholder+\pholder$ for the \emph{infix} addition operator or $\pholder!$ for the factorial \emph{postfix} operator.
We use $\Symbols$ to denote
a \emph{signature of function symbols}, i.e., $\{f, g, \ldots \}$
whose arity is given by a
mapping $ar:\Symbols\rightarrow \mathbb{N}$.
 The set of terms built from $\Symbols$ and $\Variables$ is $\Terms$.
The set of variables occurring in $t$ is $\Var(t)$.
Terms are viewed as labeled trees in the usual way.
\emph{Positions} $p$
are represented by chains of positive natural numbers used to address subterms $t|_p$
of $t$.
The \emph{set of positions} of a term $t$ is $\Pos(t)$.
The set of positions of a subterm $s$ in $t$ is denoted $\Pos_s(t)$.
The set of positions of non-variable symbols in $t$ are denoted as $\Pos_\Symbols(t)$.
Positions are ordered by the \emph{prefix ordering} $\leq$ on sequences: 
given positions $p,q$, we write $p\leq q$ iff $p$ is a prefix of $q$. 
If $p\not\leq q$ and $q\not\leq p$, we say that $p$ and $q$ are \emph{disjoint}
(written $p\parallel q$).
We write $p\cpos{}q$ to denote the \emph{concatenation} of positions.

\paragraph{Replacement maps}
Given a signature $\Symbols$, 
a \emph{replacement map} is a mapping 
$\mu$ 
satisfying that, for all  symbols $f$ in $\Symbols$, $\mu(f)\subseteq \{1,\ldots,ar(f)\}$
 \cite{Lucas_ContextSensitiveRewriting_CSUR20}.
The set of  replacement maps for the signature $\Symbols$ is $\Rmaps{\Symbols}$. 
Extreme cases 
are $\muBot$, 
 disallowing replacements in all arguments of function symbols: 
$\muBot(f)=\emptyset$ for all $f\in\Symbols$,
and $\muTop$, restricting no replacement: $\muTop(f)=\{1,\ldots,k\}$ for all $k$-ary $f\in\Symbols$.
The 
set $\Pos^\mu(t)$ of {\em $\mu$-replacing (or \emph{active}) positions}  of $t$ is
$\Pos^\mu(t)=\{\toppos\}$, if $t\in\Variables$, and
$\Pos^\mu(t)=\{\toppos\}\cup\{i.p\mid i\in\mu(f), p\in\Pos^\mu(t_i)\}$, if 
$t=f(t_1,\ldots,t_k)$.
Positions of \emph{active} non-variable symbols in $t$ are denoted as $\Pos^\mu_\Symbols(t)$.

\paragraph{First-Order Logic}
Here, $\SPredicates$ denotes a signature of \emph{predicate symbols}. 
First-order formulas  
are built using 
function symbols from $\Symbols$, 
predicate symbols  from $\SPredicates$,
and variables from $\Variables$ in the usual way.
In particular, \emph{atomic formulas} $A$ (often called \emph{atoms} 
in the realm of automated theorem proving 
\cite[page 2]{Robinson_AReviewOfAutomaticTheoremProving_MACS67}),
are expressions 
$P(t_1,\ldots,t_n)$ where $P\in\SPredicates$ and $t_1,\ldots,t_n$ are terms;
we often refer to $P$ as $\rootTerm(A)$.
A first-order theory (FO-theory for short) $\GLtheory$ is a set of \emph{sentences}
(formulas whose variables are all \emph{quantified}).
In the following, given an FO-theory 
$\GLtheory$ and a formula $\varphi$, 
$\GLtheory\vdash\varphi$ means that $\varphi$ is \emph{deducible} from 
(or a \emph{logical consequence} of) $\GLtheory$ by using 
a correct and complete deduction procedure
\cite{Fitting_FirstOrderLogicAndAutomatedTheoremProving_1997,Mendelson_IntroductionToMathematicalLogicFourtEd_1997}.

\paragraph{Horn clauses}

A \emph{literal} is an atom or the negation of an atom.
A \emph{clause} is a disjunction of literals. 
A \emph{set of clauses} $\cC$ is regarded as a conjunction of all clauses in $\cC$, where every variable in
$\cC$ is \emph{universally quantified} 
\cite{ChaLee_SymbolicLogicAndMechanicalTheoremProving_1973}. 
For every sentence $\varphi$ 
there is a sentence $\varphi'$ in \emph{clausal form}
which is inconsistent iff
$\varphi$ is \cite[Section 4.2]{ChaLee_SymbolicLogicAndMechanicalTheoremProving_1973}.
A \emph{definite} Horn clause  (with label $\alpha$) is written $\alpha:A\IF A_1,\ldots,A_n$, 
for atoms $A,A_1,\ldots,A_n$;
if $n=0$, then $\alpha$ is written $A$ rather than $A\IF$;
we often say that $\rootTerm(A)$ is the predicate \emph{defined} by $\alpha$.
A theory all whose sentences are definite Horn clauses is often called a \emph{Horn theory}.

\paragraph{Feasibility Sequences.}

A sequence $A_1,\ldots,A_n$ of atoms $A_i$, $1\leq i\leq n$ is
$\GLtheory$-\emph{feasible} 
with respect to a
theory $\GLtheory$ (or just \emph{feasible} if no confusion arises), if there is a substitution $\sigma$ such that 
$\GLtheory\vdash\sigma(A_i)$
holds for all $1\leq i\leq n$;
otherwise, it is \emph{infeasible} \cite{GutLuc_AutomaticallyProvingAndDisprovingFeasibilityConditions_IJCAR20}.\footnote{Note 
the difference with 
\emph{satisfiability} of a 
formula 
$\varphi$, 
i.e., the \emph{existence} of 
(i) an \emph{interpretation} $\SInterpretation$ and 
(ii) a valuation of variables $\alpha$ satisfying $\varphi$, i.e., the interpretation $[\varphi]^\SInterpretation_\alpha$ of $\varphi$ using $\SInterpretation$ and $\alpha$ 
is true \cite[p.\ 65]{Mendelson_IntroductionToMathematicalLogicFourtEd_1997}.}

\paragraph{Conditional pairs}
\label{SecConditionalPairs}

As in \cite[Section 5]{Lucas_LocalConfluenceOfConditionalAndGeneralizedTermRewritingSystems_JLAMP24}, in the following we use 
\emph{conditional pairs} 
$\langle s,t\rangle\IF  A_1,\ldots,A_n$, 
where $s$ and $t$ are terms and $A_1,\ldots,A_n$ are atoms.
A conditional pair  
is \emph{trivial} if $s=t$.
Conditional pairs $\langle s,t\rangle\IF\gencond$ and $\langle t,s\rangle\IF\gencond$, whose pairs are symmetric but share the same conditional part (possibly after reordering the atoms) are called \emph{symmetric}.

\section{Peaks and joins in abstract reductions}
\label{SecPeaksAndJoinsInAbstractReduction}

This section introduces some new notations to represent and then compare different kinds of local peaks and joins used in the following sections of the paper.

\begin{notation}[Local peaks and joins]
\label{NotationPeaksAndJoins}
Given relations $\genrelation$ and $\genrelationbis$ on $A$,
\begin{itemize}
\item $\peakOf{\genrelation}{\genrelationbis}$ denotes
the family of \emph{local peaks} 
$t\spcrel{\invgenrelation{}}s\spcrel{\genrelationbis{}}t'$ for some $s,t,t'\in A$, which we call a local $\peakOf{\genrelation}{\genrelationbis}$-peak.
\item 
$\joinOf{\genrelation}{\genrelationbis}$ denotes the family of \emph{joins} 
$t\spcrel{\genrelation{}}u\spcrel{\invgenrelationbis{}}t'$  for some $t,t',u\in A$ (an $\joinOf{\genrelation}{\genrelationbis}$-join).
\end{itemize}
The use of \emph{double} up and down  arrows ($\Uparrow$ and $\Downarrow$, respectively) reinforces the idea that $\genrelation$ and $\genrelationbis$ can be \emph{different}. Besides, we borrow Huet's notation  \cite[pages 802 \& 817]{Huet_ConfluentReductionsAbstractPropertiesAndApplicationsToTermRewritingSystems_JACM80} for joinability modulo\footnote{In 
\cite[Definition 2.5.1]{Ohlebusch_AdvancedTopicsInTermRewriting_2002}, the relation $\downarrow_{\sim}$ is defined as $\rews{}\circ\sim\circ\leftrews{}$ and called \emph{joinability modulo $\sim$} (another notation for the same relation is
$\heartsuit$ 
\cite[Definition 14.3.3(i)]{Terese_TermRewritingSystems_2003}). Although the use of $\sim$ below the arrow is closer to its use in the definition, the Huet-based notation is more symmetric and compact. 
}
\begin{itemize}
\item $\ejoinOf{\genrelation}{\genrelationbis}$ denotes the family of \emph{joins} 
$t\spcrel{\genrelation{}}u\equ{}u'\spcrel{\invgenrelationbis{}}t'$ 
\emph{modulo an (implicit) equivalence $\equ{}$}, 
for some $t,t',u,u'\in A$ (an $\ejoinOf{\genrelation}{\genrelationbis}$-join).
\end{itemize}
\end{notation}
\begin{rem}[Local peaks]
In this paper we are  essentially concerned with \emph{local peaks}. In general, arbitrary peaks 
$t\spcrel{(\genrelation^{-1})^*}s\spcrel{\genrelation^*}t'$ of relations $\genrelation$ are not explicitly considered.
Thus, we often talk of ``peaks'' referring to \emph{local peaks} if no confusion arises.
\end{rem}
%
\noindent
Due to their frequent use, Table \ref{TableCompactJoinabilityNotations} collects a number of simpler, more compact notations (some of them already used in the literature). 
\begin{table}[ht]
\caption{Compact notations for joinability}
\label{TableCompactJoinabilityNotations}
\centering
\begin{tabular}{@{}l*{8}{c}@{}}
\toprule
Detailed notation:
& $\joinOf{\genrelation^*}{\genrelation^*}$
& $\joinOf{\genrelation^+}{\genrelation^*}$
& $\joinOf{\genrelation^*}{\genrelation^+}$
& $\joinOf{\genrelation^+}{\genrelation^+}$
& $\ejoinOf{\genrelation^*}{\genrelation^*}$
& $\ejoinOf{\genrelation^+}{\genrelation^*}$
& $\ejoinOf{\genrelation^*}{\genrelation^+}$
& $\ejoinOf{\genrelation^+}{\genrelation^+}$
\\
\midrule
Specific notation:
& $\joinability{\genrelation}$
& $\lsjoinability{\genrelation}$
& $\rsjoinability{\genrelation}$
& $\sjoinability{\genrelation}$
& $\ejoinability{\genrelation}$
& $\lsejoinability{\genrelation}$
& $\rsejoinability{\genrelation}$
& $\sejoinability{\genrelation}$
\\
\bottomrule
\end{tabular}\\
where `\emph{ls}', `\emph{rs}', and `\emph{s}' stand
for \emph{left-strict}, \emph{right-strict}, and \emph{strict} (joinability).
\end{table}

\begin{defi}
A pair $\langle s,t\rangle$ is 
\begin{itemize}
\item $\joinOf{\genrelation}{\genrelationbis}$-joinable (or just
$\joinability{\genrelation}$-joinable if $\genrelation=\genrelationbis$) if 
$s\spcrel{\genrelation{}}u\spcrel{\invgenrelationbis{}}t$ is in 
$\joinOf{\genrelation}{\genrelationbis}$ for some $u\in A$.
\smallskip
\item $\ejoinOf{\genrelation}{\genrelationbis}$-joinable (or just
$\ejoinability{\genrelation}$-joinable if $\genrelation=\genrelationbis$)
if 
$s\spcrel{\genrelation{}}u\equ{}u'\spcrel{\invgenrelationbis{}}t$
is in $\ejoinOf{\genrelation}{\genrelationbis}$
for some $u,u'\in A$.
\end{itemize}
\end{defi}

\section{Jouannaud and Kirchner's abstract approach to confluence of $\genrewmoduloOf{\genrelation}{\genequivalence}$}
\label{SecAbstractFrameworkJK86}

Sethi and his collaborators pioneered the abstract analysis of reduction modulo an equivalence \cite{AhoSetUll_CodeOptimizationAndFiniteChurchRosserTheorems_DOC72,Sethi_TestingForTheChurchRosserProperty_JACM74}.
Huet improved their work in \cite{Huet_ConfluentReductionsAbstractPropertiesAndApplicationsToTermRewritingSystems_JACM80} and also showed how to apply the abstract notions and results to term rewriting systems.
Jouannaud \cite{Jouannaud_ConfluentAndCoherentEquationalTermRewritingSystemsApplicationToProofsInAbstractDataTypes_CAAP83} further developed Huet's abstract analysis of confluence modulo, by providing an abstract treatment of Peterson \& Stickel's work  on practical approaches to rewriting modulo in the realm of rewriting systems \cite{PetSti_CompleteSetsOfReductionsForSomeEquationalTheories_JACM81}.
Improving on \cite{Jouannaud_ConfluentAndCoherentEquationalTermRewritingSystemsApplicationToProofsInAbstractDataTypes_CAAP83,JouKir_CompletionOfASetOfRulesModuloASetOfEquations_POPL84}, in \cite{JouKir_CompletionOfASetOfRulesModuloASetOfEquations_SIAMJC86} Jouannaud and Kirchner provided the technical presentation which we follow here 
(see also \cite{Jouannaud_ConfluenceOfTerminatingRewritingComputations_TFSP24}).

\subsection{Basic notions}
\label{SecBasicNotionsJK86}

Let $\equone{\genequivalence}$ be a \emph{symmetric} 
relation on a set $A$ and 
$\equ{\genequivalence}\spcrel{=}\equoneStar{\genequivalence}$ be its reflexive and transitive closure.
$\equ{\genequivalence}$ is an \emph{equivalence} relation,   
often called \emph{$\genequivalence$-equality}.
Let $\rew{\genrelation}$  ($\genrelation$ for short) be a binary relation on $A$.
Given $\genrelation$ and $\genequivalence$, the relation $\rew{\genrewmoduloOf{\genrelation}{\genequivalence}}$ 
($\genrewmoduloOf{\genrelation}{\genequivalence}$ for short), is called \emph{reduction (with $\rew{\genrelation}$) modulo $\equ{\genequivalence}$} and 
\emph{defined} as
\begin{IEEEeqnarray}{r'C'l}
\rew{\genrewmoduloOf{\genrelation}{\genequivalence}} & = & \composeRel{\equ{\genequivalence}}{\composeRel{\rew{\genrelation}}{\equ{\genequivalence}}}
\label{LblReductionModuloFromReductionAndEquivalence}
\end{IEEEeqnarray}
Let $\oneCR$ be the \emph{symmetric} relation
\begin{IEEEeqnarray}{r'C'l}
\oneCR & = & \equone{\genequivalence}\cup\rew{\genrelation}\cup\leftrew{\genrelation}
\end{IEEEeqnarray}
and $\oneCRs$ be the reflexive and transitive closure of $\oneCR$.
Note that $\oneCRs$ is an equivalence relation.

Computing with $\genrewmoduloOf{\genrelation}{\genequivalence}$ is difficult as it may
involve searching inside an infinite $\genequivalence$-equivalence class 
$[t]_{\genequivalence}$ for some $t'$ on which a reduction step  with $\genrelation$ can be performed.
Peterson and Stickel investigated this problem for 
TRSs $\cR$ and equational theories $E$. 
Thus, an important aspect of Jouannaud and Kirchner's framework is the use of $\rew{\genrelationUpE}$ ($\genrelationUpE$ for short) as a \emph{parameter} to define and further investigate confluence modulo.
Such a relation, though, must satisfy the following \emph{fundamental assumption} \cite[page 1158]{JouKir_CompletionOfASetOfRulesModuloASetOfEquations_SIAMJC86}:
\begin{IEEEeqnarray}{r'C'l}
\genrelation\:\subseteq\:\genrelationUpE\:\subseteq\:\genrewmoduloOf{\genrelation}{\genequivalence}\label{LblFundamentalAssumptionJK86}
\end{IEEEeqnarray}
Thus, $\genrelationUpE$ can be ``instantiated'' in different ways.
Due to the extensive use of Peterson and Stickel's reduction for \etrs{s}
in Jouannaud and Kirchner's framework and its application, 
we briefly introduce it in the next section.

\subsection{Peterson \& Stickel reduction}

Given an \etrs{} $\cR=(\Symbols,E,R)$ over a signature $\Symbols$, with set of equations $E$ and set of rules $R$, we consider
\begin{enumerate}
\item rewritings $s\rew{\cR}t$, where a term $s$ is rewritten into $t$
if  $s|_p=\sigma(\lhsr)$ for some position $p$, 
rule $\lhsr\to\rhsr$,
and substitution $\sigma$ such that $t=s[\sigma(\rhsr)]_p$;
\item rewritings modulo $E$ $s\rew{\cR/E}t$, where
$s\equequ{E}s'$ for some term $s'$ such that $s'\rew{\cR}t'$ and $t'\equequ{E}t$;
and 
\item\label{LblPetersonAndStickelReduction} rewritings \emph{\`a la Peterson \& Stickel} $s\rew{\cR,E}t$,
where $s|_p\equequ{E}\sigma(\lhsr)$ for some position $p$, 
rule $\lhsr\to\rhsr$ ,
and substitution $\sigma$ such that $t=s[\sigma(\rhsr)]_p$.
\end{enumerate}
\begin{rem}
The use of $s\rew{\cR,E}t$ to denote a reduction step as in (\ref{LblPetersonAndStickelReduction}) above
can be found in \cite[Definition 9]{Jouannaud_ConfluentAndCoherentEquationalTermRewritingSystemsApplicationToProofsInAbstractDataTypes_CAAP83}.
Although the attribution to 
\cite{PetSti_CompleteSetsOfReductionsForSomeEquationalTheories_JACM81} is usual, 
no formal definition of $\rew{\cR,E}$ --as fixed by Jouannaud-- is found there,
see Sections \ref{SecRelatedWork_PS81} and \ref{SecRelatedWork_JouannaudAndKirchner} for a discussion. 
Still, in this paper we use ``\emph{Peterson \& Stickel reduction}'' to refer to $\rew{\cR,E}$.
\end{rem}
It is clear that these reduction relations satisfy the fundamental assumption
(\ref{LblFundamentalAssumptionJK86}), i.e.,
\begin{IEEEeqnarray*}{r'C'l}
\rew{\cR}\:\subseteq\:\rew{\cR,E}\:\subseteq\:\rew{\cR/E}
\end{IEEEeqnarray*}
In general, these inclusions are \emph{strict}.
\begin{exa}
\label{ExPeakNoCPs}
\label{ExPeakNoCPs_ETRS}
Consider  
$E=\{(\ref{ExPeakNoCPs_eq1}),(\ref{ExPeakNoCPs_eq2})\}$
and 
$R=\{(\ref{ExPeakNoCPs_rule1}),
(\ref{ExPeakNoCPs_rule2})\}$, where
\begin{IEEEeqnarray}{+r:C:l+}
\fS{b} & =  & \fS{f}(\fS{a})\label{ExPeakNoCPs_eq1}\\
\fS{a} & =  & \fS{c}\label{ExPeakNoCPs_eq2}\\
\fS{c} & \to & \fS{d}\label{ExPeakNoCPs_rule1}\\
\fS{b} & \to & \fS{d}\label{ExPeakNoCPs_rule2}
\end{IEEEeqnarray}
Then, 
\begin{itemize}
\item $\fS{f}(\fS{a})$ is $\rew{\cR}$-irreducible. However, 
$\fS{f}(\fS{a})\rew{\cR,E}\fS{f}(\fS{d})$ 
because 
$\fS{a}\equequ{E}\ul{\fS{c}}\rew{\cR}\fS{d}$, i.e., $\rew{\cR}\:\subset\:\rew{\cR,E}$.
\item Also, $\fS{b}\rew{\cR/E}\fS{f}(\fS{d})$, as $\fS{b}\equequ{E}\fS{f}(\fS{a})\equequ{E}\fS{f}(\ul{\fS{c}})\rew{\cR}\fS{f}(\fS{d})$, but $\fS{b}\not\rew{\cR,E}\fS{f}(\fS{d})$, i.e., $\rew{\cR,E}\:\subset\:\rew{\cR/E}$.
\end{itemize}
\end{exa}
%

\noindent
The use of $\rew{\cR,E}$ in rewriting modulo a set of equations is usually motivated as follows \cite[page 819, paragraph 5]{DurMes_OnTheChurchRosserAndCoherencePropertiesOfConditionalOrderSortedRewriteTheories_JLAP12}: checking whether
\begin{quote}
$s\rew{\cR/E}t$ involves searching through the possibly infinite equivalence class $[s]_E$ to see whether a match\footnote{The original version speaks of ``an E-match''. In our setting, we prefer ``match'' as it fits better the usual definition of rewriting modulo.}
[\emph{against the left-hand side of a rule}] 
is
found for a subterm of some $s'\in[s]_E$ and the result of rewriting $s'$ belongs to the equivalence class $[t]_E$. For this reason,
a much simpler relation $\rew{\cR,E}$ is defined, which becomes decidable if an $E$-matching algorithm exists.
\end{quote}
However, as suggested by the sentence ``if an $E$-matching algorithm exists'', $E$-matching is known to be difficult, see, e.g.,
\cite{FagHue_CompleteSetsOfUnifiersAndMatchersInEquationalTheories_CAAP83,BenKapNar_ComplexityOfMatchingProblems_RTA85,FagHue_CompleteSetsOfUnifiersAndMatchersInEquationalTheories_TCS86,KapNar_NPcompletenessOfTheSetUnificationAndMatchingProblems_CADE86,KapNar_MatchingUnificationAndComplexity_SIGSAMB87,BenKapNar_ComplexityOfMatchingProblems_JSC87,Siekmann_UnificationTheory_JSC89}.

\subsection{Abstract Church-Rosser and confluence properties}
 
The main abstract properties at stake in the analysis of confluence of reduction modulo are introduced in the following.
 
\begin{defi}
\label{DefConfluenceAndTerminationModulo}
Let $\genrelation$ and $\genequivalence$ be as above. Then, following
\cite[Definition 1]{JouKir_CompletionOfASetOfRulesModuloASetOfEquations_SIAMJC86},
\begin{figure}[t]
\centering
\begin{tabular}{@{}c@{\hspace{1.6cm}}c@{}}
\xymatrix@R=2pc{
&& \phantom{s} \ar@{}[lld]\ar@{}[rrd]
\\
t \ar@{.>}[d]^{\genrelation/\genequivalence}_{*}\ar@{|=|}[rrrr]^{*} &&&& t'\ar@{.>}[d]^{*}_{\genrelation/\genequivalence}
\\
u\ar@{|.|}[rrrr]^{*}_{\genequivalence} &&&& u'
}
&
\xymatrix@R=2pc{
&& s \ar@{->}[lld]^{\genrelation/\genequivalence}_{*} \ar@{->}[rrd]^{*}_{\genrelation/\genequivalence}
\\
t \ar@{.>}[d]_{*}^{\genrelation/\genequivalence} &&&& t'\ar@{.>}[d]^{*}_{\genrelation/\genequivalence}
\\
u \ar@{|.|}[rrrr]^{*}_{\genequivalence} &&&& u'
}
\end{tabular}
\caption{$\genrelation$ Church-Rosser modulo $\genequivalence$ and $\genrelation$ confluent modulo $\genequivalence$
in \cite[Definition 1]{JouKir_CompletionOfASetOfRulesModuloASetOfEquations_SIAMJC86}}
\label{FigChurchRosserAndEConfluenceJK86}
\end{figure}
\begin{itemize}
\item $\genrelation$ is \emph{Church-Rosser modulo $\genequivalence$}
iff for all 
$t,t'\in A$, if 
$t\oneCRs t'$ 
then  there are $u$ and $u'$ such that
$t\rews{\genrewmoduloOf{\genrelation}{\genequivalence}}u$,
$t'\rews{\genrewmoduloOf{\genrelation}{\genequivalence}}u'$ 
and $u\equ{\genequivalence} u'$,\footnote{\label{FootNoteConfluenceOfAbstractObjectsVsConfluenceInEquivalenceClasses} Definition 1
in \cite{JouKir_CompletionOfASetOfRulesModuloASetOfEquations_SIAMJC86} does \emph{not} use
the last requirement $u\equ{\genequivalence} u'$ 
as the authors assume $u$ and $u'$ to be 
\emph{$\genequivalence$-equivalence classes}
on $A$ (i.e., $u,u'\in A/\!\!\equ{\genequivalence}$) rather than $u,u'\in A$.
In order to make the difference explicit, consider 
$A=\{\fS{a},\fS{b},\fS{c}\}$, 
$\genequivalence$ be given by 
(the reflexive, transitive, and symmetric closure of) $\fS{b}\equ{\genequivalence}\fS{c}$, and $\genrelation$ be given by
$\fS{a}\genrelation \fS{b}$ and $\fS{a}\genrelation \fS{c}$.
Then, $\fS{b}\oneCRs\fS{c}$,
but $\fS{b}$ and $\fS{c}$ are $\rew{\genrelation/\genequivalence}$-irreducible.
And $\rews{\cR/E}=\{(\fS{a},\fS{a})(\fS{b},\fS{b}),(\fS{c},\fS{c}),(\fS{a},\fS{b}),(\fS{a},\fS{c})\}$. Thus, neither $\fS{b}\rews{\cR/E}\fS{c}$ nor $\fS{c}\rews{\cR/E}\fS{b}$, i.e., as a relation on $A$, $\rew{\genrelation/\genequivalence}$ is \emph{not} Church-Rosser.
However, $\fS{b}\equ{\genequivalence}\fS{c}$. 
As a relation on $A/\!\!\equ{\genequivalence}$, 
$\rew{\genrelation/\genequivalence}$ is confluent.}
see Figure \ref{FigChurchRosserAndEConfluenceJK86} (left).
\item $\genrelation$ is \emph{confluent modulo $\genequivalence$}
(or  \emph{$\genequivalence$-confluent}) iff for all 
$s,t,t'\in A$, if 
$s\rews{\genrewmoduloOf{\genrelation}{\genequivalence}}t$ 
and
$s\rews{\genrewmoduloOf{\genrelation}{\genequivalence}}t'$,
then  there are $u$ and $u'$ such that
$t\rews{\genrewmoduloOf{\genrelation}{\genequivalence}}u$,
$t'\rews{\genrewmoduloOf{\genrelation}{\genequivalence}}u'$ 
and $u\equ{\genequivalence} u'$, see footnote \ref{FootNoteConfluenceOfAbstractObjectsVsConfluenceInEquivalenceClasses}
and  Figure \ref{FigChurchRosserAndEConfluenceJK86} (right).
\item $\genrelation$ is \emph{terminating modulo $\genequivalence$} (or \emph{$\genequivalence$-terminating}, $\fS{SN}_\genequivalence(\genrelation)$\footnote{With minor variations, we borrow these acronyms from \cite[Section 2.5]{Ohlebusch_AdvancedTopicsInTermRewriting_2002} and often extend them to capture new properties not considered there.}) if
$\rew{\genrewmoduloOf{\genrelation}{\genequivalence}}$ is terminating \cite[
p.\ 1158]{JouKir_CompletionOfASetOfRulesModuloASetOfEquations_SIAMJC86}.
\end{itemize}
\end{defi}
\begin{figure}[t]
\begin{center}
\begin{tabular}{c}
\xymatrix{
t  \ar@{.>}[dd]^{\genrelationUpE}_{*}\ar@{|=|}[rr]^{*} & & t'\ar@{.>}[dd]^{*}_{\genrelationUpE}
\\
\\
u\ar@{|.|}[rr]^{*}_{\genequivalence} & &  u'
}
\end{tabular}
\end{center}
\caption{$\genrelation$ is $\genrelationUpE$-Church-Rosser modulo $\genequivalence$ ($1^{st}$ diagram in \cite[Figure 2.1]{JouKir_CompletionOfASetOfRulesModuloASetOfEquations_SIAMJC86})}
\label{FigChurchRosserProperty}
\end{figure}
%
\noindent
The following  is a \emph{parametric} definition of \emph{joinability} modulo using  
$\genrelationUpE$.

\begin{defi}[Joinability modulo]
\label{DefJoinabilityModulo_Def2_JK86}
A pair $\langle t,t'\rangle$ with $t,t'\in A$ is 
$\genrelationUpE$-joinable modulo $\genequivalence$, 
(denoted as $t\ejoinability{\genrelationUpE}t'$), if $\exists\:u,u'\in A$ s.t.\ $t\rews{\genrelationUpE}u$, 
$t'\rews{\genrelationUpE}u'$, and
$u\equ{\genequivalence} u'$ \cite[Definition 2]{JouKir_CompletionOfASetOfRulesModuloASetOfEquations_SIAMJC86}.
\end{defi}
%
\noindent
Accordingly, the following definition is given in \cite{JouKir_CompletionOfASetOfRulesModuloASetOfEquations_SIAMJC86}.

\begin{defi}[$\genrelationUpE$-Church-Rosser modulo 
$\genequivalence$]
\label{DefChurchRosserRE_JK86}
According to \cite[Definition 3]{JouKir_CompletionOfASetOfRulesModuloASetOfEquations_SIAMJC86},
$\genrelation$ is \emph{$\genrelationUpE$-Church-Rosser modulo 
$\genequivalence$} ($\fS{CR}_\genequivalence(\genrelation,\genrelationUpE)$) if
for all $t$ and $t'$, 
$t\oneCRs t'$
implies $t\ejoinability{\genrelationUpE}t'$ (Figure \ref{FigChurchRosserProperty}).
\end{defi}
%
\noindent
As expected, if $\genrelation/\genequivalence$ is used instead of $\genrelationUpE$  
in Figure \ref{FigChurchRosserProperty}, we obtain the leftmost diagram in Figure \ref{FigChurchRosserAndEConfluenceJK86}.
\begin{figure}[t]
\begin{center}
\begin{tabular}{c}
\xymatrix{
t  \ar@{.>}[dd]^{\genrelationUpE}_{!}\ar@{|=|}[rr]^{*} & & t'\ar@{.>}[dd]^{!}_{\genrelationUpE}
\\
\\
u\ar@{|.|}[rr]^{*}_{\genequivalence} & &  u'
}
\end{tabular}
\end{center}
\caption{$\genrelation$ is $\genrelationUpE$-Church-Rosser modulo $\genequivalence$ for $\genrelationUpE$ terminating}
\label{FigChurchRosserPropertyETerminatingR}
\end{figure}
Also note that, if $\genrelationUpE$ is \emph{normalizing}, then the diagram in Figure \ref{FigChurchRosserProperty} is equivalent to the diagram in Figure 
\ref{FigChurchRosserPropertyETerminatingR}, where $\rewnorm{\genrelationUpE}$ denotes the obtention of an $\genrelationUpE$-\emph{normal} form.

\begin{propC}
[{\cite[p.\ 1160, bullet 1]{JouKir_CompletionOfASetOfRulesModuloASetOfEquations_SIAMJC86}}]
\label{PropSufficientConditionForEConfluence_JK86}
If $\genrelation$ is $\genrelation$-Church-Rosser modulo 
$\genequivalence$, then 
$\genrelation$ is $\genrelationUpE$-Church-Rosser modulo 
$\genequivalence$, and then
$\genrelation$ is $\genequivalence$-confluent.
\end{propC}

\begin{proof}
If $\genrelation$ is $\genrelation$-Church-Rosser modulo 
$\genequivalence$, then for all $t$ and $t'$, $t\oneCRs t'$
implies $t\ejoinability{\genrelation}t'$.
Since $\genrelation\spcrel{\subseteq}\genrelationUpE\spcrel{\subseteq}\genrelation/\genequivalence$,
we also have $t\ejoinability{\genrelationUpE}t'$
and 
$t\ejoinability{\genrelation/\genequivalence}t'$.
\end{proof}
%
\noindent
Testing the Church-Rosser or confluence properties of $\genrelation$ modulo an equivalence $\genequivalence$ is
difficult due to the need of considering unbounded sequences of steps with $\oneCR$, $\rew{\genrelationUpE}$, or $\equone{\genequivalence}$.
Thus, a \emph{local} analysis, reducing the \emph{initial} aforementioned steps to only one, is convenient.

\subsection{Local confluence and local coherence properties}
 
Definition \ref{DefLocalConfluenceAndCoherenceProperties_JK86} below 
introduces the (parametric) 
\emph{local} properties which are used to investigate the ``global'' properties above.

\begin{defi}[Right-strict joinability modulo]
\label{DefJoinabilityProperties_JK86}
A pair $\langle t,t'\rangle$  with $t,t'\in A$ is 
\emph{right-strict} $\genrelationUpE$-joinable modulo $\genequivalence$, 
(denoted as $t\rsjoinability{\genrelationUpE}t'$), if $\exists\:u,u'$ s.t.\ $t\rews{\genrelationUpE}u$, 
$t'\rewp{\genrelationUpE}u'$, and
$u\equ{\genequivalence} u'$.
\end{defi}
\begin{figure}[t]
\begin{center}
\begin{tabular}{c@{\hspace{1.5cm}}c}
\xymatrix{
& s  \ar@{->}[ld]^{\genrelationUpE}\ar@{->}[rd]_{\genrelation} 
\\
t \ar@{.>}[d]_{*}^{\genrelationUpE} & & t'\ar@{.>}[d]^{*}_{\genrelationUpE}
\\
u \ar@{|.|}[rr]^{*}_{\genequivalence} & & u'
}
&
\xymatrix{
& s  \ar@{|-|}[rd]_{\genequivalence}
\ar@{->}[ld]^{\genrelationUpE} 
\\
t \ar@{.>}[d]_{*}^{\genrelationUpE}  
& & t' \ar@{.>}[d]^{+}_{\genrelationUpE} 
\\
u \ar@{|.|}[rr]^{*}_{\genequivalence} & & u'
}
\\
\emph{Local confluence modulo $\genequivalence$ of $\genrelationUpE$ with $\genrelation$}
&
\emph{Local coherence modulo $\genequivalence$ of $\genrelationUpE$}
\\[0.1cm]
$\fS{LCON}_\genequivalence(\genrelationUpE,\genrelation)$
& $\fS{LCOH}_\genequivalence(\genrelationUpE)$
\end{tabular}
\end{center}
\caption{Confluence and coherence properties: $3^{rd}$ and $5^{th}$ diagrams in \cite[Figure 2.1]{JouKir_CompletionOfASetOfRulesModuloASetOfEquations_SIAMJC86}}
\label{FigConfluenceAndCoherenceProperties}
\end{figure}

\begin{defi}[Local confluence and local coherence]
\label{DefLocalConfluenceAndCoherenceProperties_JK86}
For $\genrelation$, $\genequivalence$, $\genrelationUpE$, and $\genrelation/\genequivalence$ as above 
\begin{enumerate}
\item\label{DefLocalConfluenceAndCoherenceProperties_JK86_LocalConfluence} $\genrelationUpE$ is \emph{locally confluent modulo $\genequivalence$ with $\genrelation$} 
($\fS{LCON}_\genequivalence(\genrelationUpE,\genrelation)$) 
iff for all $t$, $t'$, and $t''$, if
$t\rew{\genrelationUpE}t'$ and 
$t\rew{\genrelation}t''$, then
$t'\ejoinability{\genrelationUpE}t''$
\cite[Definition 3(3)]{JouKir_CompletionOfASetOfRulesModuloASetOfEquations_SIAMJC86}, 
see Figure \ref{FigConfluenceAndCoherenceProperties} (left).
\item\label{DefLocalConfluenceAndCoherenceProperties_JK86_LocalCoherence}
 $\genrelationUpE$ is \emph{locally coherent modulo $\genequivalence$} ($\fS{LCOH}_\genequivalence(\genrelationUpE)$) iff for all $t$, $t'$, and $t''$, if
$t\rew{\genrelationUpE}t'$ and 
$t\equone{\genequivalence}t''$, then 
$t'\rsjoinability{\genrelationUpE}t''$, 
\cite[Definition 3(5)]{JouKir_CompletionOfASetOfRulesModuloASetOfEquations_SIAMJC86}, 
see Figure \ref{FigConfluenceAndCoherenceProperties} (right).
\end{enumerate}
\end{defi}
\begin{rem}[Local peaks and local cliffs]
Following \cite[Section 7.3]{DerJou_RewriteSystems_HTCS90} and
\cite[Section 11.3.1]{Jouannaud_ConfluenceOfTerminatingRewritingComputations_TFSP24} we often use \emph{local peak} to refer to $\peakOf{\genrelation}{\genrelation}$-peaks and
\emph{local cliff} to refer to $\peakOf{\genrelation}{\equone{}}$-peaks.
\end{rem}
%
\noindent
Below \cite[Definition 3]{JouKir_CompletionOfASetOfRulesModuloASetOfEquations_SIAMJC86}, Jouannaud and Kirchner remark that, if $\genequivalence$-termination of $\genrelation$ is assumed, then $\rewp{\genrelationUpE}$ in the rightmost branch of the diagram in Figure \ref{FigConfluenceAndCoherenceProperties}(right)  can be replaced by $\rews{\genrelationUpE}$
(see property \emph{5')} in \cite[page 1159]{JouKir_CompletionOfASetOfRulesModuloASetOfEquations_SIAMJC86}).
We reproduce their argument.

\begin{prop}
\label{PropEquivalenceLocalCoherenceModuloEandProperty5prime_JK86}
If $\genrelation$ is $\genequivalence$-terminating, then 
local coherence modulo $\genequivalence$ of $\genrelationUpE$ is equivalent to the commutation of the following diagram:
\[\xymatrix{
& s  \ar@{|-|}[rd]_{\genequivalence}
\ar@{->}[ld]^{\genrelationUpE} 
\\
t \ar@{.>}[d]_{*}^{\genrelationUpE}  
& & t' \ar@{.>}[d]^{*}_{\genrelationUpE} 
\\
u \ar@{|.|}[rr]^{*}_{\genequivalence} & & u'
}
\]
\end{prop}

\begin{proof}
Since $\rewp{\genrelationUpE}\spcrel{\subseteq}\rews{\genrelationUpE}$,
commutation of the rightmost diagram in Figure \ref{FigConfluenceAndCoherenceProperties} 
clearly implies commutation of the diagram above.
On the other hand, assume that such a diagram commutes but $t'$ is \emph{not} $\genrelationUpE$-reducible, as required by the rightmost diagram in Figure \ref{FigConfluenceAndCoherenceProperties}.
Then, $u\equ{\genequivalence}t'$ holds and therefore,
$u\equ{\genequivalence}t'\equone{\genequivalence}s\rew{\genrelationUpE}t\rews{\genrelationUpE}u$, i.e.,
$u\equ{\genequivalence}s\rewp{{\genrelationUpE}}u$
 contradicts $\genequivalence$-termination of $\genrelation$ because 
$\genrelationUpE\subseteq\genrelation/\genequivalence$, hence $\equ{\genequivalence}\circ\:\genrelationUpE\subseteq\:\equ{\genequivalence}\circ\:\genrelation/\genequivalence=\genrelation/\genequivalence$, i.e., 
$u\rewp{\genrelation/\genequivalence}u$.
\end{proof}

\noindent
For $\genequivalence$-terminating relations $\genrelation$, 
the following result characterizes
the $\genrelationUpE$-Church-Rosser property of $\genrelation$ modulo $\genequivalence$
as the local confluence and coherence properties modulo $\genequivalence$.

\begin{thmC}
[{\cite[Theorem 5]{JouKir_CompletionOfASetOfRulesModuloASetOfEquations_SIAMJC86}}
\label{TheoTheorem5_JK86}]
If $\genrelation$ is $\genequivalence$-terminating, then 
$\genrelation$ is 
$\genrelationUpE$-Church-Rosser modulo 
$\genequivalence$
iff  $\genrelationUpE$ is 
(i) locally confluent modulo $\genequivalence$ with $\genrelation$  and 
(ii) locally coherent modulo $\genequivalence$.
\end{thmC}
%
\noindent
In \cite[page 1162]{JouKir_CompletionOfASetOfRulesModuloASetOfEquations_SIAMJC86}, Jouannaud and Kirchner show that $E$-termination of $\genrelation$ \emph{cannot} be replaced by termination of $\genrelationUpE$ in Theorem \ref{TheoTheorem5_JK86}, see
\cite[Figures 2.4 \& 2.5]{JouKir_CompletionOfASetOfRulesModuloASetOfEquations_SIAMJC86}.

\subsection{Proving and disproving $\genequivalence$-confluence of $\genrelation$}
\label{SecProvingAndDisprovingEconfluenceOfR}

Theorem \ref{TheoTheorem5_JK86} and Proposition \ref{PropSufficientConditionForEConfluence_JK86}
yield a \emph{sufficient condition} for $\genequivalence$-confluence of $\genrelation$
on the basis of the local properties in Definition \ref{DefLocalConfluenceAndCoherenceProperties_JK86}.

\begin{cor}
\label{CoroTheorem5_JK86}
If $\genrelation$ is $\genequivalence$-terminating, then 
$\genrelation$ is 
$\genequivalence$-confluent
if  $\genrelationUpE$ is 
(i) locally confluent modulo $\genequivalence$ with $\genrelation$ and 
(ii) locally coherent modulo $\genequivalence$.
\end{cor}
\begin{proof}
By Theorem \ref{TheoTheorem5_JK86}, $\genrelation$ is $\genrelationUpE$-Church-Rosser modulo $\genrelation$.
By Proposition \ref{PropSufficientConditionForEConfluence_JK86}, $\genrelation$ is $\genequivalence$-confluent.
\end{proof}
%
\noindent
In contrast with the non-modulo case, where 
(i) confluence of $\genrelation$ implies its local confluence, and 
(ii) confluence and local confluence coincide for terminating relations $\genrelation$. \cite[Lemma 2.4]{Huet_ConfluentReductionsAbstractPropertiesAndApplicationsToTermRewritingSystems_JACM80}, 
Corollary \ref{CoroTheorem5_JK86} \emph{cannot} be turned into a characterization of $\genequivalence$-confluence of $\genrelation$, as we have the following.

\begin{prop}
\label{PropLocalConfluenceAndLocalCoherenceNotNecessaryForEConfluence}
There is an \etrs{} $\cR=(\Symbols,E,R)$ such that, if
$\genequivalence$ is $\equequ{E}$;
$\genrelation$ is $\rew{\cR}$; and
$\genrelationUpE$ is $\rew{\cR,E}$, then
$\genrelation$ is
$\genequivalence$-terminating and $\genequivalence$-confluent, but 
\begin{enumerate}
\item\label{PropLocalConfluenceAndLocalCoherenceNotNecessaryForEConfluence_LocalConfluence}
$\genrelationUpE$  is not locally confluent modulo $\genequivalence$ with $\genrelation$, and 
\item\label{PropLocalConfluenceAndLocalCoherenceNotNecessaryForEConfluence_LocalCoherence} 
$\genrelationUpE$  is not locally coherent modulo $\genequivalence$.
\end{enumerate}
\end{prop}
%
\noindent
Therefore, in general $\genequivalence$-confluence of $\genequivalence$-terminating relations $\genrelation$ is \emph{not} a necessary condition of local confluence modulo $\genequivalence$ with $\genrelation$ 
or local coherence modulo $\genequivalence$  of $\genrelationUpE$.
We prove this result by using the following example.

\begin{exa}
\label{ExLimitsPSjoinability}
Consider the \etrs{} $\cR$
\begin{IEEEeqnarray}{r'C'l}
\fS{a} & = & \fS{f}(\fS{b})
\label{ExLimitsPSjoinability_eq1}
\\
\fS{b} & \to & \fS{d}
\label{ExLimitsPSjoinability_rule1}
\\
\fS{c} & \to & \fS{a}
\label{ExLimitsPSjoinability_rule2}
\\
\fS{c} & \to & \fS{f}(\fS{d})
\label{ExLimitsPSjoinability_rule3}
\end{IEEEeqnarray}
whose rewriting modulo relation $\rew{\cR/E}$ is clearly terminating for $E=\{(\ref{ExLimitsPSjoinability_eq1})\}$.\footnote{A proof can be obtained, though, by using the techniques discussed in Section~\ref{SecETerminationOfEGTRSs}, i.e., by automatically computing a model of the first-order theory of the \egtrs{} which makes the interpretation of the one-step rewriting modulo relation $\rew{\cR/E}$ well-founded. By lack of space, we do not provide the complete development of the proof here, but it is analogous to the $E$-termination proof given in Example \ref{ExHuet80_RemarkPage818_ETermination} for the \etrs{} $\cR$ in Example \ref{ExHuet80_RemarkPage818}.}
\end{exa}
\begin{prop}
\label{PropExLimitsPSjoinability_EConfluent}
The \etrs{} $\cR$ in Example \ref{ExLimitsPSjoinability} is $E$-confluent.
\end{prop}

\begin{proof}
According to the diagram in Figure \ref{FigChurchRosserAndEConfluenceJK86} (right), let
 $s$, $t$ and $t'$ be such that 
\begin{IEEEeqnarray}{r'C'l}
t\leftrews{\cR/E} & s & \rews{\cR/E}t'\label{PropExLimitsPSjoinability_EConfluent_peak}
\end{IEEEeqnarray}
We can write $s$ as $s=\fS{f}^n(s')$ for some $n\geq 0$ with $s'\in\{\fS{a},\fS{b},\fS{c}\}$ (if $s'$ is a variable, then $s$ is $\rew{\cR/E}$-irreducible).
The following cases cover all possibilities:
\begin{enumerate}
\item Either $s'=\fS{a}$ or 
$s'=\fS{b}$. In both cases there is only one possible $\rew{\cR/E}$-sequence on $s$ just involving a single $\rew{\cR/E}$-step. Thus, no real peak (with $t\neq t'$) is possible.
\item If $s'=\fS{c}$, then, we have 
\begin{enumerate}
\item just a single step $s=\fS{f}^n(\ul{\fS{c}})\rew{(\ref{ExLimitsPSjoinability_rule3})}\fS{f}^{n+1}(\fS{d})$, which is 
$\rew{\cR/E}$-irreducible, and
\item 
$s=\fS{f}^n(\ul{\fS{c}})\rew{(\ref{ExLimitsPSjoinability_rule2})}\fS{f}^{n}(\fS{a})\equequ{E}\fS{f}^{n+1}(\ul{\fS{b}})\rew{(\ref{ExLimitsPSjoinability_rule1})}\fS{f}^{n+1}(\fS{d})$. 
\end{enumerate}
%
\noindent
Thus, all possible peaks (\ref{PropExLimitsPSjoinability_EConfluent_peak}) 
starting from $\fS{f}^n(\fS{c})$ are $\ejoinability{\cR/E}$-joinable.
\end{enumerate}
Thus, all peaks (\ref{PropExLimitsPSjoinability_EConfluent_peak}) are $\ejoinability{\cR/E}$-joinable and $\cR$ is $E$-confluent.
\end{proof}

\begin{proof}(of Proposition \ref{PropLocalConfluenceAndLocalCoherenceNotNecessaryForEConfluence})
\begin{enumerate}
\item %
The following 
local confluence peak
\begin{IEEEeqnarray}{r'C'l}
\fS{a} \leftrew{(\ref{ExLimitsPSjoinability_rule2})} & \fS{c} & \rew{(\ref{ExLimitsPSjoinability_rule3})} \fS{f}(\fS{d})\label{ExLimitsPSjoinability_criticalPeak}
\end{IEEEeqnarray}
is \emph{not} $\ejoinability{\cR,E}$-joinable, as both 
$\fS{a}$ and $\fS{f}(\fS{d})$ are 
$\rew{\cR,E}$-irreducible and not $E$-equivalent.
This shows that $\rew{\cR,E}$ is \emph{not} 
locally confluent modulo $E$ with $\rew{\cR}$.
 \item The following cliff 
 (where $\equone{E}$ is a $\rew{\eqoriented{E}}$-step using rule $\overleftarrow{(\ref{ExLimitsPSjoinability_eq1})}$, obtained from equation (\ref{ExLimitsPSjoinability_eq1}) by a right-to-left orientation):
\begin{IEEEeqnarray}{r'C'l}
\fS{f}(\fS{d})\leftrew{(\ref{ExLimitsPSjoinability_rule1}),E}
& \fS{f}(\fS{b}) & \rew{\eqoriented{E}}\fS{a}\label{ExLimitsPSjoinability_disjointCoherencePeak}
\end{IEEEeqnarray}
is not $\ejoinability{\cR,E}$-joinable, as both $\fS{f}(\fS{d})$ and 
$\fS{a}$ are $\rew{\cR,E}$-irreducible. 
Thus, $\rew{\cR,E}$ is \emph{not} 
locally coherent modulo $E$.\qedhere
\end{enumerate}
\end{proof}
%
\noindent
Thus, in general we \emph{cannot} disprove $\genequivalence$-confluence by disproving local confluence or local coherence of $\genrelationUpE$ modulo $\genequivalence$ (with $\genrelation$).
The following result is immediate from a close inspection of diagrams in Figure \ref{FigChurchRosserAndEConfluenceJK86} (right), for $\genequivalence$-confluence of $\genrelation$,
and the leftmost diagram in Figure \ref{FigConfluenceAndCoherenceProperties}, if $\genrelationUpE$ is $\genrelation/\genequivalence$ and taking into account that $\genrelation\subseteq\genrelation/\genequivalence$.

\begin{thm}
\label{TheoEConfluenceImpliesLocalConfluenceOfRewritingModulo}
If $\genrelation$ is $\genequivalence$-confluent, then 
$\genrelation/\genequivalence$ is locally confluent modulo 
$\genequivalence$ with $\genrelation$.
\end{thm}
%
%
\noindent
As a consequence of Theorem \ref{TheoEConfluenceImpliesLocalConfluenceOfRewritingModulo}, we obtain a criterion to \emph{disprove} 
$\genequivalence$-confluence of $\genrelation$.

\begin{cor}
\label{CoroDisprovingEConcludenceByNonREJoinableLocalConfluencePeaks}
If there is a non-$\ejoinability{\genrelation/\genequivalence}$-joinable
$\peakOf{\genrelationUpE}{\genrelation}$-peak, then $\genrelation$ is not $\genequivalence$-confluent.
\end{cor}
\begin{proof}
Since $\peakOf{\genrelationUpE}{\genrelation}\subseteq\peakOf{\genrelation/\genequivalence}{\genrelation}$,
the non-$\ejoinability{\genrelation/\genequivalence}$-joinability of any
$\peakOf{\genrelationUpE}{\genrelation}$-peak implies that $\genrelation/\genequivalence$ is not locally confluent modulo 
$\genequivalence$ with $\genrelation$.
By Theorem \ref{TheoEConfluenceImpliesLocalConfluenceOfRewritingModulo},
$\genrelation$ is not $\genequivalence$-confluent.
\end{proof}
%
\noindent
In contrast, cliffs $\peakOf{\genrelationUpE}{\equone{\genequivalence}}$ are always $\ejoinability{\genrelation/\genequivalence}$-joinable and \emph{cannot} be used to witness non-$\genequivalence$-confluence of $\genrelation$.

\begin{prop}
\label{PropREjoinabilityOfREpeaks}
Every cliff $\peakOf{\genrelationUpE}{\equone{\genequivalence}}$
is $\ejoinability{\genrelation/\genequivalence}$-joinable.
\end{prop}

\begin{proof}
Let $t\leftrew{\genrelationUpE}s\equone{\genequivalence}t'$. 
Since 
$\genrelationUpE\subseteq\genrelation/\genequivalence$ and $\equone{\genequivalence}\spcrel{\subseteq}\genequivalence$, we have:
\[\equone{\genequivalence}\circ\:\genrelationUpE\spcrel{\subseteq}
\equone{\genequivalence}\circ\:\genrelation/\genequivalence
\spcrel{\subseteq}
\genequivalence\circ\:\genrelation/\genequivalence
\spcrel{\subseteq}\genrelation/\genequivalence.\]
Hence, $t'\rew{\genrelation/\genequivalence}t$ and the 
cliff is $\ejoinability{\genrelation/\genequivalence}$-joinable.
\end{proof}

\subsection{Huet's approach}
\label{SecHuet80}
Huet focused on proving the following property (\emph{almost Church-Rosser modulo} in \cite[Definition 2.5.2(2)]{Ohlebusch_AdvancedTopicsInTermRewriting_2002}, denoted
 $\fS{ACR}_\genequivalence(\genrelation)$).
\begin{figure}[t]
\centering
\begin{tabular}{@{}c@{\hspace{1.2cm}}c@{\hspace{1.2cm}}c@{}}
\xymatrix@R=2pc{
s \ar@{->}[d]_{*}^{\genrelation} \ar@{|-|}[rr]^{*}_{\genequivalence} && s' \ar@{->}[d]^{*}_{\genrelation}
\\
t \ar@{.>}[d]_{*}^{\genrelation} && t'\ar@{.>}[d]^{*}_{\genrelation}
\\
u \ar@{|.|}[rr]^{*}_{\genequivalence} && u'
}
&
\xymatrix@R=2pc{
& s \ar@{->}[ld]^{\genrelation} \ar@{->}[rd]_{\genrelation}
\\
t \ar@{.>}[d]_{*}^{\genrelation} && t'\ar@{.>}[d]^{*}_{\genrelation}
\\
u \ar@{|.|}[rr]^{*}_{\genequivalence} && u'
}
&
\xymatrix@R=2pc{
& s \ar@{->}[ld]^{\genrelation} \ar@{|-|}[rd]_{\genequivalence}
\\
t \ar@{.>}[d]_{*}^{\genrelation} && t'\ar@{.>}[d]^{*}_{\genrelation}
\\
u \ar@{|.|}[rr]^{*}_{\genequivalence} && u'
}
\\
\begin{tabular}{@{}c@{}}
Confluence of $\genrelation$ modulo $\genequivalence$\\
$\fS{ACR}_\genequivalence(\genrelation)$
\end{tabular}
&
\begin{tabular}{@{}c@{}}
Property $\alpha_\genrelation$\\
$\fS{LCON}_\genequivalence(\genrelation,\genrelation)$
\end{tabular}
&
\begin{tabular}{@{}c@{}}
Property $\gamma_\genrelation$\\
$\fS{LCOH}_\genequivalence(\genrelation)$ (if $\fS{SN}_\genequivalence(\genrelation)$)
\end{tabular}
\end{tabular}
\caption{Confluence of $\genrelation$ modulo $\genequivalence$ and properties $\alpha$,
and $\gamma$ for $\genrelation$ \cite[Figs 8, 9 \& 11]{Huet_ConfluentReductionsAbstractPropertiesAndApplicationsToTermRewritingSystems_JACM80}}
\label{FigHuetConfluencePropertiesAlphaBetaGamma}
\end{figure}
\begin{defiC}[{\cite[first definition in page 802 and Figure 8]{Huet_ConfluentReductionsAbstractPropertiesAndApplicationsToTermRewritingSystems_JACM80}}]
\label{DefAbstractConfluenceModuloHuet80}
The relation $\genrelation$ is \emph{confluent modulo $\genequivalence$} ($\fS{ACR}_\genequivalence(\genrelation)$) iff
for all $s,s',t,t'\in A$, 
if $s\equ{\genequivalence}s'$, $s\rews{\genrelation}t$, and $s'\rews{\genrelation}t'$,
then there are $u,u'\in A$ such that
$t\rews{\genrelation}u$,  
$t'\rews{\genrelation}u'$, and $u\equ{\genequivalence} u'$, see
Figure \ref{FigHuetConfluencePropertiesAlphaBetaGamma} (left).
\end{defiC}
%
\noindent
Regarding the relationship between this definition and $\genequivalence$-confluence of $\genrelation$, immediately below his definition, 
Huet remarks that 
\begin{quote}
\emph{this condition is \emph{different} from 
$\rew{\genrelation}\!/\!\equ{\genequivalence}$ 
being confluent in 
$A\:/\!\equ{\genequivalence}$, 
since we do not allow $\sim$ along the $\rew{}$-derivations.
}
\end{quote}
 The following example substantiates this claim.
 
 \begin{exaC}[{\cite[Remark in p.\ 818]{Huet_ConfluentReductionsAbstractPropertiesAndApplicationsToTermRewritingSystems_JACM80}}]
\label{ExHuet80_RemarkPage818}
Let $E=\{(\ref{ExHuet80_RemarkPage818_eq1})\}$ and $\cR=\{(\ref{ExHuet80_RemarkPage818_rule1})\}$,
with
\begin{IEEEeqnarray}{r'C'l+}
\fS{a} & = & \fS{b}  \label{ExHuet80_RemarkPage818_eq1}
\\
\fS{f}(x,x) & \to & \fS{g}(x) \label{ExHuet80_RemarkPage818_rule1}
\end{IEEEeqnarray}
The following \emph{cliff}
\begin{IEEEeqnarray}{r'C'l}
\fS{g}(\fS{a})\leftarrowWith{\cR} \spcrel{\fS{f}(\fS{a},\fS{a})}  \equone{E} \fS{f}(\fS{a},\fS{b})\label{LblHuetRemarkPage818}
\end{IEEEeqnarray}
is \emph{not} $\ejoinability{\cR}$-joinable: $\fS{f}(\fS{a},\fS{b})$ and
$\fS{g}(\fS{\fS{a}})$ are non-$E$-equivalent $\rew{\cR}$-normal forms.
Hence,
$\rew{}$ is \emph{not} confluent modulo $\equequ{E}$ (Definition \ref{DefAbstractConfluenceModuloHuet80}).
However, 
$\cR$ is $E$-confluent (Example \ref{ExHuet80_RemarkPage818_EConfluence}).
\end{exaC}
%
\noindent
If $\genrelation$ is  \emph{normalizing}, 
then 
confluence of $\genrelation$ modulo $\genequivalence$ in Definition \ref{DefAbstractConfluenceModuloHuet80}
is \emph{equivalent} to the  commutation of the diagram in Figure \ref{FigChurchRosserPropertyETerminatingR},
 provided that \emph{$\genrelation$ and $\genrelationUpE$ are identified}:
 \begin{propC}
 [\label{PropEquivalenceOfConfluenceModuloAndChurchRosserPropertyOfR}
 \emph{\cite[Lemma 2.6]{Huet_ConfluentReductionsAbstractPropertiesAndApplicationsToTermRewritingSystems_JACM80}}]
For \emph{normalizing} relations $\genrelation$,  
confluence of $\genrelation$ modulo $\genequivalence$ ($\fS{ACR}_\genequivalence(\genrelation)$)
and
$\genrelation$ being $\genrelation$-Church-Rosser modulo 
$\genequivalence$ ($\fS{CR}_\genequivalence(\genrelation,\genrelation)$) 
are \emph{equivalent}.
\end{propC}
%
\noindent
If $\genrelation\circ\genequivalence$ is
terminating (which coincides with $E$-termination of $\genrelation$, i.e., that of $\genequivalence\circ\genrelation\circ\genequivalence$), then $\fS{ACR}_\genequivalence(\genrelation)$ is characterized by properties $\alpha_\genrelation$ 
(Figure \ref{FigHuetConfluencePropertiesAlphaBetaGamma}, middle) and 
$\gamma_\genrelation$ 
(Figure \ref{FigHuetConfluencePropertiesAlphaBetaGamma}, right). 

\begin{propC}
[\label{PropLemma2_8_Hue80}
\emph{\cite[Lemma  2.8]{Huet_ConfluentReductionsAbstractPropertiesAndApplicationsToTermRewritingSystems_JACM80}}]
Let $\genrelation$ be $\genequivalence$-\emph{terminating}.
Then, $\genrelation$ is confluent modulo $\genequivalence$ ($\fS{ACR}_\genequivalence(\genrelation)$)
iff 
$\alpha_\genrelation$ and $\gamma_\genrelation$ hold.
\end{propC}
%
\noindent
On the other hand, note that
\begin{enumerate}
\item 
$\fS{LCON}_\genequivalence(\genrelation,\genrelation)$ and $\alpha_{\genrelation}$ coincide.
\item 
By Proposition \ref{PropEquivalenceLocalCoherenceModuloEandProperty5prime_JK86}, 
$\fS{LCOH}_\genequivalence(\genrelation)$ and $\gamma_{\genrelation}$
coincide if $\genrelation$ is $\genequivalence$-terminating.
\end{enumerate}
Thus, by 
Proposition \ref{PropLemma2_8_Hue80} and
Theorem \ref{TheoTheorem5_JK86},
Huet's confluence of $\genrelation$ 
modulo $\genequivalence$ 
(Definition \ref{DefAbstractConfluenceModuloHuet80}) 
and the Church-Rosser property 
of $\genrelation$ modulo $\genequivalence$ in Figure \ref{FigChurchRosserProperty} coincide if
$\genrelation$ is used instead of $\genrelationUpE$.
Thus, we have the following.

\begin{cor}
\label{CoroEConfluence_Huet80}
If $\genrelation$ is $\genequivalence$-terminating, then 
$\genrelation$ is 
$\genequivalence$-confluent
if  properties $\alpha_\genrelation$ and $\gamma_\genrelation$ hold. 
\end{cor}

\noindent
Note that $\genrelation$ in Huet's approach is an abstract relation.
Although in Huet's paper $\genrelation$ is intended to represent 
term rewriting $\rew{\cR}$ for a \trs{} $\cR$,
alternative ``interpretations'' of $\genrelation$, like $\rew{\cR,E}$, or even $\rew{\cR/E}$, are possible.
Jouannaud and Kirchner are explicit about such ``indeterminacy'' by using $\genrelation^\genequivalence$ as a ``parameter'' which satisfies the fundamental assumption (\ref{LblFundamentalAssumptionJK86}), but $\genrelation$ in Huet's approach can also be used in different ways.
Taking into account that $\genequivalence$-confluence of $\genrelation$ and
$\genrelationUpE$ coincide; and also
$\genequivalence$-termination of $\genrelation$ and
$\genrelationUpE$ coincide,
Huet's Corollary \ref{CoroEConfluence_Huet80} can be rephrased 
using Jouannaud and Kirchner's setting as follows:

\begin{cor}
\label{CoroEConfluence_Huet80bis}
If $\genrelation$ is $\genequivalence$-terminating, then 
$\genrelation$ is 
$\genequivalence$-confluent
if  properties $\alpha_{\genrelationUpE}$ and 
$\gamma_{\genrelationUpE}$ hold, for $\genrelationUpE$ such that
$\genrelation\subseteq\genrelationUpE\subseteq\genrelation/\genequivalence$. 
\end{cor}
\subsection{Discussion}
\label{SecDiscussionJK86}

In our setting, the diagrams for properties $\alpha$ and $\gamma$ 
in Figure \ref{FigHuetConfluencePropertiesAlphaBetaGamma} are important because
they only involve 
$\genrelation$ and $\genequivalence$.
This is important in practice: as shown in \cite[Section 3.4]{Huet_ConfluentReductionsAbstractPropertiesAndApplicationsToTermRewritingSystems_JACM80}, in the realm of (equational) term rewriting systems $\cR=(\Symbols,E,R)$, if $\genrelation$ is taken to be 
$\rew{\cR}$, then properties $\alpha_\genrelation$ and $\gamma_\genrelation$ can be proved
as the joinability (modulo) 
of \emph{ordinary} critical pairs $\langle s,t\rangle$ \cite[page 809]{Huet_ConfluentReductionsAbstractPropertiesAndApplicationsToTermRewritingSystems_JACM80}
obtained from \emph{both}
the rules in $R$ and also from the rules obtained by orienting the equations $\elhsr\equequ{}\erhsr$ 
in $E$ as rules $\elhsr\to\erhsr$ and $\erhsr\to\elhsr$ \cite[Lemmata 3.4 and 3.5]{Huet_ConfluentReductionsAbstractPropertiesAndApplicationsToTermRewritingSystems_JACM80}.
Furthermore, joinability of critical pairs $\langle s,t\rangle$ is tested as $s\rews{\cR}\circ\equequ{E}\circ\leftrews{\cR}t$ (usually called \emph{joinability modulo}, written $\ejoinability{}$ in \cite{Huet_ConfluentReductionsAbstractPropertiesAndApplicationsToTermRewritingSystems_JACM80}).

In sharp contrast (see Figure \ref{FigConfluenceAndCoherenceProperties}), 
(local) peaks in 
local confluence modulo $\genequivalence$ of $\genrelationUpE$ with $\genrelation$ 
and local coherence modulo $\genequivalence$ of $\genrelationUpE$ 
involve 
$\genequivalence$, $\genrelation$, and \emph{also} $\genrelationUpE$. 
As a consequence,  if 
$\genequivalence$ is $\equequ{E}$, 
$\genrelation$ is $\rew{\cR}$, and
$\genrelationUpE$ is 
$\rew{\cR,E}$,
the analysis of the local properties requires
$E$-critical pairs $\langle s,t\rangle$ \cite[Definition 12]{JouKir_CompletionOfASetOfRulesModuloASetOfEquations_SIAMJC86}.
As remarked in the introduction, this has a number of problems: there can be infinitely many; $E$-unification algorithms may \emph{fail} to exist,\ldots.
Also note that 
\begin{enumerate}
\item ordinary critical pairs are also $E$-critical pairs but not vice versa, i.e., the use of $E$-critical pairs amounts to having more pairs to test.
On the other hand,
\item  joinability of $E$-critical pairs 
$\langle s,t\rangle$ is tested as $s\rews{\cR,E}\circ\equequ{E}\circ\leftrews{\cR,E}t$, which is more \emph{powerful} than joinability modulo, as $\rews{\cR}\spcrel{\subseteq}\rews{\cR,E}$.
\end{enumerate}
On the other hand, note that
\begin{enumerate}
\item 
$\fS{LCON}_\genequivalence(\genrelationUpE,\genrelationUpE)$, called \emph{Local $\genequivalence$-confluence} (of $\genrelationUpE$) in
\cite[Definition 5(3)]{JouMun_TerminationOfASetOfRulesModuloASetOfEquations_CADE84}, and $\alpha_{\genrelationUpE}$ coincide.
\item 
By Proposition \ref{PropEquivalenceLocalCoherenceModuloEandProperty5prime_JK86}, 
$\fS{LCOH}_\genequivalence(\genrelationUpE)$ and $\gamma_{\genrelationUpE}$
coincide if $\genrelationUpE$ is $\genequivalence$-terminating.
\end{enumerate}
However, in general $\fS{LCON}_\genequivalence(\genrelationUpE,\genrelation)$
does \emph{not} correspond to $\alpha_\genrelation$ or $\alpha_{\genrelationUpE}$.
\begin{rem}
As a consequence of Proposition \ref{PropLocalPeaksModulo}, though,
if $\genrelationUpE$ is taken to be $\rew{\cR,E}$, then 
$\fS{LCON}_E(\rew{\cR,E},\rew{\cR})$
and
$\fS{LCON}_E(\rew{\cR,E},\rew{\cR,E})$
\emph{coincide} for $E$-terminating \etrs{s} $\cR$.
However, in contrast to Jouannaud and Kirchner, 
Huet did not provide any treatment of $E$-confluence of \etrs{s} 
by 
using properties $\alpha$ and $\gamma$ for $\rew{\cR,E}$.
\end{rem}
%
\noindent
Thus, after presenting \egtrs{s} in the next section, in
Section~\ref{SecReductionAndCoherencePeaks} we  
make explicit the kind of peaks and joinability conditions 
we need to consider to 
prove and disprove confluence modulo of \egtrs{s}
in the remainder of the paper.

\section{Equational Generalized Term Rewriting Systems}
\label{SecEGTRSs}

The following definition introduces the kind of computational systems we consider here which can be viewed as a specialization of \emph{Generalized Term Rewriting Systems} (\gtrs{s}) introduced in \cite{Lucas_LocalConfluenceOfConditionalAndGeneralizedTermRewritingSystems_JLAMP24}
(see Section~\ref{SecConfluenceOfGTRSsAsConfluenceOfEGTRSsWithEempty} for a more detailed comparison).

\begin{defi}[Equational Generalized Term Rewriting Systems]
\label{DefCS_ECTRS}
An \emph{Equational Generalized Term Rewriting System (\egtrs{})} is a tuple
$\cR=(\Symbols,\SPredicates,\mu,E,H,R)$ where 
$\Symbols$ is a signature of function symbols;
$\SPredicates$ is a signature of predicate symbols with $\ceq,\rew{},\rews{}\in\SPredicates$;
$\mu$ is a replacement map for the signature $\Symbols$ of function symbols;
and
\begin{enumerate}
\item $E$ is a set of conditional equations $s\ceq t\IF \gencond$, for terms $s$ and $t$;
\item\label{DefCS_ECTRS_HornClauses} $H$ is a set of definite Horn clauses $A\IF \gencond$ where $A=P(t_1,\ldots,t_n)$ for some $P\in\SPredicates$ such that $P\notin\{\ceq,\rew{},\rews{}\}$ and
terms $t_1,\ldots,t_n$, $n\geq 0$; and
\item\label{DefCS_ECTRS_Rules}
 $R$ is a set of conditional rules $\lhsr\to\rhsr\IF\gencond$ for terms $\lhsr\notin\Variables$ 
and $\rhsr$.
\end{enumerate}
where $\gencond$ is a sequence of atomic formulas.
\end{defi}
%
\noindent
Note that $E\cup H\cup R$ is a set of (definite) Horn clauses.

\begin{exa}
\label{ExSumListsA_EGTRS}
The \egtrs{} 
in Example \ref{ExSumListsA} with $\mu$ as given in the example is 
$\cR=(\Symbols,\SPredicates,\mu,E,H,R)$, where
$\Symbols=\{\arityOf{\fS{0}}{0},\arityOf{\fS{s}}{1},
\pholder\seqnat\pholder
\}$,
$\SPredicates=\{\arityOf{\fS{Nat}}{1},
\pholder\!\cto\!\pholder,
\pholder\!\ceq\!\pholder,
\pholder\!\to\!\pholder,
\pholder\!\to^*\!\pholder
\}$,
$E=\{(\ref{ExSumListsA_eq1})\}$,
$H=\{(\ref{ExSumListsA_clause1}),(\ref{ExSumListsA_clause2}),(\ref{ExSumListsA_clause3})\}$,
and 
$R=\{(\ref{ExSumListsA_rule1}),(\ref{ExSumListsA_rule2}),(\ref{ExSumListsA_rule3}),(\ref{ExSumListsA_rule4})\}$.
\end{exa}

\begin{exa}
\label{ExHuet80_RemarkPage818_EGTRS}
Viewed as an \egtrs{}, the \etrs{} in Example \ref{ExHuet80_RemarkPage818},
i.e.,
\begin{IEEEeqnarray*}{r'C'l+x*}
\fS{a} & = & \fS{b} & \eqref{ExHuet80_RemarkPage818_eq1}
\\
\fS{f}(x,x) & \to & \fS{g}(x)& \eqref{ExHuet80_RemarkPage818_rule1}
\end{IEEEeqnarray*}
is $(\Symbols,\SPredicates,\mu,E,H,R)$ where
$\Symbols=\{\arityOf{\fS{a}}{0},\arityOf{\fS{b}}{0},\arityOf{\fS{f}}{2},\arityOf{\fS{g}}{1}\}$,
$\SPredicates=\{\pholder\!\ceq\!\pholder,
\pholder\!\to\!\pholder,
\pholder\!\to^*\!\pholder
\}$,
$\mu=\muTop$ (no replacement restriction for any function symbol),
$E=\{(\ref{ExHuet80_RemarkPage818_eq1})\}$,
$H=\emptyset$,
and 
$R=\{(\ref{ExHuet80_RemarkPage818_rule1})\}$.
\end{exa}
%
\noindent
As for \ctrs{s} \cite[Definition 6.1]{MidHam_CompletenessResultsForBasicNarrowing_AAECC94},  
conditional rules $\lhsr\to\rhsr \IF\gencond$ are classified according to the distribution of
variables among $\lhsr$, $\rhsr$, and $\gencond$:
\begin{itemize}
\item \emph{type 1}, if  $\Var(\rhsr)\cup\Var(\gencond)\subseteq\Var(\lhsr)$;
\item \emph{type 2}, if $\Var(\rhsr)\subseteq\Var(\lhsr)$;
\item \emph{type 3}, if $\Var(\rhsr)\subseteq\Var(\lhsr)\cup\Var(\gencond)$; and
\item \emph{type 4}, otherwise.
\end{itemize}
A rule of type $n$ is often called an $n$-rule.
An n-rule $\alpha$ is \emph{proper} if for all $m<n$,
$\alpha$ is \emph{not} an $m$-rule.
An $n$-CTRS contains only $m$-rules for some $m\leq n$ and at least a proper $n$-rule.

The following definition formalizes the idea of having predicate symbols which \emph{depend} on Horn clauses or rules of an \egtrs{} (or both).
\begin{defi}
We say that $P\in\SPredicates$
\begin{itemize}
\item \emph{depends on $H$} 
if there is $A\IF A_1,\ldots,A_n\in H$ with
$\rootTerm(A)=P$
or there is $A\IF A_1,\ldots,A_n\in E\cup H\cup R$ with $\rootTerm(A)=P$
such that $\rootTerm(A_i)$ depends on $H$ for some $1\leq i\leq n$.
\item  \emph{depends on $R$}
if $P\in\{\rew{},\rews{}\}$ 
or there is $A\IF A_1,\ldots,A_n\in E\cup H$ with $\rootTerm(A)=P$
such that $\rootTerm(A_i)$ depends on $R$ for some $1\leq i\leq n$.
\end{itemize}
\end{defi}
\begin{exa}
Consider the \egtrs{} $\cR=(\Symbols,\SPredicates,\mu,E,H,R)$ 
in Example \ref{ExSumListsA}. We have that
\begin{itemize}
\item The equality predicate $=$ does \emph{not} depend on $H$ or $R$.
\item Predicate $\cto$ depends on $R$.
\item Predicate $\to$ depends on $H$ due to rules (\ref{ExSumListsA_rule3}) and
(\ref{ExSumListsA_rule4}).
\end{itemize}
\end{exa}

\begin{exa}
\label{ExEvenEq_EGTRS}
Consider the following \egtrs{} $\cR=(\Symbols,\SPredicates,\mu,E,H,R)$, 
\begin{IEEEeqnarray}{r'C'l}
\fS{s}(\fS{s}(x)) &\ceq  & x  \IF x\geq\fS{s}(\fS{0})\label{ExEvenEq_eq1}
\\
x & \geq & 0\label{ExEvenEq_clause1}
\\
\fS{s}(x) & \geq & \fS{s}(y) \IF x\geq y\label{ExEvenEq_clause2}
\\
\fS{test}(x) & \to & \fS{pev}(x)\IF x\ceq\fS{s}(\fS{s}(\fS{0}))\qquad\label{ExEvenEq_rule1}
\\
\fS{test}(x) & \to & \fS{odd}(x)\IF x\ceq\fS{s}(\fS{0})\label{ExEvenEq_rule2}
\\
\fS{test}(x) & \to & \fS{zero}(x)\IF x\ceq\fS{0}\label{ExEvenEq_rule3}
\end{IEEEeqnarray}
where 
$\geq$ is defined by the Horn clauses (\ref{ExEvenEq_clause1}) and (\ref{ExEvenEq_clause2}).
Rules (\ref{ExEvenEq_rule1}), (\ref{ExEvenEq_rule2}), and (\ref{ExEvenEq_rule3}) define tests to check
whether a number (in Peano notation) is either
(i) \emph{positive and even}, or 
(ii) \emph{odd}, or else
(iii) \emph{zero}, using the equivalence specified by (\ref{ExEvenEq_eq1}).
We have that
\begin{itemize}
\item The equality predicate $=$ depends on $H$ (due to (\ref{ExEvenEq_eq1})) 
but not on $R$.
\item Predicate $\to$ depends on $E$ (due to rules (\ref{ExEvenEq_rule1})--(\ref{ExEvenEq_rule3})) and also on $H$, as the rules use the equality, which depends on $H$.
\end{itemize}
\end{exa}

\begin{rem}[Conditions $s\cto t$ and their interpretation]
\label{RemInterpretationOfCto}
In the (recent) literature about \emph{Conditional TRSs},
symbol $\cto$ is often used to 
specify conditions 
$s\cto t$ in rules having different \emph{interpretations}: as \emph{joinability}, \emph{reachability}, 
etc., as pioneered by Bergtstra and Klop \cite[Definition 2.1]{BerKlo_ConditionalRewriteRulesConfluenceAndTermination_JCSS86}, although they used $\Box$, instead.
In \egtrs{s}, $\cto$ would be treated as a predicate in $\SPredicates$ and
the desired interpretation is \emph{explicitly} 
obtained by including in $H$ an appropriate set 
of clauses defining $\cto$.
For instance, 
(\ref{ExSumListsA_clause3}) in Example \ref{ExSumListsA} interprets $\cto$ as \emph{reachability}.
\end{rem}

\begin{notation}
 Equations $s\ceq t\IF\gencond$ in $E$ are often \emph{transformed} 
into rules 
by choosing a \emph{left-to-right} or \emph{right-to-left} orientation:
\[\overrightarrow{E}=\{s\to t\IF\gencond\mid s\ceq t\IF\gencond\in E\} \text{ and }
\overleftarrow{E}=\{t\to s\IF\gencond\mid s\ceq t\IF\gencond\in E\}.\]
We let $\eqoriented{E}=\overrightarrow{E}\cup\overleftarrow{E}$ (we shrink $\oneconversion{}$ into $\diamond$).
Note that $\eqoriented{E}$ may contain rules $\elhsr\to\erhsr\IF\gencond$ whose left-hand side
$\elhsr$ is a \emph{variable}.
Let $\DSymbols(\eqoriented{E})=\{\rootTree(\elhsr)\in\Symbols\mid \elhsr\to\erhsr\IF\gencondbis\in\eqoriented{E}\}$.
\end{notation}
%
\noindent
In the following, a rule $\lhsr\to\rhsr\IF\gencond$ is often said to be \emph{collapsing} if $\rhsr$ is a variable (cf. \cite[Definition 9.2.3]{BaaNip_TermRewritingAndAllThat_1998}).
A set of rules is \emph{collapsing} if it contains collapsing rules.

\begin{rem}[Equations leading to ill-defined rules]
\label{RemEquationsLeadingToImproperRules}
If $\eqoriented{E}$ is a collapsing set of rules, then $\eqoriented{E}$ contains \emph{ill-defined} 
rules $\alpha:x\to\rhsr\IF\gencond$ (for some $x\in\Variables$). 
\end{rem}

\subsection{First-order theories of an \egtrs{}}
\label{SecFirstOrderTheoryOfAnEGRS}
Computational relations (e.g., $\equequ{E}$, 
$\rew{\cR}$, $\rew{\cR,E}$, $\rew{\cR/E}$,\ldots) 
induced by an \egtrs{} $\cR=(\Symbols,\SPredicates,\mu,E,H,R)$
are defined by first-order deduction of atoms 
$s = t$ (equality in $E$),
$s\to t$ (one-step rewriting in the usual sense), 
$s\rewpstickel{}t$ (rewriting \`a la Peterson \& Stickel), 
$s\rewmodulo{}t$ (rewriting modulo), 
etc., in
some FO-theory. 
We \emph{extend} 
$\SPredicates$ with 
$\rewpstickel{}$, $\rewmodulo{}$, etc., and also
$\cto_\PSabbr$, $\cto_\RMabbr$ (as they defined in terms of the previous predicates).
Our FO-theories are obtained from the generic sentences in Table \ref{TableFOSentencesEGRSs}.

\begin{table}[ht]
\caption{Generic sentences of the FO-theory of \egtrs{s}}
\label{TableFOSentencesEGRSs}
\centering
\begin{tabular}{@{}ll@{}}
\toprule
Label & Sentence\\
\midrule
$(\RuleReflexivity)^\gbop$
&
$(\forall x)~x \mathrel{\gbop} x$
\\
$(\RuleTransitivity)^\gbop$
& $(\forall x,y,z)~x \:\gbop\: y\wedge y \:\gbop\: z\Rightarrow x
\:\gbop\: z$
\\
$(\RuleSymmetry)^\gbop$
& $(\forall x,y)~y\:\gbop\:x\Rightarrow x \:\gbop\: y$
\\
$(\RuleCompatibility)^\gbop$
&
$(\forall x,y,z)~x \:\gbop\: y\wedge y
\:\gbop^*\: z\Rightarrow x
\:\gbop^*\: z$
\\
$(\RulePropagation)^\gbop_{f,i}$
&
$(\forall x_1,\ldots,x_k,y_i)~x_i\:\gbop\: y_i\Rightarrow f(x_1,\ldots,x_i,\ldots,x_k)\:\gbop\: f(x_1,\ldots,y_i,\ldots,x_k)$
\\
$(\RuleHornClause)_{A\IF A_1,\ldots,A_n}$
& $(\forall\vec{x})~A_1\wedge\cdots\wedge A_n\Rightarrow A$
\\
$(\RuleRlEq)_{\lhsr\to \rhsr\IF A_1,\ldots,A_n}$
& $(\forall x,\vec{x})~x = \lhsr\wedge A_1\wedge\cdots\wedge A_n\Rightarrow x \rpstickel \rhsr$
\\
$(\RuleRewMEq)$
& $(\forall x,x',y,y')~x = x'\wedge x'\to y' \wedge y'= y\Rightarrow x \rmodulo{} y$
\\
\bottomrule
\end{tabular}
\end{table}
\noindent
where:
\begin{itemize}
\item Sentences $(\RuleReflexivity)^\gbop$, 
$(\RuleTransitivity)^\gbop$, and 
$(\RuleSymmetry)^\gbop$, which are parametric on a binary relation $\gbop$, 
 express \emph{reflexivity}, \emph{transitivity}, and \emph{symmetry} of  $\gbop$, respectively.
 \smallskip
\item $(\RuleCompatibility)^\gbop$ expresses \emph{compatibility} of one-step 
and many-step reduction with $\gbop$.
\smallskip
\item For each $k$-ary function symbol $f$, $1\leq i\leq k$, and $x_1,\ldots,x_k$ and $y_i$ distinct variables,
$(\RulePropagation)^\gbop_{f,i}$ \emph{propagates} an $\gbop$-step to the $i$-th immediate subterm of an $f$-rooted term.
\smallskip
\item $(\RuleHornClause)_\alpha$ is the sentence associated with a clause $\alpha:A\IF A_1,\ldots,A_n$ 
with variables $\vec{x}$.
Since rewrite rules $\alpha:\lhsr\to\rhsr\IF\gencond$
are particular clauses (as $\lhsr\to\rhsr$ is just an atom of the binary predicate $\_\to\_$), we do not add any specific generic sentence for them. 
\smallskip
\item $(\RuleRlEq)_{\alpha}$
defines a Peterson \& Stickel rewriting step $s\rew{\cR,E}t$ (at the root) using rule $\alpha:\lhsr\to\rhsr\IF\gencond$ with variables $\vec{x}$. Here, $x\notin\vec{x}$.
\smallskip
\item $(\RuleRewMEq)$ defines reduction modulo $\rew{\cR/E}$ in the usual way.
\end{itemize}
\begin{defi}[Theory generators]
Consider the following 
\emph{generators of} theories with parameters
$\SigParameter$ (referring to a signature),
$\RMaParameter$ (a replacement map),
$\EqParameter$ (a set of equations), and
$\TRSParameter$ (a set of rules):

{\footnotesize
\[\begin{array}{rcl}
\GLtheory_{\text{Eq}}[\SigParameter,\RMaParameter,\EqParameter] 
& = &
\{(\RuleReflexivity)^{\equequ{}},
(\RuleSymmetry)^{\equequ{}},
(\RuleTransitivity)^=\}
\}
\cup\{(\RulePropagation)^{\equequ{}}_{f,i}\mid f\in\SigParameter,i\in\RMaParameter(f)\}\cup\{(\RuleHornClause)_e\mid e\in\EqParameter\}
\\
\GLtheory_{\text{R}}[\SigParameter,\RMaParameter,\TRSParameter] 
& = & \{
(\RuleReflexivity)^{\to^*},
(\RuleCompatibility)^{\to}\}
\cup
\{(\RulePropagation)^\to_{f,i}\mid f\in \SigParameter,i\in\RMaParameter(f)\}\cup\{(\RuleHornClause)_\alpha\mid\alpha\in \TRSParameter\}
\}
\\  
\GLtheory_{\text{R,M}}[\SigParameter,\RMaParameter,\TRSParameter]
& = & 
\{
(\RuleReflexivity)^{\rewspstickel{}},
(\RuleCompatibility)^{\rewpstickel{}}\}
\cup
\{(\RulePropagation)^{\rewpstickel{}}_{f,i}\mid f\in \SigParameter,i\in\RMaParameter(f)\}\cup\{(\RuleRlEq)_\alpha\mid\alpha\in \TRSParameter\}
\\  
\GLtheory_{\text{R/M}}[\SigParameter,\RMaParameter,\TRSParameter] 
& = & \{
(\RuleReflexivity)^{\rmStar},
(\RuleCompatibility)^{\rmodulo}
\}
\cup
\{(\RulePropagation)^{\to}_{f,i}\mid f\in \SigParameter,i\in\RMaParameter(f))\}\cup\{(\RuleHornClause)_\alpha\mid\alpha\in \TRSParameter\}
\}
\cup
\{
(\RuleRewMEq)
\}
\end{array}
\]
\noindent
}%
\end{defi}
Note that, in $\GLtheory_{\text{R/M}}[\SigParameter,\RMaParameter,\TRSParameter]$ 
propagation sentences are given for $\rew{}$ rather than for $\rmodulo$;
this is consistent with the usual definition of rewriting modulo implemented by 
(\RuleRewMEq).

\subsubsection{Different roles of rewriting predicates}

Since rules $\alpha:\lhsr\to\rhsr\IF\gencond$ in $R$ 
are used to specify 
\emph{different} computational relations (see Definition \ref{DefComputationalRelationsOfAnEGTRS}), 
conditions $s\cto t\in\gencond$
may have different \emph{interpretations} depending on the targeted relation:
the meaning of $\cto$ is based on predicate $\rew{}$ if $\alpha$ is
used to describe the usual (conditional) rewrite relation $\rew{\cR}$;
however, 
$\cto$ should be treated using $\rpstickel$ if $\alpha$ is used to describe Peterson \& Stickel's rewriting modulo $\rew{\cR,E}$;
and
$\cto$ should be treated using $\rmodulo$ if $\alpha$ is used to describe $\rew{\cR/E}$.
A simple way to deal with this situation is the following.

\begin{defi}
Given a clause $\alpha:A\IF A_1,\ldots,A_n$ 
we obtain $\alpha^\PSabbr$ (resp. $\alpha^\RMabbr$) 
by replacing 
\begin{enumerate}
\item each occurrence of $\cto$ in $A$ by $\cto_\PSabbr$ (resp.\ $\cto_\RMabbr$)
and
\item 
each occurrence of
$\cto$, $\rew{}$, and 
$\rews{}$ in 
$A_1,\ldots,A_n$ 
by $\cto_\PSabbr$, $\rewpstickel{}$, and $\rewspstickel{}$ 
(resp.\ $\cto_\RMabbr$, $\rewmodulo{}$, and $\rewmodulos{\!\!}$).
\end{enumerate}
Given an \egtrs{} $\cR=(\Symbols,\SPredicates,\mu,E,H,R)$,
\begin{itemize}
\item $H^\PSabbr$ (resp.\ $H^\RMabbr$) is obtained from $H$ by replacing 
each $\alpha\in H$ by $\alpha^\PSabbr$ (resp.\ $\alpha^\RMabbr$).
\item $R^\PSabbr$ (resp.\ $R^\RMabbr$) is obtained from $R$ by replacing 
each $\alpha\in R$ by $\alpha^\PSabbr$ (resp.\ $\alpha^\RMabbr$).
\end{itemize}
\end{defi}
%
\noindent
In the following, if $\alpha$ coincides with $\alpha^\PSabbr$ (resp.\ $\alpha^\RMabbr$), we just use $\alpha$ rather than introducing a new label.
Note, for instance, that this always happens if $\alpha$ is an \emph{unconditional} rule.
\begin{rem}[Predicate $\equequ{}$ does not depend on $R$]
\label{RemEqDoesNotDependOnR}
In the following, we will assume that $\equequ{}$ does \emph{not} 
depend on $R$. Thus, 
the ``meaning'' of equational atoms $s\equequ{} t$ does \emph{not} depend on the \emph{rules} in 
$R$.
For this reason, we do not need versions $\equequ{ps}$ or $\equequ{rm}$
of the equality predicate.
Similarly, $E^\PSabbr$ or $E^\RMabbr$ are not necessary.
\end{rem}
\begin{exa}
\label{ExSumListsA_RM}
For $\cR$ in Example \ref{ExSumListsA}, 
\[H^\RMabbr=\{(\ref{ExSumListsA_clause1}),(\ref{ExSumListsA_clause2}),(\ref{ExSumListsA_clause3})^\RMabbr\} \text{ and }
R^\RMabbr=\{(\ref{ExSumListsA_rule1}),(\ref{ExSumListsA_rule2}),(\ref{ExSumListsA_rule3})^\RMabbr,(\ref{ExSumListsA_rule4})^\RMabbr\}\] with:
\begin{IEEEeqnarray*}{+rCl+x*}
x\cto_\RMabbr y & \IF & x\rewmodulos{}y
& $(\ref{ExSumListsA_clause3})^\RMabbr$
\label{ExSumListsA_clause3_RM}
\\
\fS{sum}(\cfrozen{m}\hphantom{s}) & \to & n\IF m\cto_\RMabbr n,\fS{Nat}(n)
& $(\ref{ExSumListsA_rule3})^\RMabbr$
\label{ExSumListsA_rule3_RM}
\\
\fS{sum}(\cfrozen{ms}) & \to & m+n \IF ms \cto_\RMabbr m \fS{\seqnat} \cfrozen{ns},
 \fS{Nat}(m), \fS{sum}(\cfrozen{ns})\cto_\RMabbr n
 & $(\ref{ExSumListsA_rule4})^\RMabbr$
 \label{ExSumListsA_rule4_RM}
\end{IEEEeqnarray*}
 $H^\PSabbr$ and $R^\PSabbr$ are similarly obtained.
 \end{exa}

\subsubsection{Theories of an \egtrs}

In the following, we consider a number of theories associated with an \egtrs{} to define different computational relations.

\begin{defi}[Theories of an \egtrs]
\label{DefTheoriesOfAnEGTRS}
Given an \egtrs{} $\cR=(\Symbols,\SPredicates,\mu,E,H,R)$
whose equality predicate $\equequ{}$ does \emph{not} depend on $R$,
the following associated theories are obtained:
\[\begin{array}{rcl@{\hspace{1cm}}rcl}
\GLtheory_{E} & = & \GLtheory_{\text{Eq}}[\Symbols,\mu,E]
\cup\{(\RuleHornClause)_\alpha\mid\alpha\in H\}
\\
\GLtheory_{\cR} & = & 
\GLtheory_{\text{Eq}}[\Symbols,\mu,E]
\cup\GLtheory_{\text{R}}[\Symbols,\mu,R]
\cup\{(\RuleHornClause)_\alpha\mid\alpha\in H\}
\\
& = & 
\GLtheory_{E}\cup\GLtheory_{\text{R}}[\Symbols,\mu,R]
\\
\GLtheory_{\cR,E} & = & \GLtheory_{\text{Eq}}[\Symbols,\mu,E]
\cup
\GLtheory_{\text{R,M}}[\Symbols,\mu,R^{\PSabbr}]
\cup\{(\RuleHornClause)_\alpha\mid\alpha\in H^{\PSabbr}\}
\\
\GLtheory_{\cR/E} & = & \GLtheory_{\text{Eq}}[\Symbols,\mu,E]
\cup\GLtheory_{\text{R/M}}[\Symbols,\mu,R^{\RMabbr}]
\cup\{(\RuleHornClause)_\alpha\mid\alpha\in H^{\RMabbr}\}
\end{array}\] 
\end{defi}
\begin{exa}
\label{ExHuet80_RemarkPage818_TheoryRModulo}
For the \etrs{} in Example \ref{ExHuet80_RemarkPage818}, i.e.,
\begin{IEEEeqnarray*}{+rCl+x*}
\fS{a} & = & \fS{b} & \eqref{ExHuet80_RemarkPage818_eq1}
\\
\fS{f}(x,x) & \to & \fS{g}(x)& \eqref{ExHuet80_RemarkPage818_rule1}
\end{IEEEeqnarray*}
the corresponding \egtrs{} $\cR$ in Example \ref{ExHuet80_RemarkPage818_EGTRS},
has an associated theory $\GLtheory_{\cR/E}$ which, since $H=\emptyset$, 
is the union of
\[\GLtheory_{\text{Eq}}[\Symbols,\muTop,E]=\{
(\RuleReflexivity)^=,
(\RuleSymmetry)^=,
(\RuleTransitivity)^=,
(\RulePropagation)^=_{\fS{f},1},
(\RulePropagation)^=_{\fS{f},2},
(\RulePropagation)^=_{\fS{g},1},
(\RuleHornClause)_{(\ref{ExHuet80_RemarkPage818_eq1})}
\} 
\]
and, since $R^\RMabbr=R$,
\[\GLtheory_{\text{R/M}}[\Symbols,\muTop,R]=
\{
(\RulePropagation)^{\rew{}}_{\fS{f},1},
(\RulePropagation)^{\rew{}}_{\fS{f},2},
(\RulePropagation)^{\rew{}}_{\fS{g},1},
(\RuleHornClause)_{(\ref{ExHuet80_RemarkPage818_rule1})},
(\RuleRewMEq),
(\RuleReflexivity)^{\rewmodulos{}},
(\RuleCompatibility)^{\rewmodulo{}}
\} 
\]
see Figure \ref{FigExHuet80_RemarkPage818_Theory}.
\begin{figure}[t]
\begin{IEEEeqnarray*}{+r"Cl+x*}
(\forall x) && x\ceq x 
& $(\RuleReflexivity)^=$
\\
(\forall x,y) && x\ceq y\Rightarrow y\ceq x 
& $(\RuleSymmetry)^=$
\\
(\forall x,y,z) && x\ceq y\wedge y\ceq z\Rightarrow x\ceq z 
& $(\RuleTransitivity)^=$
\\
(\forall x_1,y_1,x_2) && x_1\ceq y_1\Rightarrow \fS{f}(x_1,x_2)\ceq \fS{f}(y_1,x_2) 
& $(\RulePropagation)^=_{\fS{f},1}$
\\
(\forall x_1,x_2,y_2) && x_2\ceq y_2\Rightarrow \fS{f}(x_1,x_2)\ceq \fS{f}(x_1,y_2) & $(\RulePropagation)^=_{\fS{f},2}$
\\
(\forall x_1,y_1) && x_1\ceq y_1\Rightarrow \fS{g}(x_1)\ceq \fS{g}(y_1) 
&  $(\RulePropagation)^=_{\fS{g},1}$
\\
&& \fS{a}\ceq\fS{b}
& $(\RuleHornClause)_{(\ref{ExHuet80_RemarkPage818_eq1})}$
\\[0.3cm]
(\forall x_1,y_1,x_2) && x_1\rew{} y_1\Rightarrow \fS{f}(x_1,x_2)\rew{} \fS{f}(y_1,x_2) 
&  $(\RulePropagation)^{\rew{}}_{\fS{f},1}$
\\
(\forall x_1,x_2,y_2) && x_2\rew{} y_2\Rightarrow \fS{f}(x_1,x_2)\rew{} \fS{f}(x_1,y_2)
&  $(\RulePropagation)^{\rew{}}_{\fS{f},2}$
\\
(\forall x_1,y_1) && x_1\rew{} y_1\Rightarrow \fS{g}(x_1)\rew{} \fS{g}(y_1)
&  $(\RulePropagation)^{\rew{}}_{\fS{g},1}$
\\
(\forall x) && \fS{f}(x,x)\rew{} \fS{g}(x)
& $(\RuleHornClause)_{(\ref{ExHuet80_RemarkPage818_rule1})}$
\\[0.3cm]
(\forall x,x',y,y') & & x \ceq x'\wedge x'\rew{} y' \wedge y'\ceq y\Rightarrow x \rewmodulo{} y & $(\RuleRewMEq)$
\\
(\forall x) && x\rewsmodulo{} x 
& $(\RuleReflexivity)^{\rewmodulos{}}$
\\
(\forall x,y,z) && x\rewmodulo{} y\wedge y\rewsmodulo{} z\Rightarrow x\rewsmodulo{} z 
& $(\RuleCompatibility)^{\rewmodulo{}}$ 
\end{IEEEeqnarray*}
\caption{Theory $\GLtheory_{\cR/E}$ for $\cR$ in Example \ref{ExHuet80_RemarkPage818}.}
\label{FigExHuet80_RemarkPage818_Theory}
\end{figure}
The expected sentences 
\begin{itemize}
\item $(\RuleReflexivity)^=$,
$(\RuleSymmetry)^=$,
$(\RuleTransitivity)^=$,
$(\RulePropagation)^=_{\fS{f},1}$,
$(\RulePropagation)^=_{\fS{f},2}$,
$(\RulePropagation)^=_{\fS{g},1}$, and
$(\RuleHornClause)_{(\ref{ExHuet80_RemarkPage818_eq1})}$
describing the \emph{equational theory};
\item 
$(\RulePropagation)^{\rew{}}_{\fS{f},1}$,
$(\RulePropagation)^{\rew{}}_{\fS{f},2}$,
$(\RulePropagation)^{\rew{}}_{\fS{g},1}$, and
$(\RuleHornClause)_{(\ref{ExHuet80_RemarkPage818_rule1})}$,
describing \emph{one-step term rewriting};
and 
\item $(\RuleRewMEq)$,
$(\RuleReflexivity)^{\rewmodulos{}}$,
$(\RuleCompatibility)^{\rewmodulo{}}$,
describing
one-step
and many-step rewriting modulo 
\end{itemize}
are obtained.
\end{exa}

\begin{rem}[Proving infeasibility as satisfiability]
As explained in \cite[Section 4.3]{GutLuc_AutomaticallyProvingAndDisprovingFeasibilityConditions_IJCAR20}, the infeasibility of an atom $\GLatom$ with variables $\vec{x}$
with respect to a first-order theory $\GLtheory$  (i.e., the absence of a substitution $\sigma$ such that $\deductionInThOf{\GLtheory}{\sigma(\GLatom)}$ holds),
can be 
proved by finding a model $\SInterpretation$ 
of $\GLtheory\cup\neg(\exists\vec{x})\GLatom$.
Such first-order models can often 
be automatically obtained by using model generators like 
\AGES{} \cite{GutLuc_AutomaticGenerationOfLogicalModelsWithAGES_CADE19} or \MaceFour{} \cite{McCune_Prove9andMace4_Unpublished10}.
\end{rem}

\begin{exa}
\label{ExPeakNoCPs_PStheory}
Consider the \etrs{} in Example \ref{ExPeakNoCPs}, i.e.,
\begin{IEEEeqnarray*}{+r:C:l+x*}
\fS{b} & =  & \fS{f}(\fS{a}) & \eqref{ExPeakNoCPs_eq1}\\
\fS{a} & =  & \fS{c} & \eqref{ExPeakNoCPs_eq2}
\\
\fS{c} & \to & \fS{d} & \eqref{ExPeakNoCPs_rule1}\\
\fS{b} & \to & \fS{d} & \eqref{ExPeakNoCPs_rule2}
\end{IEEEeqnarray*}
Since $H=H^\PSabbr=\emptyset$ and $R^\PSabbr=R$,  the theory
$\GLtheory_{\cR,E}$ 
is the union of
\[\begin{array}{rcll}
\GLtheory_{\text{Eq}}[\Symbols,\muTop,E] & = & \{
(\RuleReflexivity)^=,
(\RuleSymmetry)^=,
(\RuleTransitivity)^=,
(\RulePropagation)^=_{\fS{f},1},
(\RuleHornClause)_{(\ref{ExPeakNoCPs_eq1})},
(\RuleHornClause)_{(\ref{ExPeakNoCPs_eq2})}
\} 
& \qquad\qquad\quad\text{and}
\\
\GLtheory_{\text{R,M}}[\Symbols,\muTop,R] & = &
\{
(\RuleReflexivity)^{\rewspstickel{}},
(\RuleCompatibility)^{\rewpstickel{}},
(\RulePropagation)^{\rewpstickel{}}_{\fS{f},1},
(\RuleRlEq)_{(\ref{ExPeakNoCPs_rule1})},
(\RuleRlEq)_{(\ref{ExPeakNoCPs_rule2})}
\},
\end{array}
\]
\noindent
see Figure \ref{FigExPeakNoCPs_PSTheory}.
\begin{figure}[t]
\begin{IEEEeqnarray*}{+r"Cl+x*}
(\forall x) && x\ceq x
& $(\RuleReflexivity)^=$
\\
(\forall x,y) && x\ceq y\Rightarrow y\ceq x 
& $(\RuleSymmetry)^=$
\\
(\forall x,y,z) && x\ceq y\wedge y\ceq z\Rightarrow x\ceq z
& $(\RuleTransitivity)^=$
\\
(\forall x_1,x_2) && x_1\ceq x_2\Rightarrow \fS{f}(x_1)\ceq \fS{f}(x_2)
& $(\RulePropagation)^=_{\fS{f},1}$
\\
&& \fS{b}\ceq\fS{f}(\fS{a})
& $(\RuleHornClause)_{(\ref{ExPeakNoCPs_eq1})}$
\\
&& \fS{a}\ceq\fS{c}
& $(\RuleHornClause)_{(\ref{ExPeakNoCPs_eq2})}$
\\
(\forall x) && x\rewspstickel{} x
& $(\RuleReflexivity)^{\rewspstickel{}}$
\\
(\forall x,y,z) && x\rewpstickel{} y\wedge y\rewspstickel{} z\Rightarrow x\rewspstickel{} z
& $(\RuleCompatibility)^{\rewpstickel{}}$ 
\\
(\forall x_1,x_2) && x_1\rewpstickel{} x_2\Rightarrow \fS{f}(x_1)\rewpstickel{} \fS{f}(x_2)
&  $(\RulePropagation)^{\rewpstickel{}}_{\fS{f},1}$
\\
(\forall x) && x\ceq\fS{c}\Rightarrow x\rewpstickel{} \fS{d}
& $(\RuleRlEq)_{(\ref{ExPeakNoCPs_rule1})}$
\\
(\forall x) && x\ceq\fS{b}\Rightarrow x\rewpstickel{} \fS{d}
& $(\RuleRlEq)_{(\ref{ExPeakNoCPs_rule2})}$
\end{IEEEeqnarray*}
\caption{Theory $\GLtheory_{\cR,E}$ for $\cR$ in Example \ref{ExPeakNoCPs}.}
\label{FigExPeakNoCPs_PSTheory}
\end{figure}
In Example \ref{ExPeakNoCPs} we claim that $\fS{b}$ does not reduce into $\fS{f}(\fS{d})$ in one step of $\rew{\cR,E}$.
The following interpretation $\SInterpretation$ over the positive integers (computed by \AGES), with 
\[\begin{array}{r@{~}c@{~}l@{\hspace{1cm}}r@{~}c@{~}l@{\hspace{1cm}}r@{~}c@{~}l@{\hspace{1cm}}r@{~}c@{~}l}
\fS{a}^\SInterpretation & = & 1
& 
\fS{b}^\SInterpretation & = & 1
&
\fS{b}^\SInterpretation & = & 1
&
\fS{b}^\SInterpretation & = & 1
\\
\fS{f}^\SInterpretation(x) & = & 2
&
x\equequ{}^\SInterpretation y & \Leftrightarrow& true
&
x(\rewpstickel{})^\SInterpretation y & \Leftrightarrow & x\geq_\naturals y
&
x(\rewspstickel{})^\SInterpretation y & \Leftrightarrow & true
\end{array}
\]
satisfies 
$\GLtheory_{\cR,E}\cup\{\neg(\fS{b}\rewpstickel{}\fS{f}(\fS{d}))\}$, 
thus formally proving our claim.
\end{exa}
%
\noindent
When necessary, \AGES{} and \MaceFour{} are often useful to prove infeasibility of rules, conditional pairs, joinability of terms, etc., see \cite[Figures 4--6]{Lucas_ProvingSemanticPropertiesAsFirstOrderSatisfiability_AI19} for some examples of use of first-order satisfiability when dealing with rewriting-based systems.

\subsection{Horn theories as Elementary Inference Systems}
\label{SecHornTheoriesAsEIS}
In \cite{Lucas_SemanticPropertiesOfComputationsDefinedByElementaryInferenceSystems_HCVS25} it is observed that theories $\GLtheory$ consisting of sentences
\begin{IEEEeqnarray}{r'C'l}
(\forall\vec{x})B_1\wedge\cdots\wedge B_n\Rightarrow B
\label{LblGenericHornSentence}
\end{IEEEeqnarray}
 for atoms $B,B_1,\ldots,B_n$ and $\vec{x}$ the sequence of all variables occurring in these atoms, are implicative forms of 
\emph{definite Horn clauses} and can be
viewed as \emph{Elementary Inference Systems (\ElementaryIS{s})} $\GLinferenceOf{\GLtheory}$, \`a la Smullyan \cite{Smullyan_TheoryOfFormalSystems_1961}, consisting of inference rules 
\begin{IEEEeqnarray}{r'C'l}
\bigfracLabel{B_1\quad\cdots\quad B_n}{B}\label{LblGenericElementaryInferenceRules}
\end{IEEEeqnarray} 
for each sentence (\ref{LblGenericHornSentence}), and vice versa.
Roughly speaking, as in \emph{Natural Deduction} \cite{Gentzen_UntersuchungenUberDasLogischeSchliessenPart1_MZ35,Gentzen_InvestigationsIntoLogicalDeduction_SLFM69,Prawitz_NaturalDeductionAProofTheoreticalStudy_1965},
given an \eis{} $\GLinference$ consisting of rules (\ref{LblGenericElementaryInferenceRules}), we write 
$\proofInISof{\GLinference}{A}$ (and say that $A$ is \emph{provable} in $\GLinference$) if 
(i) $A=\sigma(B)$ for some inference rule (\ref{LblGenericElementaryInferenceRules}) 
and substitution $\sigma$, and 
(ii) for all $1\leq i\leq n$, $\proofInISof{\GLinference}{\sigma(B_i)}$ holds. 
In this way a \emph{proof tree} is built to witness provability of $A$ in $\GLinference$  \cite[Section 3.1]{Lucas_SemanticPropertiesOfComputationsDefinedByElementaryInferenceSystems_HCVS25}.
Proving \emph{atoms} in \eis{s} is \emph{stable} under substitution application.
\begin{propC}[\emph{\cite[Proposition 9]{Lucas_SemanticPropertiesOfComputationsDefinedByElementaryInferenceSystems_HCVS25}}]\label{PropProvingAtomsIsStable}
Let $\GLinference$ be an \eis, $A$ be an atom, and $\sigma$ be a substitution. If 
$\proofInISof{\GLinference}{A}$,
then $\proofInISof{\GLinference}{\sigma(A)}$.
\end{propC}
%
\noindent
Atoms $A$ are deducible from a \emph{Horn theory} $\GLtheory$,
i.e., $\deductionInThOf{\GLtheory}{A}$, 
iff  $A$ is provable in 
$\GLinferenceOf{\GLtheory}$.
\begin{propC}[\emph{\cite[Proposition 14]{Lucas_SemanticPropertiesOfComputationsDefinedByElementaryInferenceSystems_HCVS25}}]\label{PropDeductionAndProvabilityOfAtoms}
Let $\GLtheory$ be a Horn theory and $A$ be an atom with variables $\vec{x}$. 
Then, $\deductionInThOf{\GLtheory}{(\forall\vec{x})A}$
if and only if
$\proofInISof{\GLinferenceOf{\GLtheory}}{A}$.
\end{propC}

\noindent
While provability (building proof trees) has a more operational flavour, the use of 
$\GLtheory_{\cR}$, $\GLtheory_{\cR,E}$, etc., is important for other purposes (e.g., infeasibility of rules and conditional pairs, proofs of $E$-termination as the generation of well-founded models, etc.; see Section~\ref{SecETerminationOfEGTRSs}).

\subsection{Computation with \egtrs{s}}

The theories (or corresponding \ElementaryIS{s}) above are used to define the expected computational relations as follows.

\begin{defi}
\label{DefComputationalRelationsOfAnEGTRS}
Let $\cR=(\Symbols,\SPredicates,\mu,E,\mu,H,R)$ be an \egtrs{} and $s,t\in\Terms$. We write 
\begin{enumerate}
\item\label{DefComputationalRelationsOfAnEGTRS_equivalence}
$s\equequ{E}t$  
(resp.\ $s\rew{\eqoriented{E}}t$) 
if
$\deductionInThOf{\GLtheory_{E}}{s\equequ{}t}$
(resp.\ $\deductionInThOf{\left (\GLtheory_{E}\cup\GLtheory_{R}[\Symbols,\mu,\eqoriented{E}]\right )}
{s\rew{}t}$)
holds.
\item \label{DefComputationalRelationsOfAnEGTRS_rewriting}
$s\rew{\cR}t$ if $\deductionInThOf{\GLtheory_{\cR}}{s\rew{}t}$
\item\label{DefComputationalRelationsOfAnEGTRS_rewPStickel} 
$s\rew{\cR,E}t$ if $\deductionInThOf{\GLtheory_{\cR,E}}{s\rewpstickel{}t}$.
\item\label{DefComputationalRelationsOfAnEGTRS_RewritingModulo}
  $s\rew{\cR/E}t$ if $\deductionInThOf{\GLtheory_{\cR/E}}{s\rewmodulo{}t}$%
\end{enumerate}
Similarly for $s\rews{\cR}t$ (resp.\ $s\rews{\cR,E}t$ and $s\rews{\cR/E}t$).
Moreover, since $\GLtheory_{E}$, $\GLtheory_{\cR}$,  etc.,
are Horn theories, by Proposition \ref{PropDeductionAndProvabilityOfAtoms},
provability in $\GLinferenceOf{\GLtheory_{E}}$, $\GLinferenceOf{\GLtheory_{\cR}}$,  etc.,
can be used
instead of deduction 
above.
\end{defi}
%
\noindent
Note that 
(i) $\rew{\eqoriented{E}}$ is \emph{symmetric} by definition of $\eqoriented{E}$;
(ii) $\equequ{E}$ is an equivalence due to 
$(\RuleReflexivity)^{\equequ{}}$ (reflexivity),
$(\RuleSymmetry)^{\equequ{}}$ (symmetry), and
$(\RuleTransitivity)^{\equequ{}}$ (transitivity),
all 
included in $\GLtheory_{E}$; 
and
(iii) $\equequ{E}$ is the reflexive and transitive closure of $\rew{\eqoriented{E}}$.
\begin{defi}[Confluence and termination modulo of an \egtrs{}]
\label{DefConfluenceAndTerminationModuloOfEGTRSs}
We say that an \egtrs{} $\cR=(\Symbols,\SPredicates,\mu,E,H,R)$
is 
\begin{itemize}
\item 
\emph{confluent modulo $E$}
(or  \emph{$E$-confluent}) if for all terms
$s$, $t_1$, and $t_2$, if 
$s\rews{\cR/E}t_1$ 
and
$s\rews{\cR/E}t_2$,
there are $u_1$ and $u_2$ such that
$t_1\rews{\cR/E}u_1$,
$t_2\rews{\cR/E}u_2$ 
and $u_1\equequ{E} u_2$; and
\item
\emph{terminating modulo $E$} (or \emph{$E$-terminating}) if
$\rew{\cR/E}$ is terminating.
\end{itemize}
\end{defi}
\begin{rem}[Necessary conditions for $E$-termination]
\label{RemANecessaryConditionForETermination}
Regarding \etrs{s} $\cR=(\Symbols,E,R)$, in 
\cite[page 1169]{JouKir_CompletionOfASetOfRulesModuloASetOfEquations_SIAMJC86},
Jouannaud and Kirchner noticed that:
\begin{enumerate}
\item if $E$ contains 
equations $s=t$ such that $\Var(s)\neq\Var(t)$ (so that there is a variable $x$ occurring in $s$ but not occurring in $t$, or vice versa), 
then $E$-termination is not possible. 
For this reason, \emph{non-erasingness} is assumed for the equational component $E$ of an \etrs{}, i.e., for all $s\ceq t\in E$, $\Var(s)=\Var(t)$.\footnote{In the realm of unification theory, this condition was called \emph{regularity}
\cite[page 243]{Siekmann_UnificationTheory_JSC89}. This name is also used in recent papers investigating confluence modulo, e.g.,
\cite[Section 2]{DurMes_OnTheChurchRosserAndCoherencePropertiesOfConditionalOrderSortedRewriteTheories_JLAP12} and
\cite[Section 2.1]{Meseguer_StrictCoherenceOfConditionalRewritingModuloAxioms_TCS17}.
}
This is also assumed by Huet \cite[page 816, bottom]{Huet_ConfluentReductionsAbstractPropertiesAndApplicationsToTermRewritingSystems_JACM80}. Although the requirement is not justified, the first sentence of page 819, discussing  
proofs of $E$-termination, makes sense only if such a condition holds.
\item Equations of the form $x=t$ where $t$ contains at least two occurrences of $x$ are also incompatible with $E$-termination.
Thus, they are disallowed.
\end{enumerate}
\end{rem}

\subsection{Proving $E$-termination of \egtrs{s}}
\label{SecETerminationOfEGTRSs}
According to \cite{LucGut_AutomaticSynthesisOfLogicalModelsForOrderSortedFirstOrderTheories_JAR18,Lucas_UsingWellFoundedRelationsForProvingOperationalTermination_JAR20}, $E$-termination of \egtrs{s} $\cR$ can be proved by 
showing a first-order model $\SInterpretation$ of $\GLtheory_{\cR/E}$ 
where the interpretation $(\rmodulo)^\SInterpretation$ of 
predicate $\rmodulo$ denoting $\rew{\cR/E}$ is \emph{well-founded}.
Again, \AGES{} or \MaceFour{} can be used to obtain 
such first-order models.
As explained in \cite[page 293]{GutLuc_AutomaticGenerationOfLogicalModelsWithAGES_CADE19} it is possible to instruct \AGES{} to guarantee that the interpretation $P^\SInterpretation$ of 
a binary predicate $P$ is a \emph{well-founded} relation on the interpretation domain.
This is not possible with \MaceFour; however, as explained in \cite[Section 5.5.1]{Lucas_UsingWellFoundedRelationsForProvingOperationalTermination_JAR20}, it is possible to force the tool to do the same.
 
\begin{exa}
\label{ExHuet80_RemarkPage818_ETermination}
For $\cR$ in Example \ref{ExHuet80_RemarkPage818},
with $E=\{\fS{a}=\fS{b}\}$ and $R=\{\fS{f}(x,x)\to\fS{g}(x)\}$,
the interpretation $\SInterpretation$ over the naturals (obtained with \AGES{}), with 
\[\begin{array}{r@{~}c@{~}l@{\hspace{0.5cm}}r@{~}c@{~}l@{\hspace{0.5cm}}r@{~}c@{~}l@{\hspace{0.5cm}}r@{~}c@{~}l}
\fS{a}^\SInterpretation & = & 1
& 
\fS{b}^\SInterpretation & = & 1
&
\fS{f}^\SInterpretation(x,y) & = & x+2y+1
&
\fS{g}^\SInterpretation(x) & = & 3x
\\[0.1cm]
x\equequ{}^\SInterpretation y & \Leftrightarrow& x=_\naturals y
&
x\rew{}^\SInterpretation y & \Leftrightarrow & x>_\naturals y
&
x(\rewmodulo){\!\!}^\SInterpretation y & \Leftrightarrow & x>_\naturals y
&
x(\rewsmodulo){\!}^\SInterpretation y & \Leftrightarrow & x\geq_\naturals y
\end{array}
\]
is a model of $\GLtheory_{\cR/E}$ in Example \ref{ExHuet80_RemarkPage818_TheoryRModulo}. 
Since $(\rmodulo)^\SInterpretation$ is well-founded, 
$\cR$ is $E$-terminating.
\end{exa}

\subsection{Use of Jouannaud \& Kirchner's abstract framework with \egtrs{s}}

Unfortunately, the relationship between $\equequ{E}$, $\rew{\cR}$, and $\rew{\cR/E}$, 
is \emph{not} as required in the abstract framework presented in Section \ref{SecAbstractFrameworkJK86}.
In particular, $\rew{\cR/E}\:=\:(\composeRel{\equequ{E}}{\composeRel{\rew{\cR}}{\equequ{E}}})$ does \emph{not} hold.

\begin{exa}
\label{ExRModuloAndRrelativeE}
Consider the following \egtrs{}
\begin{IEEEeqnarray}{r'C'l}
\fS{a} & = & \fS{b}\label{ExRModuloAndRrelativeE_eq1}\\
\fS{a} & \to & \fS{c}\label{ExRModuloAndRrelativeE_rule1}\\
\fS{a} & \to & \fS{d}\IF \fS{b} \rews{} \fS{c}\label{ExRModuloAndRrelativeE_rule2}
\\
\fS{c} & \to & \fS{d}\label{ExRModuloAndRrelativeE_rule3}
\end{IEEEeqnarray}
We have $\fS{a}\rew{\cR}\fS{c}$; but (\ref{ExRModuloAndRrelativeE_rule2}) is 
$\GLtheory_{\cR}$-infeasible, hence $\fS{a}\not\rew{\cR}\fS{d}$. 
However, $\fS{a}\rew{\cR/E}\fS{d}$, as $\fS{b}\equequ{E}\fS{a}\rew{\cR}\fS{c}$, i.e.,
$\fS{b}\rew{\cR/E}\fS{c}$ and (\ref{ExRModuloAndRrelativeE_rule2}) can be used. Thus,
\begin{eqnarray*}
\rew{\cR}  & = &  \{(\fS{a},\fS{c}),(\fS{c},\fS{d})\}
\\
{(\equequ{E}\circ\rew{\cR}\circ\equequ{E})} & = & \{(\fS{a},\fS{c}),(\fS{b},\fS{c}),(\fS{c},\fS{d})\}
\\
\rew{\cR/E} & = & \{(\fS{a},\fS{c}),(\fS{b},\fS{c}),(\fS{a},\fS{d}),(\fS{b},\fS{d}),(\fS{c},\fS{d})\}
\end{eqnarray*}
Therefore,
$\rew{\cR/E}\:\neq(\equequ{E}\circ\rew{\cR}\circ\equequ{E})$.
\end{exa} 
\subsubsection{Rewriting modulo and rewriting in conditional systems}
\label{SecRewritingModuloAndRewritingInConditionalSystems}
Example \ref{ExRModuloAndRrelativeE} shows a \emph{mismatch} between the definition of
$\rew{\cR/E}$ for an \egtrs{} $\cR$ (Definition \ref{DefComputationalRelationsOfAnEGTRS}(\ref{DefComputationalRelationsOfAnEGTRS_RewritingModulo})) 
and the \emph{abstract} definition (\ref{LblReductionModuloFromReductionAndEquivalence}), usually understood for ETRSs.
For \egtrs{s} (and already for CTRSs), the connection 
$\rew{\cR/E}\:=\:(\equequ{E}\circ\rew{\cR}\circ\equequ{E})$ 
is \emph{lost}. 
This is because rewriting conditions in rules used to obtain $\rew{\cR}$-steps 
are treated using
$\rews{\cR}$ 
(see, e.g., \cite[Definition 7.1.4]{Ohlebusch_AdvancedTopicsInTermRewriting_2002}).
In contrast, if $\rew{\cR/E}$-steps are desired, such rewriting conditions are
evaluated using $\rews{\cR/E}$, see, e.g., \cite[page 819]{DurMes_OnTheChurchRosserAndCoherencePropertiesOfConditionalOrderSortedRewriteTheories_JLAP12}.\footnote{As remarked by a reviewer, alternative definitions of $\rew{\genrelation/\genequivalence}$ as
$\genequivalence\circ\genrelation$, see, e.g., \cite[page 277]{Jouannaud_ConfluenceOfTerminatingRewritingComputations_TFSP24}
under the name \emph{class rewriting} (see also Section~\ref{SecJouannaud2024}) do \emph{not} solve the problem.}
We overcome this problem as follows:
given an \egtrs{} $\cR=(\Symbols,\SPredicates,\mu,E,H,R)$ such that $E$ does not depend on $R$, we let
\[
\cR^\RMabbr = (\Symbols,\SPredicates\uplus\{\rewmodulo{},\rewsmodulo{}\},\mu,E,H^\RMabbr,R^\RMabbr)
\]
\begin{enumerate}
\item\label{RemRewritingModuloAndRewritingInConditionalSystems_rewriting}
 Instead of $\rew{\cR}$, a new one step rewriting relation $\rew{\cR^\RMabbr}$ for the \egtrs{} $\cR^\RMabbr$ above is used.
 Although (by abuse) we use the same notation, we \emph{cannot} use 
 \[\GLtheory_{\cR^\RMabbr}=\GLtheory_{\text{Eq}}[\Symbols,\mu,E]
\cup\GLtheory_{\text{R}}[\Symbols,\mu,R^\RMabbr]
\cup\{(\RuleHornClause)_\alpha\mid\alpha\in H^\RMabbr\}
\] 
as given in Definition \ref{DefTheoriesOfAnEGTRS} 
to define $\rew{\cR^\RMabbr}$ as in Definition \ref{DefComputationalRelationsOfAnEGTRS}(\ref{DefComputationalRelationsOfAnEGTRS_rewriting}).
This is because $\rewmodulo{}$ and $\rewsmodulo{}$, which may occur in $H^\RMabbr$ and $R^\RMabbr$, 
are \emph{not} defined in $\GLtheory_{\cR^\RMabbr}$.
We also need  
$\GLtheory_{\text{R/M}}[\Symbols,\mu,R^\RMabbr]$ to give meaning to 
$\rewmodulo{}$ and $\rewmodulos{}$.
Overall, we need $\GLtheory_{\cR^\RMabbr}\cup\GLtheory_{\text{R/M}}[\Symbols,\mu,R^\RMabbr]$, i.e.,
\begin{IEEEeqnarray}{r'C'l}
\GLtheory_{\text{Eq}}[\Symbols,\mu,E]
\cup\GLtheory_{\text{R}}[\Symbols,\mu,R^\RMabbr]
\cup\{(\RuleHornClause)_\alpha\mid\alpha\in H^\RMabbr\}
\cup\GLtheory_{\text{R/M}}[\Symbols,\mu,R^\RMabbr]
\end{IEEEeqnarray}

\item\label{RemRewritingModuloAndRewritingInConditionalSystems_PSrewriting}
Instead of $\rew{\cR,E}$, a one step rewriting relation $\rew{\cR^\RMabbr,E}$ is used.
Analogously, we need to use 
\begin{IEEEeqnarray}{r'C'l}
\GLtheory_{\text{Eq}}[\Symbols,\mu,E]
\cup\GLtheory_{\text{R,M}}[\Symbols,\mu,R^\RMabbr]
\cup\:\{(\RuleHornClause)_\alpha\mid\alpha\in H^\RMabbr\}
\cup\:\GLtheory_{\text{R/M}}[\Symbols,\mu,R^\RMabbr]
\label{RemRewritingModuloAndRewritingInConditionalSystems_PSrewritingTh}
\end{IEEEeqnarray}

\item \label{RemRewritingModuloAndRewritingInConditionalSystems_rewritingModulo}
We do \emph{not} need to define any $\rew{\cR^\RMabbr/E}$ because
$\rew{\cR/E}$ in Definition \ref{DefComputationalRelationsOfAnEGTRS}(\ref{DefComputationalRelationsOfAnEGTRS_RewritingModulo})
is defined as first-order deduction on $\GLtheory_{\cR/E}$ in Definition \ref{DefTheoriesOfAnEGTRS}, which already uses $H^\RMabbr$ and $R^\RMabbr$.
\end{enumerate}

\subsubsection{First-order theory for the analysis of $E$-confluence of \egtrs{s}}

According to the previous discussion,
we define the theory to be used to analyze $E$-confluence of \egtrs{s}.
\begin{defi}[CR-theory of an \egtrs{}]
\label{DefCRtheoryOfAnEGTRS}
Let $\cR=(\Symbols,\SPredicates,\mu,E,H,R)$ be an \egtrs{} such that $E$ does not depend on $R$. The CR-theory 
$\crtheoryOf{\cR}$ of $\cR$ is
\[\begin{array}{r@{\:}c@{\:}l}
\crtheoryOf{\cR} & = &
\GLtheory_{\cR^\RMabbr}
\cup\GLtheory_{\text{R,M}}[\Symbols,\mu,R^\RMabbr]
\cup\GLtheory_{\text{R/M}}[\Symbols,\mu,R^\RMabbr]
\\
& = & \GLtheory_{\text{Eq}}[\Symbols,\mu,E]
\cup\GLtheory_{\text{R}}[\Symbols,\mu,R^\RMabbr]
\cup\GLtheory_{\text{R,M}}[\Symbols,\mu,R^\RMabbr]
\cup\GLtheory_{\text{R/M}}[\Symbols,\mu,R^\RMabbr]
\\
&& 
\cup\:\{(\RuleHornClause)_\alpha\mid\alpha\in H^\RMabbr\}
\end{array}
\]
\end{defi}
\begin{rem}[\egtrs{s} with $E$ not depending on $R$]
Note that $\GLtheory_{\text{Eq}}[\Symbols,\mu,E]$ is a component of $\crtheoryOf{\cR}$ and $\GLtheory_{E}$ is 
$\GLtheory_{\text{Eq}}[\Symbols,\mu,E]\cup\{(\RuleHornClause)_\alpha\mid\alpha\in H\}$.
Since $\GLtheory_{E}$ defines the equivalence $\equequ{E}$ (see Definition \ref{DefComputationalRelationsOfAnEGTRS}(\ref{DefComputationalRelationsOfAnEGTRS_equivalence})), requiring $E$ not to depend on $R$ is important, as it guarantees that every Horn clause $\alpha\in H$ used to prove $s\equequ{E}t$ using $\GLtheory_{E}$ coincides with $\alpha^\RMabbr\in H^\RMabbr$, and therefore $\GLtheory_{E}\subseteq\GLtheory_{\cR^\RMabbr}
\subseteq\crtheoryOf{\cR}$.
Thus we obtain a proof of $s\equequ{E}t$ in $\crtheoryOf{\cR}$ as well
(and vice versa).
For the sake of readability, in the statement of the next definitions and 
results, unless stated otherwise we silently assume in all \egtrs{s}
that $E$ does not depend on $R$.
\end{rem}
%
\noindent
Note that
$\GLtheory_{\cR/E}\subseteq\crtheoryOf{\cR}$, but $\crtheoryOf{\cR}$ adds nothing to the definition of $\rewmodulo{}$ and $\rewmodulos{}$ given by $\GLtheory_{\cR/E}$.
Thus, we have:

\begin{prop}
Let $\cR=(\Symbols,\SPredicates,\mu,E,H,R)$ be an \egtrs{}.
Let $s$ and $t$ be terms. Then,
\begin{itemize}
\item $\deductionInThOf{\GLtheory_{E}}{s\ceq{}t}$ if and only if 
$\deductionInThOf{\crtheoryOf{\cR}}{s\ceq{}t}$.
\item 
$\deductionInThOf{\GLtheory_{\cR/E}}{s\rewmodulo{}t}$ if and only if 
$\deductionInThOf{\crtheoryOf{\cR}}{s\rewmodulo{}t}$.
\end{itemize}
\end{prop}
%
\noindent
Now, $\rew{\cR^\RMabbr}$ and 
$\rew{\cR^\RMabbr,E}$ are defined as follows:
\begin{defi}
Let $\cR=(\Symbols,\SPredicates,\mu,E,H,R)$ be an \egtrs{}.
Let $s$ and $t$ be terms. We write
\begin{enumerate}
\item $s \rew{\cR^\RMabbr}t$ (resp.\ $s \rews{\cR^\RMabbr}t$) if 
$\deductionInThOf{\crtheoryOf{\cR}}{s\rew{}t}$ (resp.\ 
$\deductionInThOf{\crtheoryOf{\cR}}{s\rews{}t}$).
\item
$s \rew{\cR^\RMabbr,E}t$ (resp.\ $s \rews{\cR^\RMabbr,E}t$) if 
$\deductionInThOf{\crtheoryOf{\cR}}{s\rewpstickel{}t}$ (resp.\ 
$\deductionInThOf{\crtheoryOf{\cR}}{s\rewspstickel{}t}$).
\end{enumerate}
\end{defi}
\begin{rem}[Infeasibility with \egtrs{s}]
If no confusion arises,  in the following we use \emph{(in)feasible} instead of 
\emph{$\crtheoryOf{\cR}$-(in)feasible}, i.e., theory $\crtheoryOf{\cR}$ is used for any statement about (in)feasibility of sequences of atoms.
\end{rem}
%
\noindent
The following result is essential in the analysis of peaks in Section
\ref{SecReductionAndCoherencePeaks}.

\begin{prop}
\label{PropUseOfRulesAndPositionsInRewritingSteps_EGTRSs}
Let $\cR=(\Symbols,\SPredicates,\mu,E,H,R)$ be an \egtrs{} and $s,t\in\Terms$.
\begin{enumerate}
\item\label{PropUseOfRulesAndPositionsInRewritingSteps_EGTRSs_R}
 If $s\rew{\cR^\RMabbr} t$, then there is $p\in\Pos^\mu(s)$ and $\lhsr\to\rhsr\IF\gencond\in R^\RMabbr$ such that
(i) $s|_p=\sigma(\lhsr)$ for some substitution $\sigma$, 
(ii) for all $A\in\gencond$, $\crtheoryOf{\cR}\vdash\sigma(A)$ holds,
and 
(iii) $t=s[\sigma(\rhsr)]_p$.
\item\label{PropUseOfRulesAndPositionsInRewritingSteps_EGTRSs_PStickel}
 If $s\to_{\cR^\RMabbr,E} t$, then there is $p\in\Pos^\mu(s)$ and $\lhsr\to\rhsr\IF\gencond\in R^\RMabbr$ such that
(i) $s|_p\equequ{E}\sigma(\lhsr)$ for some substitution $\sigma$, 
(ii) for all $A\in\gencond$, $\crtheoryOf{\cR}\vdash\sigma(A)$ holds,
and 
(iii) $t=s[\sigma(\rhsr)]_p$.
\end{enumerate}
\end{prop}

\begin{proof}
Since $s\rew{\cR^\RMabbr}t$ iff $\deductionInThOf{\crtheoryOf{\cR}}{s\rew{} t}$, we have 
two possibilities:
\begin{enumerate}
\item There is a clause $(\RuleHornClause)_\alpha$
for some rule $\alpha:\lhsr\to\rhsr\IF \gencond\in R^\RMabbr$ 
and a substitution $\sigma$ such that $\deductionInThOf{\crtheoryOf{\cR}}{\sigma(A)}$ holds
for all $A\in\gencond$, and then $s=\sigma(\lhsr)$ and $t=\sigma(\rhsr)$.
In this case, we just need to take $p=\toppos$ to obtain the desired conclusion.
\item There is a clause $(\RulePropagation)_{f,i}$ for some $f\in\Symbols$ and
$i\in\mu(f)$
so that 
(a) $s=f(s_1,\ldots,s_i,\ldots,s_k)$ for some terms $s_1,\ldots,s_k$, 
(b) $t=f(s_1,\ldots,t_i,\ldots,s_k)$,
and $\deductionInThOf{\crtheoryOf{\cR}}{s_i\rew{}t_i}$ holds.
By the induction hypothesis, $s_i\rew{\cR^\RMabbr} t_i$
holds and
there is $q\in\Pos^\mu(s_i)$, $\lhsr\to\rhsr\IF\gencond\in\cR$ and
substitution $\sigma$ such that
$s_i|_q=\sigma(\lhsr)$, $t_i=s_i[\sigma(\rhsr)]_q$, and each atom in $\sigma(\gencond)$ is deducible from $\crtheoryOf{\cR}$.
Since $s|_{i.q}=\sigma(\lhsr)$ and $t=s[\sigma(\rhsr)]_{i.q}$, the conclusion follows by letting
$p=i.q$, which belongs to $\Pos^\mu(s)$ by definition of $\Pos^\mu$.
\end{enumerate}
The case $s\rew{\cR^\RMabbr,E} t$ is similar.
\end{proof}

\subsubsection{Fitting Jouannaud and Kirchner's framework}

Consider the correspondence in Table \ref{TableARSsAndEGTRSs}.
\begin{table}[ht]
\caption{Abstract notions in Section~\ref{SecBasicNotionsJK86} applied to \egtrs{s}}
\label{TableARSsAndEGTRSs}
\centering
\begin{tabular}{@{}l*{5}{c}@{}}
\toprule
Abstract reduction:
& $\equone{\genequivalence}$
& $\equ{\genequivalence}$
& $\rew{\genrelation}$
& $\rew{\genrelationUpE}$
& $\rew{\genrelation/\genequivalence}$
\\
\midrule
Application to \egtrs{s}:
& $\rew{\eqoriented{E}}$
& $=_E$
& $\rew{\cR^\RMabbr}$
& $\rew{\cR^\RMabbr,E}$
& $\rew{\cR/E}$
\\
\bottomrule
\end{tabular}
\end{table}
We prove that
(\ref{LblReductionModuloFromReductionAndEquivalence}) and 
(\ref{LblFundamentalAssumptionJK86}) in Jouannaud and Kirchner's framework are fulfilled now.

\begin{prop}
\label{PropRewritingModuloEGRSsAsComposition}
\label{PropFundamentalAssumptionForEGRSs}
Let $\cR=(\Symbols,\SPredicates,\mu,E,H,R)$ be an \egtrs{}.
Then, 
\[\rew{\cR/E}\:=\:(\composeRel{\equequ{E}}{\composeRel{\rew{\cR^\RMabbr}}{\equequ{E}}})
\text{\qquad and  \qquad}
\rew{\cR^\RMabbr}\:\subseteq\:\rew{\cR^\RMabbr,E}\:\subseteq\:\rew{\cR/E}.\]
\end{prop}
\begin{proof}
As for the first claim, we have $s\rew{\cR/E}t$ iff $\deductionInThOf{\crtheoryOf{\cR}}{s\rewmodulo{}t}$ iff
(i) $\deductionInThOf{\crtheoryOf{\cR}}{s\equequ{} s'}$ holds for some term $s'$,
(ii) $\deductionInThOf{\crtheoryOf{\cR}}{s'\rew{} t'}$ holds for some term $t'$,
and 
(iii) $\deductionInThOf{\crtheoryOf{\cR}}{t'\equequ{} t}$ holds; 
then   $(\RuleRewMEq)$ applies to conclude  $s\rewmodulo{}t$.
Thus, we have $s\equequ{E}s'$, $s'\rew{\cR^\RMabbr}t'$, and $t'\equequ{E}t$, i.e.,
$s\composeRel{\equequ{E}}{\composeRel{\rew{\cR^\RMabbr}}{\equequ{E}}}t$.

As for the second claim, we silently assume the equivalence between deduction and provability of atoms in Proposition \ref{PropDeductionAndProvabilityOfAtoms} to proceed by induction on the structure of the proof tree. 
In order to prove the first inclusion, 
since $s\rew{\cR^\RMabbr}t$ iff $\deductionInThOf{\crtheoryOf{\cR}}{s\rew{} t}$, we have 
two possibilities:
(i) there is a rule $\alpha:\lhsr\to\rhsr\IF\gencond\in R^\RMabbr$
and a substitution $\sigma$ such that $\deductionInThOf{\crtheoryOf{\cR}}{\sigma(A)}$ holds for all $A\in \gencond$, and then $s=\sigma(\lhsr)$ and $t=\sigma(\rhsr)$ due to the use of 
$(\RuleHornClause)_\alpha\in\GLtheory_{\text{R}}[\Symbols,\mu,R^\RMabbr]\subseteq\crtheoryOf{\cR}$.
In this case, we can use 
$(\RuleRlEq)_\alpha\in\GLtheory_{\text{R,E}}[\Symbols,R^\RMabbr]\subseteq\crtheoryOf{\cR}$
by just letting $\sigma(x)=\sigma(\lhsr)$ for $x$ the
fresh variable in $(\RuleRlEq)_\alpha$ and using $(\RuleReflexivity)^=\in\GLtheory_{\text{Eq}}[\Symbols,E]\subseteq\crtheoryOf{\cR}$ to conclude $\sigma(x)\equequ{E}\sigma(\lhsr)$, as 
required to conclude $s\rewpstickel{\cR^\RMabbr,E}t$.
The second case is (ii) the use of 
$(\RulePropagation)^{\rew{}}_{f,i}\in\GLtheory_{\text{R}}[\Symbols,R^\RMabbr]\subseteq\crtheoryOf{\cR}$ for some $k$-ary
symbol $f\in\Symbols$ and $i\in\mu(f)$ 
such that $s=f(s_1,\ldots,s_i,\ldots,s_k)$ for some
terms $s_1,\ldots,s_k$, $t=f(s_1,\ldots,t_i,\ldots,s_k)$ for some term $t_i$ and 
$\deductionInThOf{\crtheoryOf{\cR}}{s_i\rew{}t_i}$ holds.
By an inductive argument, we can assume that
$\deductionInThOf{\crtheoryOf{\cR}}{s_i\rewpstickel{}t_i}$ holds.
Then, by using
$(\RulePropagation)^{\rewpstickel{}}_{f,i}\in\GLtheory_{\text{R,E}}[\Symbols,\mu,R^\RMabbr]\subseteq\crtheoryOf{\cR}$ we conclude
$\deductionInThOf{\crtheoryOf{\cR}}{s\rewpstickel{}t}$.

For the inclusion 
$\rew{\cR^\RMabbr,E}\spcrel{\subseteq}\rew{\cR/E}$, we proceed similarly. Since $s\rew{\cR^\RMabbr,E}t$ iff $\deductionInThOf{\crtheoryOf{\cR}}{s\rewpstickel{} t}$, we have 
two possibilities:
(i) $(\RuleRlEq)_\alpha$ is used, i.e., 
there is a rule $\alpha:\lhsr\to\rhsr\IF\gencond\in R^\RMabbr$
and a substitution $\sigma$ such that 
(i.1) $\deductionInThOf{\crtheoryOf{\cR}}{\sigma(x)=\sigma(\lhsr)}$ for $x$ the
fresh variable in $(\RuleRlEq)_\alpha$,
(i.2) $\deductionInThOf{\crtheoryOf{\cR}}{\sigma(A)}$ holds for all $A\in \gencond$, 
and then 
(i.3) $s=\sigma(x)\equequ{E}\sigma(\lhsr)$ and $t=\sigma(\rhsr)$ to conclude $s\rew{\cR^\RMabbr,E}t$.
By (i.2), we know that $\deductionInThOf{\crtheoryOf{\cR}}{\sigma(\lhsr)\rew{}\sigma(\rhsr)}$
by using $(\RuleHornClause)_\alpha$,
i.e., $\sigma(\lhsr)\rew{\cR^\RMabbr}\sigma(\rhsr)$ holds.
Thus, we have $s\equequ{E}\sigma(\lhsr)\rew{\cR^\RMabbr}\sigma(\rhsr)\equequ{E}t$ and,
 by the first (already proved) claim of this proposition, $s\rew{\cR/E}t$.
The second case is (ii) the use of 
$(\RulePropagation)^{\rewpstickel{}}_{f,i}\in\GLtheory_{\text{R}}[\Symbols,\mu,R^\RMabbr]\subseteq\crtheoryOf{\cR}$ for some $k$-ary
symbol $f\in\Symbols$ and $i\in\mu(f)$ 
such that $s=f(s_1,\ldots,s_i,\ldots,s_k)$ for some
terms $s_1,\ldots,s_k$, $t=f(s_1,\ldots,t_i,\ldots,s_k)$ for some term $t_i$ and 
$\deductionInThOf{\crtheoryOf{\cR}}{s_i\rewpstickel{}t_i}$ holds.
By an inductive argument, we can assume that $s_i\rew{\cR/E}t_i$ and by using the first 
claim of the proposition, $s_i\equequ{E}s'_i\rew{\cR^\RMabbr}t'_i\equequ{E}t_i$ for some terms
$s'_i$ and $t'_i$.
By using $(\RulePropagation)^=_{f,i}\in\GLtheory_{\text{Eq}}[\Symbols,\mu,R^\RMabbr]\subseteq\crtheoryOf{\cR}$, we conclude 
$s=f(s_1,\ldots,s_i,\ldots,s_k)\equequ{E}f(s_1,\ldots,s'_i,\ldots,s_k)$
and
$f(s_1,\ldots,t'_i,\ldots,s_k)\equequ{E}f(s_1,\ldots,t_i,\ldots,s_k)=t$.
Now, by $(\RulePropagation)^{\rew{}}_{f,i}$, we have
$f(s_1,\ldots,s'_i,\ldots,s_k)\rew{\cR^\RMabbr}f(s_1,\ldots,t'_i,\ldots,s_k)$.
Thus, 
$s\equequ{E}f(s_1,\ldots,s'_i,\ldots,s_k)\rew{\cR^\RMabbr}f(s_1,\ldots,t'_i,\ldots,s_k)\equequ{E}t$, i.e., $s\rew{\cR/E}t$.
\end{proof}

\subsubsection{Practical use of Jouannaud and Kirchner's framework with \egtrs{s}}
By  the first statement in Proposition \ref{PropRewritingModuloEGRSsAsComposition},
$E$-confluence of \egtrs{s} (Definition \ref{DefConfluenceAndTerminationModuloOfEGTRSs}) and $E$-confluence of $\rew{\cR^\RMabbr}$ (as an abstract relation on terms, 
Definition \ref{DefConfluenceAndTerminationModulo}) \emph{coincide}.
This enables the use of the notions in Section~\ref{SecBasicNotionsJK86} to
analyze $E$-confluence of \egtrs{s}.
Note that Corollary \ref{CoroEConfluence_Huet80} (from Huet's results) is equivalent to Corollary \ref{CoroTheorem5_JK86} if $\genrelation=\genrelationUpE\spcrel{=}\rew{\cR^\RMabbr}$.
Thus, as a consequence of 
Corollary \ref{CoroEConfluence_Huet80} we have the following
two sufficient conditions for $E$-confluence of $E$-terminating \egtrs{s}. 

\begin{cor}
\label{CoroSufficientCriteriaForEConfluenceOfEGRSs}
Let $\cR=(\Symbols,\SPredicates,\mu,E,H,R)$ be an 
$E$-terminating \egtrs{}. Then, $\cR$ is $E$-confluent if one of the following conditions hold:
\begin{enumerate}
\item\label{CoroSufficientCriteriaForEConfluenceOfEGRSs_Huet} 
$\rew{\cR^\RMabbr}$ and $\equone{E}$ (i.e., $\rew{\eqoriented{E}}$) satisfy properties $\alpha$ and $\gamma$, or
\item\label{CoroSufficientCriteriaForEConfluenceOfEGRSs_JK86}
 $\rew{\cR^\RMabbr,E}$ is locally confluent modulo $E$ with $\rew{\cR^\RMabbr}$ and locally coherent modulo $E$.
\end{enumerate}
\end{cor}
%
%
\noindent
For abstract relations $\genequivalence$, 
$\genrelation$, and $\genrelationUpE$ in Jouannaud and Kirchner's framework,
where $\genrelation$ is $\genequivalence$-terminating, and taking into account that, as discussed in Section~\ref{SecHuet80}, $\alpha_\genrelation$ is equivalent to $\fS{LCON}_\genequivalence(\genrelation,\genrelation)$ and $\gamma_\genrelation$ is equivalent to $\fS{LCOH}_\genequivalence(\genrelation)$, we have the following:
\[
\xymatrix{
\alpha_\genrelation \wedge \gamma_\genrelation 
& \ar@{<=>}[r]^{\text{Theorem  \ref{TheoTheorem5_JK86}}} &
& \fS{CR}_\genequivalence(\genrelation,\genrelation)
\ar@{=>}[d]^{\text{Proposition  \ref{PropSufficientConditionForEConfluence_JK86}}} 
\\
\fS{LCON}_\genequivalence(\genrelationUpE,\genrelation) \wedge \fS{LCOH}_\genequivalence(\genrelationUpE) 
& \ar@{<=>}[r]^{\text{Theorem  \ref{TheoTheorem5_JK86}}} &
& \fS{CR}_\genequivalence(\genrelation,\genrelationUpE)
}
\]
Thus, from a theoretical point of view, proofs of $E$-confluence of \egtrs{s} obtained according to item  (\ref{CoroSufficientCriteriaForEConfluenceOfEGRSs_JK86}) of Corollary \ref{CoroSufficientCriteriaForEConfluenceOfEGRSs}, 
which are proofs of  $\fS{CR}_\genequivalence(\genrelation,\genrelationUpE)$,
\emph{subsume} those made according to item (\ref{CoroSufficientCriteriaForEConfluenceOfEGRSs_Huet}), which
are proofs of $\fS{CR}_\genequivalence(\genrelation,\genrelation)$.
Actually, this subsumption is \emph{strict}, as shown in Section~\ref{SecLimitationsOfHuetApproach} and Example \ref{ExHuet80_RemarkPage818_EConfluence} regarding $E$-confluence of $\cR$ in Example \ref{ExHuet80_RemarkPage818}.
See also Example \ref{ExRModuloAndRrelativeE_EConfluence} regarding the proof of 
$E$-confluence of $\cR$ in Example \ref{ExRModuloAndRrelativeE}.
From a practical point of view, though, proofs obtained according to (\ref{CoroSufficientCriteriaForEConfluenceOfEGRSs_Huet}) are simpler, as they do \emph{not} involve 
$E$-unification  (for computing $E$-critical pairs) 
or $E$-matching (for checking $\ejoinability{\cR^\RMabbr,E}$-joinability), 
see Sections \ref{SecConditionalPairsAlphaAndGamma} 
to \ref{SecEConfluenceWithPSRandCpeaks}.

\subsection{\egtrs{s} with a confluent set of oriented equations}

\leavevmode\newline
Given
$\cR=(\Symbols,\SPredicates,\allowbreak\mu,E,H,R)$,
we consider the
\gtrs{}
\[
\cR_E=(\Symbols,\SPredicates_E,\mu,H_E,R_E)
\]
where
\begin{itemize}
\item $\SPredicates_E$ contains the predicate symbols from $\SPredicates$ which depend on $H$ together with $\to$, $\to^*$, 
and a fresh predicate symbol $\_\downarrow\_$;
\item $H_E$ contains the clauses from $H$ 
defining predicates in $\SPredicates_E$ together with a clause $x\downarrow y \IF x\rews{}z,y\rews{}z$ for the new predicate $\_\downarrow\_$; and
\item $R_E$ consists of rules $s\to t\IF\gencond'$ obtained from $s\to t\IF\gencond\in\:\overrightarrow{E}$ (provided that $s$ is not a variable)
by replacing every atom $u\equequ{}v$ in $\gencond$ by $u\downarrow v$.
\end{itemize}
%
\noindent
Note that, since equations in $E$ \emph{do not depend on $R$} (see Remark \ref{RemEqDoesNotDependOnR}), rules in $R_E$ do not depend on $R$ either.

If $\cR_E$ is confluent, then proofs of $s\equequ{E}t$ 
can be equivalently treated as \emph{joinability} proofs $s\joinability{\cR_E}t$.
As usual, if $\cR_E$ is terminating, then equalities $s\equequ{E}t$ can be solved by normalization of $s$ and $t$ with $\cR_E$ into $\nfOf{s}$ and $\nfOf{t}$, and then comparing $\nfOf{s}$ and $\nfOf{t}$ for syntactic equality. 

Essentially, this is what is done in, e.g., \Maude's system modules \cite[Chapter 6]{ClavelEtAl_MaudeBook_2007}, consisting of a set of equations and a set of rewrite rules.
Note, however, that in the presence of conditional rules in $\cR_E$ (coming from conditional equations in $E$), additional requirements like \emph{operational termination} \cite{LucMarMes_OperationalTerminationOfConditionalTermRewritingSystems_IPL05} can be necessary to guarantee decidability of equational goals, see \cite{LucMes_NormalFormsAndNormalTheoriesInConditionalRewriting_JLAMP16} for a discussion.
Operational termination of \gtrs{s} $\cR$ is defined in \cite[Definition 11]{Lucas_SemanticPropertiesOfComputationsDefinedByElementaryInferenceSystems_HCVS25} as the absence of infinite proof trees in $\GLinferenceOf{\rewtheoryOf{\cR}}$, the \eis{} associated with $\cR$ \cite[Definition 6]{Lucas_SemanticPropertiesOfComputationsDefinedByElementaryInferenceSystems_HCVS25}.

\subsection{\egtrs{s} without equations as \gtrs{s}}
\label{SecEGTRSsWithoutEquationsAsGTRSs}

In the ``degenerate'' case when the set of equations $E$ of an \egtrs{} $\cR=(\Symbols,\SPredicates,\mu,E,H,R)$ is empty, i.e., $E=\emptyset$, 
a number of consequences follow which are used later in this paper (see Section~\ref{SecConfluenceOfGTRSsAsConfluenceOfEGTRSsWithEempty}):
\begin{enumerate}
\item The relation $\equone{E}$, i.e., $\rew{\eqoriented{E}}$ is empty.
\item $\GLtheory_{\text{Eq}}[\Symbols,\mu,E]=\GLtheory_{\text{Eq}}[\Symbols,\mu,\emptyset]$ consists of reflexivity, transitivity, symmetry and propagation rules for the equality predicate.
Thus, it characterizes the equality predicate $\equequ{}$ as the \emph{syntactic identity}.
Therefore, 
\item Sentences $(R,E)_{\alpha}$ for rules $\alpha:\lhsr\to\rhsr\IF A_1,\ldots,A_n$ are 
equivalent to $(\forall\vec{x})\:A_1\wedge\cdots\wedge A_n\Rightarrow\lhsr\rewpstickel{}\rhsr$, i.e., predicate $\rewpstickel{}$ defines the same reduction relation (as deduction in $\GLtheory_{\cR,E}$) as $\rew{}$ (as deduction in $\GLtheory_{\cR}$).
Also,
\item Sentence $(R/E)$ is 
equivalent to $(\forall x,y)\:x\rew{}y\Rightarrow x\rewmodulo{}y$, i.e., predicate $\rewmodulo{}$ defines the same reduction relation (as deduction in $\GLtheory_{\cR/E}$) as $\rew{}$ (as deduction in $\GLtheory_{\cR}$).
\item Consequently, $\rew{\cR,E}$ and $\rew{\cR/E}$ collapse into $\rew{\cR}$ and there is no need to use $R^\PSabbr$, $R^\RMabbr$, etc. Hence, 
\item $\ejoinability{\cR/E}$-joinability, 
$\ejoinability{\cR,E}$-joinability, 
and 
$\ejoinability{\cR}$-joinability boil down into $\joinability{\cR}$-joinability.
\item Both 
(i) property $\alpha$ and 
(ii) local confluence of 
$\rew{\cR,E}$ modulo $E$ with $\rew{\cR}$ 
boil 
down into the usual definition of \emph{local confluence} of $\rew{\cR}$.
\item Local coherence of $\rew{\cR}$ 
and $\rew{\cR,E}$ modulo $E$ (and also Property $\gamma$) trivially hold, as the $\equone{}$-step of the cliff is not possible, being $\equone{E}=\emptyset$.
\item $E$-termination of $\cR$ becomes termination of $\rew{\cR}$ (investigated in \cite{Lucas_TerminationOfGeneralizedTermRewritingSystems_FSCD24}).
\end{enumerate}
Therefore, $E$-confluence
can be analyzed as confluence of the
\gtrs{} $\cR=(\Symbols,\SPredicates,\mu,H,R)$.

\section{Peaks for the analysis of $E$-confluence}
\label{SecReductionAndCoherencePeaks}

According to Corollary \ref{CoroSufficientCriteriaForEConfluenceOfEGRSs},
in order to prove $E$-confluence of ($E$-terminating) \egtrs{s} 
$\cR=(\Symbols,\SPredicates,\mu,E,H,R)$, we can investigate one of the following situations.
\begin{enumerate}
\item 
\label{ItemSecReductionAndCoherencePeaks_Huet}
 \emph{Properties $\alpha$ and $\gamma$ of $\rew{\cR^\RMabbr}$} (sometimes referred to as $\alpha_{\rew{\cR^\RMabbr}}$ and $\gamma_{\rew{\cR^\RMabbr}}$, if necessary), using 
 $\rew{\eqoriented{E}}$ (Corollary \ref{CoroSufficientCriteriaForEConfluenceOfEGRSs}(\ref{CoroSufficientCriteriaForEConfluenceOfEGRSs_Huet})). 
\item 
\emph{Local confluence modulo $E$ of $\rew{\cR^\RMabbr,E}$ with $\rew{\cR^\RMabbr}$ and local coherence modulo $E$ of $\rew{\cR^\RMabbr,E}$} (Corollary \ref{CoroSufficientCriteriaForEConfluenceOfEGRSs}(\ref{CoroSufficientCriteriaForEConfluenceOfEGRSs_JK86})). 
\end{enumerate}

\subsection{Local confluence and coherence peaks for properties $\alpha$ and $\gamma$}
\label{SecLocalConfluenceAndCoherencePeaksForAlphaAndGamma}

According to Figure \ref{FigHuetConfluencePropertiesAlphaBetaGamma} and taking into account Table \ref{TableARSsAndEGTRSs}, by Proposition \ref{PropUseOfRulesAndPositionsInRewritingSteps_EGTRSs}(\ref{PropUseOfRulesAndPositionsInRewritingSteps_EGTRSs_R}), the peaks corresponding to these properties,  are of the form
\begin{IEEEeqnarray}{r'C'l}
\raisebox{0.6cm}{
\xymatrix{
&  s 
\ar@{->}[dl]_{p_1}^>>{\alpha_1}
\ar@{->}[dr]^{p_2}_>>{\alpha_2} 
\\
t_1  & &  t_2 
  \\
}}
\label{GenericRPeakTerms}
\end{IEEEeqnarray}
for positions $p_1,p_2\in\Pos^\mu(s)$, rules 
$\alpha_1:\lhsr_1\to\rhsr_1\IF\gencond_1$ and 
$\alpha_2:\lhsr_2\to\rhsr_2\IF\gencond_2$
sharing no variable,
and a
substitution $\sigma$,
such that 
(i) $s|_{p_1}=\sigma(\lhsr_1)$ and $\sigma(\gencond_1)$ hold;  
(ii) $s|_{p_2}=\sigma(\lhsr_2)$ and $\sigma(\gencond_2)$ hold;
(iii) $t_1 = s[\sigma(\rhsr_1)]_{p_1}$; and
(iv) $t_2  = s[\sigma(\rhsr_2)]_{p_2}$.
In particular, the following properties must be checked:
\begin{enumerate}
\item \emph{Property $\alpha$}, if $\alpha_1,\alpha_2\in R^\RMabbr$; then 
(\ref{GenericRPeakTerms}) is often called a \emph{local confluence peak};
and
\item \emph{Property $\gamma$}, if $\alpha_1\in R^\RMabbr$ and $\alpha_2\in\:\eqoriented{E}$;  
then 
(\ref{GenericRPeakTerms}) is 
a \emph{local coherence peak}.\footnote{This terminology is borrowed from \cite{JouKir_CompletionOfASetOfRulesModuloASetOfEquations_SIAMJC86,Jouannaud_ConfluenceOfTerminatingRewritingComputations_TFSP24}.}
\end{enumerate}
In both cases, $\ejoin{\cR^\RMabbr}$-joinability of (\ref{GenericRPeakTerms}) must be proved.

\subsection{Local $E$-confluence and $E$-coherence peaks for local confluence modulo $E$ of $\rew{\cR^\RMabbr,E}$ with $\rew{\cR^\RMabbr}$ and local coherence modulo $E$ of $\rew{\cR^\RMabbr,E}$}
\label{SecLocalConfluenceAndCoherencePeaksForJK86}
According to Figure \ref{FigConfluenceAndCoherenceProperties} and taking into account Table \ref{TableARSsAndEGTRSs}, by  Proposition \ref{PropUseOfRulesAndPositionsInRewritingSteps_EGTRSs}(\ref{PropUseOfRulesAndPositionsInRewritingSteps_EGTRSs_PStickel}), 
peaks corresponding to these properties are of the form
\begin{IEEEeqnarray}{r'C'l}
\raisebox{0.6cm}{
\xymatrix{
&  s 
\ar@{->}[dl]_{p_1}^>>{\alpha_1,E}
\ar@{->}[dr]^{p_2}_>>{\alpha_2} 
\\
t_1  & &  t_2 
  \\
}}
\label{GenericPSPeakTerms}
\end{IEEEeqnarray}
for positions $p_1,p_2\in\Pos^\mu(s)$, rules 
$\alpha_1:\lhsr_1\to\rhsr_1\IF\gencond_1$ and 
$\alpha_2:\lhsr_2\to\rhsr_2\IF\gencond_2$
sharing no variable, 
and a
substitution $\sigma$,
such that 
(i) $s|_{p_1}\equequ{E}\sigma(\lhsr_1)$, i.e., \emph{matching modulo $E$ is used}, and $\sigma(\gencond_1)$ hold;  
(ii) $s|_{p_2}=\sigma(\lhsr_2)$ and $\sigma(\gencond_2)$ hold;
(iii) $t_1 = s[\sigma(\rhsr_1)]_{p_1}$; and
(iv) $t_2  = s[\sigma(\rhsr_2)]_{p_2}$.
Due to the use of \emph{Peterson \& Stickel} reduction  $\rew{\cR^\RMabbr,E}$ 
in the leftmost branch of the peaks, 
$\alpha_1$ always belongs to $R^\RMabbr$.
In particular, the following properties must be checked:
\begin{enumerate}
\item \emph{Local confluence of $\rew{\cR^\RMabbr,E}$  modulo $E$ with $\rew{\cR^\RMabbr}$}, 
if $\alpha_2\in R^\RMabbr$; then, (\ref{GenericPSPeakTerms}) is called a local $E$-confluence peak for short; and
\item 
\emph{Local coherence of $\rew{\cR^\RMabbr,E}$  modulo $E$}, 
if $\alpha_2\in\:\eqoriented{E}$; then, (\ref{GenericPSPeakTerms}) is 
a local $E$-coherence peak.
\end{enumerate}
In both cases, $\ejoinability{\cR^\RMabbr,E}$-joinability of (\ref{GenericPSPeakTerms}) must be proved.

\subsection{Disjoint and non-disjoint peaks}
\label{SecDisjointAndNonDisjointPeaks}

If $p_1$ and $p_2$ in (\ref{GenericRPeakTerms}) 
are \emph{disjoint}, i.e., $p_1\parallel p_2$, 
then the peak is $\joinability{\cR}$-joinable, as 
(i) $p_1$ remains active after replacements on $p_2$ (and vice versa, see \cite[Section 3]{Lucas_ContextSensitiveRewriting_CSUR20} for a justification if a replacement map $\mu$ is used) and (ii) the instances $\sigma(\gencond_1)$ 
and  $\sigma(\gencond_2)$ of rules $\alpha_1$ and $\alpha_2$ are also satisfied (see Figure \ref{FigDisjointRpeaks}). 
\begin{figure}[t]
\begin{center}
 \begin{tikzpicture}[node distance = 6cm, auto,scale=0.5,every node/.style={scale=0.8}]
\node (DP_TextT) {$s[\sigma(\rhsr_1)]_{p_1}[\sigma(\lhsr_2)]_{p_2}$};
\node [triangleW, below  = 0cm of DP_TextT, minimum width = 3cm, minimum height = 2.25cm] (DP_TriangleT) {};
\node [triangleG, below left = -0.42cm and -0.35cm of DP_TriangleT, minimum width = 1cm, minimum height = 1.24cm, scale=0.9] (DP_RHStriangleT) {};
\node [below  = -0.35 of DP_RHStriangleT] (DP_LHStextT) {\tiny $\sigma(\rhsr_1)$};
\node [triangleG, right = 0.41cm of DP_RHStriangleT, minimum width = 1cm, minimum height = 1.24cm, scale=0.9] (DP_LHSpTriangleT) {};
\node [below  = -0.35 cm of DP_LHSpTriangleT] (DP_LHSpTextT) {\tiny $\sigma(\lhsr_2)$}; 
\node [above right = 1cm and 1.8 cm of DP_TextT] (DP_TextS) {$s[\sigma(\lhsr_1)]_{p_1}[\sigma(\lhsr_2)]_{p_2}$};
\node [triangleW, below  = 0cm of DP_TextS, minimum width = 3cm, minimum height = 2.25cm] (DP_TriangleS) {};
\node [triangleG, below left = -0.42cm and -0.35cm of DP_TriangleS, minimum width = 1cm, minimum height = 1.24cm, scale=0.9] (DP_LHStriangleS) {};
\node [below  = -0.35 cm of DP_LHStriangleS] (DP_LHSt
TextS) {\tiny $\sigma(\lhsr_1)$}; 
\node [triangleG, right = 0.41cm of DP_LHStriangleS, minimum width = 1cm, minimum height = 1.24cm, scale=0.9] (DP_LHSpTriangleS) {};
\node [below  = -0.35 cm of DP_LHSpTriangleS] (DP_LHSpTextS) {\tiny $\sigma(\lhsr_2)$};
\node [below right =1cm and 1.8 cm of DP_TextS] (DP_TextTp) {$s[\sigma(\lhsr_1)]_{p_1}[\sigma(\rhsr_2)]_{p_2}$};
\node [triangleW, below  = 0cm of DP_TextTp, minimum width = 3cm, minimum height = 2.25cm] (DP_TriangleTp) {};
\node [triangleG, below left = -0.42cm and -0.35cm of DP_TriangleTp, minimum width = 1cm, minimum height = 1.24cm, scale=0.9] (DP_LHStriangleTp) {};
\node [below  = -0.35 cm of DP_LHStriangleTp] (DP_LHStextTp) {\tiny $\sigma(\lhsr_1)$}; 
\node [triangleG, right = 0.41cm of DP_LHStriangleTp, minimum width = 1cm, minimum height = 1.24cm, scale=0.9] (DP_RHSpTriangleTp) {};
\node [below  = -0.35 cm of DP_RHSpTriangleTp] (DP_RHSpTextTp) {\tiny $\sigma(\rhsr_2)$}; 
\node [below =3 cm of DP_TextS] (DP_TextU) {$s[\sigma(\rhsr_1)]_{p_1}[\sigma(\rhsr_2)]_{p_2}$};
\node [triangleW, below  = 0cm of DP_TextU, minimum width = 3cm, minimum height = 2.25cm] (DP_TriangleU) {};
\node [triangleG, below left = -0.42cm and -0.35cm of DP_TriangleU, minimum width = 1cm, minimum height = 1.24cm, scale=0.9] (DP_RHStriangleU) {};
\node [below  = -0.35 cm of DP_RHStriangleU] (DP_RHStextU) {\tiny $\sigma(\rhsr_1)$}; 
\node [triangleG, right = 0.41cm of DP_RHStriangleU, minimum width = 1cm, minimum height = 1.24cm, scale=0.9] (DP_RHSpTriangleU) {};
\node [below  = -0.35 cm of DP_RHSpTriangleU] (DP_RHSpTextU) {\tiny $\sigma(\rhsr_2)$}; 
\draw[->] ([shift={(-0.2,0.2)}]DP_TriangleS.west) -- ([shift={(0.1,0.3)}]DP_TriangleT.east);
\draw[->] ([shift={(0.2,0.2)}]DP_TriangleS.east) -- ([shift={(-0.2,0.2)}]DP_TriangleTp.west);
\draw[->] ([shift={(0.2,0)}]DP_TriangleT.east) -- ([shift={(-0.2,0.2)}]DP_TriangleU.west);
\draw[->] ([shift={(-0.2,0)}]DP_TriangleTp.west) -- ([shift={(0.2,0.2)}]DP_TriangleU.east);
 \end{tikzpicture}
 \end{center}
 \caption{Disjoint peaks (\ref{GenericRPeakTerms})}
 \label{FigDisjointRpeaks}
 \end{figure}
For (\ref{GenericPSPeakTerms}) it is similar, using $\joinability{\cR,E}$-joinability.

If $p_1$ and $p_2$ in (\ref{GenericRPeakTerms}) or (\ref{GenericPSPeakTerms}) are \emph{not disjoint}, i.e., $p_1\leq p_2$ or $p_2\leq p_1$, then we need to investigate their joinability.
Section~\ref{SecConditionalPairsAlphaAndGamma} (resp.\ \ref{SecConditionalPairsForNonDisjointPeaksJK}) investigates how to obtain a \emph{finite representation} of possibly infinitely many peaks (\ref{GenericRPeakTerms}) (resp.\ (\ref{GenericPSPeakTerms})) as instances of finitely many \emph{conditional pairs}.
\begin{rem}[$\ejoin{\cR^\RMabbr}$-joinability and $\ejoin{\cR^\RMabbr,E}$-joinability]
Note that $\ejoin{\cR^\RMabbr}$-joinability implies $\ejoin{\cR^\RMabbr,E}$-joinability.
In practice, trying $\ejoin{\cR^\RMabbr}$-joinability first can be better as it is simpler to implement, as 
 $\rew{\cR^\RMabbr,E}$-rewriting requires \emph{$E$-matching substitutions}.
\end{rem}

\subsection{About local peaks modulo}
\label{SecAboutLocalPeaksModulo}

As in \cite[Definition 5(3)]{JouMun_TerminationOfASetOfRulesModuloASetOfEquations_CADE84} with the notion of \emph{local E-confluence} (of $\rew{\cR,E}$), in 
\cite[Figure 11.5 (left)]{Jouannaud_ConfluenceOfTerminatingRewritingComputations_TFSP24}
Jouannaud uses  
peaks of the form 
\begin{IEEEeqnarray}{r'C'l}
\raisebox{0.6cm}{
\xymatrix{
&  s 
\ar@{->}[dl]_{p_1}^>>{\alpha_1,E}
\ar@{->}[dr]^{p_2}_>>{\alpha_2,E} 
\\
t_1  & &  t_2 
  \\
}}
\label{LblLocalPeakModulo}
\end{IEEEeqnarray}
which he calls \emph{local peaks modulo}.
The following result shows that such peaks
can be seen as
peaks (\ref{GenericPSPeakTerms}).
As usual, we do not consider disjoint peaks (\ref{LblLocalPeakModulo}) 
with $p_1\parallel p_2$,
as they are trivially $\ejoinability{\cR,E}$-joinable.

\begin{prop}
\label{PropLocalPeaksModulo}
For every non-disjoint local 
peak 
modulo 
$t_1\leftrew{\cR,E}s\rew{\cR,E}t_2$, 
there is a local peak 
$t_2\leftrew{\cR,E}s'\rew{\cR}t_1$ 
or
$t_1\leftrew{\cR,E}s''\rew{\cR}t_2$ for some terms $s'\equequ{E}s\equequ{E}s''$.
\end{prop}

\begin{proof}
Without loss of generality, we consider a 
peak
\begin{eqnarray*}
t_1\leftrewAtPos{p}{\alpha_1,E} & s & \rewAtPos{\toppos}{\alpha_2,E}t_2\label{PropPS2peaksAsPSRpeaks_PS2peak}
\end{eqnarray*}
with $p\in\Pos^\mu(s)$,
$s|_p\equequ{E}\sigma(\lhsr_1)$, 
$t_1=s[\sigma(\rhsr_1)]_p$,
$s\equequ{E}\sigma(\lhsr_2)$, 
and 
$t_2=\sigma(\rhsr_2)$,
for $\alpha_1:\lhsr_1\to\rhsr_1\IF\gencond_1,\alpha_2:\lhsr_2\to\rhsr_2\IF\gencond_2,\in R^\RMabbr$.
Let $s'=s[\sigma(\lhsr_1)]_p$. Then, $s'\equequ{E}s=\sigma(\lhsr_2)$
and 
\begin{eqnarray*}
t_2 =\sigma(\rhsr_2)\leftrewAtPos{\toppos}{\alpha_2,E} & s' & \rewAtPos{p}{\alpha_1}t_1
\end{eqnarray*}
as required. The second possibility arises starting from 
\begin{eqnarray*}
t_1\leftrewAtPos{\toppos}{\alpha_1,E} & s & \rewAtPos{p}{\alpha_2,E}t_2\label{PropPS2peaksAsPSRpeaks_PS2peak}
\end{eqnarray*}
and proceeding in a similar way, possibly obtaining a different term $s''$ still equivalent to $s$ (and hence to $s'$).
\end{proof}
%
\noindent
Thus, we do not
investigate 
peaks 
(\ref{LblLocalPeakModulo}).
Proposition \ref{PropLocalPeaksModulo} justifies the use of ``$E$-confluence peak'' to refer to ``hybrid'' peaks (\ref{GenericPSPeakTerms}), involving reductions with $\rew{\alpha_1,E}$ and $\rew{\alpha_2}$ when $\alpha_2$ belongs to $R^\RMabbr$,
although one would naturally use ``$E$-confluence peak'' with (\ref{LblLocalPeakModulo}), as they concern local $E$-confluence of $\rew{\cR,E}$ in \cite{JouMun_TerminationOfASetOfRulesModuloASetOfEquations_CADE84}.

\subsection{Checking $\ejoinability{\cR^\RMabbr}$-joinability and $\ejoinability{\cR^\RMabbr,E}$-joinability in practice}
\label{SecCheckingEjoinabilityAndPSjoinabilityInPractice}

Proofs of $E$-confluence require to prove either $\ejoinability{\cR^\RMabbr}$-joinability 
or $\ejoinability{\cR^\RMabbr,E}$-joinability.
Due to the use of $\rew{\cR^\RMabbr}$ and $\rew{\cR^\RMabbr,E}$, rewriting goals $s\tos{} t$ in the conditional part $\gencond$ of clauses in $H$ and rules in $R$  become $s\rewmodulos{} t$ in $H^\RMabbr$ and $R^\RMabbr$, respectively.
Thus, they should be evaluated using $\rews{\cR/E}$.
Regarding the analysis of peaks (\ref{GenericRPeakTerms}) and (\ref{GenericPSPeakTerms}), we have two situations:
\begin{enumerate}
\item The components $s$, $\sigma$, 
and $t_i$, $p_i$, and $\alpha_i$, for $i=1,2$ of the peaks should be \emph{singled out}.
We address this problem by using conditional pairs representing them, see Sections \ref{SecConditionalPairsAlphaAndGamma} and
\ref{SecConditionalPairsForNonDisjointPeaksJK}.
Here, it is crucial to consider rules from $R^\RMabbr$; otherwise, some peaks could be lost, thus leading to wrong conclusions.
However, 
\item
Regarding $\ejoinability{\cR^\RMabbr}$-joinability 
or $\ejoinability{\cR^\RMabbr,E}$-joinability tests,
since $\ejoinability{\cR}\subseteq\ejoinability{\cR^\RMabbr}$ and
$\ejoinability{\cR,E}\subseteq\ejoinability{\cR^\RMabbr,E}$,
in practice $\ejoinability{\cR}$-joinability 
or $\ejoinability{\cR,E}$-joinability tests could be used instead, 
although some joinable peaks could be missed.
\end{enumerate}
\begin{exa}[$\rew{\cR,E}$ instead of $\rew{\cR/E}$ in conditions]
\label{ExPSreachVsRMreach}
Consider the following \egtrs{} $\cR$, 
\begin{IEEEeqnarray}{r'C'l}
\fS{a}& = & \fS{f}(\fS{b})\label{ExPSreachVsRMreach_eq1}
\\
\fS{b} & \to & \fS{c}\label{ExPSreachVsRMreach_rule1}
\\
\fS{b} & \to & x \IF \fS{a} \rew{} x\label{ExPSreachVsRMreach_rule2}
\end{IEEEeqnarray}
where $\fS{a} \rew{} x$ in the conditional part of (\ref{ExPSreachVsRMreach_rule2}) 
encodes a \emph{one-step} rewriting.
Thus,
\begin{enumerate}
\item\label{ExPSreachVsRMreach_UsePSrewriting}
 If $\fS{a} \rew{} x$ is treated as $\fS{a}\rewpstickel{}x$, i.e., interpreted as $\fS{a}\rew{\cR,E}x$ (according to Definition \ref{DefComputationalRelationsOfAnEGTRS}(\ref{DefComputationalRelationsOfAnEGTRS_rewPStickel})) 
then (\ref{ExPSreachVsRMreach_rule2})  is \emph{infeasible}, as $\fS{a}$ is $\rew{\cR^\RMabbr,E}$-irreducible.
Thus, only rule (\ref{ExPSreachVsRMreach_rule1}) can be used and  no ($E$-)critical peak would be singled out.
\item If $\fS{a} \rew{} x$ is treated as $\fS{a}\rmodulo{}x$, i.e., interpreted as $\fS{a}\rew{\cR/E}x$, then there is a critical peak
\begin{IEEEeqnarray}{r'C'l}
\fS{c} \leftrewAtPos{\toppos}{(\ref{ExPSreachVsRMreach_rule1})} 
\label{ExPSreachVsRMreach_peak}
& \fS{b} & 
\rewAtPos{\toppos}{(\ref{ExPSreachVsRMreach_rule2})} \fS{f}(\fS{c})
\end{IEEEeqnarray}
where the rightmost 
step with rule (\ref{ExPSreachVsRMreach_rule2}) is possible as 
$\fS{a}\equequ{(\ref{ExPSreachVsRMreach_eq1})}\fS{f}(\ul{\fS{b}})\rew{(\ref{ExPSreachVsRMreach_rule1})}\fS{f}(\ul{\fS{c}})$.
This peak is \emph{not} $\ejoinability{\cR^\RMabbr/E}$-joinable, as both $\fS{c}$ 
and $\fS{f}(\fS{c})$ are $\rew{\cR^\RMabbr/E}$-irreducible and not $\equequ{E}$-equivalent.
As a consequence, $\cR$ is not $E$-confluent.
\end{enumerate}
\end{exa}
Using $\rew{\cR,E}$ instead of $\rew{\cR^\RMabbr,E}$ as in item (\ref{ExPSreachVsRMreach_UsePSrewriting}) of Example \ref{ExPSreachVsRMreach} 
would lead to an analysis of confluence of $\rew{\cR,E}$ modulo $E$ 
(the commutation of the diagram in Figure \ref{FigChurchRosserAndEConfluenceJK86} (right) when $\rew{\cR,E}$ is used instead of $\rew{\cR/E}$)
rather than to an analysis of $E$-confluence.
In general, confluence of $\rew{\cR,E}$ modulo $E$ does \emph{not} imply $E$-confluence.
Thus, investigating confluence of $\rew{\cR,E}$ modulo $E$ as an approximation to $E$-confluence of $\cR$ would be wrong.

\section{Conditional pairs for non-disjoint local confluence and coherence
peaks}
\label{SecConditionalPairsAlphaAndGamma}
As explained in Section~\ref{SecLocalConfluenceAndCoherencePeaksForAlphaAndGamma}, depending on the origin of rules $\alpha_1$ and $\alpha_2$ in peaks (\ref{GenericRPeakTerms}), i.e.,
\begin{IEEEeqnarray*}{r'C'l+x*}
\raisebox{0.6cm}{
\xymatrix{
&  s 
\ar@{->}[dl]_{p_1}^>>{\alpha_1}
\ar@{->}[dr]^{p_2}_>>{\alpha_2} 
\\
t_1  & &  t_2
  \\
}}
&&& \eqref{GenericRPeakTerms}
\end{IEEEeqnarray*}
 we obtain:
\begin{enumerate}
\item A local \emph{confluence} peak if both $\alpha_1$ and $\alpha_2$ belong to $R^\RMabbr$; and
\item A 
local \emph{coherence} peak, if either 
(i) $\alpha_1$ belongs to $R^\RMabbr$ and $\alpha_2$ belongs to $\eqoriented{E}$
or  
(ii) $\alpha_1$ belongs to $\eqoriented{E}$ and $\alpha_2$ belongs to $R^\RMabbr$.
\end{enumerate}
\begin{rem}[Asymmetry in coherence peaks]
\label{RemAsymmetryCoherencePeaks}
In contrast to confluence peaks,
coherence peaks are \emph{asymmetric} in the use of rules belonging to either $R^\RMabbr$ or to $\eqoriented{E}$, but \emph{not both}.
\end{rem}
If $p_1$ and $p_2$ in (\ref{GenericRPeakTerms}) are \emph{not disjoint}, then we have $p_1\leq p_2$ or $p_2\leq p_1$.
Dealing with confluence peaks, without loss of generality we can focus on \emph{one} of them (dismissing the other), provided that all rules in $R^\RMabbr$ are considered.
However, dealing with coherence peaks, due to the aforementioned \emph{asymmetry}, in order to cover all cases we need to consider \emph{both}, corresponding to the relative position (above or below) of the $R^\RMabbr$-redexes and $\eqoriented{E}$-redexes.
An alternative approach
is 
\emph{assuming}
e.g., $p_1\leq p_2$, and then investigating \emph{both} possibilities: 
the rule from $R^\RMabbr$ is applied at $p_1$ and the rule from $\eqoriented{E}$ at $p_2$; and vice versa.

Thus, let $p_2=p_1.p\in\Pos^\mu(s)$ for some position $p$.
That is, (\ref{GenericRPeakTerms}) corresponds to
\begin{IEEEeqnarray}{r'C'l}
\underbrace{\sigma(\lhsr)[\sigma(\rhsr')]_p}_{t} ~\leftrewAtPos{p}{\alpha'}   & \underbrace{\sigma(\lhsr)[\sigma(\lhsr')]_p}_s & \rewAtPos{\toppos}{\alpha}~\underbrace{\sigma(\rhsr)}_{t'}\label{NestedRMRPeakTermsRup}
\end{IEEEeqnarray}
for \emph{variable disjoint} rules 
$\alpha:\lhsr\to\rhsr\IF\gencond$ and $\alpha':\lhsr'\to\rhsr'\IF\gencond'$,
and substitution $\sigma$ such that $\sigma(\gencond)$ and $\sigma(\gencond')$ hold, and 
where, by abuse, we still use $s$, $t$, and $t'$, now as indicated.
\begin{defi}[Confluence and coherence peaks]
\label{DefConfluenceAndCoherencePeaks}
Consider the rules $\alpha$ and $\alpha'$ in (\ref{NestedRMRPeakTermsRup}).
\begin{enumerate}
\item If both $\alpha$ and $\alpha'$ belong to $R^\RMabbr$, 
then (\ref{NestedRMRPeakTermsRup}) is a \emph{confluence peak}.
\item  If either 
(i) $\alpha$ belongs to $R^\RMabbr$ and $\alpha'$ belongs to $\eqoriented{E}$, 
or 
(ii) $\alpha$ belongs to $\eqoriented{E}$ and $\alpha'$ belongs to $R^\RMabbr$,
then (\ref{NestedRMRPeakTermsRup}) is a \emph{coherence peak}.
\end{enumerate}
\end{defi}
%
\noindent
Two families of peaks are distinguished:
 \begin{defi}[Critical and variable peaks]
 Consider rule $\alpha$, $\alpha'$, and $p$ in (\ref{NestedRMRPeakTermsRup}).
 \begin{enumerate}
 \item If $p\in\Pos^\mu_\Symbols(\lhsr)$, then rule $\alpha'$  \emph{overlaps} rule $\alpha$ and (\ref{NestedRMRPeakTermsRup}) is called a \emph{critical peak};\footnote{We borrow this denomination from 
\cite[Section 2.2]{DerOkaSiv_ConfluenceOfConditionalRewriteSystems_CTRS87}.
Recently, Jouannaud has used \emph{overlapping local peak} to refer to them \cite[Proof of Theorem 11.3, item 3]{Jouannaud_ConfluenceOfTerminatingRewritingComputations_TFSP24}
} and
 \item If $p\notin\Pos^\mu_\Symbols(\lhsr)$, then rule $\alpha'$  \emph{does not overlap} rule $\alpha$ and (\ref{NestedRMRPeakTermsRup}) is called a \emph{variable peak}.\footnote{see, e.g., 
\cite[Section 2.3]{DerOkaSiv_ConfluenceOfConditionalRewriteSystems_CTRS87}. Recently, Jouannaud has used \emph{ancestor local peak} to refer to them \cite[Proof of Theorem 11.3, item 2]{Jouannaud_ConfluenceOfTerminatingRewritingComputations_TFSP24}.}
 \end{enumerate}
 \end{defi}
 %
\noindent
These peaks are captured by different conditional pairs
 obtained from rules in $R^\RMabbr$ and/or $\eqoriented{E}$.

\paragraph{Use of conditional pairs}
In the following, we investigate how to represent and prove joinability of critical and variable confluence and coherence peaks by using different kinds of conditional pairs $\langle s,t\rangle\IF\gencond$.
The following notions and results are used in the remainder of the paper.
\begin{defi}[Joinability of conditional pairs]
\label{DefJoinabilityOfConditionalPairs}
Let $\cR$
be an \egtrs{}
and $\genrelation,\genrelationbis$ be relations on terms.
A conditional pair $\pi:\langle s,t\rangle\IF\gencond$ is 
$\joinOf{\genrelation}{\genrelationbis}$-joinable
if for all
substitutions $\sigma$, if 
$\deductionInThOf{\crtheoryOf{\cR}}{\sigma(A)}$
holds for all $A\in\gencond$, 
 then $\sigma(s)$ and $\sigma(t)$ are 
 $\joinOf{\genrelation}{\genrelationbis}$-joinable.
\end{defi}
\begin{defi}[Feasible conditional pair]
\label{DefFeasibleConditionalPair}
Let $\cR$ be an \egtrs{}.
A conditional pair 
$\langle s,t\rangle\IF\gencond$ is 
\emph{feasible} 
if 
$\gencond$ is 
feasible.
Otherwise, it is \emph{infeasible}.
\end{defi}
\begin{prop}\label{PropJoinabilityOfInfeasiblePairs}
Let $\cR$ be an \egtrs{} and $\genrelation,\genrelationbis$ be relations on terms.
\begin{enumerate}
\item 
Infeasible conditional pairs are 
$\joinOf{\genrelation}{\genrelationbis}$-joinable.
\item  
Trivial conditional pairs are 
$\ejoinOf{\genrelation}{\genrelation}$-joinable.
If $\genrelation$ is reflexive, then they are $\joinOf{\genrelation}{\genrelation}$-joinable.
\end{enumerate}
\end{prop}
%
\noindent
We consider
confluence and coherence \emph{critical}
peaks in Section~\ref{SecCriticalPeaksAndCCPs}; and
confluence and coherence \emph{variable} peaks in Section~\ref{SecVariablePeaksAndCVPs}.

\subsection{Critical confluence and coherence peaks and conditional critical pairs}
\label{SecCriticalPeaksAndCCPs}

If $\alpha'$ \emph{overlaps} $\alpha$ in (\ref{NestedRMRPeakTermsRup}), i.e., there is a non-variable (active) position $p\in\Pos^\mu_\Symbols(\lhsr)$ such that $\lhsr|_p$ and $\lhsr'$ \emph{unify}, 
we obtain a
\emph{critical peak} 
(see Figure \ref{FigCriticalPeak})
\begin{IEEEeqnarray}{r'C'l}
\underbrace{\sigma(\lhsr)[\sigma(\rhsr')]_p}_{t} ~\leftrewAtPos{p}{\alpha'}   & \underbrace{\sigma(\lhsr)[\sigma(\lhsr')]_p}_s & \rewAtPos{\toppos}{\alpha}~\underbrace{\sigma(\rhsr)}_{t'}\qquad\qquad\text{where }p\in\Pos^\mu_\Symbols(\lhsr)\label{LblCriticalPeak}
\end{IEEEeqnarray}
\begin{figure}[t] 
 \begin{center}
\begin{tikzpicture}[node distance = 6cm, auto,scale=1.5,every node/.style={scale=1}]
\node (CP_Ttext) {$\sigma(\lhsr)[\sigma(\rhsr')]_p$};
\node [triangleW, below  = 0cm of CP_Ttext, minimum width = 2cm, minimum height = 2.35cm, scale=0.9] (CP_Ttriangle) {};
\node [triangleG, below  = 0cm of CP_Ttext, minimum width = 1cm, minimum height = 1.2cm, scale=0.9] (CP_LHSTriangleT) {};
\node [below  = 0.2 cm of CP_Ttext] (CP_LHStextT) {$\lhsr$};
\node [triangleW, below  = 0.05cm of CP_LHStextT, minimum width = 1.3cm, minimum height = 1.48cm, scale=0.9] (CP_SigmaXtriangleT) {};
\node [triangleG, below  = 0.05cm of CP_LHStextT, minimum width = 0.5cm, minimum height = 0.5cm, scale=0.9] (CP_RHSpTriangleT) {};
\node [below  = -0.45 cm of CP_RHSpTriangleT] (CP_RHSptext) {\tiny $\rhsr'$};
\node [above right = 1cm and 1.8 cm of CP_Ttext] (CP_Stext) {$\sigma(\lhsr)[\sigma(\lhsr')]_p$};
\node [triangleW, below  = 0cm of CP_Stext, minimum width = 2cm, minimum height = 2.35cm, scale=0.9] (CP_Striangle) {};
\node [triangleG, below  = 0cm of CP_Stext, minimum width = 1cm, minimum height = 1.2cm, scale=0.9] (CP_LHSTriangleS) {};
\node [below  = 0.2 cm of CP_Stext] (CP_LHStextS) {$\lhsr$};
\node [triangleW, below  = 0.05cm of CP_LHStextS, minimum width = 1.3cm, minimum height = 1.48cm, scale=0.9] (CP_SigmaXtriangleS) {};
\node [above right = 1cm and 1 cm] at (CP_Striangle) (PinS) {\tiny $p\in \Pos^\mu_\Symbols(\lhsr)$};
\node [triangleG, below  = 0.05cm of CP_LHStextS, minimum width = 0.5cm, minimum height = 0.5cm, scale=0.9] (CP_LHSpTriangleS) {};
\node [below  = -0.45 cm of CP_LHSpTriangleS] (CP_LHSptext) {\tiny $\lhsr'$};
\node [below right =1cm and 2.5 cm of CP_Stext] (CP_TtextTp) {$\sigma(\rhsr)$};
\node [triangleW, below  = 0cm of CP_TtextTp, minimum width = 3cm, minimum height = 2.25cm] (CP_Tptriangle) {};
\node [triangleG, below  = 0cm of CP_TtextTp, minimum width = 1.67cm, minimum height = 1.28cm, scale=0.9] (CP_RHStTriangle) {};
\node [below  = 0.35 cm of CP_TtextTp] (CP_RHStextTp) {$\rhsr$};
\draw[-*] ([shift={(0.2,0.2)}]PinS) -- ([shift={(-0.05,0)}]CP_LHSpTriangleS.north);
\draw[->] ([shift={(-0.2,0.2)}]CP_Striangle.west) -- ([shift={(0.1,0.3)}]CP_Ttriangle.east);
\node[above right = 0.08cm and 0.2cm of CP_Ttriangle.east]{\tiny $\alpha'$};
\draw[->] ([shift={(0.2,0.2)}]CP_Striangle.east) -- ([shift={(-0.2,0.2)}]CP_Tptriangle.west);
\node[above left = 0.01cm and 0.35cm of CP_Tptriangle.west]{\tiny $\alpha$};
\end{tikzpicture}
 \end{center}
 \caption{Critical peak}
 \label{FigCriticalPeak}
 \end{figure}
  \begin{enumerate}
 \item As for \emph{confluence critical  peaks}, 
 both $\alpha$ and $\alpha'$ in (\ref{LblCriticalPeak}) belong to $R^\RMabbr$.
 \item As for \emph{coherence critical  peaks}, 
either 
(i) $\alpha\in R^\RMabbr$ and $\alpha'\in\:\eqoriented{E}$, 
or 
(ii) $\alpha\in\:\eqoriented{E}$ and $\alpha'\in R^\RMabbr$.
\end{enumerate}
Table \ref{TableDifferentKindsOfCriticalPeaks} summarizes the different kinds of critical peaks and their use.
\begin{table}[ht]
\caption{Different kinds of critical peaks (\ref{LblCriticalPeak})}
\label{TableDifferentKindsOfCriticalPeaks}
\centering
\begin{tabular}{@{}ccll@{}}
\toprule
$\alpha$ & $\alpha'$ & Denomination & Use
\\
\midrule
$R^\RMabbr$ & $R^\RMabbr$ & Confluence critical peak & Property $\alpha_{\rew{\cR^\RMabbr}}$, Figure \ref{FigHuetConfluencePropertiesAlphaBetaGamma} (middle)
\\
$R^\RMabbr$ & $\eqoriented{E}$ & Coherence critical peak & Property $\gamma_{\rew{\cR^\RMabbr}}$, Figure \ref{FigHuetConfluencePropertiesAlphaBetaGamma} (right)
\\
$\eqoriented{E}$ & $R^\RMabbr$ & Coherence critical peak & Property $\gamma_{\rew{\cR^\RMabbr}}$, Figure \ref{FigHuetConfluencePropertiesAlphaBetaGamma} (right)
\\
\bottomrule
\end{tabular}
\end{table}

\subsubsection{Conditional Critical Pairs}
\label{SecCCPs}

  Critical peaks for
 unconditional rules 
 are characterized as instances of 
 (finitely many) \emph{critical pairs} obtained from the involved rules \cite[Proposition 3.7]{Huet_ConfluentReductionsAbstractPropertiesAndApplicationsToTermRewritingSystems_JACM80}.
As for critical peaks (\ref{LblCriticalPeak}),
we consider \emph{Conditional Critical Pairs (CCPs)}.

\begin{defi}[Conditional Critical Pairs]
\label{DefProperAndImproperCCPs}
Let $\alpha:\lhsr\to\rhsr\IF\gencond$ and $\alpha':\lhsr'\to\rhsr'\IF\gencond'$ be
variable disjoint conditional rules and 
$p\in\Pos^\mu_\Symbols(\lhsr)$ be an
active, non-variable position of $\lhsr$, 
such that $\lhsr|_p$ and $\lhsr'$ unify with \emph{mgu} $\theta$.
Then, 
\begin{IEEEeqnarray}{r'C'l}
\langle\theta(\lhsr)[\theta(\rhsr')]_p,\theta(\rhsr)\rangle\IF\theta(\gencond),\theta(\gencond')\label{LblGenericCCP}
\end{IEEEeqnarray}
is a \emph{Conditional Critical Pair (CCP)}, denoted $\pi_{\alpha,p,\alpha'}$. 
As usual, $p$ is often called the \emph{critical position} of the CCP.
\end{defi}
%
\noindent
Given sets of rules $U$ and $V$,
\begin{eqnarray*}
\SetOf{CCP}(U,V) & = & \{\ccpOf{\alpha,p,\alpha'} \mid \alpha:\lhsr\to\rhsr\IF\gencond \in U, p\in\Pos^\mu_\Symbols(\lhsr), \alpha'\in V\}
\label{LblGenericSetOfCCPs}
\end{eqnarray*}
is the set of CCPs obtained from rules in $V$ overlapping a rule in $U$.
\begin{rem}[Asymmetric role of rules in critical peaks and CCPs]
\label{RemAssymetricRoleOfRulesInCCPs}
The role of rules $\alpha\in U$ and $\alpha'\in V$ is \emph{asymmetric} when computing critical pairs to capture critical peaks  (\ref{LblCriticalPeak}): the \emph{critical position $p$ is taken from the left-hand side $\lhsr$ of a rule $\alpha\in U$}.
Thus,  
if $U\neq V$ we need to consider \emph{both} $\SetOf{CCP}(U,V)$ and $\SetOf{CCP}(V,U)$ to capture all critical peaks.
\end{rem}
%
\noindent
Following Definition \ref{DefConfluenceAndCoherencePeaks}, 
$\ccpOf{\alpha,p,\alpha'}$ is a 
\begin{itemize}
\item \emph{confluence CCP}, if both $\alpha$ and $\alpha'$ belong to $R^\RMabbr$; and a
\item \emph{coherence CCP}, if either $\alpha$ or $\alpha'$ 
belong to 
$\eqoriented{E}$ (and hence $\alpha'$ or $\alpha$, respectively, to $R^\RMabbr$).
\end{itemize}
\begin{rem}[Coherence CCPs from ill-defined rules]
\label{RemCoherenceCCPsFromIllDefinedRules}
As discussed in Remark \ref{RemEquationsLeadingToImproperRules},
$\eqoriented{E}$ may contain ill-defined rules.
If $\alpha$ is one of such rules (i.e., $\lhsr$ is a variable), 
\emph{no CCP} is obtained
as $\Pos^\mu_\Symbols(\lhsr)$ is \empty{empty}.
In contrast, if it is $\alpha'$, then it overlaps every non-variable position $p$ of $\lhsr$ (as $\alpha$ and $\alpha'$ are variable disjoint) and the corresponding CCPs $\ccpOf{\alpha,p,\alpha'}$ would be obtained.
The CCPs (\ref{ExEvenEq_RECCP1})--(\ref{ExEvenEq_RECCP3}) in Example \ref{ExEvenEq} below illustrate them.
\end{rem}

\subsubsection{Taxonomy of CCPs}
In the literature, a number of specific situations for CCPs have been given particular denominations which are used in the following.

\begin{defi}[Proper, improper, and root CCPs]
\label{DefProperAndImproperCCPs}
Let $\alpha$ and $\alpha'$ be two rules and $p$ be a non-variable position of the left-hand side of $\alpha$.
\begin{itemize}
\item If $\alpha$ and $\alpha'$ in $\ccpOf{\alpha,p,\alpha'}$ 
are renamed versions of \emph{the same} rule and $p=\toppos$, then $\ccpOf{\alpha,p,\alpha'}$ is called \emph{improper} (iCCP, simply written $\iccpOf{\alpha}$); 
\item Otherwise, $\ccpOf{\alpha,p,\alpha'}$ is \emph{proper} (pCCP).\footnote{This terminology is borrowed from \cite[Definition 4.2]{AveLor_ConditionalRewriteSystemsWithExtraVariablesAndDeterministicLogicPrograms_LPAR94}.}
Furthermore, 
\item CCPs $\ccpOf{\alpha,\toppos,\alpha'}$ are called \emph{root} CCPs (rCCPs).\footnote{See \cite[Exercise 6.19]{BaaNip_TermRewritingAndAllThat_1998}.}
In particular, improper CCPs are root CCPs.
\end{itemize}
\end{defi}
%
\noindent
Note that there is no improper \emph{coherence} CCP as the involved rules in coherence CCPs come from different sets ($R^\RMabbr$ and $\eqoriented{E}$).
The following fact is well-known.
\begin{prop}
\label{PropJoinabilityOfImproperCriticalPeakssIn2CTRSs}
Improper CCPs $\iccpOf{\alpha}$ of 2-rules $\alpha$ 
are trivial.
\end{prop}
\begin{proof}
If $\alpha:\lhsr\to\rhsr\IF\gencond$ is a 2-rule, then $\Var(\rhsr)\subseteq\Var(\lhsr)$, and the renamed versions $\lhsr$ and $\lhsr'$ (with each variable $x$ in $\lhsr$ renamed as $x'$ in $\lhsr'$) of the left-hand side of $\alpha$ 
 unify with \emph{mgu} $\theta(x)=x'$ for all $x\in\Var(\lhsr)$. Thus, we also have 
 $\theta(\rhsr)=\theta(\rhsr')$.
 Hence $\iccpOf{\alpha}:\langle\theta(\rhsr'),\theta(\rhsr)\rangle\IF\sigma(\gencond),\theta(\gencond')$ is trivial.
\end{proof}
%
\noindent
In contrast, dealing with \emph{proper} 3-rules $\alpha$, 
Proposition \ref{PropJoinabilityOfImproperCriticalPeakssIn2CTRSs} 
may \emph{fail} to hold.

\begin{exa}
\label{Ex4_1a_AL94}
Consider the following (oriented) 3-\ctrs{} in \cite[Example 4.1(a)]{AveLor_ConditionalRewriteSystemsWithExtraVariablesAndDeterministicLogicPrograms_LPAR94}:
\begin{IEEEeqnarray}{+r'C'l+}
0 + y & \to & y\label{Ex4_1a_AL94_rule1}
\\
\fS{s}(x) + y & \to & x+\fS{s}(y)\label{Ex4_1a_AL94_rule2}
\\
\fS{f}(x,y) & \to & z \IF x+y\cto z+z'\label{Ex4_1a_AL94_rule2}
\end{IEEEeqnarray}
The following critical peak
 \begin{IEEEeqnarray}{r'C'l}
\fS{s}(\fS{0}) ~\leftrewAtPos{\toppos}{(\ref{Ex4_1a_AL94_rule2})}   & \fS{f}(\fS{s}(\fS{0}),\fS{0}) & \rewAtPos{\toppos}{(\ref{Ex4_1a_AL94_rule2})}~\fS{0}\label{Ex4_1a_AL94_improperCriticalPeak}
\end{IEEEeqnarray}
is \emph{not} joinable, as both $\fS{s}(\fS{0})$ and $\fS{0}$ are irreducible.
Note that $\cR$ has no proper critical pair.
However, there is an improper CCP which is \emph{not} trivial (nor $\joinability{\cR}$-joinable):
\[\iccpOf{(\ref{Ex4_1a_AL94_rule2})}:\langle z_1,z_2\rangle\IF x+y\cto z_1+z'_1,x+y\cto z_2+z'_2
\]
\end{exa}

\subsubsection{Critical peaks as CCPs}

In this section, 
we investigate the relationship between peaks and conditional pairs.

\begin{prop}[Critical peaks as CCPs]
\label{PropCriticalPeaksAsCCPs}
Let 
$\cR=(\Symbols,\SPredicates,\mu,E,H,R)$ 
be an \egtrs{}.
Let 
$\alpha:\lhsr\to\rhsr\IF\gencond$ and 
$\alpha':\lhsr'\to\rhsr'\IF\gencond'$ be rules, 
$p\in\Pos^\mu_\Symbols(\lhsr)$, and $\sigma$ be such that 
\begin{IEEEeqnarray*}{r'C'l}
\sigma(\lhsr)[\sigma(\rhsr')]_p~\leftrewAtPos{p}{\alpha'}   & \sigma(\lhsr)[\sigma(\lhsr')]_p & \rewAtPos{\toppos}{\alpha}~\sigma(\rhsr)
\end{IEEEeqnarray*}
 is a critical peak and 
$\ccpOf{\alpha,p,\alpha'}
:\langle\theta(\lhsr)[\theta(\rhsr')]_p,\theta(\rhsr)\rangle\IF\theta(\gencond),\theta(\gencond')$
be a CCP. 
Then, there is a substitution $\tau$ such that
$\sigma=\tau\circ\theta$ and  $\ccpOf{\alpha,p,\alpha'}$ is feasible.
\end{prop}

\begin{proof}
The proof is basically that of 
\cite[Proposition 3.7]{Huet_ConfluentReductionsAbstractPropertiesAndApplicationsToTermRewritingSystems_JACM80}.
The feasibility of $\ccpOf{\alpha,p,\alpha'}$ follows from the fact that 
both $\sigma(\gencond)$ and $\sigma(\gencond')$ hold due to their use in a critical peak: since $\sigma$ is a refinement of $\theta$, the condition $\theta(\gencond),\theta(\gencond')$ of $\ccpOf{\alpha,p,\alpha'}$ is feasible.
\end{proof}

\begin{prop}[Feasible CCPs as critical peaks]
\label{PropFeasibleCCPsAsCriticalPeaks}
Let 
$\cR=(\Symbols,\SPredicates,\mu,E,H,R)$ 
be an \egtrs{}.
Let 
$\alpha:\lhsr\to\rhsr\IF\gencond$ and 
$\alpha':\lhsr'\to\rhsr'\IF\gencond'$ be rules, 
$p\in\Pos^\mu_\Symbols(\lhsr)$ be such that
$\ccpOf{\alpha,p,\alpha'}
:\langle\theta(\lhsr)[\theta(\rhsr')]_p,\theta(\rhsr)\rangle\IF\theta(\gencond),\theta(\gencond')$
 is a CCP.
For each substitution $\sigma$ satisfying the conditional part of 
$\ccpOf{\alpha,p,\alpha'}$, there is a critical peak
$\sigma(\theta(\lhsr[\rhsr']_p))\leftrew{\alpha'}\sigma(\theta(\lhsr))\rew{\alpha}\sigma(\theta(\rhsr))$.
\end{prop}

\subsection{Variable peaks and conditional variable pairs}
 \label{SecVariablePeaksAndCVPs}
 
 If $\alpha'$ does \emph{not} overlap $\alpha:\lhsr\to\rhsr\IF\gencond$ in (\ref{NestedRMRPeakTermsRup}), then $\alpha'$ applies to a redex at an active position $p$ which is ``below'' the position of an active variable in $\lhsr$, 
 i.e.,
there is $x\in\Var^\mu(\lhsr)$ and $q\in\Pos^\mu_x(\lhsr)$ 
such that $\lhsr|_q=x$ and $q\leq p$, i.e., $p=q\cpos{}q'$ for some position $q'$.
Thus, (\ref{NestedRMRPeakTermsRup}) is a \emph{variable peak}
 \begin{IEEEeqnarray}{r'C'l}
\underbrace{\sigma(\lhsr)[\sigma(\rhsr')]_p}_{t} ~\leftrewAtPos{p}{\alpha'}   & \underbrace{\sigma(\lhsr)[\sigma(\lhsr')]_p}_s & \rewAtPos{\toppos}{\alpha}~\underbrace{\sigma(\rhsr)}_{t'}\qquad\qquad\text{where }p\notin\Pos^\mu_\Symbols(\lhsr)\label{LblVariablePeak}
\end{IEEEeqnarray}
see Figure \ref{FigVariablePeak}. 
\begin{figure}[t] 
 \begin{center}
\begin{tikzpicture}[node distance = 6cm, auto,scale=1.5,every node/.style={scale=1}]
\node (VP_Ttext) {$\sigma(\lhsr)[\sigma(\rhsr')]_p$};
\node [triangleW, below = 0cm of VP_Ttext, minimum width = 2cm, minimum height = 2.35cm,scale=0.9] (VP_TTriangle) {};
\node [triangleG, below  = 0cm of VP_Ttext, minimum width = 0.85cm, minimum height = 1cm, scale=0.9] (VP_LHSTriangleT) {};
\node [below  = 0.25 cm of VP_Ttext] (VP_TLHStext) {$\lhsr$};
\node [triangleW, below  = 0.14cm of VP_TLHStext, minimum width = 1.3cm, minimum height = 1.32cm, scale=0.9] (VP_SigmaXtriangleT) {};
\node [triangleG, below  = 1.2cm of VP_Ttext, minimum width = 1cm, minimum height = 1.015cm, scale=0.9] (VP_LHSReduced) {};
\node [below  = -0.45 cm of VP_LHSReduced] (VP_RHSptext) {\tiny $\sigma(\rhsr')$};
\node [above right = 1cm and 1.8 cm of VP_Ttext] (VP_Stext) {$\sigma(\lhsr)[\sigma(\lhsr')]_p$};
\node [triangleW, below  = 0cm of VP_Stext, minimum width = 2cm, minimum height = 2.35cm, scale=0.9] (VP_Striangle) {};
\node [triangleG, below  = 0cm of VP_Stext, minimum width = 0.85cm, minimum height = 1cm, scale=0.9] (VP_LHSTriangleS) {};
\node [below  = 0.2 cm of VP_Stext] (VP_LHStextS) {$\lhsr$};
\node [triangleW, below  = 0.16cm of VP_LHStextS, minimum width = 1.3cm, minimum height = 1.35cm, scale=0.9] (VP_SigmaXtriangleS) {};
\node [above right = 1cm and 0.5 cm] at (VP_Striangle) (QinLhsr) {\tiny $q\in \Pos^\mu_x(\lhsr)$};
\node [above right = 0.5cm and 1 cm] at (VP_Striangle) (PinS) {\tiny $p\in \Pos^\mu(s)$};
\node [triangleG, below  = 1.2cm of VP_Stext, minimum width = 1cm, minimum height = 1.01cm, scale=0.9] (VP_LHSpTriangleS) {};
\node [below  = -0.45 cm of VP_LHSpTriangleS] (VP_LHSptext) {\tiny $\sigma(\lhsr')$};
\node [below right =1cm and 2.5 cm of VP_Stext] (VP_Ttext) {$\sigma(\rhsr)$};
\node [triangleW, below  = 0cm of VP_Ttext, minimum width = 3cm, minimum height = 2.25cm] (VP_Ttriangle) {};
\node [triangleG, below  = 0cm of VP_Ttext, minimum width = 1.67cm, minimum height = 1.28cm, scale=0.9] (VP_RHStTriangle) {};
\node [below  = 0.35 cm of VP_Ttext] (VP_RHStextTp) {$\rhsr$};
\node [triangleW, below  left = 0.85cm and 0.15cm of VP_RHStextTp, minimum width = 1.4cm, minimum height = 1.22cm, scale=0.9] (VP_SigmaXtriangleTp) {};
\node [triangleG, below left = 0.61cm and 0.05cm of VP_RHStTriangle, minimum width = 1cm, minimum height = 1.01cm, scale=0.9] (VP_LHSpTriangleT) {};
\node [below  = -0.45 cm of VP_LHSpTriangleT] (VP_LHSptext) {\tiny $\sigma(\lhsr')$}; 
\node [triangleW, below  right = 0.85cm and 0.13cm of VP_RHStextTp, minimum width = 1.4cm, minimum height = 1.22cm, scale=0.9] (VP_SigmaXtriangleTp) {};
\node [triangleG, below right = 0.61cm and 0.04cm of VP_RHStTriangle, minimum width = 1cm, minimum height = 1.01cm, scale=0.9] (VP_LHSpTriangleTBis) {};
\node [below  = -0.45 cm of VP_LHSpTriangleTBis] (VP_LHSptextBis) {\tiny $\sigma(\lhsr')$}; 

\draw[-*] ([shift={(0.2,0.2)}]QinLhsr) -- ([shift={(-0.05,0)}]VP_LHSTriangleS.south);
\draw[-*] ([shift={(0.2,0.2)}]PinS) -- ([shift={(-0.05,0)}]VP_LHSpTriangleS.north);
\draw[->] ([shift={(-0.2,0.2)}]VP_Striangle.west) -- ([shift={(0.1,0.3)}]VP_TTriangle.east);
\node[above right = 0.15cm and 0.2cm of VP_TTriangle.east]{\tiny $\alpha'$};
\draw[->] ([shift={(0.2,0.2)}]VP_Striangle.east) -- ([shift={(-0.2,0.2)}]VP_Ttriangle.west);
\node[above left = 0.08cm and 0.35cm of VP_Ttriangle.west]{\tiny $\alpha$};
\end{tikzpicture}
 \end{center}
\caption{Variable peak} 
 \label{FigVariablePeak}
 \end{figure}
Depending on the \emph{origin} of $\alpha$ and $\alpha'$ in (\ref{LblVariablePeak}), we have different kinds of variable peaks for different uses 
 (see Table \ref{TableDifferentKindsOfVariablePeaks}).
\begin{table}[ht]
\caption{Different kinds of variable peaks (\ref{LblVariablePeak})}
\label{TableDifferentKindsOfVariablePeaks}
\centering
\begin{tabular}{@{}cccll@{}}
\toprule
$\alpha$ & $\alpha'$ & $\gbop$ & Denomination & Use
\\
\midrule
$R^\RMabbr$ & $R^\RMabbr$ & $\rew{}$ & Coherence variable peak & Property $\alpha_{\rew{\cR^\RMabbr}}$, Figure \ref{FigHuetConfluencePropertiesAlphaBetaGamma} (middle)
\\
$R^\RMabbr$ & $\eqoriented{E}$ & $\equone{}$ & Coherence variable peak & Property $\gamma_{\rew{\cR^\RMabbr}}$, Figure \ref{FigHuetConfluencePropertiesAlphaBetaGamma} (right)
\\
$\eqoriented{E}$ & $R^\RMabbr$ & $\rew{}$ & Coherence variable peak & Property $\gamma_{\rew{\cR^\RMabbr}}$, Figure \ref{FigHuetConfluencePropertiesAlphaBetaGamma} (right)
\\
\bottomrule
\end{tabular}
\end{table}

\noindent
Although Huet did not use ``variable peak'', in the realm of \trs{s} 
they were considered in Case 2a of the proof of \cite[Lemma 3.1]{Huet_ConfluentReductionsAbstractPropertiesAndApplicationsToTermRewritingSystems_JACM80}, where he showed that, in contrast to critical peaks, 
variable peaks of \trs{s} are \emph{always joinable}.
In sharp contrast, 
Dershowitz, Okada, and Sivakumar noticed that, dealing with \ctrs{s}, 
this is not true \cite[Section 2.3]{DerOkaSiv_ConfluenceOfConditionalRewriteSystems_CTRS87}.

\subsubsection{Conditional Variable Pairs}

If $\alpha$ and $\alpha'$ are rules of a \gtrs{s} $\cR$ (in particular, if $\cR$ is a \ctrs{}), then variable peaks are characterized as instances of \emph{Conditional Variable Pairs} (CVPs \cite[Definition 59]{Lucas_LocalConfluenceOfConditionalAndGeneralizedTermRewritingSystems_JLAMP24}) of the form $\langle\lhsr[x']_q,\rhsr\rangle\IF x\rew{}x',\gencond$, where 
$x$ is an active variable of $\lhsr$,
$q$ is an active position of $x$ in $\lhsr$,  
and $x'$ is a fresh variable \cite[Proposition 32]{Lucas_LocalConfluenceOfConditionalAndGeneralizedTermRewritingSystems_JLAMP24}.
Note that reductions using $\alpha'$ are subsumed by the condition $x\rew{}x'$, expressing one-step rewriting using an arbitrary rule in $\cR$.

In this paper, though, we consider variable peaks (\ref{LblVariablePeak}) 
for rules 
$\alpha$ and $\alpha'$ 
coming from \emph{possibly different} sets of rules ($R^\RMabbr$ and $\eqoriented{E}$).
For this reason, we use an additional parameter $\gbop$ in CVPs, representing some kind of rewriting (e.g., $\rew{\cR^\RMabbr}$ or $\rew{\eqoriented{E}}$).
Given a rule $\alpha:\lhsr\to\rhsr\IF\gencond$, 
an active variable $x\in\Var^\mu(\lhsr)$ in $\lhsr$, 
a fresh variable $x'$ (not occurring in $\alpha$),
an active position $q\in\Pos^\mu_x(\lhsr)$ of such a variable, 
and a binary predicate $\gbop$, 
a
\emph{Conditional $\gbop$-Variable Pair} 
$\cvpOf{\gbop}{\alpha,x,q}$ is as follows:\footnote{In \cite[footnote 1, page 1166]{JouKir_CompletionOfASetOfRulesModuloASetOfEquations_SIAMJC86}, Jouannaud and Kirchner use \emph{variable critical pairs} to refer to
those coming from overlaps of ``rules'' like, e.g., $x\to x+0$ (see Remark \ref{RemCoherenceCCPsFromIllDefinedRules}). However, they do not further consider them in \cite{JouKir_CompletionOfASetOfRulesModuloASetOfEquations_SIAMJC86}.
Note
that ``variable critical pairs'' and CVPs capture rather different situations.}
\begin{IEEEeqnarray}{r'C'l}
\cvpOf{\gbop}{\alpha,x,q}:\langle s[x']_q,t\rangle\IF x\mathrel{\gbop} x', \gencond\label{LblGenericCVP}
\end{IEEEeqnarray}
We often talk about conditional variable pairs rather than 
conditional $\gbop$-variable pairs; furthermore,
we also use $\gbop$-CVP, or just CVP if no confusion arises as an abbreviation for ``Conditional Variable Pair''.
\begin{rem}[CVPs from ill-defined rules]
If an ill-defined rule $\alpha:x\to\rhsr\IF\gencond$ is in $\eqoriented{E}$,
a CVP $\cvpOf{\gbop}{\alpha,x,\toppos}:\langle x',\rhsr\rangle\IF x\spcrel{\gbop} x',\gencond$ is obtained.
\end{rem}
%
\noindent
Given a set of rules $U$, we let
\begin{IEEEeqnarray}{r'C'l}
\SetOf{CVP}^{\:\gbop}(U) & = & \{\cvpOf{\gbop}{\alpha,x,q} \mid \alpha:\lhsr\to\rhsr\IF\gencond \in U, x\in\Var^\mu(\lhsr), q\in\Pos^\mu_x(\lhsr)\}
\end{IEEEeqnarray}
be the set of CVPs obtained from rules in $U$.

\begin{rem}[Asymmetric role of rules in variable peaks and CVPs]
\label{RemAssymetricRoleOfRulesInCVPs}
Rules $\alpha:\lhsr\to\rhsr\IF\gencond$ and $\alpha'$ in variable peaks  (\ref{LblVariablePeak}) play different \emph{roles} to obtain CVPs: variable $x$ and position $q$ are taken from $\lhsr$, whereas $\gbop$ represents reduction with $\alpha'$.
Thus, if $\alpha$ and $\alpha'$ are taken from different sets $U$ and $U'$, then both $\SetOf{CVP}^{\:\gbop'}(U)$ and $\SetOf{CVP}^{\:\gbop}(U')$ are necessary to capture the considered variable peaks.
In particular, for variable 
peaks (\ref{LblVariablePeak}), we use 
$\SetOf{CVP}^{\:\equone{}}(R^\RMabbr)$ (if $\alpha\in R^\RMabbr$ and $\alpha'\in\eqoriented{E}$) and 
$\SetOf{CVP}^{\:\rew{}}(\eqoriented{E})$ 
(for $\alpha\in \eqoriented{E}$ and $\alpha'\in R^\RMabbr$).
\end{rem}

\subsubsection{Variable peaks as CVPs} We show that variable peaks are captured as CVPs.

\begin{prop}[Variable peaks as CVPs]
\label{PropVariablePeaksAsCVPs}
Let 
$\cR=(\Symbols,\SPredicates,\mu,E,H,R)$ 
be an \egtrs{}.
Let 
$\alpha:\lhsr\to\rhsr\IF\gencond$ and 
$\alpha':\lhsr'\to\rhsr'\IF\gencond'$ be rules, 
and $p\notin\Pos^\mu_\Symbols(\lhsr)$ define a variable peak 
 \begin{IEEEeqnarray*}{r'C'l}
\sigma(\lhsr)[\sigma(\rhsr')]_p ~\leftrewAtPos{p}{\alpha'}   & \sigma(\lhsr)[\sigma(\lhsr')]_p & \rewAtPos{\toppos}{\alpha}~\sigma(\rhsr)
\end{IEEEeqnarray*}
 for some substitution $\sigma$,
$q\in\Pos^\mu_x(\lhsr)$ for some $x\in\Var^\mu(\lhsr)$, 
and $p=q.p'$ for some position $p'$.
Then, $\cvpOf{\gbop}{\alpha,x,q}:\langle\lhsr[x']_q,\rhsr\rangle\IF x\spcrel{\gbop} x',\gencond$, 
where $x'$ is a fresh variable, is 
feasible. In particular,
\begin{enumerate}
\item  
$\sigma(x)\rewAtPos{p'}{\alpha'}\sigma(x)[\sigma(\rhsr')]_{p'}=\sigma(x')$  holds if $\alpha'\in R^\RMabbr$ and $\spcrel{\gbop}$ is $\rew{}$; or
\item $\sigma(x)\rewAtPos{p'}{\alpha'}\sigma(x)[\sigma(\rhsr')]_{p'}=\sigma(x')$ holds if $\alpha'\in\:\eqoriented{E}$ and $\spcrel{\gbop}$ is $\equone{}$.
\end{enumerate}
\end{prop}

\begin{proof}
By definition of variable peak (see (\ref{LblVariablePeak})), we have 
$s=\sigma(\lhsr)=\sigma(\lhsr[x]_q)=\sigma(\lhsr)[\sigma(\lhsr')]_p$,
$t=\sigma(\lhsr)[\sigma(\rhsr')]_p$,
$t'=\sigma(\rhsr)$,
and $\sigma(\gencond)$ holds.
Hence, $\sigma(x)=s|_q\spcrel{\gbop}t|_q=\sigma(x)[\sigma(\rhsr')]_{p'}$.
Since $x'$ is a fresh variable, we can extend $\sigma$ to
$\sigma(x')=t|_q$.
Hence, $\sigma(x)\spcrel{\gbop}\sigma(x')$.
Thus, since $\sigma(\gencond)$ holds by the definition of variable peak, the feasibility of $\cvpOf{\gbop}{\alpha,x,q}$ follows. 
\end{proof}

\begin{prop}[Feasible CVPs as variable peaks]
\label{PropFeasibleCVPsAsVariablePeaks}
Let 
$\cR$ 
be an \egtrs{}, 
$\alpha:\lhsr\to\rhsr\IF\gencond$ be a rule,
$x\in\Var^\mu(\lhsr)$, and
$q\in\Pos^\mu_x(\lhsr)$ be such that
$\cvpOf{\rew{}}{\alpha,x,q}$   
(resp.\ $\cvpOf{\equone{}}{\alpha,x,q}$) is a CVP.
Every substitution $\sigma$ satisfying its conditional part determines a variable peak 
 \begin{IEEEeqnarray*}{r'C'l}
\sigma(\lhsr)[\sigma(\rhsr')]_p ~\leftrewAtPos{p}{\alpha'}   & \sigma(\lhsr)[\sigma(\lhsr')]_p & \rewAtPos{\toppos}{\alpha}~\sigma(\rhsr)
\end{IEEEeqnarray*}
 for some $p\geq q$ and
$\alpha'\in R^\RMabbr$ (resp.\ $\alpha'\in\: \eqoriented{E}$).
\end{prop}

\subsection{Replacement restrictions, conditional rules, and the joinability of variable peaks}
\label{SecAboutJoinabilityOfVariablePeaks}

If only unconditional rules $\lhsr\to\rhsr$ are allowed, and 
$\mu=\muTop$, then  a variable peak (\ref{LblVariablePeak}) is 
always $\joinability{\cR}$-joinable.
Let $x\in\Var(\lhsr)$ and $q\in\Pos_x(\lhsr)$ be
such that $\lhsr|_q=x$ and $q\leq p$, i.e., $p=q\cpos{}q'$ for some position $q'$. 
In general we can write $\lhsr=\lhsr[x]_{q_1}\cdots[x]_{q_m}$ 
(or  $\lhsr=\lhsr[x]_{Q_x}$ for short) for (disjoint) positions $q_i$ of $x$ in $\lhsr$ for some $m\geq 1$ and $Q_x=\{q_1,\ldots,q_m\}$
(where $q=q_i$ for some $1\leq i\leq m$).
Hence, $s=\sigma(\lhsr)[u]_{Q_x}$, for $u=\sigma(\lhsr')$ and
$t = \sigma(\lhsr)[u]_{Q_x-\{q\}}[v]_{q}$, where $v=\sigma(\rhsr')$. 
Similarly, let $P_x$ be the (possibly empty) set of positions of variable $x$ in $\rhsr$.
Thus, $t'=\sigma(\rhsr)[u]_{P_x}$ and we have the following  sequences proving $\joinability{\cR}$-joinability of $t$ and $t'$
(for \trs{s} this is observed in Case 2a in the proof of \cite[Lemma 3.1]{Huet_ConfluentReductionsAbstractPropertiesAndApplicationsToTermRewritingSystems_JACM80}):
\begin{IEEEeqnarray}{r'C'l}
t & = & \sigma(\lhsr)[u]_{Q_x-\{q\}}[v]_{q}\rews{\alpha'}\sigma(\lhsr)[v]_{Q_x}
\label{JoinabilityOfVariablePeaks_RewritingInstancesOfXinL}
\\
& \rewAtPos{\toppos}{\alpha} & \sigma(\rhsr)[v]_{P_x} = w \qquad\qquad\qquad\qquad\qquad\text{ and }
\label{JoinabilityOfVariablePeaks_AlphaAppliesToT}
\\
t' & = & \sigma(\rhsr)[u]_{P_x}\rews{\alpha'} \sigma(\rhsr)[v]_{P_x} = w
\label{JoinabilityOfVariablePeaks_RewritingInstancesOfXinR}
\end{IEEEeqnarray}

\noindent
Dealing with \emph{arbitrary} \egtrs{s}, 
possibly using replacement restrictions and conditional rules, it is well-known that
such $\joinability{\cR}$-joinability may \emph{fail} to be obtained:
\begin{enumerate}
\item\label{JoinabilityOfVariablePeaks_RewritingInstancesOfXinL_fails} 
If $x$ occurs \emph{frozen} in $\lhsr$, then sequence (\ref{JoinabilityOfVariablePeaks_RewritingInstancesOfXinL}) \emph{cannot} be completed, as rewritings on frozen positions would be necessary. For \cstrs{s}, see \cite[Section 4.1]{LucVitGut_ProvingAndDisprovingConfluenceOfContextSensitiveRewriting_JLAMP22}.
\item\label{JoinabilityOfVariablePeaks_SatisfactionRuleCondition_fails}
 If $\alpha$ is a \emph{conditional rule} $\lhsr\to\rhsr\IF\gencond$, 
then the step (\ref{JoinabilityOfVariablePeaks_AlphaAppliesToT}) may \emph{fail} if $\varsigma(\gencond)$ does not hold for the matching substitution $\varsigma$ to be used with $\alpha$, where $\varsigma(x)=w$ and $\varsigma(y)=\sigma(y)$ for all $y\neq x$.  For \ctrs{s}, this is observed 
in \cite[Section 2.3]{DerOkaSiv_ConfluenceOfConditionalRewriteSystems_CTRS87}.
Finally,
\item\label{JoinabilityOfVariablePeaks_RewritingInstancesOfXinR_fails}
If $x$ occurs \emph{frozen} in $\rhsr$, then 
(\ref{JoinabilityOfVariablePeaks_RewritingInstancesOfXinR}) \emph{cannot} be completed, see again \cite[Section 4.1]{LucVitGut_ProvingAndDisprovingConfluenceOfContextSensitiveRewriting_JLAMP22}.
\end{enumerate}
Thus, in the analysis of peaks (\ref{NestedRMRPeakTermsRup}), besides CCPs to investigate $\ejoinability{\cR}$-joinability of critical peaks (\ref{LblCriticalPeak}) we also need CVPs to investigate $\ejoinability{\cR}$-joinability of variable peaks (\ref{LblVariablePeak}).

\subsection{Proving $\ejoinability{\cR}$-joinability of conditional variable pairs}
\label{SecProvingJoinabilityOfCVPs}

The following definitions establish some restrictions on the rules of an \egtrs{} which are useful to avoid some of the problems discussed in the previous section.

\begin{defi}[$\mu$-homogeneous term; left $\mu$-homogeneous rule]
Given a replacement map $\mu$, we say that a term $t$ is 
\emph{$\mu$-homogeneous} if no variable occurs 
both active and frozen in $t$, i.e., $\Var^\mu(t)\cap\NVar{\mu}(t)=\emptyset$.
A rule $\lhsr\to\rhsr\IF\gencond$ is \emph{left $\mu$-homogeneous} if $\lhsr$ is $\mu$-homogeneous.
\end{defi}
\begin{defi}[$\mu$-compatible rule {\rm \cite[Definition 2]{Lucas_ConfluenceOfAlmostParallelClosedGeneralizedTermRewritingSystems_CADE25}}]
\label{DefMuCompatibleRule}
A rule $\lhsr\to\rhsr\IF\gencond$ is \emph{$\mu$-compatible} if 
(i) 
no active variable in $\lhsr$ is frozen in $\rhsr$, i.e., $\Var^\mu(\lhsr)\cap\NVar{\mu}(\rhsr)=\emptyset$, and
(ii) 
no active variable in $\lhsr$ occurs in $\gencond$, i.e., $\Var^\mu(\lhsr)\cap\Var(\gencond)=\emptyset$.
\end{defi}

\begin{exa}
\label{ExSumListsA_MuHomogeneousAndCompatible}
All rules in $R~\cup\eqoriented{E}$ in Example \ref{ExSumListsA} are left $\mu$-homogeneous and $\mu$-compatible.
\end{exa}
%
\noindent
Left $\mu$-homogeneous rules \emph{avoid} the problem discussed in item (\ref{JoinabilityOfVariablePeaks_RewritingInstancesOfXinL_fails}) of Section~\ref{SecAboutJoinabilityOfVariablePeaks}.
For $\mu$-compatible rules, the problem discussed in item (\ref{JoinabilityOfVariablePeaks_SatisfactionRuleCondition_fails}) of Section~\ref{SecAboutJoinabilityOfVariablePeaks} disappears 
due to condition (ii) in Definition \ref{DefMuCompatibleRule}; 
and the problem discussed in item (\ref{JoinabilityOfVariablePeaks_RewritingInstancesOfXinR_fails}) of Section~\ref{SecAboutJoinabilityOfVariablePeaks} disappears due to condition (i).
Thus, we have the following.

\begin{prop}[$\joinability{\cR^\RMabbr}$-joinable variable peaks I]
\label{PropJoinableVariablePeaksForCriticalVariablesNotInCondition}
Let 
$\cR=(\Symbols,\SPredicates,\mu,E,H,R)$ 
be an \egtrs{}.
Let 
$\alpha:\lhsr\to\rhsr\IF\gencond,\alpha':\lhsr'\to\rhsr'\IF\gencond'\in R^\RMabbr$ be such that 
$\alpha$ is left $\mu$-homogeneous and $\mu$-compatible.
Let $p\notin\Pos^\mu_\Symbols(\lhsr)$ be such that 
 \begin{IEEEeqnarray}{r'C'l}
\sigma(\lhsr)[\sigma(\rhsr')]_p ~\leftrewAtPos{p}{\alpha'}   & \sigma(\lhsr)[\sigma(\lhsr')]_p & \rewAtPos{\toppos}{\alpha}~\sigma(\rhsr)
\label{PropJoinableVariablePeaksForCriticalVariablesNotInCondition_LblVariablePeak}
\end{IEEEeqnarray}
is  a variable peak for some
substitution $\sigma$, 
$x\in\Var^\mu(\lhsr)$, and
$q\in\Pos^\mu_x(\lhsr)$ such that $p=q.p'$ for some position $p'$.
Then, (\ref{PropJoinableVariablePeaksForCriticalVariablesNotInCondition_LblVariablePeak}) is 
$\lsjoinability{\cR^\RMabbr}$-joinable.
\end{prop}
\begin{proof}
Since $p\in\Pos^\mu(\sigma(\lhsr))-\Pos^\mu_\Symbols(\lhsr)$, we have 
$\lhsr|_q=x$ for some $x\in\Var^\mu(\lhsr)$ and $q\in\Pos^\mu_x(\lhsr)$
such that
$p=q\cpos{}u$ for some position $u$.
Note that $\sigma(x)=\sigma(x)[\sigma(\lhsr')]_u\rewAtPos{u}{\alpha'}\sigma(x)[\sigma(\rhsr')]_u$.
Let $Q=\{q_1,\ldots,q_n\}$ be the set of positions of occurrences of $x$ in $\lhsr$
(i.e., $q\in Q$).
Since $\lhsr$ is $\mu$-homogeneous, $Q\subseteq\Pos^\mu(\lhsr)$.
Thus, $\sigma(\lhsr)=\sigma(\lhsr)[\sigma(x)]_Q$
and we have $t = \sigma(\lhsr)[\sigma(\rhsr')]_p\rews{\alpha'}\sigma'(\lhsr)$, where
$\sigma'(x)=\sigma(x)[\sigma(\rhsr')]_u$ and $\sigma'(y)=\sigma(y)$ for all $y\neq x$.
By $\mu$-compatibility, 
$x$ does not occur in the conditional part $\gencond$ of $\alpha$.
Since $\sigma(\gencond)$ holds, we have that $\sigma'(\gencond)$ holds as well.
Thus, we have
$\sigma'(\lhsr)\rewAtPos{\toppos}{\alpha}\sigma'(\rhsr)$, i.e.,
$t\rews{\cR^\RMabbr}\sigma'(\rhsr)$.
We also have $t'=\sigma(\rhsr)\rews{\alpha'}\sigma'(\rhsr)$ by just reducing any possible occurrence of $\sigma(x)$ in $t'$, i.e., $\sigma(\rhsr)$, using $\alpha'$ (which is possible by $\mu$-compatibility). 
Hence, $t\rews{\alpha'}\sigma'(\lhsr)\rewAtPos{\toppos}{\alpha}\sigma'(\rhsr)\leftrews{\alpha'}\sigma(\rhsr)=t'$ proves $\lsjoinability{\cR^\RMabbr}$-joinability of (\ref{LblVariablePeak}).
\end{proof}
\begin{cor}
\label{CoroJoinableVariablePairsForCriticalVariablesNotInCond}
Let $\cR$  
be an \egtrs{} and 
  $\alpha\in R^\showRMabbr$.
Then, $\pi^\to_{\alpha,x,q}$ is $\joinability{\cR^\RMabbr}$-joinable if 
$\alpha$ is left $\mu$-homogeneous and $\mu$-compatible.
\end{cor}
%
\noindent
The following examples show that the assumptions in Proposition \ref{PropJoinableVariablePeaksForCriticalVariablesNotInCondition} cannot be dropped without risking the result.

\begin{exa}
\label{Ex7_3_3_Ohl02}
The rules of the following (oriented) CTRS $\cR$ \cite[Example 7.3.3]{Ohlebusch_AdvancedTopicsInTermRewriting_2002}:
\begin{eqnarray}
\fS{a} & \to & \fS{b}\label{Ex7_3_3_Ohl02_rule1}\\
\fS{f}(x) & \to & \fS{c}\IF x\cto\fS{a}\label{Ex7_3_3_Ohl02_rule2}
\end{eqnarray}
do \emph{not} overlap.
However,
we have a non-$\joinability{\cR}$-joinable variable peak
\begin{eqnarray*}
\fS{f}(\fS{b})\leftrewAtPos{1}{(\ref{Ex7_3_3_Ohl02_rule1})} & \fS{f}(\fS{a}) & \rewAtPos{\toppos}{(\ref{Ex7_3_3_Ohl02_rule2})}\fS{c}\label{Ex7_3_3_Ohl02_VariablePeak}
\end{eqnarray*}
as $\fS{f}(\fS{b})$ and $\fS{c}$ 
are \emph{irreducible}.
Note that rule (\ref{Ex7_3_3_Ohl02_rule2}) is \emph{not} compatible.
\end{exa}

\begin{exa}
Regarding the role of replacement restrictions, consider the one-rule \trs{}
\begin{IEEEeqnarray}{+r'C'l+}
\fS{f}(x,x) & \to & \fS{g}(x)\label{ExRoleRepRestrictionsVariablePeak_rule1}
\end{IEEEeqnarray}
\end{exa}
\begin{itemize}
\item If $\mu(\fS{f})=\mu(\fS{g})=\{1\}$, then the following variable peak  is \emph{not} $\joinability{\cR}$-joinable:
\begin{eqnarray*}
\fS{f}(\fS{g}(x),\cfrozen{\fS{f}(x,x)})\leftrewAtPos{1}{(\ref{ExRoleRepRestrictionsVariablePeak_rule1})} & \fS{f}(\fS{f}(x,\cfrozen{x}),\cfrozen{\fS{f}(x,x)}) & \rewAtPos{\toppos}{(\ref{ExRoleRepRestrictionsVariablePeak_rule1})}\fS{g}(\fS{f}(x,\cfrozen{x}))\label{ExRoleRepRestrictionsVariablePeak_VariablePeak}
\end{eqnarray*}
as $\fS{f}(\fS{g}(x),\cfrozen{\fS{f}(x,x)})$ is irreducible and $\fS{g}(\fS{f}(x,\cfrozen{x}))$ only rewrites to $\fS{g}(\fS{g}(x))$.
 Note that (\ref{ExRoleRepRestrictionsVariablePeak_rule1}) is $\mu$-compatible but \emph{not} left $\mu$-homogeneous.
\item If $\mu(\fS{f})=\{1,2\}$ and $\mu(\fS{g})=\emptyset$, then 
\begin{eqnarray*}
\fS{f}(\fS{g}(\cfrozen{x}),\fS{f}(x,x))\leftrewAtPos{1}{(\ref{ExRoleRepRestrictionsVariablePeak_rule1})} & \fS{f}(\fS{f}(x,x),\fS{f}(x,x)) & \rewAtPos{\toppos}{(\ref{ExRoleRepRestrictionsVariablePeak_rule1})}\fS{g}(\cfrozen{\fS{f}(x,\cfrozen{x})})\label{ExRoleRepRestrictionsVariablePeak_VariablePeak}
\end{eqnarray*}
 is \emph{not} $\joinability{\cR}$-joinable, as $\fS{g}(\cfrozen{\fS{f}(x,\cfrozen{x})})$ is irreducible and $\fS{f}(\fS{g}(\cfrozen{x}),\fS{f}(x,x))$ cannot be rewritten into $\fS{g}(\cfrozen{\fS{f}(x,\cfrozen{x})})$. Now, (\ref{ExRoleRepRestrictionsVariablePeak_rule1}) is left $\mu$-homogeneous, but not $\mu$-compatible.
\end{itemize}
\begin{prop}[$\ejoinability{\cR}$-joinable variable peaks II]
\label{PropJoinableVariableEupCPeaksForCriticalVariablesNotInCondition}
Let 
$\cR=(\Symbols,\SPredicates,\mu,E,H,R)$ 
be an \egtrs{}.
Let 
$\alpha:\lhsr\to\rhsr\IF\gencond\in\:\eqoriented{E}$ be 
left $\mu$-homogeneous and $\mu$-compatible and 
$\alpha':\lhsr'\to\rhsr'\IF\gencond'\in R^\RMabbr$. 
Let $p\notin\Pos^\mu_\Symbols(\lhsr)$ be such that 
 \begin{IEEEeqnarray}{r'C'l}
\sigma(\lhsr)[\sigma(\rhsr')]_p ~\leftrewAtPos{p}{\alpha'}   & \sigma(\lhsr)[\sigma(\lhsr')]_p & \rewAtPos{\toppos}{\alpha}~\sigma(\rhsr)
\label{PropJoinableVariableEupCPeaksForCriticalVariablesNotInCondition_LblVariablePeak}
\end{IEEEeqnarray}
is  a variable peak for some
substitution $\sigma$, 
$x\in\Var^\mu(\lhsr)$, and
$q\in\Pos^\mu_x(\lhsr)$ such that $p=q.p'$ for some position $p'$.
Then, (\ref{PropJoinableVariableEupCPeaksForCriticalVariablesNotInCondition_LblVariablePeak}) is 
 $\ejoinability{\cR^\RMabbr}$-joinable.
\end{prop}
\begin{proof}
Since $p\notin\Pos^\mu_\Symbols(\lhsr)$, there is $x\in\Var^\mu(\lhsr)$
and $q\in\Pos^\mu_x(\lhsr)$ such that $\lhsr|_q=x$ and
$p=q\cpos{}u$ for some position $u$.
Note that $\sigma(x)=\sigma(x)[\sigma(\lhsr')]_u\rewAtPos{u}{\alpha'}\sigma(x)[\sigma(\rhsr')]_u$.
Let $Q=\{q_1,\ldots,q_n\}$ be the set of positions of occurrences of $x$ in $\lhsr$
(i.e., $q\in Q$).
Since $\lhsr$ is $\mu$-homogeneous, $Q\subseteq\Pos^\mu_x(\lhsr)$.
Thus, $\sigma(\lhsr)=\sigma(\lhsr)[\sigma(x)]_Q$
and we have $t = \sigma(\lhsr)[\sigma(\rhsr')]_p\rews{\alpha'}\sigma'(\lhsr)$, where
$\sigma'(x)=\sigma(x)[\sigma(\rhsr')]_u$ and $\sigma'(y)=\sigma(y)$ for all $y\neq x$.
By $\mu$-compatibility of $\alpha$, 
$x$ does not occur in its conditional part $\gencond$.
Thus, since $\sigma(\gencond)$ holds, $\sigma'(\gencond)$ holds as well.
Hence,
$\sigma'(\lhsr)\rewAtPos{\toppos}{\alpha}\sigma'(\rhsr)$. Since $\alpha\in\:\eqoriented{E}$,
$t\rews{\cR^\RMabbr}\circ\rew{\eqoriented{E}}\sigma'(\rhsr)$.
We also have $t'=\sigma(\rhsr)\rews{\alpha'}\sigma'(\rhsr)$ by just reducing any possible occurrence of $\sigma(x)$ in $t'$ using $\alpha'$ (which is possible by $\mu$-compatibility of $\alpha$), i.e.,
$t\rews{\alpha'}\sigma'(\lhsr)\equequ{E}\sigma'(\rhsr)\leftrews{\alpha'}t'$,
 thus proving $\ejoinability{\cR^\RMabbr}$-joinability of (\ref{PropJoinableVariableEupCPeaksForCriticalVariablesNotInCondition_LblVariablePeak}).
\end{proof}
\begin{cor}
\label{CoroJoinableVariablePairsII}
Let $\cR$ 
be an \egtrs{} and 
 $\alpha\in\: \eqoriented{E}$.
Then, $\pi^\to_{\alpha,x,q}$ is $\ejoinability{\cR^\RMabbr}$-joinable if 
$\alpha$ is left $\mu$-homogeneous and $\mu$-compatible.
\end{cor}

\begin{defi}[$\mu$-linear term; $\mu$-left-linear rule and \egtrs\ {\rm \cite[Definition 1]{Lucas_ConfluenceOfAlmostParallelClosedGeneralizedTermRewritingSystems_CADE25}}]
\label{DefMuLeftLinearRule}
A term 
is \emph{$\mu$-linear} 
if no active variable in it
occurs more than once (active or not).
A rule $\lhsr\to\rhsr\IF\gencond$ is \emph{$\mu$-left-linear} if $\lhsr$ is $\mu$-linear.
An \egtrs{} $\cR=(\Symbols,\SPredicates,\mu,E,H,R)$ 
is $\mu$-left-linear if all rules in $R$ 
are.
\end{defi}
%
\noindent
Linear terms are, of course, $\mu$-left-linear, but not vice versa.
For instance, $t=\fS{f}(x,x)$ is $\muBot$-left-linear, as no occurrence of $x$ is active in $t$ due to $\muBot(\fS{f})=\emptyset$.

\begin{prop}[$\rsejoinability{\cR^\RMabbr,E}$- and $\rsejoinability{\cR^\RMabbr}$-joinable variable peaks III]
\label{PropJoinableVariableEdownCPeaksForCriticalVariablesNotInCondition}
Let 
$\cR=(\Symbols,\SPredicates,\mu,E,H,R)$ 
be an \egtrs{}.
Let 
$\alpha:\lhsr\to\rhsr\IF\gencond\in R^\RMabbr$ be 
left $\mu$-homogeneous and $\mu$-compatible and 
$\alpha':\lhsr'\to\rhsr'\IF\gencond'\in\:\eqoriented{E}$.
Let $p\notin\Pos^\mu_\Symbols(\lhsr)$ be such that 
 \begin{IEEEeqnarray}{r'C'l}
\sigma(\lhsr)[\sigma(\rhsr')]_p ~\leftrewAtPos{p}{\alpha'}   & \sigma(\lhsr)[\sigma(\lhsr')]_p & \rewAtPos{\toppos}{\alpha}~\sigma(\rhsr)
\label{PropJoinableVariableEdownCPeaksForCriticalVariablesNotInCondition_LblVariablePeak}
\end{IEEEeqnarray}
is  a variable peak for some
substitution $\sigma$, 
$x\in\Var^\mu(\lhsr)$, and
$q\in\Pos^\mu_x(\lhsr)$ such that $p=q.p'$ for some position $p'$.
Then, (\ref{PropJoinableVariableEdownCPeaksForCriticalVariablesNotInCondition_LblVariablePeak}) is $\rsejoinability{\cR^\RMabbr,E}$-joinable.
Furthermore, if $\alpha$ is $\mu$-left-linear, then (\ref{PropJoinableVariableEdownCPeaksForCriticalVariablesNotInCondition_LblVariablePeak}) is 
$\rsejoinability{\cR^\RMabbr}$-joinable.
\end{prop}
\begin{proof}
Since $p\notin\Pos^\mu_\Symbols(\lhsr)$, there is 
$x\in\Var^\mu(\lhsr)$ and $q\in\Pos^\mu_x(\lhsr)$ such that 
$\lhsr|_q=x$ and
$p=q\cpos{}u$ for some position $u$.
Note that $\sigma(x)=\sigma(x)[\sigma(\lhsr')]_u\rewAtPos{u}{\alpha'}\sigma(x)[\sigma(\rhsr')]_u$.
Let $U=\{u_1,\ldots,u_n\}$, for some $n\geq 0$, 
be the set of positions of occurrences of $x$ in $\rhsr$.
Thus, $t=\sigma(\rhsr)=\sigma(\rhsr)[\sigma(x)]_U$.
By $\mu$-compability, no occurrence of $x$ is frozen in $\rhsr$.
Thus, we have
$t\rews{\alpha'}\sigma'(\rhsr)$ for some substitution $\sigma'$ such that
$\sigma'(x)=\sigma(x)[\sigma(\rhsr')]_u$
and $\sigma'(y)=\sigma(y)$ if $y\neq x$, i.e., for each occurrence of $x$ in $\rhsr$,
$\sigma(x)$ is reduced using $\alpha'$, i.e., for all $z\in\Variables$,
$\sigma(z)\equequ{E}\sigma'(z)$.
Since $\lhsr$ is $\mu$-homogeneous, all occurrences of $x$ in $\lhsr$ are active and we have
$t'=\sigma(\lhsr)[\sigma(\rhsr')]_p\rews{\alpha'}\sigma'(\lhsr)$, i.e.,
$t'\equequ{E}\sigma'(\lhsr)$.
By $\mu$-compability, $x$ does not occur in the conditional part $\gencond$ of $\alpha$.
Since $\sigma(\gencond)$ holds, 
we have that $\sigma'(\gencond)$ holds as well.
Thus, $t\equequ{E}\sigma'(\rhsr)\leftrewAtPos{\toppos}{\alpha}\sigma'(\lhsr)\equequ{E}t'$, i.e.,
$t\equequ{E}\sigma'(\rhsr)\leftrew{\cR^\RMabbr,E}t'$.
Hence, 
(\ref{PropJoinableVariableEdownCPeaksForCriticalVariablesNotInCondition_LblVariablePeak}) is $\rsejoinability{\cR^\RMabbr,E}$-joinable.

If $\lhsr$ is $\mu$-left-linear, then $q$ is the only position of $x$ in $\lhsr$
and 
$t\equequ{E}\sigma'(\rhsr)\leftrewAtPos{\toppos}{\alpha}\sigma'(\lhsr)=t'$,
thus proving 
$\rsejoinability{\cR^\RMabbr}$-joinability
of (\ref{PropJoinableVariableEdownCPeaksForCriticalVariablesNotInCondition_LblVariablePeak}).
\end{proof}
%
\noindent
Note that $\mu$-left-linear rules are left $\mu$-homogeneous. Thus, we have:

\begin{cor}
\label{CoroJoinableVariablePairsIII}
Let $\cR=(\Symbols,\SPredicates,\mu,E,H,R)$ be an \egtrs{} and 
 $\alpha\in R^\RMabbr$.
Then, $\cvpOf{\equone{}}{\alpha,x,q}$ is $\ejoinability{\cR^\RMabbr}$-joinable if 
$\alpha$ is $\mu$-left-linear and $\mu$-compatible.
\end{cor}

\section{$E$-Confluence of $\cR$ as the Church-Rosser property modulo $E$ of $\rew{\cR^\RMabbr}$}
\label{SecEConfluenceWithRandCpeaks}

Corollary \ref{CoroSufficientCriteriaForEConfluenceOfEGRSs}(\ref{CoroSufficientCriteriaForEConfluenceOfEGRSs_Huet}) shows that
$E$-confluence of \egtrs{s} $\cR=(\Symbols,\SPredicates,\mu,E,H,R)$
follows if  the $\alpha$ and $\gamma$ properties can be proved for $\rew{\cR^\RMabbr}$, using $\rew{\eqoriented{E}}$.
Section~\ref{SecConditionalPairsAlphaAndGamma} shows that such properties can be guaranteed as the $\ejoinability{\cR^\RMabbr}$-joinability of
critical peaks (\ref{LblCriticalPeak}) and variable peaks (\ref{LblVariablePeak}), i.e.,
 peaks
\begin{IEEEeqnarray*}{r'C'l+}
\raisebox{0.6cm}{\resizebox{6.5cm}{!}{
\xymatrix{
&  \sigma(\lhsr)[\sigma(\lhsr')]_p 
\ar@{->}[dl]_{p}^>>{\alpha'}
\ar@{->}[dr]^{\toppos}_>>{\alpha} 
\\
\sigma(\lhsr)[\sigma(\lhsr')]_p\qquad  & &  \quad\sigma(\rhsr) 
  \\
}}}
\end{IEEEeqnarray*}
where $\alpha:\lhsr\to\rhsr\IF\gencond$ and $\alpha':\lhsr'\to\rhsr'\IF\gencond'$ are rules and 
\begin{enumerate}
\item $p\in\Pos^\mu_\Symbols(\lhsr)$ for critical peaks (\ref{LblCriticalPeak}) 
and  
\item $p\notin\Pos^\mu_\Symbols(\lhsr)$ for variable peaks (\ref{LblVariablePeak}).
\end{enumerate}
Propositions \ref{PropCriticalPeaksAsCCPs} and \ref{PropVariablePeaksAsCVPs} establish that this can be done by means of appropriate (feasible) CCPs and CVPs.
Now, according to Remarks \ref{RemAssymetricRoleOfRulesInCCPs} and \ref{RemAssymetricRoleOfRulesInCVPs}, 
we consider the following sets of CCPs and CVPs.

\begin{defi}
\label{DefSetsOfCCPsAndCVPsHuet}
Given an \egtrs{} $\cR=(\Symbols,\SPredicates,\mu,E,H,R)$, we consider 
\begin{enumerate}
\item Local confluence conditional pairs:
\begin{enumerate}
\item $\CCP(\cR)$ is the subset of \emph{feasible} CCPs from $\SetOf{CCP}(R^\RMabbr,R^\RMabbr)$.
\item $\CVPofR(\cR)$ is the subset of \emph{feasible} CVPs from $\SetOf{CVP}^\rew{}(R^\RMabbr)$.
\end{enumerate}
\item Local coherence conditional pairs:
\begin{enumerate}
\item $\CCP(\cR,E)$  is the subset 
of \emph{feasible} CCPs from 
$\SetOf{CCP}(R^\RMabbr,\eqoriented{E})$.
\item $\CCP(E,\cR)$  is the subset 
of \emph{feasible} \emph{non-root} CCPs from 
$\SetOf{CCP}(\eqoriented{E},R^\RMabbr)$, being root CCPs 
$\ccpOf{\alpha,\toppos,\alpha'}$ for $\alpha\in\:\eqoriented{E}$ and $\alpha'\in R^\RMabbr$ already considered in 
$\CCP(\cR,E)$ as $\ccpOf{\alpha',\toppos,\alpha}$.
\item $\CVPofR(E)$ is the subset of \emph{feasible} CVPs from $\SetOf{CVP}^\rew{}(\eqoriented{E})$.
\item $\CVPofEqOne(\cR)$) is the subset of \emph{feasible} CVPs from $\SetOf{CVP}^{\:\equone{}}(R^\RMabbr)$.
\end{enumerate}
\end{enumerate}
\end{defi}
%
\noindent
Table \ref{TableConditionalPairsForPeaksHuet} summarizes their relationship with previously considered local confluence and  local coherence peaks.

\begin{rem}[Refining $\CCP(\cR)$]
\label{RemRefiningCCPR}
In this paper, 
we check CCPs $\langle s,t\rangle\IF\gencond$ from $\CCP(\cR)$ for  
$\ejoinability{\cR^\RMabbr}$-joinability (to prove $E$-confluence) or non-$\ejoinability{\cR/E}$-joinability (to disprove it).
In both cases $s$ and $t$ are treated \emph{symmetrically}.
Thus, in $\CCP(\cR)$,
\begin{enumerate}
\item we take only one between the two root CCPs 
$\ccpOf{\alpha,\toppos,\alpha'}$ or $\ccpOf{\alpha',\toppos,\alpha}$ for rules $\alpha$ 
and $\alpha'$.
Also, 
\item we dismiss all \emph{trivial} CCPs $\langle t,t\rangle\IF\gencond$, in particular iCCPs $\iccpOf{\alpha}$ of 2-rules $\alpha\in R^\RMabbr$.
\end{enumerate}
\end{rem}
%
\noindent
Properties $\alpha$ and $\gamma$
of $\rew{\cR^\RMabbr}$ for \egtrs{s} $\cR$ (see Figure \ref{FigHuetConfluencePropertiesAlphaBetaGamma}) are characterized by the $\ejoinability{\cR}$-joinability of CCPs and CVPs as follows.
\begin{table}[ht]
\caption{Conditional pairs for local confluence and local coherence peaks}
\label{TableConditionalPairsForPeaksHuet}
\centering
\begin{tabular}{@{}lccl@{~}ll@{}}
\toprule
Property & Peak & Ref.\ & Rule & Rule' & Cond.\ Pairs
\\
\midrule
$\fS{LCON}_E(\rew{\cR^\RMabbr},\rew{\cR^\RMabbr})$
& Confluence critical p. & (\ref{LblCriticalPeak})
& $R^\RMabbr$ & $R^\RMabbr$
& $\CCP(\cR)$
\\
& Confluence variable p. & (\ref{LblVariablePeak})
& $R^\RMabbr$ & $R^\RMabbr$
& $\CVPofR(\cR)$
\\
\midrule
$\fS{LCOH}_E(\rew{\cR^\RMabbr})$
& Coherence critical p. & (\ref{LblCriticalPeak})
& $\eqoriented{E}$ & $R^\RMabbr$
& $\CCP(E,\cR)$
\\
& Coherence variable p. & (\ref{LblVariablePeak})
& $\eqoriented{E}$ & $R^\RMabbr$
& $\CVPofR(E)$
\\
& Coherence critical p. & (\ref{LblCriticalPeak})
& $R^\RMabbr$ & $\eqoriented{E}$
& $\CCP(\cR,E)$
\\
& Coherence variable p. & (\ref{LblVariablePeak})
& $R^\RMabbr$ & $\eqoriented{E}$
& $\CVPofEqOne(\cR)$
\\
\bottomrule
\end{tabular}
\end{table}
\begin{prop}
\label{PropEGTRSsAndPropertiesAlphaAndGamma}
Let $\cR=(\Symbols,\SPredicates,\mu,E,H,R)$ be an \egtrs{}.
Then, 
\begin{enumerate}
\item\label{PropEGTRSsAndPropertiesAlphaAndGamma_Alpha} 
$\cR$ has property $\alpha_{\rew{\cR^\RMabbr}}$ iff 
all conditional pairs in $\CCP(\cR)\cup\CVPofR(\cR)$ are $\ejoinability{\cR}$-joinable. 
\item\label{PropEGTRSsAndPropertiesAlphaAndGamma_Gamma}
 $\cR$ has property $\gamma_{\rew{\cR^\RMabbr}}$ iff 
all conditional pairs in $\CCP(\cR,E)\cup\CCP(E,\cR)\cup\CVPofR(E)\cup\CVPofEqOne(\cR)$ are $\ejoinability{\cR}$-joinable. 
\end{enumerate}
\end{prop}

\begin{proof}\leavevmode
\begin{enumerate}
\item Here we consider peaks (\ref{GenericRPeakTerms}) with $\alpha,\alpha'\in R^\RMabbr$.
(\emph{If} part) Disjoint peaks are joinable as shown above.
Regarding critical peaks (\ref{LblCriticalPeak}), by Proposition \ref{PropCriticalPeaksAsCCPs}, the joinability of all CCPs in $\CCP(\cR)$ implies joinability of critical peaks. 
Regarding variable peaks (\ref{LblVariablePeak}), by Proposition \ref{PropVariablePeaksAsCVPs}, the joinability of all feasible CVPs in 
$\CVPofR(\cR)$ implies joinability of critical peaks. 
Thus, Property $\alpha$ holds.

(\emph{Only if} part) By contradiction. Assume that property $\alpha$ holds but there is $\langle t,t'\rangle\IF\gencond\in\CCP(\cR)\cup\CVPofR(\cR)$ 
which is not $\ejoinability{\cR}$-joinable.
Then, there is a substitution $\sigma$ satisfying the conditional part of $\pi$ but such that $\sigma(t)$ and $\sigma(t')$ are not $\ejoinability{\cR}$-joinable.
Thus, there is a critical or variable peak with $s=\sigma(\lhsr)$ for some rule $\lhsr\to\rhsr\IF\gencond\in R^\RMabbr$ which is not $\ejoinability{\cR}$-joinable, contradicting Property $\alpha$.
\item
Here we consider peaks (\ref{GenericRPeakTerms}),
where either 
(i) $\alpha\in R^\RMabbr$ and $\alpha'\in\:\eqoriented{E}$
or
(ii) $\alpha\in\:\eqoriented{E}$ and $\alpha'\in R^\RMabbr$.
As for the \emph{if} part, regarding critical peaks, in case (i),  Proposition \ref{PropCriticalPeaksAsCCPs} guarantees that joinability of all CCPs in  $\CCP(\cR,E)$ implies joinability of the corresponding critical peaks (and joinability of CCPs in  $\CCP(E,\cR)$ for case (ii)). 
Regarding variable peaks, by Proposition \ref{PropVariablePeaksAsCVPs}, the joinability of all CVPs in 
$\CVPofEqOne(\cR)$ implies joinability of variable peaks in case (i)
and joinability of all CVPs in 
$\CVPofR(E)$ implies joinability of variable peaks in case (ii). 
Thus, Property $\gamma$ holds.
The \emph{only if} part is similar to that of $\alpha$ property.
\qedhere
\end{enumerate}
\end{proof}

\noindent
Now we can provide our first criteria for proving and disproving 
$E$-confluence of \egtrs{s}, which uses CCPs and CVPs 
together with Huet's abstract  results.
\begin{thm}
\label{TheoEConfluenceOfETerminatingEGTRSsHuet}
Let $\cR=(\Symbols,\SPredicates,\mu,E,H,R)$ be an \egtrs{}.
\begin{enumerate}
\item\label{TheoEConfluenceOfETerminatingEGTRSsHuet_EConfluence} 
If $\cR$ is $E$-terminating, then it is $E$-confluent if all pairs in 
\[\CCP(\cR)\cup\CVPofR(\cR)\cup\CCP(E,\cR)\cup\CCP(\cR,E)
\cup
\CVPofR(E)\cup\CVPofEqOne(\cR)
\]
\noindent
are $\ejoinability{\cR^\RMabbr}$-joinable.

\item\label{TheoEConfluenceOfETerminatingEGTRSsHuet_NonEConfluent} 
If there is a \emph{non}-$\ejoinability{\cR/E}$-joinable $\pi\in\CCP(\cR)\cup\CVPofR(\cR)$, then $\cR$ is not $E$-confluent.
\end{enumerate}
\end{thm}
\begin{proof}
(1) By Proposition \ref{PropEGTRSsAndPropertiesAlphaAndGamma}, $\rew{\cR^\RMabbr}$ 
satisfies properties $\alpha$ and $\gamma$. 
By Corollary \ref{CoroEConfluence_Huet80}, $\cR$ is $E$-confluent.
(2) By Propositions \ref{PropFeasibleCCPsAsCriticalPeaks} and
\ref{PropFeasibleCVPsAsVariablePeaks}, the existence of a non-$\ejoinability{\cR/E}$-joinable $\pi\in\CCP(\cR)\cup\CVPofR(\cR)$ implies non-$\ejoinability{\cR/E}$-joinability of a 
peak
(\ref{GenericRPeakTerms}). By Corollary \ref{CoroDisprovingEConcludenceByNonREJoinableLocalConfluencePeaks}, $\cR$ is not $E$-confluent.
\end{proof}
\begin{exa}\label{ExEvenEq}
Consider the \egtrs{} $\cR=(\Symbols,\SPredicates,\mu,E,H,R)$
in Example \ref{ExEvenEq_EGTRS}, i.e.,
\begin{IEEEeqnarray*}{r'C'l+x*}
\fS{s}(\fS{s}(x)) &\ceq  & x  \IF x\geq\fS{s}(\fS{0}) & \eqref{ExEvenEq_eq1}
\\
x & \geq & 0 & \eqref{ExEvenEq_clause1}
\\
\fS{s}(x) & \geq & \fS{s}(y) \IF x\geq y & \eqref{ExEvenEq_clause2}
\\
\fS{test}(x) & \to & \fS{pev}(x)\IF x\ceq\fS{s}(\fS{s}(\fS{0}))\qquad & \eqref{ExEvenEq_rule1}
\\
\fS{test}(x) & \to & \fS{odd}(x)\IF x\ceq\fS{s}(\fS{0}) & \eqref{ExEvenEq_rule2}
\\
\fS{test}(x) & \to & \fS{zero}(x)\IF x\ceq\fS{0} & \eqref{ExEvenEq_rule3}
\end{IEEEeqnarray*}
Its conditional pairs are shown in Table \ref{TableExEvenEq_ConditionalPairs}. 
\begin{table}[ht]
\caption{CCPs and CVPs for $\cR$ in Example \ref{ExEvenEq_EGTRS}}
\label{TableExEvenEq_ConditionalPairs}
\centering
\begingroup
\newcommand{\tabeqn}[1]{\refstepcounter{equation}\label{#1}(\theequation)}
\begin{tabular}{@{}rlr@{}}
\toprule
$\ccpOf{(\ref{ExEvenEq_rule1}),\toppos,(\ref{ExEvenEq_rule2})}:$
& $\langle\fS{odd}(x),\fS{pev}(x)\rangle\IF x\ceq\fS{s}(\fS{s}(\fS{0})),x\ceq\fS{s}(\fS{0})$
& \tabeqn{ExEvenEq_CCP1}
\\
$\ccpOf{(\ref{ExEvenEq_rule1}),\toppos,(\ref{ExEvenEq_rule3})}:$
& $\langle\fS{zero}(x),\fS{pev}(x)\rangle\IF x\ceq\fS{s}(\fS{s}(\fS{0})),x\ceq\fS{0}$
& \tabeqn{ExEvenEq_CCP2}
\\
$\ccpOf{(\ref{ExEvenEq_rule2}),\toppos,(\ref{ExEvenEq_rule3})}:$
& $\langle\fS{zero}(x),\fS{odd}(x)\rangle\IF x\ceq\fS{s}(\fS{0}),x\ceq\fS{0}$
& \tabeqn{ExEvenEq_CCP3}
\\
$\cvpROf{(\ref{ExEvenEq_rule1}),x,1}:$
& $\langle\fS{test}(x'),\fS{pev}(x)\rangle\IF x\rew{}x',x\ceq\fS{s}(\fS{s}(\fS{0}))$
& \tabeqn{ExEvenEq_CVP1}
\\
$\cvpROf{(\ref{ExEvenEq_rule2}),x,1}:$
& $\langle\fS{test}(x'),\fS{odd}(x)\rangle\IF x\rew{}x',x\ceq\fS{s}(\fS{0})$
& \tabeqn{ExEvenEq_CVP2}
\\
$\cvpROf{(\ref{ExEvenEq_rule3}),x,1}:$
& $\langle\fS{test}(x'),\fS{zero}(x)\rangle\IF x\rew{}x',x\ceq\fS{0}$
& \tabeqn{ExEvenEq_CVP3}
\\
$\ccpOf{(\ref{ExEvenEq_rule1}),\toppos,\overleftarrow{(\ref{ExEvenEq_eq1})}}:$
& $\langle\fS{s}(\fS{s}(x)),\fS{pev}(x)\rangle \IF x\ceq\fS{s}(\fS{s}(\fS{0})),\fS{test}(x)\geq\fS{s}(\fS{0})$
& \tabeqn{ExEvenEq_RECCP1}
\\
$\ccpOf{(\ref{ExEvenEq_rule2}),\toppos,\overleftarrow{(\ref{ExEvenEq_eq1})}}:$
& $\langle\fS{s}(\fS{s}(x)),\fS{odd}(x)\rangle \IF x\ceq\fS{s}(\fS{0}),\fS{test}(x)\geq\fS{s}(\fS{0})$
& \tabeqn{ExEvenEq_RECCP2}
\\
$\ccpOf{(\ref{ExEvenEq_rule3}),\toppos,\overleftarrow{(\ref{ExEvenEq_eq1})}}:$
& $\langle\fS{s}(\fS{s}(x)),\fS{zero}(x)\rangle \IF x\ceq\fS{0},\fS{test}(x)\geq\fS{s}(\fS{0})$
& \tabeqn{ExEvenEq_RECCP3}
\\
$\cvpROf{\overrightarrow{(\ref{ExEvenEq_eq1})},x,1.1}:$
& $\langle\fS{s}(\fS{s}(x')),x\rangle\IF x\rew{}x',x\geq\fS{s}(\fS{0})$
& \tabeqn{ExEvenEq_ECVP1}
\\
$\cvpROf{\overleftarrow{(\ref{ExEvenEq_eq1})},x,1.1}:$
& $\langle x',\fS{s}(\fS{s}(x))\rangle\IF x\rew{}x',x\geq\fS{s}(\fS{0})$
& \tabeqn{ExEvenEq_ECVP2}
\\
$\cvpEqOneOf{(\ref{ExEvenEq_rule1}),x,1}:$
& $\langle\fS{test}(x'),\fS{pev}(x)\rangle\IF x\equone{}x',x\ceq\fS{s}(\fS{s}(\fS{0}))$
& \tabeqn{ExEvenEq_RECVP1}
\\
$\cvpEqOneOf{(\ref{ExEvenEq_rule2}),x,1}:$
& $\langle\fS{test}(x'),\fS{odd}(x)\rangle\IF x\equone{}x',x\ceq\fS{s}(\fS{0})$
& \tabeqn{ExEvenEq_RECVP2}
\\
$\cvpEqOneOf{(\ref{ExEvenEq_rule3}),x,1}:$
& $\langle\fS{test}(x'),\fS{zero}(x)\rangle\IF x\equone{}x',x\ceq\fS{0}$
& \tabeqn{ExEvenEq_RECVP3}
\\
\bottomrule
\end{tabular}
\endgroup
\end{table}
We have:
\begin{itemize}
\item $\CCP(\cR)=\emptyset$, as
(\ref{ExEvenEq_CCP1}), (\ref{ExEvenEq_CCP2}), and (\ref{ExEvenEq_CCP3}) are \emph{infeasible} and
all rules in $R$ are 2-rules.
\item $\CVPofR(\cR)=\emptyset$, as
(\ref{ExEvenEq_CVP1}),(\ref{ExEvenEq_CVP2}), and (\ref{ExEvenEq_CVP3}) are all \emph{infeasible}: terms $s$ which can be rewritten with $\rew{\cR^\RMabbr}$ cannot be proved equivalent to 
$\fS{s}(\fS{s}(\fS{0}))$, $\fS{s}(\fS{0})$ or $\fS{0}$, 
i.e., $s\ceq\fS{s}(\fS{s}(\fS{0}))$, $s\ceq\fS{s}(\fS{0})$, and $s\ceq\fS{0}$ cannot be proved. 
Indeed, if $s$ can be rewritten with $\rew{\cR^\RMabbr}$, at least a subterm $\fS{test}(s')$ is inside $s$ for some term $s'$. 
Since symbols $\fS{test}$ are not removed but just replaced by $\fS{pev}$, $\fS{odd}$, or $\fS{zero}$, no $\equequ{E}$-match with a term 
$\fS{s}^n(\fS{0})$ is possible for any $n\geq 0$. 
\item $\CCP(\cR,E)=\emptyset$,
as (\ref{ExEvenEq_RECCP1}),(\ref{ExEvenEq_RECCP2}), and (\ref{ExEvenEq_RECCP3})
are all infeasible: condition $\fS{test}(x)\geq\fS{s}(\fS{0})$, occurring in all of them, is infeasible.
\item $\CCP(E,\cR)=\emptyset$, as no rule in $R$ overlaps a rule in $\eqoriented{E}$.
\item $\CVPofR(E)=\emptyset$ as (\ref{ExEvenEq_ECVP1}) and (\ref{ExEvenEq_ECVP2}) are \emph{infeasible}: terms $s$ which can be rewritten with $\rew{\cR^\RMabbr}$ cannot be compared with $\fS{s}(\fS{0})$ using $\geq$, i.e., $s\geq\fS{s}(\fS{0})$ cannot be proved.
\item $\CVPofEqOne(\cR)=\{(\ref{ExEvenEq_RECVP1}),(\ref{ExEvenEq_RECVP2}),(\ref{ExEvenEq_RECVP3})\}$ consists of
$\ejoinability{\cR}$-joinable CVPs.
For instance, consider (\ref{ExEvenEq_RECVP1}): if $s\equone{E}t$ holds for some terms $s$ and $t$, then either 
\begin{enumerate}
\item $s\rew{\overrightarrow{(\ref{ExEvenEq_eq1})}}t$ (hence $s=\fS{s}(\fS{s}(t))$)
and $t\geq\fS{s}(\fS{0})$ holds; furthermore, 
$s\equequ{E}\fS{s}(\fS{s}(\fS{0}))$ also holds.
This implies that $s=\fS{s}^{2n}(\fS{0})$ for some $n> 1$
and $t=\fS{s}^{2n-2}(\fS{0})$.
Since $n>1$, we have $2n-2\geq 2$ and
$\fS{test}(t)\rew{(\ref{ExEvenEq_rule1})}\fS{pev}(t)\equone{E}\fS{pev}(s)$, i.e., (\ref{ExEvenEq_RECVP1}) is 
$\ejoinability{\cR}$-joinable.
\item Dually, if $s\rew{\overleftarrow{(\ref{ExEvenEq_eq1})}}t$ (hence $t=\fS{s}(\fS{s}(s))$)
and $s\geq\fS{s}(\fS{0})$ holds; furthermore, 
$s\equequ{E}\fS{s}(\fS{s}(\fS{0}))$ also holds.
This implies that $s=\fS{s}^{2n}(\fS{0})$ for some $n>0$
and $t=\fS{s}^{2n+2}(\fS{0})$.
Since $n\geq 0$, $2n+2\geq 2$ and
$\fS{test}(t)\rew{(\ref{ExEvenEq_rule1})}\fS{pev}(t)\equone{E}\fS{pev}(s)$, i.e., (\ref{ExEvenEq_RECVP1}) is 
$\ejoinability{\cR}$-joinable.
\end{enumerate}
For (\ref{ExEvenEq_RECVP2}) and (\ref{ExEvenEq_RECVP3}), we 
reason similarly
to conclude their $\ejoinability{\cR}$-joinability.
\end{itemize}
\noindent
Overall, since $\cR$ is clearly $E$-terminating, Theorem \ref{TheoEConfluenceOfETerminatingEGTRSsHuet}(\ref{TheoEConfluenceOfETerminatingEGTRSsHuet_EConfluence}) proves it $E$-confluent.
\end{exa}

\begin{cor}
\label{CoroEConfluenceOfETerminatingLeftLinearEGTRSsHuet}
Let $\cR=(\Symbols,\SPredicates,\mu,E,H,R)$ be an \egtrs{}
such that $R^\RMabbr$ is $\mu$-left-linear and $\mu$-compatible
and $\eqoriented{E}$ is left $\mu$-homogeneous and $\mu$-compatible.
If $\cR$ is $E$-terminating and all pairs in 
$\CCP(\cR) 
\cup\CCP(E,\cR)\cup\CCP(\cR,E)
$
are $\ejoinability{\cR}$-joinable, then $\cR$ is $E$-confluent.
\end{cor}

\begin{proof}
Due to $\mu$-left-linearity 
(which implies $\mu$-homogeneity) 
and $\mu$-compatibility 
of rules in $R^\RMabbr$, 
by Corollary \ref{CoroJoinableVariablePairsForCriticalVariablesNotInCond}
and Corollary \ref{CoroJoinableVariablePairsIII}, the CVPs in $\CVPofR(\cR)\cup\CVPofEqOne(\cR)$ are $\ejoinability{\cR^\RMabbr}$-joinable.
Due to $\mu$-left-homogeneity 
and $\mu$-compatibility 
of rules in $\eqoriented{E}$, by Corollary \ref{CoroJoinableVariablePairsII}, the CVPs in $\CVPofR(E)$ are $\ejoinability{\cR^\RMabbr}$-joinable.
By Theorem \ref{TheoEConfluenceOfETerminatingEGTRSsHuet}(\ref{TheoEConfluenceOfETerminatingEGTRSsHuet_EConfluence}), $\cR$ is $E$-confluent.
\end{proof}
%
\noindent
Corollary \ref{CoroEConfluenceOfETerminatingLeftLinearEGTRSsHuet} does \emph{not} apply to Example \ref{ExEvenEq}: rules are \emph{not} $\mu$-compatible because variable $x$ in the left-hand sides is active (since $\mu=\muTop$) and occurs in the conditional part of rules.
\begin{exa}
\label{ExSumListsA_CVPtoR}
\label{ExSumListsA_EConfluent}
For $\cR$ in Example \ref{ExSumListsA},
$\eqoriented{E}$ and
 $R$
are $\mu$-left-linear and $\mu$-compatible. The only pCCP in
$\SetOf{CCP}(R^\RMabbr,R^\RMabbr)$ is
\begin{eqnarray*}
\ccpOf{(\ref{ExSumListsA_rule3}),\toppos,(\ref{ExSumListsA_rule4})}: && 
\hspace{-0.5cm}\langle m+n,n'\rangle\IF 
ms\cto_\showRMabbr n', \fS{Nat}(n'), ms\cto_\showRMabbr m \fS{\seqnat} \cfrozen{ns}, 
\fS{Nat}(m),
\fS{sum}(\cfrozen{ns})\cto_\showRMabbr n~~~\label{ExSumListsA_CCP_R_1}
\end{eqnarray*}
which
is \emph{infeasible}: assume that a substitution 
$\sigma$ satisfies the conditional part of 
(\ref{ExSumListsA_CCP_R_1}):
if $\sigma(ms)\cto_\showRMabbr\sigma(n')$ and
$\fS{Nat}(\sigma(n'))$ hold, then $\sigma(ms)\rews{\cR/E}\sigma(n')$ for
$\sigma(n')=\fS{s}^p(\fS{0})$ for some $p\geq 0$.
Hence, $\sigma(ms)$ is \emph{not} 
$\rews{\cR/E}$-reducible to 
$\sigma(m)\fS{\seqnat} \cfrozen{\sigma(ns)}$.
 There also are two iCCPs:
\begin{IEEEeqnarray}{r'C'l}
\iccpOf{(\ref{ExSumListsA_rule3})}: &&
\hspace{-0.5cm}\langle n',n\rangle\IF 
m\cto_\showRMabbr n, \fS{Nat}(n),
m\cto_\showRMabbr n', \fS{Nat}(n'),
~~\label{ExSumListsA_iCCP_R_1}
\\
\iccpOf{(\ref{ExSumListsA_rule4})}: &&
\hspace{-0.5cm}\langle m+n,m'+n'\rangle\IF 
ms\cto_\showRMabbr m \fS{\seqnat} \cfrozen{ns}, 
ms\cto_\showRMabbr m' \fS{\seqnat} \cfrozen{ns'}, 
\fS{Nat}(m),  \nonumber
\\
&& \hspace{2.85cm}\fS{Nat}(m'),\fS{sum}(\cfrozen{ns})\cto_\showRMabbr n,
\fS{sum}(\cfrozen{ns'})\cto_\showRMabbr n'~~\label{ExSumListsA_iCCP_R_2}
\end{IEEEeqnarray}
Thus, $\CCP(\cR)=\{(\ref{ExSumListsA_iCCP_R_1}),(\ref{ExSumListsA_iCCP_R_2})\}$.
However, (\ref{ExSumListsA_iCCP_R_1}) and (\ref{ExSumListsA_iCCP_R_2}) are $\ejoinability{\cR^\RMabbr}$-joinable.
This can be proved by induction on the structure of $\sigma(m)$ (resp.\ $\sigma(ms)$) for all
all substitutions $\sigma$ satisfying the conditional part of (\ref{ExSumListsA_iCCP_R_1}) (resp.\ (\ref{ExSumListsA_iCCP_R_2})).
We have $\CCP(E,\cR)=\emptyset$, as no rule in $R^\showRMabbr$ overlaps a rule in $\eqoriented{E}$. 
Also, no rule in $\eqoriented{E}$ overlaps a rule in $R^\showRMabbr$, i.e., $\CCP(\cR,E)=\emptyset$.
The system can be seen $E$-terminating. By Corollary \ref{CoroEConfluenceOfETerminatingLeftLinearEGTRSsHuet},  $\cR$ is $E$-confluent.
\end{exa}
\begin{rem}[Replacement restrictions may ease E-confluence proofs]
The use of replacement restrictions in Example \ref{ExSumListsA} is crucial to obtain the proof as done in Example \ref{ExSumListsA_EConfluent} using the techniques described in this section.
For instance, if no replacement restriction is imposed, i.e., $\mu=\muTop$, then
rules (\ref{ExSumListsA_rule3}) and (\ref{ExSumListsA_rule4}) are \emph{not} compatible, as, e.g., variable $m$ (resp.\ $ms$) occurs
in the conditional part of rule (\ref{ExSumListsA_rule3}) (resp.\ (\ref{ExSumListsA_rule4})) and $\ejoinability{\cR}$-joinability of 
$\cvpOf{\rew{}}{(\ref{ExSumListsA_rule3}),m,1}$ and
$\cvpOf{\rew{}}{(\ref{ExSumListsA_rule4}),ms,1}$
does \emph{not} follow from Corollary \ref{CoroJoinableVariablePairsForCriticalVariablesNotInCond} and should be proved.
Since Example \ref{ExSumListsA} is a slightly adapted version of 
\cite[Example 1]{Lucas_ConfluenceOfConditionalRewritingModulo_CSL24}, where replacement restrictions were \emph{not} used, the interested reader can check the proof of $E$-confluence in \cite[Example 50 \& Appendix]{Lucas_ConfluenceOfConditionalRewritingModulo_CSL24} to compare the cases with and without replacement restrictions. 
\end{rem}

\subsection{Limitations of proving $E$-confluence as $\ejoinability{\cR^\RMabbr}$-joinability of peaks (\ref{GenericRPeakTerms}).}
\label{SecLimitationsOfHuetApproach}
Theorem \ref{TheoEConfluenceOfETerminatingEGTRSsHuet}
has limitations which are exemplified by means of two examples as follows.
\begin{enumerate}
\item 
\emph{$E$-confluence} of the \etrs{} $\cR$ in Example \ref{ExHuet80_RemarkPage818}, i.e.,
\begin{IEEEeqnarray*}{r'C'l+x*}
\fS{a} & = & \fS{b} & \eqref{ExHuet80_RemarkPage818_eq1}
\\
\fS{f}(x,x) & \to & \fS{g}(x)& \eqref{ExHuet80_RemarkPage818_rule1}
\end{IEEEeqnarray*}
\emph{cannot} be proved by using Theorem \ref{TheoEConfluenceOfETerminatingEGTRSsHuet}(\ref{TheoEConfluenceOfETerminatingEGTRSsHuet_EConfluence}), 
as the following coherence CVP 
\begin{IEEEeqnarray}{r'C'l}
\cvpEqOneOf{(\ref{ExHuet80_RemarkPage818_rule1}),x,1} : && \langle\fS{f}(x',x),\fS{g}(x)\rangle \IF x\equone{}x'\label{ExHuet80_RemarkPage818_ECVP1}
\end{IEEEeqnarray}
is \emph{not}  
$\ejoinability{\cR^\RMabbr}$-joinable: $\sigma=\{x\mapsto\fS{a},x'\mapsto\fS{b}\}$  
satisfies its conditional part, but
$\sigma(\fS{f}(x',x)=\fS{f}(\fS{b},\fS{a})$
and $\sigma(\fS{g}(x))=\fS{g}(\fS{a})$ are $\rew{\cR^\RMabbr}$-irreducible 
and not $\equequ{E}$-equivalent.
\item
Non-$E$-confluence of $\cR$ in Example \ref{ExPeakNoCPs}, i.e.,
\begin{IEEEeqnarray*}{r'C'l+x*}
\fS{b} & = &\fS{f}(\fS{a}) 
&
\eqref{ExPeakNoCPs_eq1}
\\
\fS{a} & = & \fS{c}
&
\eqref{ExPeakNoCPs_eq2}
\\
\fS{c} & \to & \fS{d}
&
\eqref{ExPeakNoCPs_rule1}
\\
\fS{b} & \to & \fS{d}
&
\eqref{ExPeakNoCPs_rule2}
\end{IEEEeqnarray*}
witnessed by the  non-$\ejoinability{\cR/E}$-joinable peak $\fS{d}\leftrew{\cR,E}\fS{f}(\fS{c})\rew{\cR}\fS{f}(\fS{d})$,
\emph{cannot}  be proved by Theorem \ref{TheoEConfluenceOfETerminatingEGTRSsHuet}(\ref{TheoEConfluenceOfETerminatingEGTRSsHuet_NonEConfluent}):
$\CCP(\cR)=\emptyset$ (rules do not overlap), 
and 
$\CVPofR(\cR)=\emptyset$ (rules are ground).
\end{enumerate}
\medskip
\noindent
The use of Jouannaud and Kirchner's notions of 
\emph{local confluence modulo $E$ of $\rew{\cR,E}$ with $\rew{\cR}$} 
and
\emph{local coherence modulo $E$ of $\rew{\cR,E}$} 
permits to obtain $E$-confluence criteria exploiting $\ejoinability{\cR,E}$-joinability of conditional pairs.
In the following, we develop this approach,  which permits to solve the two aforementioned problems:
\begin{enumerate}
\item 
$E$-confluence of $\cR$ in Example \ref{ExHuet80_RemarkPage818} is proved 
in Example \ref{ExHuet80_RemarkPage818_EConfluence}.
\item
Non-$E$-confluence of $\cR$ in Example \ref{ExPeakNoCPs} is proved in Example \ref{ExPeakNoCPs_notEconfluent}.
\end{enumerate}

\section{Conditional pairs for non-disjoint local $E$-confluence and $E$-coherence peaks }
\label{SecConditionalPairsForNonDisjointPeaksJK}

As explained in Section~\ref{SecLocalConfluenceAndCoherencePeaksForJK86},
in this section we investigate $\joinability{\cR^\RMabbr,E}$-joinability 
of (non-disjoint) peaks (\ref{GenericPSPeakTerms}), i.e.,
\begin{IEEEeqnarray*}{r'C'l+x*}
\raisebox{0.6cm}{
\xymatrix{
&  s 
\ar@{->}[dl]_{p_1}^>>{\alpha_1,E}
\ar@{->}[dr]^{p_2}_>>{\alpha_2} 
\\
t_1  & &  t_2 
  \\
}}
& & &
\eqref{GenericPSPeakTerms}
\end{IEEEeqnarray*}
to prove 
(i) local confluence modulo $E$ of $\rew{\cR,E}$ with $\rew{\cR}$
and
(ii) local coherence modulo $E$ of $\rew{\cR,E}$, 
thus obtaining a new criterion for $E$-confluence of \egtrs{s}.
We proceed as in Section~\ref{SecConditionalPairsAlphaAndGamma}, but with some important differences.
First, note that $\alpha_1$ always belongs to $R^\RMabbr$, as it is used in a Peterson \& Stickel rewriting step.
Still, as explained in Section~\ref{SecLocalConfluenceAndCoherencePeaksForJK86}, 
depending on the origin of rule $\alpha_2$ in 
(\ref{GenericPSPeakTerms}), 
if $\alpha_2$ belongs to $R^\RMabbr$, then (\ref{GenericPSPeakTerms}) is a local \emph{$E$-confluence} peak;
if  $\alpha_2$ belongs to $\eqoriented{E}$,
then (\ref{GenericPSPeakTerms}) is a 
local \emph{$E$-coherence} peak. Now,
\begin{enumerate}
\item 
If $p_1\geq p_2$, then $p_1=p_2.p$ for some position $p$ and
rewriting with $\rew{\eqoriented{E}}$ or $\rew{\cR^\RMabbr}$ occurs \emph{above or on} the $\rew{\cR^\RMabbr,E}$-step, i.e., (\ref{GenericPSPeakTerms}) can be written as follows:
\begin{IEEEeqnarray}{r'C'l}
\underbrace{\sigma(\lhsr)[\sigma(\rhsr')]_p}_{t} ~\leftrewAtPos{p}{\alpha',E}   & \underbrace{\sigma(\lhsr)[w]_p}_s & \rewAtPos{\toppos}{\alpha}~\underbrace{\sigma(\rhsr)}_{t'}
\label{NestedPSRupPeak}
\end{IEEEeqnarray}
where $\alpha:\lhsr\to\rhsr\IF\gencond\in\:\eqoriented{E}\:\cup\:R^\RMabbr$,
$\sigma(\gencond)$ holds, 
$w=\sigma(\lhsr)|_p$ for some $p\geq\toppos$,
and $w=_E\sigma(\lhsr')$
for some $\alpha':\lhsr'\to\rhsr'\IF\gencond'\in R^{\RMabbr}$ such that $\sigma(\gencond')$ holds.
\item If $p_1\leq p_2$, then $p_2=p_1.p$ for some position $p$  and rewriting with $\rew{\eqoriented{E}}$ or $\rew{\cR^\RMabbr}$ occurs strictly \emph{below} the $\rew{\cR^\RMabbr,E}$-step, i.e.,
then (\ref{GenericPSPeakTerms}) can be written as follows:
\begin{IEEEeqnarray}{r'C'l}
\underbrace{\sigma(\rhsr')}_{t}~\leftrewAtPos{\toppos}{\alpha',E}   & \underbrace{w[\sigma(\lhsr)]_p}_{w} & \rewAtPos{p}{\alpha}~\underbrace{w[\sigma(\rhsr)]_p}_{t'}
\label{LblPreEDownPeak}
\end{IEEEeqnarray}
where $\alpha$, $\alpha'$ and $\sigma$ are as in the previous item and $w=w[\sigma(\lhsr)]_p=_E\sigma(\lhsr')$ for some $p>\toppos$, 
as $p=\toppos$ is considered in (\ref{NestedPSRupPeak}).
\end{enumerate}
 Notice that $\alpha'$ belongs to $R^\RMabbr$. 
  \begin{defi}[$E$-confluence and $E$-coherence peaks]
 Consider rules $\alpha'\in R^\RMabbr$ and $\alpha$ in (\ref{NestedPSRupPeak}) and (\ref{LblPreEDownPeak}).
 \begin{enumerate}
\item If $\alpha$ belongs to $R^\RMabbr$, 
then (\ref{NestedPSRupPeak}) (resp.\ (\ref{LblPreEDownPeak})) is called an \emph{$E$-confluence} peak.
\item If  $\alpha$ belongs to $\eqoriented{E}$,
then
(\ref{NestedPSRupPeak}) (resp.\ (\ref{LblPreEDownPeak})) is called an \emph{$E$-coherence} peak.
\end{enumerate}
\end{defi}
%
 \noindent
 We have the following special cases for $E$-coherence peaks:
 
 \begin{prop}
 \label{PropJoinabilityOfEdownECohPeaks}
 Let $\cR=(\Symbols,\SPredicates,\mu,E,H,R)$ be an \egtrs{}. 
\begin{enumerate}
\item 
\label{PropJoinabilityOfEdownECohPeaks_Eroot}
Peaks (\ref{NestedPSRupPeak}) 
such that $\alpha\in\:\eqoriented{E}$ and $p=\toppos$ are 
$\rsjoinability{\cR^\RMabbr,E}$-joinable.
\item 
\label{PropJoinabilityOfEdownECohPeaks_Edown}
Peaks (\ref{LblPreEDownPeak}) 
such that $\alpha\in\:\eqoriented{E}$ 
are 
$\rsjoinability{\cR^\RMabbr,E}$-joinable.
\end{enumerate}
\end{prop}
\begin{proof}\leavevmode
\begin{enumerate}
\item 
Since
$p=\toppos$ and 
(i) $s=w=\sigma(\lhsr)$ in (\ref{NestedPSRupPeak})
and
(ii) $\alpha\in\eqoriented{E}$, we have
$t'=\sigma(\rhsr)\:\equequ{E}\sigma(\lhsr)\rew{\alpha',E}t$,
as required.
\item 
For peaks (\ref{LblPreEDownPeak}), since $\alpha\in\:\eqoriented{E}$,
and $w\rewAtPos{p}{\alpha}t'$, we have $t'\equequ{E}w 
=_E \sigma(\lhsr')\rewAtPos{\toppos}{\alpha}\sigma(\rhsr')=t$, i.e., $t'\rew{\alpha',E}t$,
as required.\qedhere
\end{enumerate}
\end{proof}
%
\noindent
Depending on 
position $p$ in (\ref{NestedPSRupPeak}), two families of peaks are distinguished.

 \begin{defi}[$E$-critical and $E$-variable peaks]
 Consider $\alpha$, $\alpha'$, and $p$ in (\ref{NestedPSRupPeak}).
 \begin{enumerate}
 \item If $p\in\Pos^\mu_\Symbols(\lhsr)$, then rule $\alpha'$  \emph{$E$-overlaps} rule $\alpha$ and (\ref{NestedPSRupPeak}) is called an \emph{$E$-critical peak}.
 \item If $p\notin\Pos^\mu_\Symbols(\lhsr)$, then $\alpha'$  \emph{does not $E$-overlap} $\alpha$ and (\ref{NestedPSRupPeak}) is called an \emph{$E$-variable peak}.
 \end{enumerate}
 \end{defi}
 %
 \noindent
 Note that $E$-critical and $E$-variable peaks (\ref{NestedPSRupPeak}) 
 occur when analyzing \emph{both}
 $E$-confluence and $E$-coherence properties.
 For the sake of readability, though, in the following we will drop ``$E$-'' from ``$E$-confluence'' and ``$E$-coherence'' to just talk about, e.g., confluence $E$-critical peaks, coherence $E$-variable peaks, etc.
 
Sections \ref{SecECCPsAndLCTRS} and
\ref{SecCVPsForEvariablePeaks} show how these $E$-critical and $E$-variable peaks are captured by different conditional pairs
 obtained from rules in $R^\RMabbr$ and/or $\eqoriented{E}$.
  Peaks (\ref{LblPreEDownPeak}) do \emph{not} represent $E$-overlaps or the lack of $E$-overlaps as such and are treated in Section~\ref{SecEDownPeaks}.  

\subsection{E-critical confluence and coherence peaks and logic-based conditional critical pairs}
\label{SecECCPsAndLCTRS}
Note that $\sigma(\lhsr)|_p=\sigma(\lhsr|_p)=w=_E\sigma(\lhsr')$, in 
(\ref{NestedPSRupPeak}), i.e., $\lhsr|_p$ and $\lhsr'$ \emph{$E$-unify}.
If $\alpha'$ \emph{$E$-overlaps} $\alpha$, 
i.e., $p\in\Pos^\mu_\Symbols(\lhsr)$, then 
we obtain an  \emph{E-critical peak}\footnote{We follow here the terminology for $E$-critical pairs in \cite[Section 3.4, page 1167]{JouKir_CompletionOfASetOfRulesModuloASetOfEquations_SIAMJC86}.}
\begin{IEEEeqnarray}{r'C'l}
\underbrace{\sigma(\lhsr)[\sigma(\rhsr')]_p}_{t} ~\leftrewAtPos{p}{\alpha',E}   & \underbrace{\sigma(\lhsr)[w]_p}_s & \rewAtPos{\toppos}{\alpha}~\underbrace{\sigma(\rhsr)}_{t'}\qquad\qquad\text{where }p\in\Pos^\mu_\Symbols(\lhsr)\label{LblECriticalPeak}
\end{IEEEeqnarray}
see Figure \ref{FigECriticalPeak}. 
\begin{figure}[t]
 \begin{center}
\begin{tikzpicture}[node distance = 6cm, auto,scale=1.5,every node/.style={scale=1}]
\node (CP_Ttext) {$\sigma(\lhsr)[\sigma(\rhsr')]_p$};
\node [triangleW, below  = 0cm of CP_Ttext, minimum width = 2cm, minimum height = 2.35cm, scale=0.9] (CP_Ttriangle) {};
\node [triangleG, below  = 0cm of CP_Ttext, minimum width = 1cm, minimum height = 1.2cm, scale=0.9] (CP_LHSTriangleT) {};
\node [below  = 0.2 cm of CP_Ttext] (CP_LHStextT) {$\lhsr$};
\node [triangleW, below  = 0.05cm of CP_LHStextT, minimum width = 1.3cm, minimum height = 1.48cm, scale=0.9] (CP_SigmaXtriangleT) {};
\node [triangleG, below  = 0.05cm of CP_LHStextT, minimum width = 0.5cm, minimum height = 0.5cm, scale=0.9] (CP_RHSpTriangleT) {};
\node [below  = -0.45 cm of CP_RHSpTriangleT] (CP_RHSptext) {\tiny $\rhsr'$};
\node [above right = 1cm and 1.8 cm of CP_Ttext] (CP_Stext) {$\sigma(\lhsr)[w]_p$};
\node [triangleW, below  = 0cm of CP_Stext, minimum width = 2cm, minimum height = 2.35cm, scale=0.9] (CP_Striangle) {};
\node [triangleG, below  = 0cm of CP_Stext, minimum width = 1cm, minimum height = 1.2cm, scale=0.9] (CP_LHSTriangleS) {};
\node [below  = 0.2 cm of CP_Stext] (CP_LHStextS) {$\lhsr$};
\node [triangleW, below  = 0.05cm of CP_LHStextS, minimum width = 1.3cm, minimum height = 1.48cm, scale=0.9] (CP_SigmaXtriangleS) {};
\node [above right = 1cm and 1 cm] at (CP_Striangle) (PinS) {\tiny $p\in \Pos^\mu_\Symbols(\lhsr)$};
\node [below  = -0.9 cm of CP_SigmaXtriangleS] (CP_Wtext) {\tiny $w$};
\node [below right =1cm and 3 cm of CP_Stext] (CP_TtextTp) {$\sigma(\rhsr)$};
\node [triangleW, below  = 0cm of CP_TtextTp, minimum width = 3cm, minimum height = 2.25cm] (CP_Tptriangle) {};
\node [triangleG, below  = 0cm of CP_TtextTp, minimum width = 1.67cm, minimum height = 1.28cm, scale=0.9] (CP_RHStTriangle) {};
\node [below  = 0.35 cm of CP_TtextTp] (CP_RHStextTp) {$\rhsr$};
\draw[-*] ([shift={(0.2,0.2)}]PinS) -- ([shift={(-0.05,0)}]CP_SigmaXtriangleS.north);
\draw[->] ([shift={(-0.2,0.2)}]CP_Striangle.west) -- ([shift={(0.1,0.3)}]CP_Ttriangle.east);
\node[above right = 0.08cm and 0.2cm of CP_Ttriangle.east]{\tiny $\alpha',E$};
\draw[->] ([shift={(0.2,0.2)}]CP_Striangle.east) -- ([shift={(-0.2,0.2)}]CP_Tptriangle.west);
\node[above left = 0.01cm and 0.35cm of CP_Tptriangle.west]{\tiny $\alpha$};
\end{tikzpicture}
 \end{center}
 \caption{E-critical peak (where $w\equequ{E}\sigma(\lhsr')$)}
 \label{FigECriticalPeak}
 \end{figure}
 Depending on the \emph{origin} of rule $\alpha$ in (\ref{LblECriticalPeak}), we have different kinds of critical peaks for different uses
 (see Table \ref{TableDifferentKindsOfECriticalPeaks}).
\begin{table}[ht]
\caption{Different kinds of E-critical peaks (\ref{LblECriticalPeak})}
\label{TableDifferentKindsOfECriticalPeaks}
\centering
\begin{tabular}{@{}ccll@{}}
\toprule
$\alpha$ & $\alpha'$ & Denomination & Use
\\
\midrule
$R^\RMabbr$ & $R^\RMabbr$ & Confluence $E$-Critical peak
& $\fS{LCON}_E(\rew{\cR^\RMabbr,E},\rew{\cR^\RMabbr})$
\\
$\eqoriented{E}$ & $R^\RMabbr$ & Coherence $E$-Critical peak
& $\fS{LCOH}_E(\rew{\cR^\RMabbr,E})$
\\
\bottomrule
\end{tabular}
\end{table}

\subsubsection{Conditional $E$-Critical Pairs}
\label{SecECCPs}
Peaks (\ref{LblECriticalPeak}) involve the use of \emph{$E$-unifiers} in the computation of appropriate critical pairs by replacing syntactic unifiers by $E$-unifiers $\theta$
leading to the notion of \emph{$E$-critical pairs (ECPs)} for \etrs{s}, first introduced by Peterson \& Stickel \cite{StiPet_CompleteSetsOfReductionsForEquationalTheoriesWithCompleteUnificationAlgorithms_TR77} and
\cite[Definition 9.2]{PetSti_CompleteSetsOfReductionsForSomeEquationalTheories_JACM81} and then coined as ECPs in \cite[Definition 10]{Jouannaud_ConfluentAndCoherentEquationalTermRewritingSystemsApplicationToProofsInAbstractDataTypes_CAAP83} and 
\cite[Definition 12]{JouKir_CompletionOfASetOfRulesModuloASetOfEquations_SIAMJC86}.
Given terms $s$ and $t$, $\fS{CSU}_E(s,t)$ denotes a \emph{complete} (possibly infinite) set of $E$-unifiers of $s$ and $t$, i.e., for each $E$-unifier $\sigma$ of $s$ and $t$, there is 
$\theta\in\fS{CSU}_E(s,t)$ such that $\theta(s)\equequ{E}\theta(t)$ 
and a substitution $\tau$ such that 
for all $x\in\Variables$, 
$\sigma(x)\equequ{E}\tau(\theta(x))$
\cite[Definition 10.1.4]{BaaNip_TermRewritingAndAllThat_1998}.
\begin{defi}[Conditional E-critical pairs, \text{c.f.\ \cite[Definition 6]{DurMes_OnTheChurchRosserAndCoherencePropertiesOfConditionalOrderSortedRewriteTheories_JLAP12}}]
\label{DefECCP}

Given variable disjoint  rules $\alpha:\lhsr\to\rhsr\IF\gencond$ and $\alpha':\lhsr'\to\rhsr'\IF\gencond'$, 
$p\in\Pos^\mu_\Symbols(s)$, 
and $\theta\in\fS{CSU}_E(\lhsr|_p,\lhsr')$,
the \emph{Conditional $E$-Critical Pair}  (ECCP) 
$\eccpOf{\alpha,p,\alpha',\theta}$ is
\begin{IEEEeqnarray}{r'C'l}
\langle\theta(\lhsr)[\theta(\rhsr')]_p,\theta(\rhsr)\rangle\IF\theta(\gencond),\theta(\gencond')\label{LblGenericConditionalECriticalPair}
\end{IEEEeqnarray}
\end{defi}
%
\noindent
Of course, in order to obtain each $\theta$ above, an appropriate $E$-unification algorithm is necessary. 
However, such an algorithm does not always exist.
This problem was addressed by several authors, in particular for associative and/or commutative unification \cite{Plotkin_BuildingInEquationalTheories_MI72,%
Slagle_AutomatedTheoremProvingForTheoriesWithSimplifiersCommutativityAndAssociativity_JACM74,Stickel_ACompleteUnificationAlgorithmForAssociativeCommutativeFunctions_IJCAI75,%
Stickel_AUnificationAlgorithmForAssociativeCommutativeFunctions_JACM81},
and recently in the realm of \Maude{} for associative unification \cite{Eker_AssociativeUnificationInMaude_JLAMP22}.
Stickel \cite{Stickel_AUnificationAlgorithmForAssociativeCommutativeFunctions_JACM81} introduced a unification algorithm for associative and commutative (AC) theories which is proved complete for a particular class of expressions incorporating associative and commutative functions \cite[Theorem 3]{Stickel_AUnificationAlgorithmForAssociativeCommutativeFunctions_JACM81}.
Fages proved termination of the algorithm \cite{Fages_AssociativeCommutativeUnification_CADE84}.
Other algorithms for AC unification have been introduced since then, see \cite{BaaSny_UnificationTheory_HAR01} for a wider account.
However, in the realm of \egtrs{s}, $E$-unification algorithms would greatly benefit from taking into account the conditional part of rules involved in the computation of ECCPs.
\begin{exa}\label{ExInnermostArithmetic}
The following \egtrs{}
\begin{IEEEeqnarray}{r'C'l}
x + (y+z) & = & (x+y)+z\label{ExInnermostArithmetic_eq1}
\\
\fS{Nat}(\fS{0}) && \label{ExInnermostArithmetic_horn1}
\\
\fS{Nat}(\fS{s}(n)) & \IF & \fS{Nat}(n)\label{ExInnermostArithmetic_horn2}
\\
x+\fS{0} & \to &x\IF\fS{Nat}(x) \label{ExInnermostArithmetic_rule1}
\\
\fS{0}+x & \to &x\IF\fS{Nat}(x) \label{ExInnermostArithmetic_rule2}
\\
\fS{s}(x)+y & \to &\fS{s}(x+y)\IF\fS{Nat}(x),\fS{Nat}(y)\qquad\label{ExInnermostArithmetic_rule3}
\end{IEEEeqnarray}
encodes \emph{eager} evaluation of addition, as rules apply only on completely evaluated arguments of the form $\fS{s}^n(\fS{0})$ for some $n\geq 0$.
Mimicking Stickel's example \cite[page 71]{Stickel_ACompleteUnificationAlgorithmForAssociativeCommutativeFunctions_IJCAI75}, 
$x+\fS{0}$ and $\fS{s}(x')+y'$ have infinitely many associative unifiers
$\theta_0=\{x\mapsto\fS{s}(x'), y'\mapsto\fS{0}\}$,
$\theta_1=\{x\mapsto\fS{s}(x')+\fS{0}, y'\mapsto\fS{0}+\fS{0}\}$, \ldots
According to Definition \ref{DefECCP}
there are 
\emph{infinitely many} ECCPs:
\[
\begin{array}{rl}
\eccpOf{(\ref{ExInnermostArithmetic_rule1}),\toppos,(\ref{ExInnermostArithmetic_rule3}),\theta_0}: & \langle\fS{s}(x'+\fS{0}),\fS{s}(x')\rangle\IF\fS{Nat}(\fS{s}(x')),\fS{Nat}(x'),\fS{Nat}(\fS{0})
\\[0.2cm]
\eccpOf{(\ref{ExInnermostArithmetic_rule1}),\toppos,(\ref{ExInnermostArithmetic_rule3}),\theta_1}: & \langle\fS{s}(x'+(\fS{0}+\fS{0})),\fS{s}(x')+\fS{0}\rangle\IF\fS{Nat}(\fS{s}(x')+\fS{0}),\fS{Nat}(x'),\fS{Nat}(\fS{0}+\fS{0})
\\
\end{array}
\]
However, if conditions in rules 
would be used to \emph{constrain} the computation of 
$E$-unifiers, then only $\theta_0$ above leads to a feasible (and $\ejoinability{\cR^\RMabbr}$-joinable) ECCP $\eccpOf{(\ref{ExInnermostArithmetic_rule1}),\toppos,(\ref{ExInnermostArithmetic_rule3}),\theta_0}$.
\end{exa}
%
\noindent
Algorithms for constrained unification have been investigated in, e.g., \cite{KajFujKas_SolvingAUnificationProblemUnderConstrainedSubstitutionsUsingTreeAutomata_JSC97}.
Also, research on order-sorted unification \cite{MesGogSmo_OrderSortedUnification_JSC89} is connected with this problem.

\subsubsection{Logic-Based Conditional Critical Pairs}
\label{SecLCCPs}

In order to avoid the computation of possibly infinitely many $E$-unifiers of two terms, 
we introduce the following.

\begin{defi}[Logic-based conditional critical pair]
Let $\alpha:\lhsr\to\rhsr\IF\gencond$ 
and
$\alpha':\lhsr\to\rhsr\IF\gencond'$ be rules
sharing no variables, 
together with an active non-variable position $p\in\Pos^\mu_\Symbols(\lhsr)$.
The \emph{Logic-based Conditional Critical Pair} 
(LCCP for short) 
$\pi_{\alpha,p,\alpha'}$ of $\alpha$ at position $p$ with $\alpha'$ 
is:
\begin{IEEEeqnarray}{r'C'l}
\lccpOf{\alpha,p,\alpha'}:\langle \lhsr[\rhsr']_p,\rhsr\rangle\IF \lhsr|_p\equequ{} \lhsr', \gencond,\gencond'\label{LblLogicBaseDCP}
\end{IEEEeqnarray}
\end{defi}
%
\noindent
We avoid the \emph{explicit} use of $E$-unifiers
by including the \emph{$E$-overlapping test} as an atom
$\lhsr|_p=\lhsr'$ in the conditional part of (\ref{LblLogicBaseDCP}).
We often  remove from the conditional part of LCCPs trivial equational components $t=t$ introduced in this way.
\begin{rem}[Terminology]
In contrast to ECCPs above, 
LCCPs do not explicitly use $E$-unifiers for the $E$-unification problem $\lhsr|_p =\lhsr'$ corresponding to targeted $E$-critical peaks. Still, they include an explicit representation of the $E$-unification problem in the conditional part, hence we keep ``critical'' as part of the denomination of these conditional pairs which are intended to capture $E$-critical peaks.
However, we use ``\emph{logic-based}'' to avoid confusion with \emph{$E$-critical pairs} for \etrs{s}
\cite[Definition 12]{JouKir_CompletionOfASetOfRulesModuloASetOfEquations_SIAMJC86}
and with ECCPs above.
\end{rem}
%
\noindent
In sharp contrast to CCPs (see Proposition \ref{PropJoinabilityOfImproperCriticalPeakssIn2CTRSs}), 
improper LCCPs $\ilccpOf{\alpha}:\langle\rhsr',\rhsr\rangle\IF\lhsr=\lhsr',\gencond,\gencond'$ of 2-rules 
$\alpha:\lhsr\to\rhsr\IF\gencond\in R^\RMabbr$, with renamed version $\alpha':\lhsr'\to\rhsr'\IF\gencond'$ are \emph{not} trivial: terms $\rhsr$ and $\rhsr'$ are identical but only \emph{up to renaming}.
Furthermore, in sharp contrast to the unequational case (i.e., \gtrs{s}), 
unconditional rules may define improper LCCPs which are \emph{not} $\ejoinability{\cR/E}$-joinable.\footnote{This observation and the example below are due to an anonymous referee.}

\begin{exa}\label{ExImproperLCCPsForETRSs}
Consider the \etrs{}
\begin{eqnarray}
\fS{f}(\fS{a}) & = & \fS{f}(\fS{b})\label{ExImproperLCCPsForETRSs_eq1}
\\
\fS{f}(x) & \to & x\label{ExImproperLCCPsForETRSs_rule1}
\end{eqnarray}
The improper LCCP 
\begin{eqnarray}
\ilccpOf{(\ref{ExImproperLCCPsForETRSs_rule1})}: && \langle x,x'\rangle\IF\fS{f}(x)\equequ{}\fS{f}(x')
\end{eqnarray}
is \emph{not} $\ejoinability{\cR/E}$-joinable as $\sigma=\{x\mapsto\fS{a},x'\mapsto\fS{b}\}$ satisfies the conditional part, but $\langle\sigma(x),\sigma(x')\rangle=\langle\fS{a},\fS{b}\rangle$ is \emph{not} $\ejoinability{\cR/E}$-joinable as both $\fS{a}$ and $\fS{b}$ are $\rew{\cR/E}$-irreducible.
\end{exa}
%
\noindent
However, we have the following obvious result.

\begin{prop}
\label{PropRightGroundRulesHaveTrivialfImproperLCCPs}
Improper LCCPs $\ilccpOf{\alpha}$ of right-ground (possibly conditional) rules $\alpha$ 
are trivial and hence $\joinability{\cR^\RMabbr}$-joinable.
\end{prop}

\subsubsection{ECCPs vs.\ LCCPs}
\label{SecECCPsVsLCCPs}

Regarding the relationship between ECCPs and LCCPs, we have the following.

\begin{prop}
Let $\cR$ be an \egtrs{}. Let $\alpha$ and $\alpha'$ be rules.
\begin{enumerate}
\item If for all $\theta\in\fS{CSU}_E(\lhsr|_p,\lhsr')$, 
$\eccpOf{\alpha,p,\alpha',\theta}$
is 
$\joinability{\cR^\RMabbr,E}$-joinable, then $\lccpOf{\alpha,p,\alpha'}$ is $\joinability{\cR^\RMabbr,E}$-joinable.
\item If $\lccpOf{\alpha,p,\alpha'}$
is 
$\joinability{\cR^\RMabbr,E}$-joinable, 
then 
for all $\theta\in\fS{CSU}_E(\lhsr|_p,\lhsr')$, $\eccpOf{\alpha,p,\alpha',\theta}$ is $\joinability{\cR^\RMabbr,E}$-joinable.
\end{enumerate}
\end{prop}

\begin{proof}\leavevmode
\begin{enumerate}
\item For $\alpha:\lhsr\to\rhsr\IF\gencond$, 
$\alpha':\lhsr'\to\rhsr'\IF\gencond'$, $p\in\Pos^\mu_\Symbols(\lhsr)$, and
$\theta\in\fS{CSU}_E(\lhsr|_p,\lhsr')$, we obtain
 $\eccpOf{\alpha,p,\alpha',\theta}:\langle\theta(\lhsr)[\theta(\rhsr')]_p,\theta(\rhsr)\rangle\IF\theta(\gencond),\theta(\gencond')$.
Let  $\sigma$ be such that both $\sigma(\theta(\gencond))$ and $\sigma(\theta(\gencond'))$ 
hold. 
Since $\theta(\lhsr|_p)\equequ{E}\theta(\lhsr')$ holds (as $\theta\in\fS{CSU}_E(\lhsr|_p,\lhsr')$), 
by stability of equality (c.f.\ Proposition \ref{PropProvingAtomsIsStable}), $\sigma(\theta(\lhsr|_p))\equequ{E}\sigma(\theta(\lhsr'))$ holds as well.
Thus, the conditional part of $\lccpOf{\alpha,p,\alpha'}:\langle \lhsr[\rhsr']_p,\rhsr\rangle\IF \lhsr|_p\equequ{} \lhsr', \gencond,\gencond'$ is satisfied by $\sigma\circ\theta$.
By $\joinability{\cR^\RMabbr,E}$-joinability of 
$\eccpOf{\alpha,p,\alpha',\theta}$, the $\joinability{\cR^\RMabbr,E}$-joinability of
$\lccpOf{\alpha,p,\alpha'}$ follows.
\item Let $\sigma$ be a substitution satisfying the conditional part of $\lccpOf{\alpha,p,\alpha'}$.
In particular, $\sigma(\lhsr|_p)\equequ{E}\sigma(\lhsr')$ holds, i.e., $\sigma$ is a $E$-unifier of $\lhsr|_p$ and $\lhsr'$.
Let $\theta\in\fS{CSU}_E(\lhsr|_p,\lhsr')$ be such that  $\eccpOf{\alpha,p,\alpha',\theta}$ is obtained and $\sigma=\tau\circ\theta$ for some substitution $\tau$.
Thus, the conditional part of $\eccpOf{\alpha,p,\alpha',\theta}$ also holds and by $\joinability{\cR^\RMabbr,E}$-joinability of $\lccpOf{\alpha,p,\alpha'}$, the
$\joinability{\cR^\RMabbr,E}$-joinability of
$\eccpOf{\alpha,p,\alpha',\theta}$ follows.\qedhere
\end{enumerate}
\end{proof}
%
\noindent
Every LCCP $\lccpOf{\alpha,p,\alpha'}:\langle s,t\rangle\IF\gencond$ 
includes an $E$-unification problem $\lhsr|_p=\lhsr'$ in $\gencond$.
Even if $\fS{CSU}_E(\lhsr|_p,\lhsr')$ is not computable or it is infinite, it is still possible to prove
$\gencond$ infeasible to finally conclude joinability of  $\lccpOf{\alpha,p,\alpha'}$.
For ECCPs $\eccpOf{\alpha,p,\alpha',\theta}$, exploiting infeasibility is also used 
to prove joinability (see, e.g.,
\cite[Theorem 3]{DurMes_OnTheChurchRosserAndCoherencePropertiesOfConditionalOrderSortedRewriteTheories_JLAP12}).
However, the generation of ECCPs is necessary before checking them for infeasibility, which makes the approach unaffordable in the aforementioned conditions for $\fS{CSU}_E(\lhsr|_p,\lhsr')$, as $\theta\in\fS{CSU}_E(\lhsr|_p,\lhsr')$.
This does not happen for LCCPs.
Of course, if $\gencond$ in $\lccpOf{\alpha,p,\alpha'}$ above  is feasible, a proof of joinability of $\lccpOf{\alpha,p,\alpha'}$ 
would require checking joinability of $\sigma(s)$ and $\sigma(t)$ for 
all subtitutions $\sigma$ satisfying $\gencond$.

\subsubsection{E-critical peaks as LCCPs}

Given a replacement map $\mu$ and
sets of rules $U$ and $V$, 
we let
\begin{IEEEeqnarray}{+r'C'l+}
\GLCCP(U,V) & = & \{\lccpOf{\alpha,p,\alpha'}\mid\alpha:\lhsr\to\rhsr\IF\gencond\in U,p\in\Pos^\mu_\Symbols(\lhsr),\alpha'\in V
\}\label{LblGLCCPs}
\end{IEEEeqnarray}
As for CCPs (Definition \ref{DefProperAndImproperCCPs}), we distinguish between proper and improper LCCPs in the same way.
We also say that $\lccpOf{\alpha,p,\alpha'}$ is a \emph{coherence} LCCP 
if either $\alpha$ or $\alpha'$ belong to $\eqoriented{E}$.

\begin{prop}[$E$-critical peaks as LCCPs]
\label{PropCriticalPSrUpPeaksAndConditionalCriticalPairs}
Let $\cR=(\Symbols,\SPredicates,\mu,E,H,R)$ be an \egtrs{}. 
Let $\alpha:\lhsr\to\rhsr\IF\gencond\in R^\RMabbr\:\cup\eqoriented{E}$ and $\alpha':\lhsr'\to\rhsr'\IF\gencond'\in R^\RMabbr$,  $p\in\Pos^\mu_\Symbols(\lhsr)$ 
be such that (\ref{LblECriticalPeak}) is an $E$-critical peak
for some substitution $\sigma$,
and $\lccpOf{\alpha,p,\alpha'}:\langle \lhsr[\rhsr']_p,\rhsr\rangle\IF \lhsr|_p\equequ{} \lhsr', \gencond,\gencond'$.
Then, $\sigma$ satisfies the conditional part of $\lccpOf{\alpha,p,\alpha'}$; and $t=\sigma(\lhsr[\rhsr']_p)$ and $t'=\sigma(\rhsr)$
in (\ref{LblECriticalPeak}).
\end{prop}
\begin{proof}
Let (\ref{LblECriticalPeak}) be
\[\begin{array}{rcl}
t = \sigma(\lhsr)[\sigma(\rhsr')]_p ~\leftarrowWith{\cR^\RMabbr,E}   & \sigma(\lhsr)[w]_p & \rewAtPos{\toppos}{\cR^\RMabbr}~\sigma(\rhsr)=t'
\end{array}
\]
where 
$p\in\Pos^\mu_\Symbols(\lhsr)$ and $\sigma(\gencond)$ and $\sigma(\gencond')$ are 
satisfied.
Since $\sigma(\lhsr)|_p=\sigma(\lhsr|_p)=w=_E\sigma(\lhsr')$ in (\ref{LblECriticalPeak}),
$\sigma$ satisfies the conditional part of $\lccpOf{\alpha,p,\alpha'}$.
\end{proof}
%
\noindent
The following result is useful to analyze joinability of LCCPs (as infeasibility).
Here, given a set of rules $S$, $\DSymbols(S)=\{\rootTerm(\lhsr)\mid \lhsr\to\rhsr\IF\gencond\in S\}$ is the set of function symbols \emph{defined} by means of some rule in the usual sense.

\begin{prop}
\label{PropInfeasibilityOfLCCPsWithEqCondWithNoEqSymbol}
Let $\cR$ 
be an \egtrs{}, 
$\alpha:\lhsr\to\rhsr\IF\gencond,\alpha':\lhsr'\to\rhsr'\IF\gencond'\in R^\showRMabbr$ 
be such that $\rootTree(\lhsr|_p)\notin\DSymbols(\eqoriented{E})$
and $\rootTree(\lhsr|_p)\neq\rootTree(\lhsr')$.
If 
$\eqoriented{E}$ is \emph{not} collapsing,
then $\lccpOf{\alpha,p,\alpha'}$ is infeasible.
\end{prop}

\begin{proof}
Since $p\in\Pos^\mu_\Symbols(\lhsr)$ 
and 
$\alpha'\in R^\showRMabbr$, 
we have $\lhsr|_p,\lhsr'\notin\Variables$.
If $\sigma$ satisfies $\lhsr|_p=\lhsr'$,
then, 
\begin{IEEEeqnarray}{r'C'l}
\sigma(\lhsr|_p)\equequ{E}
\sigma(\lhsr')\label{PropInfeasibilityOfLCCPsWithEqCondWithNoEqSymbol_Sequence}
\end{IEEEeqnarray}
Note that $\rootTree(\lhsr|_p)\neq\rootTree(\lhsr')$ and $\rootTree(\lhsr|_p)\notin\DSymbols(\eqoriented{E})$.
Reductions with $\rew{\eqoriented{E}}$ are not applied at the root and cannot change 
$\rootTree(\sigma(\lhsr|_p))$, which is $\rootTree(\lhsr|_p)$.
Since $\lhsr'$ is not a variable and $\rootTree(\sigma(\lhsr'))=\rootTree(\lhsr')\neq\rootTree(\lhsr|_p)$, sequence (\ref{PropInfeasibilityOfLCCPsWithEqCondWithNoEqSymbol_Sequence}) is not possible.
Thus, $\lccpOf{\alpha,p,\alpha'}$ is infeasible.
\end{proof}
%
\noindent
The proof of the following result 
is analogous to 
that of Proposition \ref{PropInfeasibilityOfLCCPsWithEqCondWithNoEqSymbol}.
\begin{prop}
\label{PropInfeasibilityOfELCCPsWithEqCondWithNoEqSymbol}
Let $\cR$
be an \egtrs{},
$\alpha:\lhsr\to\rhsr\IF\gencond
\in\:\eqoriented{E}$ and
$\alpha':\lhsr'\to\rhsr'\IF\gencond'\in R^\RMabbr$
be such that
$\rootTree(\lhsr')\notin\DSymbols(\eqoriented{E})$,
and $\rootTree(\lhsr|_p)\neq\rootTree(\lhsr')$.
If $\eqoriented{E}$ is \emph{not} collapsing, 
then $\lccpOf{\alpha,p,\alpha'}$ is 
infeasible.
\end{prop}
%
\noindent
Non-collapsingness of $\eqoriented{E}$ is necessary for
Proposition \ref{PropInfeasibilityOfELCCPsWithEqCondWithNoEqSymbol} to hold.

\begin{exa}\label{ExNeedOfNonCollapsingness}
Consider the following \egtrs{}
\begin{IEEEeqnarray}{r'C'l}
\fS{0}+x & = & x\label{ExNeedOfNonCollapsingness_eq1}\\
\fS{f}(\fS{0}) & \to & \fS{0}\label{ExNeedOfNonCollapsingness_rule1}
\end{IEEEeqnarray}
where $\DSymbols(\eqoriented{E})=\{\arityOf{+}{2}\}$
and the root symbol $\fS{f}$ in the left-hand side of (\ref{ExNeedOfNonCollapsingness_rule1}) is not in $\DSymbols(\eqoriented{E})$.
The LCCP 
$\lccpOf{\overrightarrow{(\ref{ExNeedOfNonCollapsingness_eq1})},\toppos,(\ref{ExNeedOfNonCollapsingness_rule1})}:\langle (\fS{0},x\rangle\IF\fS{0}+x\equequ{}\fS{f}(\fS{0})$ 
is 
feasible as $\sigma=\{x\mapsto\fS{f}(\fS{0})\}$ satisfies its conditional part. 
\end{exa}

\subsection{$E$-variable confluence and coherence peaks and conditional variable pairs}
\label{SecCVPsForEvariablePeaks}

If $p\notin\Pos^\mu_\Symbols(\lhsr)$ in (\ref{NestedPSRupPeak}), then there is $x\in\Var^\mu(\lhsr)$ 
and $q\in\Pos^\mu(\lhsr)$ such that $\lhsr|_q=x$ and $q\leq p$.
Then, we say
that (\ref{NestedPSRupPeak}) is an \emph{$E$-variable peak}
\begin{IEEEeqnarray}{r'C'l}
\underbrace{\sigma(\lhsr)[\sigma(\rhsr')]_p}_{t} ~\leftrewAtPos{p}{\alpha',E}   & \underbrace{\sigma(\lhsr)[w]_p}_s & \rewAtPos{\toppos}{\alpha}~\underbrace{\sigma(\rhsr)}_{t'}\qquad\qquad\text{where }p\notin\Pos^\mu_\Symbols(\lhsr)\label{LblEVariablePeak}
\end{IEEEeqnarray}
%
\noindent
Depending on the \emph{origin} of rules in $\alpha$ and $\alpha'$ in (\ref{LblEVariablePeak}), we have different kinds of \mbox{$E$-variable} peaks for different uses
(see Table \ref{TableDifferentKindsOfEVariablePeaks}).
 \begin{figure}[t] 
 \begin{center}
\begin{tikzpicture}[node distance = 6cm, auto,scale=1.5,every node/.style={scale=1}]
\node (VP_Ttext) {$\sigma(\lhsr)[\sigma(\rhsr')]_p$};
\node [triangleW, below = 0cm of VP_Ttext, minimum width = 2cm, minimum height = 2.35cm,scale=0.9] (VP_TTriangle) {};
\node [triangleG, below  = 0cm of VP_Ttext, minimum width = 0.85cm, minimum height = 1cm, scale=0.9] (VP_LHSTriangleT) {};
\node [below  = 0.25 cm of VP_Ttext] (VP_TLHStext) {$\lhsr$};
\node [triangleW, below  = 0.14cm of VP_TLHStext, minimum width = 1.3cm, minimum height = 1.32cm, scale=0.9] (VP_SigmaXtriangleT) {};
\node [triangleG, below  = 1.2cm of VP_Ttext, minimum width = 1cm, minimum height = 1.015cm, scale=0.9] (VP_LHSReduced) {};
\node [below  = -0.45 cm of VP_LHSReduced] (VP_RHSptext) {\tiny $\sigma(\rhsr')$};
\node [above right = 1cm and 1.8 cm of VP_Ttext] (VP_Stext) {$\sigma(\lhsr)[w]_p$};
\node [triangleW, below  = 0cm of VP_Stext, minimum width = 2cm, minimum height = 2.35cm, scale=0.9] (VP_Striangle) {};
\node [triangleG, below  = 0cm of VP_Stext, minimum width = 0.85cm, minimum height = 1cm, scale=0.9] (VP_LHSTriangleS) {};
\node [below  = 0.2 cm of VP_Stext] (VP_LHStextS) {$\lhsr$};
\node [triangleW, below  = 0.16cm of VP_LHStextS, minimum width = 1.3cm, minimum height = 1.35cm, scale=0.9] (VP_SigmaXtriangleS) {};
\node [above right = 1cm and 0.5 cm] at (VP_Striangle) (QinLhsr) {\tiny $q\in \Pos^\mu_x(\lhsr)$};
\node [above right = 0.5cm and 1 cm] at (VP_Striangle) (PinS) {\tiny $p\in \Pos^\mu(s)$};
\node [triangleW, below  = 1.2cm of VP_Stext, minimum width = 1cm, minimum height = 1.01cm, scale=0.9] (VP_LHSpTriangleS) {};
\node [below  = -0.45 cm of VP_LHSpTriangleS] (VP_LHSptext) {\tiny $w$};
\node [below right =1cm and 3 cm of VP_Stext] (VP_Ttext) {$\sigma(\rhsr)$};
\node [triangleW, below  = 0cm of VP_Ttext, minimum width = 3cm, minimum height = 2.25cm] (VP_Ttriangle) {};
\node [triangleG, below  = 0cm of VP_Ttext, minimum width = 1.67cm, minimum height = 1.28cm, scale=0.9] (VP_RHStTriangle) {};
\node [below  = 0.35 cm of VP_Ttext] (VP_RHStextTp) {$\rhsr$};
\node [triangleW, below  left = 0.85cm and 0.15cm of VP_RHStextTp, minimum width = 1.4cm, minimum height = 1.22cm, scale=0.9] (VP_SigmaXtriangleTp) {};
\node [triangleW, below left = 0.61cm and 0.05cm of VP_RHStTriangle, minimum width = 1cm, minimum height = 1.01cm, scale=0.9] (VP_LHSpTriangleT) {};
\node [below  = -0.45 cm of VP_LHSpTriangleT] (VP_LHSptext) {\tiny $w$}; 
\node [triangleW, below  right = 0.85cm and 0.13cm of VP_RHStextTp, minimum width = 1.4cm, minimum height = 1.22cm, scale=0.9] (VP_SigmaXtriangleTp) {};
\node [triangleW, below right = 0.61cm and 0.04cm of VP_RHStTriangle, minimum width = 1cm, minimum height = 1.01cm, scale=0.9] (VP_LHSpTriangleTBis) {};
\node [below  = -0.45 cm of VP_LHSpTriangleTBis] (VP_LHSptextBis) {\tiny $w$}; 
\draw[-*] ([shift={(0.2,0.2)}]QinLhsr) -- ([shift={(-0.05,0)}]VP_LHSTriangleS.south);
\draw[-*] ([shift={(0.2,0.2)}]PinS) -- ([shift={(-0.05,0)}]VP_LHSpTriangleS.north);
\draw[->] ([shift={(-0.2,0.2)}]VP_Striangle.west) -- ([shift={(0.1,0.3)}]VP_TTriangle.east);
\node[above right = 0.15cm and 0.2cm of VP_TTriangle.east]{\tiny $\alpha',E$};
\draw[->] ([shift={(0.2,0.2)}]VP_Striangle.east) -- ([shift={(-0.2,0.2)}]VP_Ttriangle.west);
\node[above left = 0.08cm and 0.35cm of VP_Ttriangle.west]{\tiny $\alpha$};
\end{tikzpicture}
 \end{center}
\caption{E-variable peak (where $w\equequ{E}\sigma(\lhsr')$)}
 \label{FigEVariablePeak}
 \end{figure}
\begin{table}[ht]
\caption{Different kinds of E-variable peaks (\ref{LblEVariablePeak})}
\label{TableDifferentKindsOfEVariablePeaks}
\centering
\begin{tabular}{@{}cccll@{}}
\toprule
$\alpha$ & $\alpha'$ & $\gbop$ & Denomination & Use
\\
\midrule
$R^\RMabbr$ & $R^\RMabbr$ & $\rewpstickel{}$ & Confluence E-Variable peak & $\fS{LCON}_E(\rew{\cR^\RMabbr,E},\rew{\cR^\RMabbr})$
\\
$\eqoriented{E}$ & $R^\RMabbr$ & $\rewpstickel{}$ & Coherence E-Variable peak & $\fS{LCOH}_E(\rew{\cR^\RMabbr,E})$
\\
\bottomrule
\end{tabular}
\end{table}
In contrast to variable peaks (\ref{LblVariablePeak}), though, with
$E$-variable peaks (\ref{LblEVariablePeak})
we use CVPs
$\cvpOf{\rewpstickel{}}{\alpha,x,p}$
because reductions below active variables of the left-hand side of  rule $\alpha$ use $\rew{\cR^\RMabbr,E}$ instead of $\rew{\cR^\RMabbr}$
(see Figure \ref{FigEVariablePeak}).
Proposition \ref{PropVariablePeaksAsCVPs} is immediately 
extended to 
$E$-variable peaks to show them captured as $\rewpstickel{}$-CVPs.
\begin{prop}[E-Variable peaks as $\rewpstickel{}$-CVPs]
\label{PropEVariablePeaksAsPStickelCVPs}
Let 
$\cR$ 
be an \egtrs{} and
(\ref{LblEVariablePeak}) be an  $E$-variable peak.
Then, $\cvpOf{\rewpstickel{}}{\alpha,x,q}$ as in (\ref{LblGenericCVP})
is such that 
$\sigma(x)\rew{\cR,E}\sigma(x')$ 
holds.
Furthermore, 
$s=\sigma(\lhsr[x]_p)$,
$t=\sigma(\lhsr[x']_p)$,
$t'=\sigma(\rhsr)$, and 
$\cvpOf{\rewpstickel{}}{\alpha,x,q}$ is feasible.
\end{prop}
%
\noindent
Proposition \ref{PropJoinableVariablePeaksForCriticalVariablesNotInCondition} 
has the following counterpart for
E-variable peaks.

\begin{prop}[Joinable $E$-variable peaks]
\label{PropJoinableVariablePSrUpPeaksForCriticalVariablesNotInCondition}
An $E$-variable peak (\ref{LblEVariablePeak}) is 
$\lsjoinability{\cR^\RMabbr,E}$-joinable 
if 
$\alpha$ is $\mu$-left-homogeneous and $\mu$-compatible.
Furthermore, if $x\in\Var^\mu(\rhsr)$, then (\ref{LblEVariablePeak}) is 
$\sjoinability{\cR^\RMabbr,E}$-joinable.
\end{prop}
\begin{cor}
\label{CoroJoinableVariablePSrUpPeaksForCriticalVariablesNotInCondition}
Let $\cR$ 
be an \egtrs{} and $\alpha$
be a rule.
Then, $\cvpPStickelOf{\alpha,x,q}$ is $\lsjoinability{\cR^\RMabbr,E}$-joinable if 
$\alpha$ is $\mu$-left-homogeneous and $\mu$-compatible.
Furthermore, if $x\in\Var^\mu(\rhsr)$, then $\cvpPStickelOf{\alpha,x,q}$ is 
$\sjoinability{\cR^\RMabbr,E}$-joinable
\end{cor}
%
\noindent
For $E$-variable peaks  (\ref{LblEVariablePeak}),
the proof of the following result is similar to that of Proposition \ref{PropJoinableVariableEupCPeaksForCriticalVariablesNotInCondition} for variable peaks (\ref{LblVariablePeak}).

\begin{prop}\label{PropJoinabilityOfEVariablePeaksForEquationalRules}
Let $\cR$ 
be an \egtrs{}.
An $E$-variable peak (\ref{LblEVariablePeak}).
for
$\alpha\in\:\eqoriented{E}$ is $\ejoinability{\cR^\RMabbr,E}$-joinable if $\alpha$ is $\mu$-left-homogeneous and $\mu$-compatible. 
Furthermore,
\begin{enumerate}
\item If $x$ occurs more than once in $\lhsr$, then $(\ref{LblEVariablePeak})$ is $\lsjoinability{\cR^\RMabbr,E}$-joinable.
\item 
If $x\in\Var^\mu(\rhsr)$, then $(\ref{LblEVariablePeak})$ is $\rsjoinability{\cR^\RMabbr,E}$-joinable.
\end{enumerate}
\end{prop}

\begin{cor}
\label{PropJoinabilityOfPStickelCVOsForEquationalRules}
Let $\cR$ 
be an \egtrs{}
and
$\alpha\in\:\eqoriented{E}$.
Then, 
$\cvpPStickelOf{\alpha,x,p}$ 
is $\ejoinability{\cR^\RMabbr,E}$-joinable if $\alpha$ is $\mu$-left-homogeneous and $\mu$-compatible. 
Furthermore,
\begin{enumerate}
\item If $x$ occurs more than once in $\lhsr$, then $\cvpPStickelOf{\lhsr,x,p}$ is $\lsjoinability{\cR^\RMabbr,E}$-joinable.
\item 
If $x\in\Var^\mu(\rhsr)$, then $\cvpPStickelOf{\lhsr,x,p}$ is $\rsjoinability{\cR^\RMabbr,E}$-joinable.
\end{enumerate}
\end{cor}

\subsection{$E$-down peaks}
\label{SecEDownPeaks}

Due to the use of $E$-matching in the $\rew{\alpha',E}$-step of a peak (\ref{LblPreEDownPeak}), i.e.,
\begin{IEEEeqnarray*}{+r'C'l+x*}
\underbrace{\sigma(\rhsr')}_{t}~\leftrewAtPos{\toppos}{\alpha',E}   & \underbrace{w[\sigma(\lhsr)]_p}_{w} & \rewAtPos{p}{\alpha}~\underbrace{w[\sigma(\rhsr)]_p}_{t'}
& \eqref{LblPreEDownPeak}
\end{IEEEeqnarray*}
it is not possible to qualify position $p>\toppos$, where the $\rew{\alpha}$-step  occurs, as witnessing an $E$-critical or $E$-variable peak. 
Furthermore, in view of Proposition \ref{PropJoinabilityOfEdownECohPeaks}(\ref{PropJoinabilityOfEdownECohPeaks_Edown}), only the use of rules $\alpha$ in $R^\RMabbr$ (in an $E$-confluence peak) can be harmful (regarding $\ejoinability{\cR^\RMabbr,E}$-joinability).
The following example illustrates the situation.

\begin{exa}
\label{ExPeakNoCPs_EdownPeaks}
For the \etrs{} $\cR$ in Example \ref{ExPeakNoCPs}, i.e.,
\begin{IEEEeqnarray*}{r'C'l+x*}
\fS{b} & = &\fS{f}(\fS{a}) 
&
\eqref{ExPeakNoCPs_eq1}
\\
\fS{a} & = & \fS{c}
&
\eqref{ExPeakNoCPs_eq2}
\\
\fS{c} & \to & \fS{d}
&
\eqref{ExPeakNoCPs_rule1}
\\
\fS{b} & \to & \fS{d}
&
\eqref{ExPeakNoCPs_rule2}
\end{IEEEeqnarray*}
there is no confluence $E$-critical peak, as constants $\fS{c}$ and $\fS{b}$ in the left-hand sides of rules
(\ref{ExPeakNoCPs_rule1}) and (\ref{ExPeakNoCPs_rule2}) do \emph{not} $E$-unify.
There is no confluence $E$-variable peak, as the left-hand sides of the rules are \emph{ground}.
However, the following local $E$-confluence peak
\begin{IEEEeqnarray}{r'C'l}
\fS{d}~\leftrewAtPos{\toppos}{(\ref{ExPeakNoCPs_rule2}),E}   & \fS{f}(\fS{c})  & \rewAtPos{1}{(\ref{ExPeakNoCPs_rule1})}\fS{f}(\fS{d})
\label{ExPeakNoCPs_PSRDownPeak}
\end{IEEEeqnarray}
is not $\ejoinability{\cR/E}$-joinable, as $\fS{d}$ and $\fS{f}(\fS{d})$ are neither $\rew{\cR/E}$-reducible nor $E$-equivalent.
By Corollary \ref{CoroDisprovingEConcludenceByNonREJoinableLocalConfluencePeaks}, 
$\cR$ is \emph{not} $E$-confluent.
\end{exa}
%
\noindent
This example shows that $E$-critical and $E$-variable peaks do \emph{not} cover all causes of non-local confluence of $\rew{\cR^\RMabbr,E}$ modulo $E$ with $\rew{\cR^\RMabbr}$, and hence of non-$E$-confluence.

\begin{defi}[$E$-down peak]
A peak (\ref{LblPreEDownPeak}), i.e.,
\begin{IEEEeqnarray}{r'C'l}
\sigma(\rhsr')~\leftrewAtPos{\toppos}{\alpha',E}   & w[\sigma(\lhsr)]_p & \rewAtPos{p}{\alpha}~w[\sigma(\rhsr)]_p
\label{LblEDownPeak}
\end{IEEEeqnarray}
with $\alpha:\lhsr\to\rhsr\IF\gencond,\alpha':\lhsr'\to\rhsr'\IF\gencond'\in R^\RMabbr$, $w\equequ{E}\sigma(\lhsr')$ and $p>\toppos$
is an \emph{E-down peak}.
\end{defi}
\noindent
We use \emph{E-down peak}, as the $\rew{\alpha'}$-step (using a rule from $R^\RMabbr$) 
is performed \emph{below} the root of $w$, where the root $\rew{\alpha,E}$-step is performed
(Figure \ref{FigPSdownPeak}).
\begin{figure}[t]
 \begin{tikzpicture}[node distance = 6cm, auto,scale=0.5,every node/.style={scale=1}]
\node (ED_Ttext) {$\sigma(\rhsr')$};
\node [triangleW, below = 0cm of ED_Ttext, minimum width = 2cm, minimum height = 2.35cm,scale=0.9] (ED_Ttriangle) {};
\node [triangleG, below  = 0cm of ED_Ttext, minimum width = 0.85cm, minimum height = 1cm, scale=0.9] (ED_RHSTriangleT) {};
\node [below  = 0.3 cm of ED_Ttext] (ED_RHStext) {$\rhsr'$};
\node [above right =1cm and 2.8 cm of ED_Ttext] (ED_Stext) {$w[\sigma(\lhsr)]_p$};
\node [triangleW, below  = 0cm of ED_Stext, minimum width = 2cm, minimum height = 2.35cm, scale=0.9] (ED_Striangle) {};
\node [triangleG, below  = 0.34cm of ED_Stext, minimum width = 1.45cm, minimum height = 1.97cm, scale=0.9] (ED_LHSRprimeOut) {};
\node [triangleG, below left  = -0.43cm and -0.55cm of ED_LHSRprimeOut, minimum width = 0.8cm, minimum height = 1.05cm, scale=0.9] (ED_LHSRprimeIn) {};
\node [below  left = -0.05 cm and -0.8cm of ED_LHSRprimeOut] (ED_LHSRpInTxt) {\tiny $\sigma(\lhsr)$};
\node [above right = 1.4cm and 1 cm] at (ED_LHSRprimeOut) (PinLhsrOut) {\tiny $p?$};
\node [above right = 0.8cm and 1 cm] at (ED_LHSRprimeIn) (PinLhsrIn) {\tiny $p?$};
\draw[dotted] (ED_LHSRprimeOut.north) -- (ED_LHSRprimeIn.north);
\node [below right =1cm and 2.5 cm of ED_Stext] (ED_TpText) {$w[\sigma(\rhsr)]_p$};
\node [triangleW, below  = 0cm of ED_TpText, minimum width = 2cm, minimum height = 2.35cm, scale=0.9] (ED_TpTriangle) {};
\node [triangleG, below  = 0.34cm of ED_TpText, minimum width = 1.45cm, minimum height = 1.97cm, scale=0.9] (ED_RHSRprimeOut) {};
\node [triangleG, below left  = -0.43cm and -0.55cm of ED_RHSRprimeOut, minimum width = 0.8cm, minimum height = 1.05cm, scale=0.9] (ED_RHSRprimeIn) {};
\node [below  left = -0.05 cm and -0.8cm of ED_RHSRprimeOut] (ED_RHSRpInTxt) {\tiny $\sigma(\rhsr)$};
\node [above right = 1.4cm and 1 cm] at (ED_RHSRprimeOut) (PinRhsrOut) {\tiny $p?$};
\node [above right = 0.8cm and 1 cm] at (ED_RHSRprimeIn) (PinRhsrIn) {\tiny $p?$};
\draw[dotted] (ED_RHSRprimeOut.north) -- (ED_RHSRprimeIn.north);
\draw[-*] ([shift={(0.1,0.1)}]PinLhsrOut) -- ([shift={(-0.1,0)}]ED_LHSRprimeOut.north);
\draw[-*] ([shift={(0.2,0.2)}]PinLhsrIn) -- ([shift={(-0.1,0)}]ED_LHSRprimeIn.north);
\draw[-*] ([shift={(0.1,0.1)}]PinRhsrOut) -- ([shift={(-0.1,0)}]ED_RHSRprimeOut.north);
\draw[-*] ([shift={(0.2,0.2)}]PinRhsrIn) -- ([shift={(-0.1,0)}]ED_RHSRprimeIn.north);
\draw[->] ([shift={(-0.2,0.2)}]ED_Striangle.west) -- ([shift={(0.1,0.3)}]ED_Ttriangle.east);
\node[above right = -0.15cm and 0.2cm of ED_Ttriangle.east]{\tiny $\alpha',E$};
\draw[->] ([shift={(0.2,0.2)}]ED_Striangle.east) -- ([shift={(-0.2,0.2)}]ED_TpTriangle.west);
\node[above left = -0.1cm and 0.35cm of ED_TpTriangle.west]{\tiny $\alpha$};
 \end{tikzpicture}
 \caption{E-down peak (where $w\equequ{E}\sigma(\lhsr')$ and $\alpha\in R^\RMabbr$)}
 \label{FigPSdownPeak}
\end{figure}
Peak (\ref{ExPeakNoCPs_PSRDownPeak}) is an example of an $E$-down peak.
If $E=\emptyset$, E-down peaks are captured as critical or variable peaks.
This explains why they do not play any role in \gtrs{s}.
\begin{prop}
\label{PropEDownPeaksInGTRSs}
Let $\cR$ be an \egtrs{} such that $E=\emptyset$.
Then, every E-down peak (\ref{LblEDownPeak}) is either a critical peak (\ref{LblCriticalPeak}) or a variable peak (\ref{LblVariablePeak}).
\end{prop}
\begin{proof}
According to the discussion in Section~\ref{SecEGTRSsWithoutEquationsAsGTRSs}, since $\equequ{E}$ is the equality of terms, $w$ in (\ref{LblEDownPeak}) is $w=\sigma(\lhsr')$, i.e., it becomes 
\begin{IEEEeqnarray}{r'C'l}
\sigma(\rhsr')~\leftrewAtPos{\toppos}{\alpha'}   & \sigma(\lhsr')[\sigma(\lhsr)]_p & \rewAtPos{p}{\alpha}~\sigma(\lhsr')[\sigma(\rhsr)]_p
\label{LblEDownPeakGTRS}
\end{IEEEeqnarray}
with $\alpha,\alpha'\in R^\RMabbr$, which is (symmetric to) a peak (\ref{NestedRMRPeakTermsRup}) which, as explained in Sections \ref{SecCriticalPeaksAndCCPs} and \ref{SecVariablePeaksAndCVPs}, can be investigated either as a critical peak (\ref{LblCriticalPeak})  or as a variable peak (\ref{LblVariablePeak}).
\end{proof}
%
\noindent
Although $E$-down peaks combine possible \emph{rule $E$-overlaps} and 
the application of rules `below' a variable, 
these two sources of divergence do not admit a neat capture by means of LCCPs and CVPs.
We introduce \emph{down conditional pairs}
to capture them at once.
\begin{defi}[Down conditional pair]
\label{DefDownConditionalPair}
Let $\cR=(\Symbols,\SPredicates,\mu,E,H,R)$ be an \egtrs{}.
Rule  $\alpha:\lhsr\to\rhsr\IF\gencond\in R^\RMabbr$ 
defines a \emph{Down Conditional Pair} (DCP for short) as follows:
\begin{IEEEeqnarray}{r'C'l}
\dcpOf{\alpha}: \langle \rhsr,x'\rangle\IF x\ceq\lhsr,x\innerrew{} x',\gencond\label{LblDownConditionalCriticalPair}
\end{IEEEeqnarray}
where $x$ and $x'$ are fresh variables.
The set of DCPs obtained in this way 
is
$\SetOf{DCP}(\cR)$.
\end{defi}
%
\noindent
In Definition \ref{DefDownConditionalPair}, predicate symbol $\innerrew{}$
refers to \emph{inner} (i.e., \emph{non-root}) $\rew{\cR^\RMabbr}$-rewriting steps (by Proposition \ref{PropJoinabilityOfEdownECohPeaks}.(\ref{PropJoinabilityOfEdownECohPeaks_Edown}), we do not need to consider $\innerrew{\eqoriented{E}}$-steps).
As usual, $\innerrew{}$ is defined by deduction on $\crtheoryOf{\cR}$ extended with the theory 
\begin{IEEEeqnarray}{r'C'l}
\{\RuleInnerRed_{f,i,\alpha}\mid f\in\Symbols,i\in\mu(f),\alpha\in R^\RMabbr\}
\cup
\{(\RulePropagation)^{\innerrew{}}_{f,i}\mid f\in\Symbols,i\in\mu(f)\}
\end{IEEEeqnarray}
where sentences 
$\RuleInnerRed_{f,i,\alpha}$ for rules 
$\alpha:\lhsr\to\rhsr\IF A_1,\ldots,A_n$
are as follows:
\[
\begin{array}{rcl}
 (\forall\vec{x}\vec{y})A_1\wedge\cdots \wedge A_n\Rightarrow f(x_1,\ldots,x_{i-1},\lhsr,x_{i+1},\ldots,x_k)\innerrew{}f(x_1,\ldots,x_{i-1},\rhsr,x_{i+1},\ldots,x_k)
\end{array}
\]
where 
$\vec{x}$, i.e., $x_1,\ldots,x_k$ are fresh, pairwise distinct variables,
which are not in $\vec{y}$, the variables occurring in $\alpha$.

\begin{prop}[$E$-down peaks as DCPs]
\label{PropRdownPeaksAndDownConditionalCriticalPairs}
Let $\cR$ 
be an \egtrs{}. 
The $E$-down peak 
\[\begin{array}{rcl}
t =  \sigma(\rhsr')~\leftrewAtPos{\toppos}{\alpha',E}   & w[\sigma(\lhsr)]_p & \rewAtPos{p}{\alpha}~w[\sigma(\rhsr)]_p=t'
\end{array}
\]
where $p>\toppos$ and $\alpha,\alpha'\in R^\RMabbr$ is 
$\ejoinability{\cR^\RMabbr,E}$-joinable  
($\ejoinability{\cR/E}$-joinable) iff
$\dcpOf{\alpha'}$
is $\ejoinability{\cR^\RMabbr,E}$-joinable  ($\ejoinability{\cR/E}$-joinable).  
\end{prop}
\begin{proof}
As for the \emph{only if} part, since $w=w[\sigma(\lhsr)]_p=_E\sigma(\lhsr')$, and $x,x'\notin\Var(\alpha,\alpha')$, we can let $\sigma(x)=w[\sigma(\lhsr)]_p$ and $\sigma(x')=w[\sigma(\rhsr)]_p$.
Hence, the conditional part of $\dcpOf{\alpha'}$ is satisfied by $\sigma$.
By $\ejoinability{\cR^\RMabbr,E}$-joinability of the $E$-down peak, 
$\dcpOf{\alpha'}$
is also $\ejoinability{\cR^\RMabbr,E}$-joinable (and hence $\ejoinability{\cR^\RMabbr/E}$-joinable).

As for the \emph{if} part, let $\sigma$ be a substitution satisfying the conditional part of 
$\dcpOf{\alpha'}$.
In particular, $\sigma(x)=_E\sigma(\lhsr')$ holds and, since $\sigma(\gencond')$ also  holds,  we have $\sigma(x)\rew{\cR^\RMabbr,E}\sigma(\rhsr')=t$.
We also have $\sigma(x)\innerrew{\cR^\RMabbr}\sigma(x')$, i.e.,
$\sigma(x)=w[\sigma(\lhsr)]_p$
and 
$\sigma(x')=w[\sigma(\rhsr)]_p=t'$
for some $p>\toppos$ and rule $\lhsr\to\rhsr\IF\gencond$ such that $\sigma(\gencond)$ holds.
By $\ejoinability{\cR^\RMabbr,E}$-joinability of $\dcpOf{\alpha'}$, the $E$-down peak 
is also $\ejoinability{\cR^\RMabbr,E}$-joinable (and hence $\ejoinability{\cR^\RMabbr/E}$-joinable).
\end{proof}

\section{$E$-Confluence of $\cR$ as the Church-Rosser property modulo $E$ of $\rew{\cR^\RMabbr,E}$}
\label{SecEConfluenceWithPSRandCpeaks}

By Corollary \ref{CoroSufficientCriteriaForEConfluenceOfEGRSs}(\ref{CoroSufficientCriteriaForEConfluenceOfEGRSs_JK86}) 
$E$-confluence  of $E$-terminating \egtrs{s} $\cR$ 
follows if 
$\rew{\cR^\RMabbr,E}$ is 
(i) locally confluent modulo $E$ with $\rew{\cR^\RMabbr}$ and 
(ii) locally coherent modulo $E$. 
Section~\ref{SecConditionalPairsForNonDisjointPeaksJK} shows that 
(i) and (ii) 
are guaranteed if 
\begin{enumerate}
\item
 all $E$-critical peaks 
 and 
 $E$-variable peaks,  
 i.e.,
 peaks
\begin{IEEEeqnarray*}{r'C'l+}
\raisebox{0.6cm}{\resizebox{6.5cm}{!}{
\xymatrix{
&  \sigma(\lhsr)[w]_p 
\ar@{->}[dl]_{p}^>>{\alpha',E}
\ar@{->}[dr]^{\toppos}_>>{\alpha} 
\\
\sigma(\lhsr)[\sigma(\rhsr')]_p\qquad  & &  \quad\sigma(\rhsr)
  \\
}}}
\end{IEEEeqnarray*}
where $\alpha:\lhsr\to\rhsr\IF\gencond$ is a rule in $\eqoriented{E}\cup\:R^\RMabbr$, $\alpha':\lhsr'\to\rhsr'\IF\gencond'$ belongs to $R^\RMabbr$, 
and $w\equequ{E}\sigma(\lhsr')$; and
$p\in\Pos^\mu_\Symbols(\lhsr)$ for $E$-critical peaks (\ref{LblECriticalPeak}) 
and  $p\notin\Pos^\mu_\Symbols(\lhsr)$ for $E$-variable peaks (\ref{LblEVariablePeak}),  are $\ejoinability{\cR^\RMabbr,E}$-joinable. 
Propositions \ref{PropCriticalPeaksAsCCPs} and \ref{PropVariablePeaksAsCVPs} establish that this can be done by means of appropriate (feasible) CCPs and CVPs.

\item Besides, we also need to consider
$E$-down peaks (\ref{LblEDownPeak}), i.e.,
\begin{IEEEeqnarray*}{r'C'l+}
\raisebox{0.6cm}{\resizebox{6.5cm}{!}{
\xymatrix{
&  w[\sigma(\lhsr)]_p 
\ar@{->}[dl]_{\toppos}^>>{\alpha',E}
\ar@{->}[dr]^{p}_>>{\alpha} 
\\
\sigma(\rhsr')  & &  \quad w[\sigma(\rhsr)]_p  
 \\
}}}
\end{IEEEeqnarray*}
where both $\alpha:\lhsr\to\rhsr\IF\gencond$ and $\alpha':\lhsr'\to\rhsr'\IF\gencond'$ belong to $R^\RMabbr$,
$w\equequ{E}\sigma(\lhsr')$, and
$p\in\Pos^\mu_\Symbols(\lhsr)-\{\toppos\}$.
Proposition \ref{PropRdownPeaksAndDownConditionalCriticalPairs} establishes that this can be done using DCPs.
\end{enumerate}
Accordingly, using the definition of $\GLCCP$ in (\ref{LblGLCCPs}), 
we introduce the following:
\begin{defi}
\label{DefLCCPsAndCVPsForEGTRSs}
Given an \egtrs{} $\cR=(\Symbols,\SPredicates,\mu,E,H,R)$, we let 
\begin{enumerate} 
\item Local $E$-confluence conditional pairs:
\begin{enumerate}[align=left]
\item\label{DefLCCPsAndCVPsForEGTRSs_LCCPofR} 
$\LCCP(\cR)$ 
is the subset of \emph{feasible} LCCPs from $\GLCCP(R^\RMabbr,R^\RMabbr)$.
\item\label{DefLCCPsAndCVPsForEGTRSs_CVPpstickel}
 $\CVPofPStickel(\cR)$ 
is the subset of \emph{feasible} CVPs from 
$\SetOf{CVP}^{\rewpstickel{}}(R^\RMabbr)$.
\item\label{DefLCCPsAndCVPsForEGTRSs_DCP}
 $\DCP(\cR)$  be the subset 
of \emph{feasible} DCPs from 
$\SetOf{DCP}(R^\RMabbr)$.
\end{enumerate}
\item Local $E$-coherence conditional pairs:
\begin{enumerate}[align=left]
\item\label{DefLCCPsAndCVPsForEGTRSs_LCCPofER}
 $\LCCP(E,\cR)$ is the subset of \emph{feasible} LCCPs  
$\lccpOf{\alpha,p,\alpha'}\in\GLCCP(\eqoriented{E},R^\RMabbr)$
such that $p\neq\toppos$ (Proposition \ref{PropJoinabilityOfEdownECohPeaks}(\ref{PropJoinabilityOfEdownECohPeaks_Eroot}) justifies this refinement).
\item\label{DefLCCPsAndCVPsForEGTRSs_CVPpstickel}
$\CVPofPStickel(E)$
is the subset of \emph{feasible} CVPs from 
$\SetOf{CVP}^{\rewpstickel{}}(\eqoriented{E})$.
\end{enumerate}
\end{enumerate}
\end{defi}
%
\noindent
As done for $\CCP(\cR)$ (see Remark \ref{RemRefiningCCPR}), 
some refinements on $\LCCP(\cR)$ can be introduced:
\begin{enumerate}
\item 
We take only one of 
$\lccpOf{\alpha,\toppos,\alpha'}$ and $\lccpOf{\alpha',\toppos,\alpha}$ for rules $\alpha$ 
and $\alpha'$.
\item We dismiss from $\LCCP(\cR)$ all \emph{trivial} (proper or improper) 
LCCPs.
As shown in Example \ref{ExImproperLCCPsForETRSs}, though, 
improper LCCPs $\ilccpOf{\alpha}$ \emph{cannot}
be dismissed
unless $\alpha$ is right-ground (Proposition \ref{PropRightGroundRulesHaveTrivialfImproperLCCPs}).  In general, for each non-right-ground rule $\alpha$,  we need to check whether $\ilccpOf{\alpha}$ is trivial.
\end{enumerate}
Table \ref{TableConditionalPairsForEcriticalAndCoherencePeaks} summarizes their relationship with previously considered local $E$-confluence and local $E$-coherence peaks.
\begin{table}[ht]
\caption{Conditional pairs for local $E$-confluence and local $E$-coherence peaks}
\label{TableConditionalPairsForEcriticalAndCoherencePeaks}
\centering
\begin{tabular}{@{}lccl@{~}ll@{}}
\toprule
Property & Peak & Ref.\ & Rule & Rule' & Cond.\ Pairs
\\
\midrule
$\fS{LCON}_E(\rew{\cR^\RMabbr,E},\rew{\cR^\RMabbr})$
& Confluence $E$-critical p. & (\ref{LblECriticalPeak})
& $R^\RMabbr$ & $R^\RMabbr$
& $\LCCP(\cR)$
\\
& Confluence $E$-variable p. & (\ref{LblEVariablePeak})
& $R^\RMabbr$ & $R^\RMabbr$
& $\CVP^{\rpstickel{}}(\cR)$
\\
& Confluence $E$-down p. & (\ref{LblEDownPeak})
& $R^\RMabbr$ & $R^\RMabbr$
& $\DCP(\cR)$
\\
\midrule
$\fS{LCOH}_E(\rew{\cR^\RMabbr,E})$
& Coherence $E$-critical p. & (\ref{LblECriticalPeak})
& $\eqoriented{E}$ & $R^\RMabbr$
& $\LCCP(E,\cR^\RMabbr)$
\\
& Coherence $E$-variable p. & (\ref{LblEVariablePeak})
& $\eqoriented{E}$ & $R^\RMabbr$
& $\CVP^{\rpstickel{}}(E)$
\\
\bottomrule
\end{tabular}
\end{table}
\begin{thm}
\label{TheoEGTRSsAndLocaConfluenceAndCoherenceModulo}
Let $\cR$ 
be an \egtrs{}.
Then, 
\begin{enumerate}
\item\label{TheoEGTRSsAndLocaConfluenceAndCoherenceModulo_LocalConfluence}
 $\rew{\cR,E}$ is locally confluent modulo $E$ with $\rew{\cR}$ iff 
every $\pi\in\LCCP(\cR)\cup\CVPofPStickel(\cR)\cup\DCP(\cR)$ is $\ejoinability{\cR^\RMabbr,E}$-joinable. 
\item\label{TheoEGTRSsAndLocaConfluenceAndCoherenceModulo_LocalCoherence} 
$\rew{\cR,E}$ is locally coherent modulo $E$ iff 
every $\pi\in\LCCP(E,\cR)\cup\CVPofPStickel(E)$ is $\rsejoinability{\cR^\RMabbr,E}$-joinable (or just $\ejoinability{\cR^\RMabbr,E}$-joinable, if $\cR$ is $E$-terminating).  
\end{enumerate}
\end{thm}
\begin{proof}\leavevmode
\begin{enumerate}
\item  
Using the correspondence in Table \ref{TableARSsAndEGTRSs} 
with the diagram in Figure \ref{FigConfluenceAndCoherenceProperties} (left), and dismissing disjoint peaks, 
$\rew{\cR^\RMabbr,E}$ is \emph{locally confluent modulo $E$ with $\rew{\cR^\RMabbr}$} 
if and only if all  
E-critical peaks (\ref{LblECriticalPeak}),
E-variable peaks (\ref{LblEVariablePeak}), and
E-down peaks (\ref{LblEDownPeak}), 
involving rules from $R^\RMabbr$ only,
are $\ejoinability{\cR,E}$-joinable.
By Proposition \ref{PropCriticalPSrUpPeaksAndConditionalCriticalPairs}, E-critical peaks (\ref{LblECriticalPeak}) are $\ejoinability{\cR,E}$-joinable if all
LCCPs in 
$\LCCP(\cR)$ are.
By Proposition \ref{PropEVariablePeaksAsPStickelCVPs}, 
$E$-variable peaks (\ref{LblEVariablePeak}) are $\ejoinability{\cR,E}$-joinable
if all CVPs in $\SetOf{CVP}^{\rewpstickel{}}(R^\RMabbr)$ are.
By Proposition \ref{PropRdownPeaksAndDownConditionalCriticalPairs}
$E$-down peaks  (\ref{LblEDownPeak}) are $\ejoinability{\cR,E}$-joinable if all DCPs in $\SetOf{DCP}(R^\RMabbr)$ are.
\item
The diagram in Figure \ref{FigConfluenceAndCoherenceProperties} (right) shows that, 
$\rew{\cR^\RMabbr,E}$ is locally coherent modulo $E$ 
if and only if, 
for all 
E-critical peaks (\ref{LblECriticalPeak}) and
E-variable peaks (\ref{LblEVariablePeak}) 
where $\alpha\in\:\eqoriented{E}$  and $\alpha'\in R^\RMabbr$ 
are $\rsejoinability{\cR^\RMabbr,E}$-joinable.
By Proposition \ref{PropCriticalPSrUpPeaksAndConditionalCriticalPairs}, $E$-critical peaks (\ref{LblECriticalPeak}) are $\ejoinability{\cR,E}$-joinable if all
LCCPs in $\GLCCP(\eqoriented{E},R^\RMabbr)$ are.
By Proposition \ref{PropEVariablePeaksAsPStickelCVPs}, 
$E$-variable peaks (\ref{LblEVariablePeak}) are $\rsejoinability{\cR,E}$-joinable
if all CVPs in $\SetOf{CVP}^{\:\rewpstickel{}}(E)$ are.
 If $E$-termination is assumed, then, 
 by Proposition \ref{PropEquivalenceLocalCoherenceModuloEandProperty5prime_JK86}, 
$\ejoinability{\cR^\RMabbr,E}$-joinability can be used instead.
\qedhere
\end{enumerate}
\end{proof}

\noindent
As a consequence of Theorem 
\ref{TheoEGTRSsAndLocaConfluenceAndCoherenceModulo},
and Corollaries \ref{CoroJoinableVariablePSrUpPeaksForCriticalVariablesNotInCondition}
and 
\ref{PropJoinabilityOfPStickelCVOsForEquationalRules}, we have the following.

\begin{cor}
\label{CoroEGTRSswithRleftHomogeneousRules}
\sloppy
Let $\cR$ 
be an  \egtrs{}.
\begin{enumerate}
\item\label{CoroEGTRSswithRleftHomogeneousRules_LocalConfluenceModuloE}
 $\rew{\cR^\RMabbr,E}$ is locally confluent modulo $E$
with $\rew{\cR^\RMabbr}$ if all rules in $R^\showRMabbr$ are
$\mu$-left-homogeneous and $\mu$-compatible,
and
every $\pi\in\LCCP(\cR)\cup\DCP(\cR)$ is
$\ejoinability{\cR^\RMabbr,E}$-joinable. 
\item\label{CoroEGTRSswithRleftHomogeneousRules_LocalCoherenceModuloE}
 $\rew{\cR^\RMabbr,E}$ is locally coherent modulo $E$ if 
$\cR$ is $E$-terminating, 
all rules in $\eqoriented{E}$ are $\mu$-left-homogeneous and $\mu$-compatible, 
and
every $\pi\in\LCCP(E,\cR)$ 
is $\ejoinability{\cR^\RMabbr,E}$-joinable.
\end{enumerate}
\end{cor}
\subsection{Proving $E$-confluence}

As done by Jouannaud and Kirchner in the proof of their \cite[Theorem 16]{JouKir_CompletionOfASetOfRulesModuloASetOfEquations_SIAMJC86} for \etrs{s}, if local coherence modulo $E$ of $\rew{\cR,E}$ is guaranteed, 
then we can dismiss $\DCP(\cR)$ in proofs of $E$-confluence of $E$-terminating \egtrs{s}.
In the proof of the following result, as in  Jouannaud and Kirchner's proofs \cite{JouKir_CompletionOfASetOfRulesModuloASetOfEquations_SIAMJC86}, we use \emph{multiset induction} based on a well-founded set $(A,>)$  from which the (possibly repeated) components of the multiset $M$ are obtained. 

\begin{thm}
\label{TheoEConfluenceWithoutDCPs}
Let $\cR$ 
be an $E$-terminating \egtrs{}.
If every $\pi\in\LCCP(\cR)
\cup
\CVPofPStickel(\cR)\cup\LCCP(E,\cR)\cup\CVPofPStickel(E)
$ is $\ejoinability{\cR^\RMabbr,E}$-joinable, 
then $\rew{\cR^\RMabbr}$  has the $\rew{\cR^\RMabbr,E}$-Church-Rosser property modulo $E$. Thus, $\cR$ is $E$-confluent.
\end{thm}

\begin{proof}
By Theorem \ref{TheoEGTRSsAndLocaConfluenceAndCoherenceModulo}(\ref{TheoEGTRSsAndLocaConfluenceAndCoherenceModulo_LocalCoherence}), $\rew{\cR^\RMabbr,E}$ is locally coherent modulo $E$.
According to Theorem \ref{TheoTheorem5_JK86}, 
since $\rew{\cR^\RMabbr,E}$ is locally coherent modulo $E$, 
we prove that $\rew{\cR^\RMabbr,E}$ is locally confluent modulo $E$ with $\rew{\cR^\RMabbr}$.
The
$\ejoinability{\cR^\RMabbr,E}$-joinability of E-critical and E-variable peaks (\ref{LblECriticalPeak}) and (\ref{LblEVariablePeak}) 
follows from $\ejoinability{\cR^\RMabbr,E}$-joinability of conditional pairs in 
$\LCCP(\cR)
\cup
\CVPofPStickel(\cR)$.
Regarding E-down-peaks (\ref{LblEDownPeak}), we write them as follows:
\begin{IEEEeqnarray}{r'C'l}
t =  \sigma(\rhsr')~\leftrewAtPos{\toppos}{\alpha'}\sigma(\lhsr')=w_n\equone{E}
\cdots\equone{E} 
w_2\equone{E}w_1=   w[\sigma(\lhsr)]_p \rewAtPos{p}{\alpha}~w[\sigma(\rhsr)]_p=t'
\label{NestedPSRdownPeakDetail}
\end{IEEEeqnarray}
Note that we have $t\oneCRs{} w_{n}\oneCRs{}\cdots\oneCRs{} t_1\oneCRs{} t'$.
The proof becomes simpler if the righmost $\rewAtPos{p}{\alpha}$-step is weakened into a $\rewAtPos{p}{\alpha,E}$-step so that we can use local coherence modulo $E$ with the peak $w_2\equone{E}w_1\rewAtPos{p}{\alpha,E}$.\footnote{This key suggestion is due to an anonymous referee.} 
If $n=0$, then (\ref{NestedPSRdownPeakDetail}) can be written as a single  $E$-critical (resp.\ $E$-variable) peak 
\begin{eqnarray}
t\leftrewAtPos{\toppos}{\alpha'}\sigma(\lhsr')\rewAtPos{p}{\alpha,E}t'\label{TheoEConfluenceWithoutDCPs_baseInduction}
\end{eqnarray}
 Note that, since rewriting with $\alpha'$ occurs at the root position, (\ref{TheoEConfluenceWithoutDCPs_baseInduction}) is \emph{not} an $E$-down critical peak.
Thus, $\ejoinability{\cR^\RMabbr,E}$-joinability of $t$ and $t'$ follows from $\ejoinability{\cR^\RMabbr,E}$-joinability 
conditional pairs in 
$\LCCP(\cR)
\cup
\CVPofPStickel(\cR)$.

If $n>0$, then we prove 
$\ejoinability{\cR^\RMabbr,E}$-joinability of (\ref{NestedPSRdownPeakDetail})  by multiset induction.
Since $\rew{\cR^\RMabbr,E}$ is locally coherent modulo $E$, from (\ref{NestedPSRdownPeakDetail}), we obtain the following diagram:
\begin{center}
\resizebox{\linewidth}{!}{%
\xymatrix@C=1.4em{
w_n\ar@{->}[dd]_{\toppos}^>>{\alpha'}\ar@{|-|}[rr]_{E}  
\ar@{->}[ddrr]^{}_{\cR^\RMabbr,E}
& & w_{n-1}\ar@{->}[dr]^{}_{\cR^\RMabbr,E}
\ar@{|-|}[r]_{E} 
&  \cdots \ar@{|-|}[r]_{E} 
& w_2 \ar@{->}[ddrr]^{}_{\cR^\RMabbr,E} \ar@{|-|}[rr]_{E} 
& & w_1
\ar@{->}[ddrr]^{p}_>>{\alpha,E}
\\
&&&\cdots \ar@{->}[dr]^{}_{\cR^\RMabbr,E} 
&&&&&&&
\\
t \ar@{->}[ddrrr]^{*}_>>>{\cR^\RMabbr,E} & \ejoinability{\cR^\RMabbr,E} & t_{n-1} &  \cdots & t_2 & \ejoinability{\cR^\RMabbr,E} & t_1 &\ejoinability{\cR^\RMabbr,E} & t'\ar@{->}[ddlll]_{*}^>>>{\cR^\RMabbr,E}
\\
\\
&& & \nfOf{t}\ar@{|-|}[rr]^{*}_{E} && \nfOf{t'}
}}
\end{center}
where, $E$-coherence peaks $w_{i+1}\equone{E} w_{i}\rew{\cR^\RMabbr,E}t_{i-1}$ are $\ejoinability{\cR^\RMabbr,E}$-joinable for all $1\leq i\leq n-1$ by local coherence modulo $E$ (and we assume $t_0=t'$).
The last peak to be considered in the upper row of the diagram above, i.e., 
the leftmost uppermost peak is \emph{not} an $E$-coherence peak, but an  
$E$-confluence peak
\begin{eqnarray*}
t\leftrewAtPos{\toppos}{\alpha'}w_n\rewAtPos{}{\alpha,E}t_{n-1}
\end{eqnarray*}
As happens with (\ref{TheoEConfluenceWithoutDCPs_baseInduction}), it cannot be an $E$-down critical peak and it is $\ejoinability{\cR^\RMabbr,E}$-joinable by 
$\ejoinability{\cR^\RMabbr,E}$-joinability of conditional pairs in 
$\LCCP(\cR)
\cup
\CVPofPStickel(\cR)$.
Since the multiset $\{t,t_{n-1},\ldots,t_1,t'\}$ which  
satisfies $t\oneCRs{} t_{n-1}\oneCRs{}t_{n-2}\oneCRs{} \cdots\oneCRs{} t_1\oneCRs{} t'$ 
is strictly smaller than $\{t,w_n,\ldots,w_1,t'\}$
with respect to the well-founded relation $\rew{\cR/E}$ (due to $E$-termination of $\cR$) and $\rew{\cR}\subseteq\rew{\cR,E}\subseteq\rew{\cR/E}$, by the induction hypothesis, we conclude that $t$ and $t'$ are $\ejoinability{\cR^\RMabbr,E}$-joinable. 
Thus, the $E$-down peak (\ref{NestedPSRdownPeakDetail}) is $\ejoinability{\cR^\RMabbr,E}$-joinable. 
\end{proof}
%
\noindent
As a consequence of Theorem \ref{TheoEConfluenceWithoutDCPs} and 
Corollary \ref{CoroJoinableVariablePSrUpPeaksForCriticalVariablesNotInCondition}, 
we have the following.
\begin{cor}
\label{CoroEConfluenceWithoutDCPs}
Let $\cR$ be an $E$-terminating \egtrs{}
such that $R\:\:\cup\eqoriented{E}$ is left $\mu$-homogeneous and $\mu$-compatible.
If every $\pi\in\LCCP(\cR)\cup\LCCP(E,\cR)$
is $\ejoinability{\cR^\RMabbr,E}$-joinable, 
then $\cR$ is $E$-confluent.
\end{cor}
\begin{exa}
\label{ExHuet80_RemarkPage818_EConfluence}
The \etrs{} $\cR$ in Example \ref{ExHuet80_RemarkPage818}, 
with $E=\{\fS{a}=\fS{b}\}$ and $R=\{\fS{f}(x,x)\to\fS{g}(x)\}$
satisfies the assumptions of Corollary \ref{CoroEConfluenceWithoutDCPs} (using $\mu=\muTop$).
We have that
$\LCCP(\cR)$ consists of the improper LCCP 
\begin{eqnarray}
\iccpOf{(\ref{ExHuet80_RemarkPage818_rule1})} && \langle\fS{g}(x),\fS{g}(x')\rangle\IF\fS{f}(x,x)=\fS{f}(x',x')
\end{eqnarray}
which is $\ejoinability{\cR^\RMabbr}$-joinable as every substitution $\sigma$ satisfying the (equational) conditional part $\fS{f}(x,x)=\fS{f}(x',x')$ must instantiate $x$ and $x'$ to terms $t$ and $t'$ such that $t\equequ{E}t'$.
Then, $\sigma(\fS{g}(x))=\fS{g}(t)\equequ{E}\fS{g}(t')=\sigma(\fS{g}(x'))$, i.e., 
$\iccpOf{(\ref{ExHuet80_RemarkPage818_rule1})}$ is $\ejoinability{\cR^\RMabbr}$-joinable.
Finally, $\LCCP(E,\cR)=\emptyset$ as all LCCPs in $\GLCCP(\eqoriented{E},R^\RMabbr)$ are \emph{root} LCCPs.
$E$-termination of $\cR$ is proved in Example \ref{ExHuet80_RemarkPage818_ETermination}.
By Corollary \ref{CoroEConfluenceWithoutDCPs}, $\cR$ is $E$-confluent.
\end{exa}
\begin{exa}
\label{ExRModuloAndRrelativeE_EConfluence}
The $E$-confluence of $\cR$ in Example \ref{ExRModuloAndRrelativeE}, i.e.,
\begin{IEEEeqnarray*}{r'C'l+x*}
\fS{a} & = & \fS{b}
&
\eqref{ExRModuloAndRrelativeE_eq1}
\\
\fS{a} & \to & \fS{c}
&
\eqref{ExRModuloAndRrelativeE_rule1}
\\
\fS{a}  & \to &  \fS{d}\IF \fS{b} \rews{} \fS{c}
&
\eqref{ExRModuloAndRrelativeE_rule2}
\\
\fS{c} & \to & \fS{d}
&
(\ref{ExRModuloAndRrelativeE_rule3})
\end{IEEEeqnarray*}
cannot be proved by using the results in Section~\ref{SecEConfluenceWithRandCpeaks}.
For instance, the coherence critical pair
\begin{eqnarray*}
\ccpOf{(\ref{ExRModuloAndRrelativeE_rule1}),\toppos,\overrightarrow{(\ref{ExRModuloAndRrelativeE_eq1})}}:
&& \langle\fS{b},\fS{c}\rangle
\end{eqnarray*}
is not $\ejoinability{\cR^\RMabbr}$-joinable as 
$\fS{b}$ is $\cR^\RMabbr$-irreducible and not $E$-equivalent to $\fS{c}$
or any reduct of $\fS{c}$.
However, now we can prove $E$-confluence of $\cR$.
By Proposition \ref{PropRightGroundRulesHaveTrivialfImproperLCCPs}, all improper LCCPs are trivial.
By Propositions \ref{PropInfeasibilityOfLCCPsWithEqCondWithNoEqSymbol}
and \ref{PropInfeasibilityOfELCCPsWithEqCondWithNoEqSymbol},
the LCCPs 
$\lccpOf{(\ref{ExRModuloAndRrelativeE_rule1}),\toppos,(\ref{ExRModuloAndRrelativeE_rule3})}$,
and
$\lccpOf{(\ref{ExRModuloAndRrelativeE_rule2}),\toppos,(\ref{ExRModuloAndRrelativeE_rule3})}$,
are infeasible.
We have a feasible LCCP
\begin{IEEEeqnarray}{r'C'l}
\lccpOf{(\ref{ExRModuloAndRrelativeE_rule1}),\toppos,(\ref{ExRModuloAndRrelativeE_rule2})}:
&& \langle\fS{d},\fS{c}\rangle\IF
\fS{b}\rewmodulos{}\fS{c}
\label{ExRModuloAndRrelativeE_LCCP1}
\end{IEEEeqnarray}
Therefore, $\LCCP(\cR)=\{(\ref{ExRModuloAndRrelativeE_LCCP1})\}$
and
$\LCCP(E,\cR)=\emptyset$ (as $\GLCCP(\eqoriented{E},R^\RMabbr)$ consists of \emph{root} LCCPs only).
Since $\ul{\fS{c}}\rew{(\ref{ExRModuloAndRrelativeE_rule3})}\fS{d}$,
(\ref{ExRModuloAndRrelativeE_LCCP1}) is 
$\ejoinability{\cR^\RMabbr,E}$-joinable.
$E$-termination of $\cR$ is proved with
the following interpretation $\SInterpretation$ with domain $\SemDomain=\{-1,0,1\}$
(computed with \AGES), 
with 
\[\begin{array}{r@{~}c@{~}l@{\hspace{0.65cm}}r@{~}c@{~}l@{\hspace{0.65cm}}r@{~}c@{~}l@{\hspace{0.65cm}}r@{~}c@{~}l}
\fS{a}^\SInterpretation & = & -1
& 
\fS{b}^\SInterpretation & = & -1
&
\fS{c}^\SInterpretation & = & 0
\\
\fS{d}^\SInterpretation & = & 1
&
x\equequ{}^\SInterpretation y & \Leftrightarrow& x=_\SemDomain y
&
x\rew{}^\SInterpretation y & \Leftrightarrow & 9y\geq_\naturals 10x+9
\\
&& &
x(\rewmodulo){\!\!}^\SInterpretation y & \Leftrightarrow & 18y>_\naturals 18x+12
&
x(\rewsmodulo){\!}^\SInterpretation y & \Leftrightarrow & \emph{true}
\end{array}
\]
which  
models 
$\GLtheory_{\cR/E}$,
with $(\rmodulo)^\SInterpretation=\{(0,1),(-1,1),(-1,0)\}$ 
well-founded (on $\SemDomain$).
By Corollary \ref{CoroEConfluenceWithoutDCPs}, $\cR$ is $E$-confluent.
\end{exa}

\subsection{Disproving $E$-confluence}

The following result shows how to disprove $E$-confluence.

\begin{thm}
\label{TheoremNonEConfluenceOfEGRSs}
Let $\cR$ 
be an $\egtrs{}$.
If there is $\pi\in\LCCP(\cR)\cup\CVPofPStickel(\cR)\cup\DCP(\cR)$ which is \emph{not}
$\ejoinability{\cR^\RMabbr/E}$-joinable, then $\cR$ is not $E$-confluent.
\end{thm}

\begin{proof}
Since conditional pairs in $\LCCP(\cR)\cup\CVPofPStickel(\cR)\cup\DCP(\cR)$ represent $\peakOf{\cR^\RMabbr,E}{\cR^\RMabbr}$-peaks, the claim follows by Corollary \ref{CoroDisprovingEConcludenceByNonREJoinableLocalConfluencePeaks}.
\end{proof}
%
\noindent
The following result shows that Theorem \ref{TheoremNonEConfluenceOfEGRSs} subsumes Theorem \ref{TheoEConfluenceOfETerminatingEGTRSsHuet}(\ref{TheoEConfluenceOfETerminatingEGTRSsHuet_NonEConfluent})
regarding proofs of non-$E$-confluence of \egtrs{s}.

\begin{prop}
\label{PropDisprovingEConfluenceAlaHuetOrAlaJK}
If $E$-confluence of an \egtrs{} $\cR$ is disproved using 
Theorem \ref{TheoEConfluenceOfETerminatingEGTRSsHuet}(\ref{TheoEConfluenceOfETerminatingEGTRSsHuet_NonEConfluent}), 
then Theorem \ref{TheoremNonEConfluenceOfEGRSs} also disproves it.
\end{prop}

\begin{proof}
Theorem \ref{TheoEConfluenceOfETerminatingEGTRSsHuet}(\ref{TheoEConfluenceOfETerminatingEGTRSsHuet_NonEConfluent}) disproves $E$-confluence of $\cR$ as the 
non-$\ejoinability{\cR/E}$-joinability of 
some $\pi\in\CCP(\cR)\cup\CVPofR(\cR)$.
We have the following:
\begin{enumerate}
\item If $\ccpOf{\alpha,p,\alpha'}\in\CCP(\cR)$ for variable disjoint rules $\alpha:\lhsr\to\rhsr\IF\gencond$ and $\alpha':\lhsr'\to\rhsr'\IF\gencond'$ 
and $p\in\Pos^\mu_\Symbols(\lhsr)$, then $\lhsr|_p$ and $\lhsr'$ unify with \emph{mgu} $\theta$.
If $\ccpOf{\alpha,p,\alpha'}:\langle \theta(\lhsr[\rhsr']_p),\theta(\rhsr)\rangle\IF\theta(\gencond),\theta(\gencond')$ 
is \emph{not} $\ejoinability{\cR/E}$-joinable, then there is a substitution $\sigma$ satisfying both $\theta(\gencond)$ and $\theta(\gencond')$ such that $\sigma(\theta(\lhsr[\rhsr']_p))$ and $\sigma(\theta(\rhsr))$ are not 
$\ejoinability{\cR/E}$-joinable.
Then, $\lccpOf{\alpha,p,\alpha'}\in\LCCP(\cR)$ is not $\ejoinability{\cR/E}$-joinable either, as $\sigma\circ\theta$ satisfies the condition $\lhsr|_p=\lhsr'$ in the conditional part of $\lccpOf{\alpha,p,\alpha'}$ and also $\gencond$ and $\gencond'$. Since $\sigma(\theta(\lhsr[\rhsr']_p))$ and $\sigma(\theta(\rhsr))$ are not $\ejoinability{\cR/E}$-joinable, $\lccpOf{\alpha,p,\alpha'}$ is not $\ejoinability{\cR/E}$-joinable.
\item If $\cvpROf{\alpha,x,p}\in\CVPofR(\cR)$ for some rule $\alpha:\lhsr\to\rhsr\IF\gencond$, $x\in\Var^\mu(\lhsr)$, and $p\in\Pos^\mu_x(\lhsr)$, i.e., $\langle\lhsr[x']_p,\rhsr\rangle\IF x\rew{}x',\gencond$, 
is not $\ejoinability{\cR/E}$-joinable, then there is a substitution $\sigma$ such that $\sigma(x)\rew{\cR^\RMabbr}\sigma(x')$ and $\sigma(\gencond)$ holds,
but $\sigma(\lhsr[x']_p)$ and $\sigma(\rhsr)$ are not $\ejoinability{\cR/E}$-joinable.
Since $\rew{\cR^\RMabbr}\subseteq\rew{\cR^\RMabbr,E}$, $\sigma$ also satisfies the conditional part of 
$\cvpPStickelOf{\alpha,x,p}\in\CVPofPStickel(\cR)$, i.e., $\langle\lhsr[x']_p,\rhsr\rangle\IF x\rewpstickel{}x',\gencond$. Thus, $\cvpPStickelOf{\alpha,x,p}$ is not $\ejoinability{\cR/E}$-joinable.
\end{enumerate}
Therefore, Theorem \ref{TheoremNonEConfluenceOfEGRSs} also disproves 
$E$-confluence of $\cR$.
\end{proof}
\begin{rem}[Theorem \ref{TheoremNonEConfluenceOfEGRSs} vs.\ Theorem \ref{TheoEConfluenceOfETerminatingEGTRSsHuet}(\ref{TheoEConfluenceOfETerminatingEGTRSsHuet_NonEConfluent})]
Despite Proposition \ref{PropDisprovingEConfluenceAlaHuetOrAlaJK},  \emph{in practice} it is often easier to prove non-$\ejoinability{\cR/E}$-joinability of pairs in $\CCP(\cR)\cup\CVPofR(\cR)$
than non-$\ejoinability{\cR/E}$-joinability of pairs in $\LCCP(\cR)\cup\CVPofPStickel(\cR)$.
\end{rem}
%
\noindent
The following example illustrates the use of DCPs when disproving $E$-confluence.
\begin{exa}
\label{ExPeakNoCPs_notEconfluent}
We know that the \etrs{} $\cR$ in Example \ref{ExPeakNoCPs}, i.e.,
\begin{IEEEeqnarray*}{r'C'l+x*}
\fS{b} & = &\fS{f}(\fS{a}) 
&
\eqref{ExPeakNoCPs_eq1}
\\
\fS{a} & = & \fS{c}
&
\eqref{ExPeakNoCPs_eq2}
\\
\fS{c} & \to & \fS{d}
&
\eqref{ExPeakNoCPs_rule1}
\\
\fS{b} & \to & \fS{d}
&
\eqref{ExPeakNoCPs_rule2}
\end{IEEEeqnarray*}
is \emph{not} $E$-confluent (see Example \ref{ExPeakNoCPs_EdownPeaks}).
Note that $\LCCP(\cR)$ is empty, as the only proper LCCP $\lccpOf{(\ref{ExPeakNoCPs_rule1}),\toppos,(\ref{ExPeakNoCPs_rule2})}$ is trivial
and all rules are right-ground (all improper LCCPs are trivial as well).
Also $\CVPofPStickel(\cR)$ is empty as all rules in $R$ are ground. 
However,
\begin{IEEEeqnarray}{r'C'l}
\dcpOf{(\ref{ExPeakNoCPs_rule2})} && \langle \fS{d},x'\rangle \IF x \ceq \fS{b}, x\innerrew{} x'
\label{ExPeakNoCPs_DCP2}
\end{IEEEeqnarray}
is \emph{not} 
$\ejoinability{\cR/E}$-joinable:
$\sigma=\{x\mapsto\fS{f}(\fS{c}), x'\mapsto\fS{f}(\fS{d})\}$
satisfies its conditional part, 
as 
(i) $\sigma(x)=\fS{f}(\fS{c})\equequ{E}\fS{f}(\fS{a})\equequ{E}\fS{b}$
and 
(ii) $\sigma(x)=\fS{f}(\ul{\fS{c}})\innerrew{(\ref{ExPeakNoCPs_rule1})}\fS{f}(\fS{d})=\sigma(x')$.
However, $\fS{d}$ and $\fS{f}(\fS{d})$ are \emph{not} 
$\cR^\RMabbr/E$-reducible (hence not $\cR^\RMabbr,E$-reducible) nor $E$-equivalent. 
Hence they are not
 $\ejoinability{\cR^\RMabbr/E}$-joinable.
Thus, $\DCP(\cR)$, including $\dcpOf{(\ref{ExPeakNoCPs_rule2})}$, is the only 
set of conditional pairs that can be used to disprove 
$E$-confluence of $\cR$.
\end{exa}

\subsection{Using $R^\RMabbr$ is necessary, even with $E$-confluent \egtrs{s}}

The \egtrs{} in 
Example \ref{ExRModuloAndRrelativeE} was used to show the need of considering $\rewmodulos{}$ instead of $\rews{}$ to interpret reachability conditions in rules of \egtrs{s}. Since Example \ref{ExRModuloAndRrelativeE_EConfluence} shows that $\cR$ is $E$-confluent, it also shows that, in general, we cannot ``simplify'' the treatment of such reachability conditions to use $\rews{}$ back.
The following example shows that this is also true for $\rewspstickel{}$.

\begin{exa}
\label{ExLimitsPSjoinabilityBis}
The following \egtrs{} $\cR$:
\begin{IEEEeqnarray}{r'C'l}
\fS{a} & = & \fS{f}(\fS{b})
\label{ExLimitsPSjoinabilityBis_eq1}
\\
\fS{b} & \to & \fS{d}\IF\fS{a}\rews{}\fS{f}(\fS{e})
\label{ExLimitsPSjoinabilityBis_rule1}
\\
\fS{b} & \to & \fS{e}
\label{ExLimitsPSjoinabilityBis_rule2}
\\
\fS{c} & \to & \fS{a}
\label{ExLimitsPSjoinabilityBis_rule3}
\\
\fS{c} & \to & \fS{f}(\fS{d})
\label{ExLimitsPSjoinabilityBis_rule4}
\\
\fS{e} & \to & \fS{d}
\label{ExLimitsPSjoinabilityBis_rule5}
\end{IEEEeqnarray}
is $E$-confluent if $\rews{}$ in rule (\ref{ExLimitsPSjoinabilityBis_rule1}) 
is interpreted as $\rewsmodulo{}$ to define $\rew{\cR/E}$, i.e., $R^\RMabbr$ is used.
However, 
if $\rewspstickel{}$ (i.e., $R^\PSabbr$) is used, then $E$-confluence is lost, as the 
peak
\begin{IEEEeqnarray}{r'C'l}
\fS{f}(\fS{d})\leftrewAtPos{\toppos}{(\ref{ExLimitsPSjoinabilityBis_rule4})} & \fS{c} & \rewAtPos{\toppos}{(\ref{ExLimitsPSjoinabilityBis_rule3})}\fS{a}
\end{IEEEeqnarray}
is \emph{not} $\ejoinability{\cR^\PSabbr/E}$-joinable. Note that $\fS{f}(\fS{d})$ is $\rew{\cR^\PSabbr/E}$-irreducible (also with $R^\RMabbr$). 
As for $\fS{a}$, we have 
$\fS{a}\equequ{(\ref{ExLimitsPSjoinabilityBis_eq1})}\fS{f}(\fS{b})$, but 
$\fS{b}\not\rew{(\ref{ExLimitsPSjoinabilityBis_rule1})}\fS{d}$ because 
$\fS{a}$ in the conditional part of (\ref{ExLimitsPSjoinabilityBis_rule1}) 
is $\rewpstickel{\cR^\PSabbr}$-irreducible.
\end{exa}

\section{Application to Equational Term Rewriting Systems}
\label{SecApplicationToETRSs}

An \etrs{} $\cR=(\Symbols,E,R)$ can be seen as an \egtrs{} $\cR=(\Symbols,\{\rew,\rews{}\},\muTop,\emptyset,E,R)$.
Thus, we expect that existing results on $E$-confluence of \etrs{s} can be obtained from our results for \egtrs{s}.
Since equations and rules in \etrs{s} are unconditional, all rules in $\eqoriented{E}$ and $R$ are
trivially left $\muTop$-homogeneous and $\muTop$-compatible (for the sake of readability, we drop $\muTop$ in the following); 
and, being $R$ a set of unconditional rules, we have $R=R^\PSabbr=R^\RMabbr$. 
Since unconditional critical pairs coincide with our CCPs when unconditional rules are used, we write $\CP$ instead of $\CCP$ to refer to the corresponding sets.
Thus, from 
Theorem \ref{TheoEConfluenceOfETerminatingEGTRSsHuet} (using Corollaries \ref{CoroJoinableVariablePairsForCriticalVariablesNotInCond}, 
\ref{CoroJoinableVariablePairsII}, and \ref{CoroJoinableVariablePairsIII}), we obtain. 

\begin{cor}
\label{CoroEConfluenceOfETRSs_Huet}
Let $\cR$ be an \etrs.
\begin{enumerate}
\item\label{CoroEConfluenceOfETRSs_Huet_ETRSEConfluence}
 If $\cR$ is $E$-terminating and every
$\pi\in\CP(\cR)\cup\CP(\cR,E)\cup\CP(E,\cR)\cup\CVPofEqOne(\cR)
$
is $\ejoinability{\cR}$-joinable, then $\cR$ is $E$-confluent.
\item\label{CoroEConfluenceOfETRSs_Huet_EConfluence}
 If $\cR$ is  left-linear and $E$-terminating and every
$\pi\in\CP(\cR) \cup\CP(\cR,E)\cup\CP(E,\cR)
$
is $\ejoinability{\cR}$-joinable, then $\cR$ is $E$-confluent.
\item\label{CoroEConfluenceOfETRSs_Huet_NonEConfluence} 
If there is a \emph{non}-$\ejoinability{\cR/E}$-joinable $\pi\in\CP(\cR)$, then $\cR$ is not $E$-confluent.
\end{enumerate}
\end{cor}
%
\noindent
Similarly, from Corollary \ref{CoroEConfluenceWithoutDCPs}
and Theorem \ref{TheoremNonEConfluenceOfEGRSs}
(using Corollary \ref{CoroJoinableVariablePSrUpPeaksForCriticalVariablesNotInCondition}),
we obtain:

\begin{cor}
\label{CoroEConfluenceOfETRSs_JK}
Let $\cR$ be an \etrs{}.
\begin{enumerate}
\item\label{CoroEConfluenceOfETRSs_JK_EConfluence} If $\cR$ is $E$-terminating and every $\pi\in\LCCP(\cR)\cup\LCCP(E,\cR)$
is $\ejoinability{\cR,E}$-joinable, 
then $\cR$ is $E$-confluent.
\item\label{CoroEConfluenceOfETRSs_JK_NonEConfluence} If there is a \emph{non}-$\ejoinability{\cR/E}$-joinable $\pi\in\LCCP(\cR)\cup\DCP(\cR)$, then $\cR$ is not $E$-confluent.
\end{enumerate}
\end{cor}
%
\noindent
The next two sections compare these results 
with \cite{Huet_ConfluentReductionsAbstractPropertiesAndApplicationsToTermRewritingSystems_JACM80} 
and 
\cite{JouKir_CompletionOfASetOfRulesModuloASetOfEquations_SIAMJC86}.

\subsection{Huet \cite{Huet_ConfluentReductionsAbstractPropertiesAndApplicationsToTermRewritingSystems_JACM80}}
\label{SecHuet}

In his results about $E$-confluence of \emph{Equational First-Order Theories} $\cR=(\Symbols,E,R)$ where $R$ is a set of rewrite rules
and $E$ is a set of equations $s=t$ such that $\Var(s)=\Var(t)$,
 Huet uses critical pairs previously introduced  in the definition of pages 809 (for $\cR$) and 817 (for $\cR$ with $E$)
and which coincide with 
\[\CP(\cR)\cup\CP(E,\cR)\cup\CP(\cR,E)\]
Then, \cite[Section 3.4]{Huet_ConfluentReductionsAbstractPropertiesAndApplicationsToTermRewritingSystems_JACM80} develops a number of results.

\subsubsection{Property $\alpha$, see Figure \ref{FigHuetConfluencePropertiesAlphaBetaGamma} (middle)}
\label{SecHuetPropertyAlpha}
The $\ejoinability{\cR}$-joinability of pairs in $\CP(\cR)$ proves \emph{property $\alpha$} \cite[Lemma 3.4]{Huet_ConfluentReductionsAbstractPropertiesAndApplicationsToTermRewritingSystems_JACM80}.
When applied to \etrs{s}, our Proposition \ref{PropEGTRSsAndPropertiesAlphaAndGamma}(\ref{PropEGTRSsAndPropertiesAlphaAndGamma_Alpha}) is equivalent to this result, as all CVPs in $\CVPofR(\cR)$ are $\ejoinability{\cR}$-joinable due to Corollary \ref{CoroJoinableVariablePairsForCriticalVariablesNotInCond}.
\begin{exa}
\label{ExHuet80_RemarkPage818variant}
Consider the following variant $\cR$ of the \etrs{} in Example \ref{ExHuet80_RemarkPage818}:
\begin{IEEEeqnarray}{r'C'l}
\fS{a} & = & \fS{b}\label{ExHuet80_RemarkPage818variant_eq1}
\\
\fS{f}(x,x) & \to & \fS{g}(x)\label{ExHuet80_RemarkPage818variant_rule1}
\\
\fS{f}(x,x) & \to & \fS{h}(x)\label{ExHuet80_RemarkPage818variant_rule1b}
\\
\fS{a} & \to & \fS{c}\label{ExHuet80_RemarkPage818variant_rule2}
\\
\fS{b} & \to & \fS{c}\label{ExHuet80_RemarkPage818variant_rule3}
\\
\fS{g}(x) & \to & x\label{ExHuet80_RemarkPage818variant_rule4}
\\
\fS{h}(x) & \to & x\label{ExHuet80_RemarkPage818variant_rule5}
\end{IEEEeqnarray}
We have
$\CCP(\cR)=\{\ccpOf{(\ref{ExHuet80_RemarkPage818variant_rule1}),\toppos,(\ref{ExHuet80_RemarkPage818variant_rule1b})}\}$,
for
$\ccpOf{(\ref{ExHuet80_RemarkPage818variant_rule1}),\toppos,(\ref{ExHuet80_RemarkPage818variant_rule1b})}: 
\langle\fS{h}(x),\fS{g}(x)\rangle$, 
which is $\joinability{\cR}$-joinable as 
$\fS{h}(x)\rew{(\ref{ExHuet80_RemarkPage818variant_rule5})}x\leftrew{(\ref{ExHuet80_RemarkPage818variant_rule4})}\fS{g}(x)$.
By Corollary \ref{CoroJoinableVariablePairsForCriticalVariablesNotInCond},
all CVPs in $\CVPofR(\cR)$ are $\ejoinability{\cR}$-joinable.
Thus, by Proposition \ref{PropEGTRSsAndPropertiesAlphaAndGamma}, 
property $\alpha$ holds for $\cR$.
\end{exa}

\subsubsection{Property $\gamma$, see Figure \ref{FigHuetConfluencePropertiesAlphaBetaGamma} (right)}
\label{SecHuetPropertyGamma}

If $R$ is a set of \emph{left-linear} 
rules, then $\ejoinability{\cR}$-joinability of pairs in $\CP(E,\cR)\cup\:\CP(\cR,E)$ guarantees Property $\gamma$ \cite[Lemma 3.5]{Huet_ConfluentReductionsAbstractPropertiesAndApplicationsToTermRewritingSystems_JACM80}.
Our Proposition \ref{PropEGTRSsAndPropertiesAlphaAndGamma}.(\ref{PropEGTRSsAndPropertiesAlphaAndGamma_Gamma}) subsumes  this result, as 
all CVPs in $\CVPofR(E)$ are 
$\ejoinability{\cR}$-joinable due to 
Corollary \ref{CoroJoinableVariablePairsII}
 and
all CVPs in $\CVPofEqOne(R)$ are 
$\ejoinability{\cR}$-joinable due to 
Corollary \ref{CoroJoinableVariablePairsIII}.
However, Huet's Lemma 3.5 does \emph{not} apply to  $\cR$ and $E$ in Example \ref{ExHuet80_RemarkPage818variant}, with non-left-linear rules.

\begin{exa}
\label{ExHuet80_RemarkPage818variant_EqCCPsAndEqCVPs}
For $\cR$ and $E$ in Example \ref{ExHuet80_RemarkPage818variant}, 
$\CCP(\cR,E)=\{(\ref{ExHuet80_RemarkPage818variant_CCP_RE_1}),
(\ref{ExHuet80_RemarkPage818variant_CCP_RE_2})\}$,
$\CVPofR(E)=
\CCP(E,\cR)=\emptyset$ 
and $\CVPofEqOne(\cR)=\{
(\ref{ExHuet80_RemarkPage818variant_ECVP1}),(\ref{ExHuet80_RemarkPage818variant_ECVP2}),(\ref{ExHuet80_RemarkPage818variant_ECVP3}),(\ref{ExHuet80_RemarkPage818variant_ECVP4}),(\ref{ExHuet80_RemarkPage818variant_ECVP5}),(\ref{ExHuet80_RemarkPage818variant_ECVP6})
\}$ with
\begin{IEEEeqnarray}{r'C'l}
\pi_{(\ref{ExHuet80_RemarkPage818variant_rule2}),\toppos,\overrightarrow{(\ref{ExHuet80_RemarkPage818variant_eq1})}}: && \langle\fS{b},\fS{c}\rangle\label{ExHuet80_RemarkPage818variant_CCP_RE_1}
\\
\pi_{(\ref{ExHuet80_RemarkPage818variant_rule3}),\toppos,\overleftarrow{(\ref{ExHuet80_RemarkPage818variant_eq1})}}: 
&&\langle\fS{a},\fS{c}\rangle\label{ExHuet80_RemarkPage818variant_CCP_RE_2}
\\
\cvpEqOneOf{(\ref{ExHuet80_RemarkPage818variant_rule1}),x,1}:
&& \langle\fS{f}(x',x),\fS{g}(x)\rangle \IF  x\equone{}x'\label{ExHuet80_RemarkPage818variant_ECVP1}
\\
\cvpEqOneOf{(\ref{ExHuet80_RemarkPage818variant_rule1}),x,2}:
&& \langle\fS{f}(x,x'),\fS{g}(x)\rangle \IF  x\equone{}x'\label{ExHuet80_RemarkPage818variant_ECVP2}
\\
\cvpEqOneOf{(\ref{ExHuet80_RemarkPage818variant_rule1b}),x,1}:
&& \langle\fS{f}(x',x),\fS{h}(x)\rangle \IF  x\equone{}x'\label{ExHuet80_RemarkPage818variant_ECVP3}
\\
\cvpEqOneOf{(\ref{ExHuet80_RemarkPage818variant_rule1b}),x,2}:
&& \langle\fS{f}(x,x'),\fS{h}(x)\rangle \IF  x\equone{}x'\label{ExHuet80_RemarkPage818variant_ECVP4}
\\
\cvpEqOneOf{(\ref{ExHuet80_RemarkPage818variant_rule4}),x,1}:
&& \langle\fS{g}(x'),x\rangle \IF  x\equone{}x'\label{ExHuet80_RemarkPage818variant_ECVP5}
\\
\cvpEqOneOf{(\ref{ExHuet80_RemarkPage818variant_rule5}),x,1}:
&& \langle\fS{h}(x'),x\rangle \IF  x\equone{}x'\label{ExHuet80_RemarkPage818variant_ECVP6}
\end{IEEEeqnarray}
Note that (\ref{ExHuet80_RemarkPage818variant_CCP_RE_1}) 
and (\ref{ExHuet80_RemarkPage818variant_CCP_RE_2}), from $\CCP(\cR,E)$, 
are 
$\joinability{\cR^\RMabbr}$-joinable.
The CVPs (\ref{ExHuet80_RemarkPage818variant_ECVP1})--
(\ref{ExHuet80_RemarkPage818variant_ECVP4}), from $\CVPofEqOne(E)$, are
all $\joinability{\cR}$-joinable:
if $\sigma$ satisfies their common conditional part, 
$x\equone{}x'$ then $\sigma=\{ x\mapsto C[\fS{a}], x'\mapsto C[\fS{b}]\}$
or $\sigma=\{ x\mapsto C[\fS{b}], x'\mapsto C[\fS{a}]\}$
for some context $C[\:]$.
Then (in particular for (\ref{ExHuet80_RemarkPage818variant_ECVP1}); for the others it is similar), 
\[\begin{array}{rcl}
\sigma(\fS{f}(x',x)) & = & \fS{f}(C[\fS{b}],C[\fS{a}])\rewp{}\fS{f}(C[\fS{c}],C[\fS{c}])\rew{}\fS{g}(C[\fS{c}])
\\ 
\sigma(\fS{g}(x)) & = & \fS{g}(C[\fS{a}])\rew{}\fS{g}(C[\fS{c}])
\end{array}
\]
and similarly for the other $\sigma$. 
By Corollary \ref{CoroJoinableVariablePairsIII}, 
(\ref{ExHuet80_RemarkPage818variant_ECVP5}) and
(\ref{ExHuet80_RemarkPage818variant_ECVP6}) are
$\ejoinability{\cR}$-joinable.
\end{exa}

\subsubsection{$E$-confluence}

In \cite{Huet_ConfluentReductionsAbstractPropertiesAndApplicationsToTermRewritingSystems_JACM80}, Huet investigates how to prove \emph{confluence of $\rew{\cR}$ modulo $E$} (see Definition \ref{DefAbstractConfluenceModuloHuet80}) for \etrs{s} $\cR$.
His main result is as follows: 

\begin{thmC}[\emph{\cite[Theorem 3.3]{Huet_ConfluentReductionsAbstractPropertiesAndApplicationsToTermRewritingSystems_JACM80}}]
\label{TheoTheorem3_3_Huet80}
Let $\cR=(\Symbols,E,R)$ be an \etrs{} such that
\begin{enumerate}
\item $R$ is a set of left-linear rules.
\item $E$ is a set of equations $s=t$ such that $\Var(s)=\Var(t)$.
\item $\rew{\cR}\circ\equequ{E}$ is terminating.
\end{enumerate}
Then, \emph{$\rew{\cR}$ is confluent modulo $E$}  
if and only if for all its critical pairs $\langle s,t\rangle$, we have
$\nfOf{s}\equequ{E}\nfOf{t}$, where  $\nfOf{s}$ and $\nfOf{t}$ 
are the $E$-normal forms of $s$ and $t$, respectively.
\end{thmC}
%
\noindent
By ``all its critical pairs'', Huet means those in 
$\CP(\cR)\cup\CP(E,\cR)\cup\CP(\cR,E)$.
Thus, 
Corollary \ref{CoroEConfluenceOfETRSs_Huet}(\ref{CoroEConfluenceOfETRSs_Huet_EConfluence})
provides all $E$-confluence proofs which can be obtained with Theorem \ref{TheoTheorem3_3_Huet80}.
However, 
Corollary \ref{CoroEConfluenceOfETRSs_Huet} can often be used 
when Theorem \ref{TheoTheorem3_3_Huet80} fails.

\begin{exa}
\label{ExHuet80_RemarkPage818variant_EConfluence}
All conditional pairs in $\CCP(\cR)$ for $\cR$ in Example \ref{ExHuet80_RemarkPage818variant} are 
$\ejoinability{\cR^\RMabbr}$-joinable.
All conditional pairs in $\CCP(\cR,E)$, $\CCP(E,\cR)$, and
$\CVPofEqOne(E)$ are 
$\ejoinability{\cR^\RMabbr}$-joinable (Example \ref{ExHuet80_RemarkPage818variant_EqCCPsAndEqCVPs}).
$E$-termination of $\cR$ is proved as done 
in Example \ref{ExHuet80_RemarkPage818_ETermination} for $\cR$ 
in Example \ref{ExHuet80_RemarkPage818}. 
By Corollary \ref{CoroEConfluenceOfETRSs_Huet}(\ref{CoroEConfluenceOfETRSs_Huet_ETRSEConfluence}), $\cR$ is $E$-confluent.
Theorem \ref{TheoTheorem3_3_Huet80} does not apply as $\cR$ is not left-linear.
\end{exa}
%
\noindent
Since confluence modulo $E$ is just a sufficient condition for $E$-confluence, disproving confluence modulo $E$ does \emph{not} imply non-$E$-confluence.

\subsection{Jouannaud and Kirchner \cite{JouKir_CompletionOfASetOfRulesModuloASetOfEquations_SIAMJC86}}
\label{SecJouannaudAndKirchner}

The main result for ETRSs 
 is \cite[Theorem 16]{JouKir_CompletionOfASetOfRulesModuloASetOfEquations_SIAMJC86} which makes use of the following sets of pairs, for $E$ a set of equations and $S$, $S'$ sets of rules \cite[page 1167, first paragraph]{JouKir_CompletionOfASetOfRulesModuloASetOfEquations_SIAMJC86} (we use our notation to ease readability):\footnote{Given sets $S$ and $S'$ of unconditional rules, Jouannaud and Kirchner write $\SetOf{CP}(S,S')$ to denote the set of critical pairs for rules $\alpha:\lhsr\to\rhsr\in S$ \emph{on} all rules $\alpha':\lhsr'\to\rhsr'\in S'$, i.e., 
$\alpha$ overlaps $\alpha'$ at position $p\in\Pos_\Symbols(\lhsr')$ to produce a critical pair $\langle\theta(\lhsr')[\rhsr]_p,\theta(\rhsr')\rangle$, see the last paragraph of \cite[page 1166]{JouKir_CompletionOfASetOfRulesModuloASetOfEquations_SIAMJC86}. In our notation, such a pair is denoted $\ccpOf{\alpha',p,\alpha}$.
Therefore, their sets $\SetOf{CP}(S,S')$ correspond to $\CP(S',S)$ in the notation of this paper. Here, we consistently use our notation and just swap the sets of rules as originally written in \cite{JouKir_CompletionOfASetOfRulesModuloASetOfEquations_SIAMJC86}.}
 \begin{itemize}
 \item $\CP(S,S')$ is a set of nontrivial critical pairs $\ccpOf{\alpha,p,\alpha'}$ with $\alpha\in S$ and $\alpha'\in S'$
 \item $\CP(S,E)$ is a set of nontrivial critical pairs $\ccpOf{\alpha,p,\alpha'}$ with $\alpha\in S$ and $\alpha'\in\eqoriented{E}$. 
\item $\CP(E,S)$ is a set of nontrivial critical pairs  $\ccpOf{\alpha,p,\alpha'}$ with $\alpha\in\eqoriented{E}$, $p\neq\toppos$, and $\alpha'\in S$. 
 \item $\ECP(S,S')$ is a set of nontrivial $E$-critical pairs $\ecpOf{\alpha,p,\alpha',\theta}$ with $\alpha\in S$ and $\alpha'\in S'$.
 \item $\ECP(E,S)$ is a set of nontrivial $E$-critical pairs $\ecpOf{\alpha,p,\alpha',\theta}$ with $\alpha\in\eqoriented{E}$, $p\neq\toppos$, and $\alpha'\in S$.
 \end{itemize}
 \subsubsection{Proving $E$-confluence}
 \label{SecJouannaudAndKirchner_Econfluence}
 
 In  \cite[Theorem 16]{JouKir_CompletionOfASetOfRulesModuloASetOfEquations_SIAMJC86}, 
a characterization of the $R^E$-Church-Rosser modulo $E$ property is obtained for $E$-terminating sets of rules $R=L\cup N$ such that all rules in $L$ are left-linear and $R^E$ is defined as $\rew{L}\cup\rew{N,E}$.
 \begin{thmC}
 [\label{Theorem16_JK86}
\emph{\cite[Theorem 16]{JouKir_CompletionOfASetOfRulesModuloASetOfEquations_SIAMJC86}}]
Assume $E$ to be a set of equations such that a complete and finite unification algorithm exists and $E$-congruence classes are finite. 
Let $R=L\cup N$ be an $E$-terminating set of rules such that all rules in $L$ are left-linear. 
Let $R^E$ be defined as $\rew{L}\cup\rew{N,E}$. 
Then, $R$ is $R^E$-Church-Rosser modulo $E$ iff
\begin{enumerate}
\item\label{Theorem16_JK86_CriticalPairs}
 every confluence pair in
\begin{IEEEeqnarray*}{r'C'l}
\CP(L,L)\cup\CP(N,L)\cup\ECP(N,N)\cup\ECP(L,N)
\end{IEEEeqnarray*}
is $\ejoinability{R/E}$-joinable.\footnote{Jouannaud and Kirchner use ``$R/E$-confluent modulo $E$''. According to \cite[Definition 2]{JouKir_CompletionOfASetOfRulesModuloASetOfEquations_SIAMJC86}, this is equivalent to $\ejoinability{R/E}$-joinability.
}
\item\label{Theorem16_JK86_CoherencePairs}
 for all coherence pairs
\begin{IEEEeqnarray*}{r'C'l}
\langle s,t\rangle\in\CP(E,L)\cup\ECP(E,N)
\end{IEEEeqnarray*}
there is $t'$ such that $t\rew{R^E}t'$ and $t$ and $t'$ are $\ejoinability{R/E}$-joinable.
\item\label{Theorem16_JK86_CoherencePairsBis}
 for all coherence pairs
\begin{IEEEeqnarray*}{r'C'l}
\langle s,t\rangle\in\CP(L,E)
\end{IEEEeqnarray*}
there is $s'$ such that $s\rew{R^E}s'$ and $s$ and $s'$ are $\ejoinability{R/E}$-joinable.
\end{enumerate}
\end{thmC}
%
%
\noindent
Since the $\genrelationUpE$-Church-Rosser property modulo $\genequivalence$ of 
 $\genrelation$ 
 implies $\genequivalence$-confluence of $\genrelation$ (Proposition \ref{PropSufficientConditionForEConfluence_JK86}), Theorem \ref{Theorem16_JK86} provides a sufficient condition for $E$-confluence of \etrs{s}.
 
If we let $L=\emptyset$ and $N=R$ in Theorem \ref{Theorem16_JK86}, then
$\CP(L,L)=\CP(N,L)=\ECP(L,N)=\CP(E,L)=\CP(L,E)=\emptyset$, i.e., only 
$\ECP(N,N)$ and $\ECP(E,N)$ must be considered in Theorem \ref{Theorem16_JK86}.
Then, \cite[Theorem 16]{JouKir_CompletionOfASetOfRulesModuloASetOfEquations_SIAMJC86}
can be used to prove $E$-confluence essentially as our 
Corollary \ref{CoroEConfluenceOfETRSs_JK}(\ref{CoroEConfluenceOfETRSs_JK_EConfluence}) for \etrs{s}, i.e.,
as the $\ejoinability{\cR,E}$-joinability of 
every $\pi\in\LCCP(\cR)\:\cup\:\LCCP(E,\cR)$.

 \subsubsection{Disproving $E$-confluence?}

Unfortunately, the attempt to use  the ``\emph{iff}'' formulation of Theorem \ref{Theorem16_JK86} to disprove $E$-confluence 
by showing that an \emph{arbitrarily chosen} critical pair in 
\begin{eqnarray*}
\CP(L,L)\cup\CP(N,L)\cup\ECP(N,N)\cup\ECP(L,N)\cup\CP(E,L)\cup\ECP(E,N)\cup\CP(L,E)\label{LblCPsForCRpropertyOfETRSs}
\end{eqnarray*}
\emph{fails} to fulfill its corresponding joinability condition
may lead to wrong conclusions.

\begin{exa}
\label{ExLimitsPSjoinability_CPs}

Consider again the $E$-confluent \etrs{} $\cR$ in Example \ref{ExLimitsPSjoinability}, i.e.,
\begin{IEEEeqnarray*}{r'C'l+x*}
\fS{a} & = & \fS{f}(\fS{b}) 
&
\eqref{ExLimitsPSjoinability_eq1}
\\
\fS{b} & \to & \fS{d}
&
\eqref{ExLimitsPSjoinability_rule1}
\\
\fS{c} & \to & \fS{a}
&
\eqref{ExLimitsPSjoinability_rule2}
\\
\fS{c} & \to & \fS{f}(\fS{d})
&
\eqref{ExLimitsPSjoinability_rule3}
\end{IEEEeqnarray*}
With $L=\emptyset$ and $N=R$, we have $R^E=\rew{R,E}$.
The coherence $E$-critical pair $\ccpOf{\overleftarrow{(\ref{ExLimitsPSjoinability_eq1})},1,(\ref{ExLimitsPSjoinability_rule1}),\epsilon}$: 
\begin{IEEEeqnarray}{r'C'l}
 \langle\fS{f}(\fS{d}),\fS{a}\rangle\label{ExLimitsPSjoinability_CP2}
\end{IEEEeqnarray}
belongs to $ECP(E,N)$. In Theorem \ref{Theorem16_JK86}, the joinability condition for 
(\ref{ExLimitsPSjoinability_CP2}) requires a first $\rew{R,E}$-step on 
$\fS{a}$,
which is impossible, as $\fS{a}$ is $\rew{R,E}$-irreducible.
By Theorem \ref{Theorem16_JK86}, $\cR$ is \emph{not} $\rew{\cR,E}$-Church-Rosser modulo $E$
(see also Proposition \ref{PropLocalConfluenceAndLocalCoherenceNotNecessaryForEConfluence}).
Misinterpreting these facts as disproving $E$-confluence of $\cR$ would contradict Proposition \ref{PropExLimitsPSjoinability_EConfluent}.
\end{exa}
%
\noindent
Indeed, Theorem \ref{Theorem16_JK86}, can be used to disprove the $\rew{\cR,E}$-Church-Rosser  modulo $E$ property, but not to disprove $E$-confluence of $\cR$.

\section{Application to Conditional Term Rewriting Systems}
\label{SecApplicationToConditionalETRSs}

A \gtrs{} $\cR=(\Symbols,\SPredicates,\mu,H,R)$ can be seen as an \egtrs{} $\cR=(\Symbols,\SPredicates,\mu,\emptyset,H,R)$ where no equality predicate is used
and the set of equations $E$ is empty.
Thus, we expect that our results for $E$-confluence (in particular $\emptyset$-confluence) 
of \egtrs{s} \emph{collapse} into the existing results on (local) confluence of \gtrs{s} 
\cite{Lucas_LocalConfluenceOfConditionalAndGeneralizedTermRewritingSystems_JLAMP24}.

On the other hand, existing research about confluence and 
$E$-confluence of 
\emph{membership equational specifications} and
\emph{rewrite theories}, also extending \ctrs{s} with \emph{membership} and
\emph{equational} components, has been conducted in the realm of 
Membership Equational Logic \cite{Meseguer_MembershipAlgebraAsALogicalFrameworkForEquationalSpecification_WADT97}
and 
Rewriting Logic \cite{Meseguer_ConditionalRewritingLogicAsAUnifiedModelOfConcurrency_TCS92,Meseguer_TwentyYearsOfRewritingLogic_JLAP12} and, in particular, of Generalized Rewrite Theories \cite{BruMes_GeneralizedRewriteTheories_ICALP03,BruMes_SemanticFoundationsForGeneralizedRewriteTheories_TCS06}.

The following sections discuss the application of our results in these settings.

\subsection{Confluence of \gtrs{s} as $\emptyset$-confluence of \egtrs{s}}
\label{SecConfluenceOfGTRSsAsConfluenceOfEGTRSsWithEempty}

As discussed in Section~\ref{SecEGTRSsWithoutEquationsAsGTRSs},
if the set of equations $E$ is empty, then a number of consequences follow:
\begin{itemize}
\item The relation $\equone{E}$ is empty and its reflexive and transitive closure
$\equequ{E}$ becomes the identity of terms.
\item $\rew{\cR,E}$ and $\rew{\cR/E}$ collapse into $\rew{\cR}$.
\item We can assume $R^\PSabbr=R^\RMabbr=R$ and similarly for $H$.
\item $\joinability{\cR}$ is equivalent to $\ejoinability{\cR}$, $\ejoinability{\cR,E}$ and $\ejoinability{\cR/E}$.
\end{itemize}
We use these facts to show how the results in Sections \ref{SecEConfluenceWithRandCpeaks} and \ref{SecEConfluenceWithPSRandCpeaks} boil down into the characterizations of (local) confluence of \gtrs{s}.

\subsubsection{Regarding the results in Section~\ref{SecEConfluenceWithRandCpeaks}}

\begin{enumerate}
\item $\CCP(\cR)$ and $\CVPofR(\cR)$  
contain exactly the CCPs and CVPs in \cite[Definition 59]{Lucas_LocalConfluenceOfConditionalAndGeneralizedTermRewritingSystems_JLAMP24}.
\item Property $\alpha$ of $\rew{\cR}$
boils down into the usual definition of \emph{local confluence} of $\rew{\cR}$. 
Thus, Proposition \ref{PropEGTRSsAndPropertiesAlphaAndGamma}(\ref{PropEGTRSsAndPropertiesAlphaAndGamma_Alpha}) provides the characterization of local confluence of $\rew{\cR}$ given by \cite[Theorem 62]{Lucas_LocalConfluenceOfConditionalAndGeneralizedTermRewritingSystems_JLAMP24}.
\item\label{SecConfluenceOfGTRSsAsConfluenceOfEGTRSsWithEempty_Huet_CoherencePairsEmpty} Since $E$ is empty, $\CCP(E,\cR)=\CCP(\cR,E)=\CVPofR(E)=\CVPofEqOne(\cR)=\emptyset$. Regarding $\CVPofEqOne(\cR)$, all its CVPs 
$\cvpEqOneOf{\alpha,x,p}:\langle\lhsr[x']_p,\rhsr\rangle\IF x\equone{}x',\gencond$ for rules $\alpha:\lhsr\to\rhsr\IF\gencond\in R$ are infeasible as $\equone{E}$, i.e., $\rew{\eqoriented{E}}$, is empty.
\item Since (i) $\rew{\cR}$ and $\rew{\cR/E}$ coincide, and also 
(ii) $\ejoinability{\cR}$, $\ejoinability{\cR/E}$ boil down into $\joinability{\cR}$, 
and due to item (\ref{SecConfluenceOfGTRSsAsConfluenceOfEGTRSsWithEempty_Huet_CoherencePairsEmpty}) above, 
Theorem \ref{TheoEConfluenceOfETerminatingEGTRSsHuet} can be formulated as a \emph{characterization} of confluence of $\rew{\cR}$.
\item Since $\emptyset$-termination of an \egtrs{} is just termination of $\cR$ viewed as a \gtrs, 
Theorem \ref{TheoEConfluenceOfETerminatingEGTRSsHuet} therefore coincides with \cite[Theorem 68]{Lucas_LocalConfluenceOfConditionalAndGeneralizedTermRewritingSystems_JLAMP24}, characterizing confluence of terminating \gtrs{s}. 
\end{enumerate} 

\subsubsection{Regarding the results in Section~\ref{SecEConfluenceWithPSRandCpeaks}}

As a consequence of Proposition \ref{PropEDownPeaksInGTRSs}, 
we have the following.

\begin{cor}
\label{CoroDCPsInGTRSs}
Let $\cR$ be an \egtrs{} with $E=\emptyset$.
If all pairs in $\LCCP(\cR)\cup\CVPofPStickel(\cR)$ are $\ejoinability{\cR,E}$-joinable (equivalently $\joinability{\cR}$-joinable), then every pair in $\DCP(\cR)$ is $\joinability{\cR}$-joinable.
\end{cor}

\begin{enumerate}
\item Each \emph{feasible} LCCP $\lccpOf{\alpha,p,\alpha'}:\langle\lhsr[\rhsr']_p,\rhsr\rangle\IF \lhsr|_p=\lhsr',\gencond,\gencond'$ is ``equivalent'' to the CCP
$\ccpOf{\alpha,p,\alpha'}:\langle\theta(\lhsr[\rhsr']_p),\theta(\rhsr)\rangle\IF\theta(\gencond),\theta(\gencond')$ as every substitution $\sigma$ satisfying $\lhsr|_p=\lhsr'$ (which now represents a \emph{syntactic} unification problem, often written $\lhsr|_p=^?\lhsr'$) 
is a unifier of both terms and hence $\sigma=\tau\circ\theta$ for some substitution $\tau$ and mgu $\theta$ of $\lhsr|_p$ and $\lhsr'$.
Therefore, we can say that $\LCCP(\cR)$ and $\CCP(\cR)$ coincide.
\item Since $\rew{\cR}$ and $\rew{\cR,E}$ coincide, and also 
$\ejoinability{\cR,E}$ and $\joinability{\cR}$ coincide, each $\cvpPStickelOf{\alpha,x,p}\in\CVPofPStickel(\cR)$ is $\joinability{\cR}$-joinable if and only if $\cvpOf{}{\alpha,x,p}$ in \cite[Definition 59]{Lucas_LocalConfluenceOfConditionalAndGeneralizedTermRewritingSystems_JLAMP24} (structurally identical to $\cvpPStickelOf{\alpha,x,p}$) is.
\item By Corollary \ref{CoroDCPsInGTRSs}, $\joinability{\cR}$-joinability of all pairs in $\LCCP(\cR)\cup\CVPofPStickel(\cR)$ implies $\joinability{\cR}$-joinability of all pairs in $\DCP(\cR)$.
\item 
Since local confluence of 
$\rew{\cR,E}$ modulo $E$ with $\rew{\cR}$ 
boils
down into local confluence of $\rew{\cR}$, Theorem \ref{TheoEGTRSsAndLocaConfluenceAndCoherenceModulo}(\ref{TheoEGTRSsAndLocaConfluenceAndCoherenceModulo_LocalConfluence}) characterizes local confluence of $\cR$ as done by \cite[Theorem 62]{Lucas_LocalConfluenceOfConditionalAndGeneralizedTermRewritingSystems_JLAMP24}.
\item\label{SecConfluenceOfGTRSsAsConfluenceOfEGTRSsWithEempty_JK86_CoherencePairsEmpty} Since $E$ is empty, $\LCCP(E,\cR)=\CVPofPStickel(E)=\emptyset$. 
\item Since $\emptyset$-termination of an \egtrs{} is just termination of $\cR$ viewed as a \gtrs, using  Corollary \ref{CoroDCPsInGTRSs} again, we conclude that
Theorems \ref{TheoEConfluenceWithoutDCPs} (where $\LCCP(E,\cR)$ and $\CVPofPStickel(E)$ are empty)
and \ref{TheoremNonEConfluenceOfEGRSs} (where $\DCP(\cR)$ is subsumed by $\LCCP(\cR)$ and $\CVPofPStickel(\cR)$)
can be merged to obtain a characterization of confluence which
actually coincide with \cite[Theorem 68]{Lucas_LocalConfluenceOfConditionalAndGeneralizedTermRewritingSystems_JLAMP24}, characterizing confluence of terminating \gtrs{s}. 
\end{enumerate}

\subsection{Bouhoula, Jouannaud, and Meseguer \cite{BouJouMes_SpecificationAndProofInMembershipEquationalLogic_TAPSOFT97,BouJouMes_SpecificationAndProofInMembershipEquationalLogic_TCS00}}

According to \cite[Section 2.1]{BouJouMes_SpecificationAndProofInMembershipEquationalLogic_TCS00} (where missing details can be found),
a many-kinded signature is a pair $(\SKinds,\SSymbols)$,
where $\SSymbols$ is a $\SKinds^*\times\SKinds$-indexed family of sets $\{\SSymbols_{\vec{K}\to K}\}_{\vec{K}\in\SKinds^*,K\in\SKinds}$.
A signature in Membership Equational Logic (MEL, \cite{Meseguer_MembershipAlgebraAsALogicalFrameworkForEquationalSpecification_WADT97}) is a pair $\SPSignature$ consisting of a many-kinded signature $(\SKinds,\SSymbols)$ and a disjoint $\SKinds$-kinded family of sets of sorts $\SSorts=\{\SSorts_K\}_{K\in\SKinds}$ \cite[Definition 1]{BouJouMes_SpecificationAndProofInMembershipEquationalLogic_TCS00}.

Roughly speaking, a \mel{} specification consists of a set $M$ of
\emph{membership} axioms $t:s\IF\gencond$ for some term $t$ and sort $s$,  where the conditional part $\gencond$ may consist of (i) equations $u=v$ for terms $u$ and $v$ and (ii) other membership $t':s'$.
Conditional equations $u=v\IF\gencond$ are also allowed with $\gencond$ as above to obtain a set $E$ of such conditional equations.

\subsubsection{\mel{} specifications as \egtrs{s}}

As remarked in \cite[Section 3.4]{BouJouMes_SpecificationAndProofInMembershipEquationalLogic_TCS00}, \mel{} essentially is many-sorted Horn logic with equality. 
This could also be said of \gtrs{s}, actually (sorts $s\in\SSorts$
can be encoded in \gtrs{s} as monadic predicates $\_:s$ for each $s\in S$).
\emph{Deductions} with \mel{} specifications are defined by the \emph{rules of deduction} in \cite[Figure 4]{BouJouMes_SpecificationAndProofInMembershipEquationalLogic_TCS00}.
Roughly speaking, after removing the explicit universal quantifications in there, 
such rules of deduction can be seen as an \emph{Elementary Inference System} (see Section~\ref{SecHornTheoriesAsEIS}).
Thus, a \mel{} specification $\cE$ can be seen as 
an \egtrs{} $\cR_\cE=(\Symbols,\SPredicates,\muTop,E,H,\emptyset)$, without replacement restrictions and with an empty set of rewrite rules, where the ``deductive behavior'' of $\cE$ is encoded in $E$ and $H$ as follows (for simplicity, we assume that there is \emph{no overloading}, i.e., $f$ can only belong to a \emph{single} set 
$\SSymbols_{\SKinds^*\times\SKinds}$).
\begin{enumerate}
\item $\Symbols=\bigcup_{\vec{K}\in\SKinds^*,K\in\SKinds}\SSymbols_{\vec{K}\to K}$.
\item $\SPredicates=\{\_=\_\}\cup\{\_:K\mid K\in\SKinds\}\cup\{\_:s\mid s\in\SSorts_K, K\in\SKinds\}$.
\item $E$ coincides with
the conditional equations in $\cE$, which become \emph{Replacement} rules in \cite[Figure 4]{BouJouMes_SpecificationAndProofInMembershipEquationalLogic_TCS00} as part of
$\GLtheory_E$
in Definition \ref{DefTheoriesOfAnEGTRS} 
(or the corresponding \eis{} $\GLinferenceOf{\GLtheory_E}$).
Also \emph{Reflexivity}, \emph{Symmetry}, and \emph{Transitivity} rules are available as part of 
$\GLtheory_E$.
And \emph{Congruence} rule is equivalent to propagation rules $(\RulePropagation)^=_{f,i}$  for all $f\in\Symbols$ and $1\leq i\leq ar(f)$, which are also part of $\GLtheory_E$.
\item $H$ consists of the membership axioms of $\cE$, which become \emph{Membership} rules as part of $\GLtheory_E$, together with Horn clauses $x:s\IF y:s, x=y$ for each sort $s$, which correspond to the \emph{Subject reduction} rules.
\end{enumerate}

\subsubsection{Computing with \mel{} specifications as rewriting with \gtrs{s}}

In order to \emph{compute} with \mel{} specifications $\cE$, the so-called \emph{Conditional Rewriting/Membership Systems (\crms, \cite[Definition 12]{BouJouMes_SpecificationAndProofInMembershipEquationalLogic_TCS00})} are used. 
Given a \mel{} specification, a \crms{} is obtained as follows:
(i) membership axioms $t:s\IF\gencond$ are translated into $t:s\IF\gencond'$, 
where equational components $u'=v'$ in $\gencond$ above are translated as joinability goals $u'\downarrow v'$ in $\gencond'$; 
then
(ii) conditional equations $u=v\IF\gencond$ are transformed into \emph{conditional rewrite rules} $u\to v\IF\gencond'$ where $\gencond'$ is obtained from $\gencond$ as in (i)  \cite[Section 4.1]{BouJouMes_SpecificationAndProofInMembershipEquationalLogic_TCS00}.
Note that equational components \emph{disappear} to become conditional rules or joinability goals.
Computations are defined by the \emph{rules of deduction} in \cite[Figure 6]{BouJouMes_SpecificationAndProofInMembershipEquationalLogic_TCS00}.
Thus, a \crms{} $\cR$ for a \mel{} specification $\cE$ can be seen as 
a \gtrs{} $\cR_\cE=(\Symbols,\SPredicates,\muTop,H,R)$, with $\Symbols$ as before and again without replacement restrictions, where
\begin{enumerate}
\item $\SPredicates=\{\_\downarrow\_,\_\to\_,\_\tos{}\_\}\cup\{\_:K\mid K\in\SKinds\}\cup\{\_:s\mid s\in\SSorts_K, K\in\SKinds\}$.
\item $H$ consists of (parenthesized names refer to corresponding rules in \cite[Figure 6]{BouJouMes_SpecificationAndProofInMembershipEquationalLogic_TCS00}):
\begin{enumerate}[align=left]
\item atomic clauses $x:s$ for all variables $x$ of sort $s$ defined in $\cE$ (\emph{Variable} rules).
\item clauses $x:s\IF x:s'$ for all sorts $s,s'\in\SSorts$ satisfying $s\leq s'$ (\emph{Subsort} rules).
\item the membership axioms of $\cR$ (\emph{Membership} rules)
\item clauses $M:s\IF N:s, M=N$ for each sort $s$ (\emph{Subject reduction} rules).
\item A clause $x\downarrow y\IF x\rews{}z, y\rews{}z$ (not explicit in \cite[Figure 6]{BouJouMes_SpecificationAndProofInMembershipEquationalLogic_TCS00}) gives meaning to joinability $\downarrow$ in terms of reachability, as usual.
\end{enumerate}
\item $R$ consists of rules in $\cR$ (\emph{Replacement} rules).
\end{enumerate}
Rather than \emph{Reflexivity} and \emph{Transitivity} rules for $\to$, 
we use  $(\RuleReflexivity)^{\tos{}}$ and $(\RuleCompatibility)^\to$.
Also, instead of \emph{Congruence} rule we use $(\RulePropagation)^\to_{f,i}$ for all $f\in\Symbols$ and $1\leq i\leq ar(f)$.

Thus, $\rew{\cR}$ and $\rew{\cR_\cE}$ coincide.
Confluence properties of \crms{} computations are investigated in \cite[Section 6]{BouJouMes_SpecificationAndProofInMembershipEquationalLogic_TCS00}.
However, the focus is on a specific Church-Rosser property 
together
with a sort-decreasingess condition 
\cite[Definition 45]{BouJouMes_SpecificationAndProofInMembershipEquationalLogic_TCS00}.
Furthermore, the involved techniques are not directly connected with that of this paper; they rather correspond to the research about confluence of \gtrs{s} developed in \cite{Lucas_LocalConfluenceOfConditionalAndGeneralizedTermRewritingSystems_JLAMP24}, although the sort-decreasingness aspect is not investigated there.

\subsection{Dur\'an and Meseguer \cite{DurMes_OnTheChurchRosserAndCoherencePropertiesOfConditionalOrderSortedRewriteTheories_JLAP12}}

The $E$-confluence of 
\emph{conditional rewrite theories} $\cR=(\Symbols,E,R)$ has been investigated in \cite{DurMes_OnTheChurchRosserAndCoherencePropertiesOfConditionalOrderSortedRewriteTheories_JLAP12}.
Conditional equations $s=t\IF\gencond$
can be specified, but they are treated as conditional
rewrite rules (in $R$) by imposing some specific orientation 
(e.g., $s\to t\IF\gencond$). Only unconditional equations $s=t$ 
(called \emph{axioms})  which are
\emph{linear} and \emph{regular} (i.e., $\Var(s)=\Var(t)$
\cite[page 819]{DurMes_OnTheChurchRosserAndCoherencePropertiesOfConditionalOrderSortedRewriteTheories_JLAP12})
are used in $E$ (denoted $A$ in \cite{DurMes_OnTheChurchRosserAndCoherencePropertiesOfConditionalOrderSortedRewriteTheories_JLAP12}).

\begin{rem}[Moving equations from $E$ to $R$]
In the realm of \Maude, for which this research is intended to be used, it is a usual practice to move to $R$ all equations in $E$ (as rules) except those (denoted $A$) for which appropriate $E$-matching algorithms exist.
Typically, $A$ would consist of associative, commutative, unit axioms, etc., see \cite[page 818]{DurMes_OnTheChurchRosserAndCoherencePropertiesOfConditionalOrderSortedRewriteTheories_JLAP12}.
\end{rem}
%
\noindent
The main result about $E$-confluence is \cite[Theorem 2]{DurMes_OnTheChurchRosserAndCoherencePropertiesOfConditionalOrderSortedRewriteTheories_JLAP12}, characterizing $E$-confluence 
as the $\joinability{\cR,A}$-joinability (\cite[Definition 7]{DurMes_OnTheChurchRosserAndCoherencePropertiesOfConditionalOrderSortedRewriteTheories_JLAP12}) of the set of 
\emph{Equational Conditional Critical Pairs} 
with respect to the equational theory 
$A$ (see Section~\ref{SecECCPsAndLCTRS}) 
obtained (if an appropriate $A$-unification algorithm exists) from rules 
in $R$ (which may include oriented conditional equations). 
However, a number of restrictions are imposed \cite[Section 3.1]{DurMes_OnTheChurchRosserAndCoherencePropertiesOfConditionalOrderSortedRewriteTheories_JLAP12}:
\begin{enumerate}
\item 
$A$ is a set of \emph{linear and regular unconditional} equations;
\item  
$R$ is \emph{strongly $A$-coherent}, i.e., the following diagram commutes
\cite[page 819]{DurMes_OnTheChurchRosserAndCoherencePropertiesOfConditionalOrderSortedRewriteTheories_JLAP12}:
\begin{xy}
\xymatrix{
 u  \ar@{->}[rr]_{\cR/A}  \ar@{|-|}[d]_{*}^{A} & & v\ar@{|.|}[d]^{*}_{A} 
\\
u' \ar@{.>}[rr]_{\cR,A} & & v'
}
\end{xy}

\item  
the rules in $R$ are \emph{strongly deterministic} \cite[Definition 1]{DurMes_OnTheChurchRosserAndCoherencePropertiesOfConditionalOrderSortedRewriteTheories_JLAP12};
and
\item  
$\cR$ is \emph{quasi-decreasing} \cite[Definition 2]{DurMes_OnTheChurchRosserAndCoherencePropertiesOfConditionalOrderSortedRewriteTheories_JLAP12}, which is quite a strong termination condition. Together with strong $A$-coherence above, it entails $E$-termination.
\end{enumerate}
Such criteria, in particular quasi-decreasingness, can be used to decide the appropriate orientation of equations in $E$ into rules to be added to $R$.
By lack of space, we cannot provide a detailed comparison with our approach. However, sometimes our results can be advantageously used to improve on \cite{DurMes_OnTheChurchRosserAndCoherencePropertiesOfConditionalOrderSortedRewriteTheories_JLAP12}:
\begin{itemize}
\item $E$-confluence of $\cR$ in Example \ref{ExRModuloAndRrelativeE}, with $E=\{\fS{a} = \fS{b}\}$ and 
$R^\RMabbr=\{
\fS{a} \to \fS{c},
\fS{a}  \to  \fS{d}\IF \fS{b} \rewmodulos{} \fS{c},
\fS{c} \to \fS{d}\}$,
cannot be proved with \cite[Theorem 2]{DurMes_OnTheChurchRosserAndCoherencePropertiesOfConditionalOrderSortedRewriteTheories_JLAP12} as $\cR$ is \emph{not} strongly deterministic \cite[Definition 1]{DurMes_OnTheChurchRosserAndCoherencePropertiesOfConditionalOrderSortedRewriteTheories_JLAP12}: condition $\fS{b}\rews{}\fS{c}$ in rule (\ref{ExRModuloAndRrelativeE_rule2})
does not fulfill the requirement of $\rew{\cR,E}$-irreducibility of $\fS{c}$. 
We prove it $E$-confluent in Example \ref{ExRModuloAndRrelativeE_EConfluence}.
\item
The $E$-confluence of 
$\cR$ in Example \ref{ExPeakNoCPs}, with $E=\{\fS{b}=\fS{f}(\fS{a}), \fS{a}=\fS{c}\}$ and $R=\{\fS{c}\to\fS{d},\fS{b}\to\fS{d}\}$ cannot be \emph{disproved} by using
\cite[Theorem 2]{DurMes_OnTheChurchRosserAndCoherencePropertiesOfConditionalOrderSortedRewriteTheories_JLAP12}.
Note that $E$ satisfies the requirements for axioms $A$ in \cite{DurMes_OnTheChurchRosserAndCoherencePropertiesOfConditionalOrderSortedRewriteTheories_JLAP12}.
We have:
$\fS{b}\equequ{E}\fS{f}(\ul{\fS{a}})\rew{\cR}\fS{f}(\fS{d})$, i.e., $\fS{b}\rew{\cR/E}\fS{f}(\fS{d})$,
but the only $\rew{\cR,E}$-step on $\fS{b}$ is $\fS{b}\rew{\cR,E}\fS{d}$, and
$\fS{f}(\fS{d})\neq_E\fS{d}$. Thus, $\cR$ is not strongly $E$-coherent
and \cite[Theorem 2]{DurMes_OnTheChurchRosserAndCoherencePropertiesOfConditionalOrderSortedRewriteTheories_JLAP12} does not apply.
However, in Example \ref{ExPeakNoCPs_notEconfluent} we prove $\cR$ non-$E$-confluent.
\end{itemize}
The definitions of $\rew{\cR,A}$ and $\rew{\cR/A}$ in \cite[page 819]{DurMes_OnTheChurchRosserAndCoherencePropertiesOfConditionalOrderSortedRewriteTheories_JLAP12} are equivalent to our Definition \ref{DefComputationalRelationsOfAnEGTRS}(\ref{DefComputationalRelationsOfAnEGTRS_rewPStickel}) 
and
\ref{DefComputationalRelationsOfAnEGTRS}(\ref{DefComputationalRelationsOfAnEGTRS_RewritingModulo}), respectively.
In this setting, they notice a number of important facts (without proof), see \cite[page 820, end of Section 2]{DurMes_OnTheChurchRosserAndCoherencePropertiesOfConditionalOrderSortedRewriteTheories_JLAP12}:
\begin{enumerate}
\item If $\cR$ is $A$-coherent, then $A$-confluence is equivalent to confluence of $\rew{\cR,A}$ modulo $A$,  
i.e., the commutation of the rightmost diagram in Figure \ref{FigChurchRosserAndEConfluenceJK86} when $\rew{\cR,A}$ and $\equequ{A}$ are used instead of $\rew{\genrelation/\genequivalence}$ and $\equone{\genequivalence}^*$, see the central diagram in \cite[page 820]{DurMes_OnTheChurchRosserAndCoherencePropertiesOfConditionalOrderSortedRewriteTheories_JLAP12}.
\item $A$-termination is equivalent to termination of $\rew{\cR,A}$.
\end{enumerate}
Unfortunately, Dur\'an and Meseguer discuss no method for proving  $A$-coherence, which is always assumed.
In \cite[Section 4]{DurMes_OnTheChurchRosserAndCoherencePropertiesOfConditionalOrderSortedRewriteTheories_JLAP12} they focus on a notion of coherence which is more appropriate for \Maude{} computations.
\section{Related work}
\label{SecRelatedWork}
Figure \ref{FigEarlyResearchOnConfluenceModulo}
provides a summary of the evolution of the research leading to methods for proving confluence modulo of \etrs{s}. 
\begin{figure}[t]
\scalebox{0.8}{
\begin{xy}
\xymatrix{
& \text{Church \& Rosser 
1936 \cite{ChuRos_SomePropertiesOfConversion_TAMS36}}
\ar@{->}[d]
\\
\!\textsc{Abstract reduction} & \text{Newman, 1942 \cite{Newman_OnTheoriesWithACombinatorialDefinitionOfEquivalence_AM42}}
\ar@{->}[ld]
\ar@{->}[rd]
&
\textsc{Equational reasoning}
\\
\text{Hindley, 1964 \cite{Hindley_TheChurchRosserPropertyAndAResultInCombinatoryLogic_PhD64}} 
\ar@{->}[d]
& & 
\text{Knuth \& Bendix, 
1970 \cite{KnuBen_SimpleWordProblemsInUniversalAlgebra_CPAA70}}
\ar@{->}[d]
\\
\text{Rosen, 1970 \cite{Rosen_TreeManipulatingSystemsAndChurchRosserTheorems_STOC70}} 
\ar@{->}[d]
& & 
\text{Lankford, 1975 \cite{Lankford_CanonicalAlgebraicSimplificationInComputationalLogic_TR75}}
\ar@{->}[d] 
\\
\text{Sethi, 1974 \cite{Sethi_TestingForTheChurchRosserProperty_JACM74}} 
\ar@{->}[dr]
&
& 
\text{Lankford \& Ballantine, 1977 \cite{LanBal_DecisionProceduresForSimpleEquationalTheoriesWithPermutativeEquations_TR77}}
\ar@{->}[dl]
\\
& 
\text{Huet, 1977 \cite{Huet_ConfluentReductionsAbstractPropertiesAndApplicationsToTermRewritingSystems_FOCS77}}
\ar@{->}[d]
\\
& 
\text{Peterson \& Stickel, 1981 \cite{PetSti_CompleteSetsOfReductionsForSomeEquationalTheories_JACM81}}
\ar@{->}[d]
\\
& \text{Jouannaud, 1983 \cite{Jouannaud_ConfluentAndCoherentEquationalTermRewritingSystemsApplicationToProofsInAbstractDataTypes_CAAP83}}
\ar@{->}[d]
\\
& 
\text{Jouannaud \& Kirchner, 1986 \cite{JouKir_CompletionOfASetOfRulesModuloASetOfEquations_SIAMJC86}}
}
\end{xy}
}
\caption{Early research on confluence modulo}
\label{FigEarlyResearchOnConfluenceModulo}
\end{figure}
Interestingly, after the seminal papers by Church and Rosser \cite{ChuRos_SomePropertiesOfConversion_TAMS36} 
and Newman \cite{Newman_OnTheoriesWithACombinatorialDefinitionOfEquivalence_AM42},
 the research follows two main branches:
 \begin{enumerate}
 \item Hindley's PhD thesis \cite{Hindley_TheChurchRosserPropertyAndAResultInCombinatoryLogic_PhD64},
 and also \cite{Hindley_AnAbstractFormOfTheChurchRosserTheoremI_JSL69,Hindley_AnAbstractFormOfTheChurchRosserTheoremApplications_JSL74}, explicitly 
 investigated \emph{abstract relations}, and different ``confluence properties'' of these abstract relations  considered for the purpose of proving a Church-Rosser property.
 This research was continued by
 Rosen \cite{Rosen_TreeManipulatingSystemsAndChurchRosserTheorems_STOC70,%
 Rosen_TreeManipulatingSystemsAndChurchRosserTheorems_JACM73} 
 and then
 Sethi \cite{Sethi_TestingForTheChurchRosserProperty_JACM74}
 who considered the joint use of reduction and equivalence relations.
 \item Knuth \& Bendix' paper \cite{KnuBen_SimpleWordProblemsInUniversalAlgebra_CPAA70} explored efficient techniques for equational reasoning, leading to a \emph{completion algorithm} aiming at transforming a set of equations $E$ into a \trs{} $\cR_E$ so that equality goals for $E$ can be solved using $\cR_E$.
Afterwards, in order to improve Knuth \& Bendix' completion,
Lankford \cite{Lankford_CanonicalAlgebraicSimplificationInComputationalLogic_TR75}, and then 
Lankford and Ballantine \cite{LanBal_DecisionProceduresForSimpleEquationalTheoriesWithPermutativeEquations_TR77} 
investigated the use of equations $E$ together with rules in a \trs{} $\cR$ in equational reasoning. 
In particular, Lankford and Ballantine introduced \emph{rewriting of equivalence classes} which corresponds to the current notion of rewriting modulo.
 \end{enumerate}
Remarkably, neither \cite{Rosen_TreeManipulatingSystemsAndChurchRosserTheorems_STOC70,%
 Rosen_TreeManipulatingSystemsAndChurchRosserTheorems_JACM73} nor 
  \cite{Sethi_TestingForTheChurchRosserProperty_JACM74} cite Knuth \& Bendix.
Similarly, \cite{Lankford_CanonicalAlgebraicSimplificationInComputationalLogic_TR75} does not cite Hindley, Rosen, or Sethi's work.
A kind of convergence arrived with Huet's work \cite{Huet_ConfluentReductionsAbstractPropertiesAndApplicationsToTermRewritingSystems_FOCS77}, who cites all the aforementioned authors and improves their work in several aspects.
However, he followed Sethi's approach without mixing equivalence and reduction steps, not providing a notion of \emph{rewriting modulo} as is currently understood.
Peterson \& Stickel provided a term-based definition of Lankford and Ballantine's equivalence class rewriting, which is, essentially, what is used today.
Jouannaud \cite{Jouannaud_ConfluentAndCoherentEquationalTermRewritingSystemsApplicationToProofsInAbstractDataTypes_CAAP83}, 
and then Jouannaud and Kirchner 
\cite{JouKir_CompletionOfASetOfRulesModuloASetOfEquations_SIAMJC86}
gave stable definitions and notations for \etrs{s}.

In the following, we further develop this evolution and also discuss subsequent developments.

\subsection{Hindley (1964)}
\label{SecRelatedWork_Hindley}

As pioneered by Church and Rosser for \emph{conversion} of $\lambda$-expressions \cite{ChuRos_SomePropertiesOfConversion_TAMS36}, in his PhD thesis \cite{Hindley_TheChurchRosserPropertyAndAResultInCombinatoryLogic_PhD64}, 
Hindley extensively uses diagrams to represent properties of reductions with (abstract) binary relations 
(he uses $r$ instead of $\rew{\genrelation}$, $\geq_r$ instead of $\rews{\genrelation}$, and $\sim_r$ instead of $\conversion{\genrelation}$).
No additional equivalence relation is considered beyond $\sim_r$, derived from $r$ as \emph{conversion}.
To introduce abstract properties of binary relations, 
Hindley uses the diagrammatic notation. 
In particular, the Church-Rosser property (CR) in \cite[page 2]{Hindley_TheChurchRosserPropertyAndAResultInCombinatoryLogic_PhD64} is exemplified in a diagram in page 3 which essentially is as follows:
\begin{IEEEeqnarray}{r'C'l}
\raisebox{0.6cm}{\xymatrix{
t  \ar@{.>}[dr]_{*}\ar@{<->}[rr]^{*} & & t'\ar@{.>}[dl]^{*}
\\
& u 
}}
\label{LblCRinHindley64}
\end{IEEEeqnarray}
And Property (D) in page 3 is depicted in a diagram which is similar to the usual \emph{local confluence} diagram.
Confluence is called (B) in page 16, where it is proved equivalent to CR (Lemma 1.1).
In page 62, in a section on ``other people's results'', he reproduces Newman's main results in \cite{Newman_OnTheoriesWithACombinatorialDefinitionOfEquivalence_AM42} and, in particular, \cite[page 62, item (1)]{Newman_OnTheoriesWithACombinatorialDefinitionOfEquivalence_AM42} corresponds to what we often call ``Newman's Lemma'' today: termination plus local confluence implies confluence (provided that the two relations $r$ and $s$ used in \cite[page 62, item (1)]{Newman_OnTheoriesWithACombinatorialDefinitionOfEquivalence_AM42} are identified).

\subsection{Knuth and Bendix (1970)} 

In contrast to Hindley, who focused on abstract relations $r$ 
from which an equivalence $\sim_r$ is derived as conversion, in 
\cite{KnuBen_SimpleWordProblemsInUniversalAlgebra_CPAA70}, Knuth and Bendix investigated the use of a set $E$ of \emph{identities} on terms (i.e., equations) as a set 
 $R$ of \emph{reductions}, i.e., ``the right-hand side of the identity represents a word [term] smaller (\ldots) than the left-hand side''
\cite[page 263, Summary]{KnuBen_SimpleWordProblemsInUniversalAlgebra_CPAA70}
to prove equational goals $s\equequ{E}t$ by first obtaining $R$-normal forms $\theNfOf{s}$ and $\theNfOf{t}$ and then checking them for equality. 
Note that,
\begin{enumerate}
\item\label{KnuthBendix_Reductions_SmallerCond} 
Being (the left-hand side of an identity) \emph{smaller} (than the right-hand side) is relative to an \emph{order} $>$ 
defined by associating \emph{weights} to the symbols occurring in the identities.

\item\label{KnuthBendix_Reductions_CompletenessCond}
$R$ must be a \emph{complete} set of reductions , i.e.,
\emph{no two distinct irreducible words [i.e., terms] are equivalent with respect to $R$}
\cite[Section 4]{KnuBen_SimpleWordProblemsInUniversalAlgebra_CPAA70}.

\end{enumerate}
Condition (\ref{KnuthBendix_Reductions_SmallerCond}) corresponds to requiring \emph{termination} of $R$.
Condition (\ref{KnuthBendix_Reductions_CompletenessCond})
is the
\emph{unique normal form} property  ($\pUN$) 
\cite[Definition 1.2(iii)]{KloVri_UniqueNormalFormsForLambdaCalculusWithSurjectivePairing_IC89}.
Since $R$ is terminating,
$\pUN$ and confluence coincide \cite[Figure 2.4]{Ohlebusch_AdvancedTopicsInTermRewriting_2002}.
Thus, 
complete set of reductions $R$ are \emph{confluent} and \emph{terminating} \trs{s}.
Knuth and Bendix then show how to prove a set of reductions complete by analyzing \emph{superpositions} among rules \cite[Theorem 5 \& Corollary]{KnuBen_SimpleWordProblemsInUniversalAlgebra_CPAA70}. 

In \cite[Section 6]{KnuBen_SimpleWordProblemsInUniversalAlgebra_CPAA70} a method to obtain the \emph{extension} of an initial set of reductions obtained from a set of equations to a \emph{complete set} of reductions is described. 
Later on, this was called a \emph{completion procedure}
\cite[Chapter 7]{BaaNip_TermRewritingAndAllThat_1998}
or \emph{completion algorithm} \cite[Chapter 7]{Terese_TermRewritingSystems_2003}.
The procedure starts as an initial attempt to obtain a terminating \trs{} $R$ by appropriately orienting the equations $s=t$ in $E$ as $s\to t$ (if $s>t$) or $t\to s$ (if $t>s$).
However, this may fail for some equational theories. 
For instance, for the equation $f(x,y)=f(y,x)$ defining commutativity of $f$, 
neither $f(x,y)\to f(y,x)$ nor $f(y,x)\to f(x,y)$ lead to a terminating \trs.
This led, in particular, to the development of the idea of doing completion with respect to
\emph{rewriting modulo a set of equations}, see Sections \ref{SecLankford} and \ref{SecLankfordBallantine}.

\subsection{Rosen (1970, 1973). Aho, Sethi, and Ullman (1972). Sethi (1974)}
\label{SecRelatedWork_Sethi}

Continuing Hindley's approach,
Rosen's \emph{General Replacement Systems} (GRSs \cite{Rosen_TreeManipulatingSystemsAndChurchRosserTheorems_STOC70,Rosen_TreeManipulatingSystemsAndChurchRosserTheorems_JACM73}) are what we call \emph{Abstract Reduction Systems} (ARSs) today \cite[Chapter 1]{Terese_TermRewritingSystems_2003}.
Rosen investigated confluence of GRSs and its application to prove confluence of computations with 
\emph{Rule Schemata} \cite[Definition 6.1]{Rosen_TreeManipulatingSystemsAndChurchRosserTheorems_JACM73}, a predecessor of Term Rewriting Systems.

In \cite{AhoSetUll_CodeOptimizationAndFiniteChurchRosserTheorems_DOC72} Aho, Sethi, and Ullman considered \emph{Finite Church-Rosser (FCR) pairs} consisting of a
\emph{terminating} reduction relation $\to$ 
and an equivalence $\equiv$ (with no particular relation with $\to$) on a set ``of programs'' $S$\footnote{The shape of such programs remains unspecified, so that it can be considered as an abstract set.} such that the diagram in Figure \ref{FigFCRproperty_ASU72} (left) commutes \cite[page 98]{AhoSetUll_CodeOptimizationAndFiniteChurchRosserTheorems_DOC72}.
\begin{figure}[t]
\begin{center}
\begin{tabular}{c@{\hspace{2cm}}c}
\xymatrix{
t  \ar@{->}[dd]_{!}\ar@{=}[rr] & & t'\ar@{->}[dd]^{!}
\\
\\
u\ar@{::}[rr] & &  u'
}
&
\xymatrix{
s \ar@{->}[d]  \ar@{=}[rr] & & s' \ar@{->}[d]^{\#}
\\
t \ar@{.>}[d]_{*} & &  t'\ar@{.>}[d]^{*}
\\
u  \ar@{::}[rr] & & u'
}
\end{tabular}

\vskip 10pt
\refstepcounter{figure}
\label{FigFCRproperty_ASU72}
\figurename~\thefigure: Conditions for Finite Church-Rosser Pairs \cite{AhoSetUll_CodeOptimizationAndFiniteChurchRosserTheorems_DOC72}, where \\
$=$ is used instead of $\equiv$ and $\stackrel{\#}{\to}$ is the union of $\to$ and equality
\end{center}
\end{figure}
They prove that this is equivalent to requiring commutation of the rightmost diagram of Figure \ref{FigFCRproperty_ASU72} \cite[Theorem 3 \& Figure 3]{AhoSetUll_CodeOptimizationAndFiniteChurchRosserTheorems_DOC72}.

\begin{rem}
Note that diagram (\ref{LblCRinHindley64}) and the leftmost diagram of Figure \ref{FigFCRproperty_ASU72} are both intended to represent a ``\emph{Church-Rosser property}'', but they are meaningfully different.
In (\ref{LblCRinHindley64}) only one relation (and its conversion)
is used.
In the leftmost diagram of Figure \ref{FigFCRproperty_ASU72} two independent relations $\equiv$ and $\to$ are considered.
\end{rem}
%
\noindent
Sethi \cite{Sethi_TestingForTheChurchRosserProperty_JACM74} continued this research and introduced new criteria for guaranteeing the Church-Rosser property
of \emph{Finite Replacement System} (FRS) $(S,\to,\equiv)$, where 
$\to$ is \emph{terminating}. Sethi used FCR to designate the leftmost property in Figure \ref{FigFCRproperty_ASU72}. 
Most properties are described using diagrams.
In particular,
$P1$ \cite[Figure 3(b)]{Sethi_TestingForTheChurchRosserProperty_JACM74} corresponds to Property $\beta$ of \cite[Figure 9 (right)]{Huet_ConfluentReductionsAbstractPropertiesAndApplicationsToTermRewritingSystems_JACM80} and
$P3$ \cite[Figure 4(left)]{Sethi_TestingForTheChurchRosserProperty_JACM74} to Property $\alpha$.
His Theorem 2.2 (\emph{an FRS satisfying $P1$ and $P3$ is FCR}) corresponds to  \cite[Lemma 2.7]{Huet_ConfluentReductionsAbstractPropertiesAndApplicationsToTermRewritingSystems_JACM80} (if $\to$ is terminating and satisfies $\alpha$ and $\beta$, then it is confluent modulo $\equiv$).
In \cite{AhoSetUll_CodeOptimizationAndFiniteChurchRosserTheorems_DOC72,Sethi_TestingForTheChurchRosserProperty_JACM74}, though, $\rew{}$ and $\equiv$ are \emph{not} mixed in their use: $\rew{}$ is used to \emph{simplify} (programs) and then $\equiv$ is used to test whether simplified programs are equivalent.
No ``reduction modulo'' is considered.

\subsection{Plotkin (1972). Slagle (1974)}

Plotkin investigated the efficient use of resolution-based theorem-provers in the presence of equational theories \cite{Plotkin_BuildingInEquationalTheories_MI72}. He considered the use of $E$-unification instead of unification in a modified resolution mechanism.
He noticed that, in this case, \emph{rather than a single most general unifier, there will be an infinite set of maximally general ones} \cite[page 74]{Plotkin_BuildingInEquationalTheories_MI72}. 
Important properties 
like \emph{associativity},  which can be represented by appropriate equational theories, \emph{fail} to 
have a complete and finite $E$-unification algorithm which is able to produce all $E$-unifiers for all $E$-unification problems $s=^?_Et$.
In general, \emph{for each equational theory one must invent a special unification algorithm with these equations built in}.

Slagle 
\cite{Slagle_AutomatedTheoremProvingForTheoriesWithSimplifiersCommutativityAndAssociativity_JACM74}
used \emph{oriented equations} as \emph{simplifiers} of clauses in resolution-based theorem proving, so that
\cite[page 626]{Slagle_AutomatedTheoremProvingForTheoriesWithSimplifiersCommutativityAndAssociativity_JACM74}:
\begin{enumerate}
\item\label{Slagle_Simplifiers_cond1} for any expression, \emph{the (simplification) process terminates}, and
\item\label{Slagle_Simplifiers_cond2}  \emph{terminates with the same expression no matter in what order the immediate collapses \emph{[i.e., one-step rewritings]} are made}.
\end{enumerate}
Item (\ref{Slagle_Simplifiers_cond1}) corresponds to \emph{termination} and item (\ref{Slagle_Simplifiers_cond2}) corresponds to the 
\emph{unique normal form with respect to reduction} property ($\pUNred$
\cite[Definition 1.2(ii)]{KloVri_UniqueNormalFormsForLambdaCalculusWithSurjectivePairing_IC89}).
In the presence of termination, 
$\pUNred$ and confluence coincide.
In Sections 5 to 8, Slagle also considered the use of $E$-unification in resolution-based theorem proving and developed algorithms for associative and/or commutative functions  which also performed 
simplification/normalization steps.

\subsection{Lankford (1975)} 
\label{SecLankford}

In \cite{Lankford_CanonicalInference_TR75}, Lankford revisited Knuth and Bendix' work and identified complete sets of reductions as \emph{sets of rewrite rules} 
satisfying 
\cite[page 2]{Lankford_CanonicalInference_TR75}:
\begin{enumerate}
\item a \emph{finite termination property} (FTP): the set of rewrite rules \emph{always leads to a finite sequence of simplifications}, and 
\item a \emph{unique termination property} (UTP), i.e., $\pUN$ as in \cite{KnuBen_SimpleWordProblemsInUniversalAlgebra_CPAA70}.
\end{enumerate}
In 
\cite[pages 2-3]{Lankford_CanonicalInference_TR75},
Lankford shows that the aforementioned Knuth and Bendix' method to prove UTP for a (terminating) set of rewrite rules $R$ based on the notion of superposition can be adapted to use \emph{unification}.
For this purpose, 
Lankford introduced the notion of a \emph{special equality inference} $s=t$ of a set $\cR$ of rewrite rules, which is obtained from two rules in $\cR$ exactly as Huet's critical pairs \cite[Section II.2]{Huet_ConfluentReductionsAbstractPropertiesAndApplicationsToTermRewritingSystems_FOCS77}, except for two slight differences: (i) he does \emph{not} mention that rules involved in the obtention of special equality inferences cannot share variables (but his examples in the second half of page 3 and page 4 show that this is implicit for him, as rules are renamed if necessary) and (ii) he considers special equality inferences obtained from the same rule at the \emph{top position} $\toppos$, which is explicitly excluded by Huet, as only trivial critical pairs $\langle t,t\rangle$ are obtained.
Then, in \cite[page 3]{Lankford_CanonicalInference_TR75}, he presents
\begin{quote}
\emph{1.4 The Unique Termination Algorithm.} If $\cR$ is a set of rewrite rules such that each sequence of simplifications by $\cR$ is finite, then $\cR$ has the unique termination property iff each special equality inference $s=t$ of $\cR$ has the property that $s$ and $t$ simplify to identical terms.
\end{quote}
which (since UTP and confluence coincide due to FTP) is just the usual algorithm to check confluence as the joinability of critical pairs $\langle s,t\rangle$ of a terminating \trs{} $\cR$ (compare with \cite[Theorem 2]{Huet_ConfluentReductionsAbstractPropertiesAndApplicationsToTermRewritingSystems_FOCS77}).

\subsection{Lankford and Ballantine (1977)}
\label{SecLankfordBallantine}
In 
\cite[page 4]{Lankford_SomeApproachesToEqualityForComputationalLogicASurveyAndAssessment_TR77}, and already in \cite{Lankford_CanonicalAlgebraicSimplificationInComputationalLogic_TR75},
Lankford observes that the presence of a \emph{commutativity} equation (e.g., $x+y=y+x$) in an equational theory $E$ makes Knuth and Bendix' approach unable to obtain a complete set of reductions to solve equalities as no orientation for such an equation ``makes it terminating''.
Then, Lankford and Ballantyne proposed the use of reduction of \emph{equivalence classes of terms} and the corresponding 
\begin{quote}
extension of \emph{complete set of reductions concepts, properties, and techniques to equivalence classes of terms}.
\cite[page 3]{LanBal_DecisionProceduresForSimpleEquationalTheoriesWithPermutativeEquations_TR77}.
\end{quote}
In particular,
\begin{quote}
\emph{a mathematical characterization of the unique termination property for finite sets of certain equivalence class rewrite rules}.
\end{quote}
Furthermore, they \emph{generalize Knuth and Bendix' completion technique to equivalence classes of rewrite rules}.
In \cite{LanBal_DecisionProceduresForSimpleEquationalTheoriesWithPermutativeEquations_TR77}, this is restricted, though, to sets of \emph{permutative equations} $s=t$ which are those where both $s$ and $t$ have the \emph{same number} of each (function or variable) symbol occurring in them
(e.g.,
commutativity and associativity axioms, but \emph{not} distributivity). 

\subsubsection{Equivalence rewriting relative to $E$}

In this setting, a \trs{} $\cR$ with rules $\lhsr\to\rhsr$ and
finite sets of expressions $[\lhsr]\to[\rhsr]$ are considered, where $[\_]$ denotes an equivalence class (of a given considered equational theory, which we call $E$ in the following)  \cite[page 160]{BalLan_NewDecisionAlgorithmsForFinitelyPresentedCommutativeSemigroups_CMA81},
see also \cite[page 4]{LanBal_DecisionProceduresForSimpleEquationalTheoriesWithPermutativeEquations_TR77}, using a different notation.
Then, \emph{equivalence rewriting relative to $E$} is defined as follows
\begin{defi}[Equivalence rewriting relative to $E$ \emph{\cite[page 4]{LanBal_DecisionProceduresForSimpleEquationalTheoriesWithPermutativeEquations_TR77}}] 
\label{DefEquivalenceRewriting_LB77}
$[s]$ rewrites into $[t]$, written $[s]\to[t]$, 
if there is a substitution 
$\sigma$, 
$s'\in[s]$,
$p\in\Pos(s')$,
$t'\in[t]$,  
$\lhsr'\in[\lhsr]$,
and
$\rhsr'\in[\rhsr]$
such that $s'=s'[\sigma(\lhsr')]_p$
and $t'=s'[\sigma(\rhsr')]_p$.
\end{defi}
%
\noindent
This definition can be found in
\cite[Section 11.1]{BaaNip_TermRewritingAndAllThat_1998} as follows:
\begin{IEEEeqnarray}{r'C'l}
[s]_E\rew{\cR/E}[t]_E & :\Leftrightarrow & (\exists s',t')\: s\equequ{E}s'\rew{\cR}t'\equequ{E}t\label{DefRewritingModulo_BN98}
\end{IEEEeqnarray}
but the usual notation $\rew{\cR/E}$ and denomination  ``\emph{rewriting modulo}'' are
used. 
\begin{rem}[Class rewriting]
Recently, Jouannaud calls \emph{class rewriting} to $\equequ{E}\circ\rew{\cR}$ (also denoted $\rew{\cR/E}$ \cite[page 277]{Jouannaud_ConfluenceOfTerminatingRewritingComputations_TFSP24}).
Although equivalence rewriting relative to $E$ and class rewriting are not formally identical,  in the following we often use the shorter term ``class rewriting'' instead of the original ``equivalence rewriting relative to $E$''.
\end{rem}

\subsubsection{Critical pairs for class rewriting}

As promised, FTP and UTP are generalized to  equivalence rewriting relative to an equational theory in \cite[page 4]{LanBal_DecisionProceduresForSimpleEquationalTheoriesWithPermutativeEquations_TR77}:
\begin{enumerate}
\item\label{LblFTPinLB77}
 the \emph{finite termination property} holds if there are no infinite sequences $[t_1]\rew{}[t_2]\rew{}\cdots$, and
\item\label{LblUTPinLB77}
 the \emph{unique termination property} holds if for each equivalence class $[t]$, if $[t]$ 
equationally rewrites into irreducible classes $[u]$ and $[v]$, then $[u]$ and $[v]$ \emph{coincide}.
\end{enumerate}
Hence, a set $\cR$ of rewrite rules relative to an equational theory $E$ is \emph{a complete set of reductions relative to $E$} if the two properties above hold.
Note that (\ref{LblFTPinLB77}) is $E$-termination of $R$.
Lankford and Ballantine provide two results to prove UTP. For the second result, on the basis of the notion of \emph{paramodulation} 
\cite{RobWos_ParamodulationAndTheoremProvingInFirstOrderTheoriesWithEquality_TR68,RobWos_ParamodulationAndTheoremProvingInFirstOrderTheoriesWithEquality_MI69,RobWos_ParamodulationAndTheoremProvingInFirstOrderTheoriesWithEquality_AR83} and also
\cite[Section 4]{Slagle_AutomatedTheoremProvingForTheoriesWithSimplifiersCommutativityAndAssociativity_JACM74} 
they introduce a notion of \emph{critical pair}.

\begin{defi}[Critical pair for class rewriting \text{\cite[page 8, last paragraph]{LanBal_DecisionProceduresForSimpleEquationalTheoriesWithPermutativeEquations_TR77}}]
\label{CriticalPair_LB77}
A critical pair is a pair $\langle [s],[t]\rangle$ such that $s=t$ is a \emph{paramodulant} obtained as described in the paragraph above the definition.
\end{defi}
%
\noindent
Then, if $\cR$ and $E$ satisfy FTP, then UTP is satisfied  iff for all critical pairs $\langle[s],[t]\rangle$, the normalization of $[s]$ and $[t]$ relative to $E$ leads to the same equivalence class $[u]$.
Clearly, this is an $E$-confluence criterion in terms of class rewriting.

\subsection{Huet (1977)}

Remarkably, the background of \cite{Huet_ConfluentReductionsAbstractPropertiesAndApplicationsToTermRewritingSystems_FOCS77} is both \cite{Hindley_TheChurchRosserPropertyAndAResultInCombinatoryLogic_PhD64,AhoSetUll_CodeOptimizationAndFiniteChurchRosserTheorems_DOC72,Rosen_TreeManipulatingSystemsAndChurchRosserTheorems_JACM73,Sethi_TestingForTheChurchRosserProperty_JACM74} and   \cite{KnuBen_SimpleWordProblemsInUniversalAlgebra_CPAA70}.
As in \cite{Sethi_TestingForTheChurchRosserProperty_JACM74}, in \cite[Section 3]{Huet_ConfluentReductionsAbstractPropertiesAndApplicationsToTermRewritingSystems_FOCS77} Huet considers independent equivalences $\sim$ and reductions $\to$, so that equivalent objects $s\sim t$ are then transformed by reduction $s\rews{}s'$ and $t\rews{}t'$ and then tested for equivalence: $s'\sim t'$.
As part of his analysis of abstract relations, Huet mentions $\rew{}/\sim$, implicitly defined as allowing ``$\sim$ along the $\to$-derivations'' \cite[after Definition D10]{Huet_ConfluentReductionsAbstractPropertiesAndApplicationsToTermRewritingSystems_FOCS77}, but he does not explicitly investigate it.
\begin{rem}[The notation $\rew{}/\sim$]
\label{RemTheNotationRewSim}
It is unclear, then, which is the origin of the notation $\rew{}/\sim$ referred to 
what we call \emph{rewriting modulo}.
Actually, in personal communication, Nachum Dershowitz said the following about that question:
\begin{quote}
\emph{My guess is and was that it was a natural use of the virgule notation for quotient relations, which dates back to the nineteenth century, applied to rewrite relations} \cite{Dershowitz_PersonalCommunication_2026_01_04}.
\end{quote}
\end{rem}
%
\noindent
In \cite[Section II.2]{Huet_ConfluentReductionsAbstractPropertiesAndApplicationsToTermRewritingSystems_FOCS77}, 
Huet introduces \emph{critical pairs} (as we know them today) 
to deal with Knuth and Bendix' \emph{superpositions of rules} (i.e., overlaps).
Besides the well-known use of critical pairs for (dis)proving confluence of \trs{s} $\cR$, 
Huet uses such critical pairs to provide a sufficient condition for his confluence of $\rew{\cR}$ modulo $E$, which he explicitly differentiates from 
$E$-confluence: in \cite[page 819]{Huet_ConfluentReductionsAbstractPropertiesAndApplicationsToTermRewritingSystems_JACM80}, referring to the completion methods developed in 
\cite{LanBal_DecisionProceduresForSimpleEquationalTheoriesWithCommutativeAxioms_TR1977} and \cite{LanBal_DecisionProceduresForSimpleEquationalTheoriesWithPermutativeEquations_TR77}, he says that
\begin{quote}
the condition checked in these papers is the confluence of $\to/\sim$, rather than the confluence of $\to$ modulo $\sim$. 
\end{quote}
In this case, however, critical pairs are obtained not only from the rules in $\cR$ but also from rules in $\cR$ and rules obtained by orienting equations in $E$, as explained in Section~\ref{SecHuet} above.

\subsection{Peterson and Stickel (1981)} 
\label{SecRelatedWork_PS81}
In \cite{PetSti_CompleteSetsOfReductionsForSomeEquationalTheories_JACM81}, a \emph{reduction} 
is just a set of \emph{pairs of terms}, actually a rewrite rule $\lhsr\to\rhsr$  \cite[Definition 8.1]{PetSti_CompleteSetsOfReductionsForSomeEquationalTheories_JACM81}, although the usual restrictions $\lhsr\notin\Variables$ and $\Var(\rhsr)\subseteq\Var(\lhsr)$ are not required. No requirement of termination as in Knuth \& Bendix, or Lankford above is imposed.

 \subsubsection{Rewriting modulo in \cite{PetSti_CompleteSetsOfReductionsForSomeEquationalTheories_JACM81}}

Ballantyne and Lankford's notion of reducing \emph{equivalence classes of terms}  (Definition \ref{DefEquivalenceRewriting_LB77}) 
is similar to \cite[Definition 8.3]{PetSti_CompleteSetsOfReductionsForSomeEquationalTheories_JACM81}.\footnote{Lankford and Ballantine \cite{LanBal_DecisionProceduresForSimpleEquationalTheoriesWithPermutativeEquations_TR77,LanBal_DecisionProceduresForSimpleEquationalTheoriesWithCommutativeAssociativeAxioms_TR1977,BalLan_NewDecisionAlgorithmsForFinitelyPresentedCommutativeSemigroups_CMA81,LanButBal_AProgressReportOnNewDecisionAlgorithmsForFinitelyPresentedAbelianGroups_CADE84}, though, restricted the attention to specific classes of equational theories, see
\cite[Section 2]{PetSti_CompleteSetsOfReductionsForSomeEquationalTheories_JACM81} for a detailed discussion.}
 However, Peterson and Stickel first define the notion on the basis of a relation on terms. 

\begin{defi}[``Rewriting modulo'' \emph{\cite[Definition 8.2]{PetSti_CompleteSetsOfReductionsForSomeEquationalTheories_JACM81}}]
\label{DefRewritingModulo_Definition8_2_PS81}
Let $\cR$ be a \trs{} and $E$ be a set of equations.
If $s$ and $t$ are terms, we say that $s\spcrel{\to\!\!(\cR,E)}t$ if there is a term $s'\in[s]_E$,
$p\in\Pos(s')$, $\lhsr\to\rhsr\in\cR$, and a substitution $\sigma$ such that $s'|_p=\sigma(\lhsr)$ and $s'[\sigma(\rhsr)]_p\in[t]_E$.
\end{defi}
%
\noindent
Note that, despite the notation $s\spcrel{\to\!\!(\cR,E)}t$, this is what we currently call \emph{rewriting modulo} of \emph{terms} (and would denote $s\rew{\cR/E}t$).

 \subsubsection{Implicit definition of $\rew{\cR,E}$}

Peterson and Stickel's definition of a \emph{complete set of reductions} \cite[Definition 8.11]{PetSti_CompleteSetsOfReductionsForSomeEquationalTheories_JACM81}
is based on equivalence rewriting.
Then, they investigate how to prove completeness of \emph{$E$-compatible reductions}.

\begin{defiC}[\emph{\cite[Definition 9.1]{PetSti_CompleteSetsOfReductionsForSomeEquationalTheories_JACM81}}]
\label{TCompatibility_PS81}

We say that $\cR$ is \emph{$E$-compatible} if whenever $s\spcrel{\to\!\!(\cR,E)}t$ [i.e., $s\rew{\cR/E}t$], there exist $p\in\Pos(s)$, 
a substitution $\sigma$, 
and $\lhsr\to\rhsr\in\cR$ such that
\begin{IEEEeqnarray}{r'C'l}
s|_p & \equequ{E} & \sigma(\lhsr)\label{TCompatibility_PS81_PSrewriting_Ematching}
\\
{}[t] & \rews{} & [s[\sigma(\rhsr)]_p]
\label{TCompatibility_PS81_EJoinability}
\end{IEEEeqnarray}
\end{defiC}
%
\noindent
Note that what we usually call \emph{Peterson \& Stickel reduction} (denoted $\rew{\cR,E}$) is ``piece-wise'' introduced in Definition \ref{TCompatibility_PS81}, \emph{without giving any explicit definition}:
(i) the $E$-matching condition is (\ref{TCompatibility_PS81_PSrewriting_Ematching}) and
(ii) the obtained $\rew{\cR,E}$-reduct is inside the $E$-equivalence class on the rightmost part of (\ref{TCompatibility_PS81_EJoinability}), i.e., $s'=s[\sigma(\rhsr)]_p$.
Using the notation in this paper, $E$-compatibility holds if the diagram in Figure \ref{FigECompatibility_Definition9_1_PS81} commutes.
\begin{figure}[t]
\[\xymatrix{
s\ar@{->}[rr]_{\cR/E} 
\ar@{.>}[dd]^{\cR,E} & & 
t\ar@{.>}[dd]^{*}_{\cR/E}
\\
\\
s' \ar@{|.|}[rr]^{*}_{E} & & t'
}
\]
\caption{Diagram for $E$-compatibility in \cite[Definition 9.1]{PetSti_CompleteSetsOfReductionsForSomeEquationalTheories_JACM81}}
\label{FigECompatibility_Definition9_1_PS81}
\end{figure}

\subsubsection{Extensions of rules}
\label{SecExtensionsOfRules_PS81}

In \cite[Theorem 8.12]{PetSti_CompleteSetsOfReductionsForSomeEquationalTheories_JACM81}, they proposed the use of 
\emph{extensions 
of rules} \cite[Definition 8.8]{PetSti_CompleteSetsOfReductionsForSomeEquationalTheories_JACM81} to 
prove completeness of reductions as the joinability (using class rewriting) of 
\emph{root $E$-overlaps} among extensions of rules of $\cR$.
However, in \cite[last paragraph of Section 8]{PetSti_CompleteSetsOfReductionsForSomeEquationalTheories_JACM81}, Peterson and Stickel remark
\begin{quote}
The difficulty with this theorem is that every reduction has infinitely many variable extensions, so the hypothesis cannot be checked.
\end{quote}
Although they used extensions in their completion algorithm for associative and commutative theories \cite[Sections 10--11]{PetSti_CompleteSetsOfReductionsForSomeEquationalTheories_JACM81}, 
their formal results to guarantee $E$-completeness are based on the notion of ($E$-)critical pair.

\subsubsection{($E$-)Critical pairs}
In \cite[Definition 9.2]{PetSti_CompleteSetsOfReductionsForSomeEquationalTheories_JACM81}, Peterson and Stickel introduce the following notion of \emph{critical pair}.

\begin{defiC}[\text{\cite[Definition 9.2]{PetSti_CompleteSetsOfReductionsForSomeEquationalTheories_JACM81}}]
\label{DefCriticalPair_Definition9_2_PS81}
Suppose $\lhsr\to\rhsr,\lhsr'\to\rhsr'\in\cR$, $p\in\Pos_\Symbols(\lhsr)$ and $\sigma$ is a unifier in a complete set of $E$-unifiers of $\lhsr|_p$ and $\lhsr'$.
Then, we say that $\langle[\sigma(\lhsr[\rhsr']_p),[\sigma(\rhsr)]\rangle$ is a \emph{critical pair} for $\cR$ and $E$.
\end{defiC}
%
\noindent
Note that $E$-unification is used instead of syntactic unification.
Except for the notation of critical pairs $\langle[s],[t]\rangle$ involving $E$-equivalence classes $[s]$ and $[t]$ rather than just terms (`a la Lankford and Ballantine's Definition \ref{CriticalPair_LB77}),
Definition \ref{DefCriticalPair_Definition9_2_PS81} is, essentially, the definition of $E$-critical pairs $\langle s,t\rangle$.
Peterson and Stickel prove that joinability of 
$[s]$ and $[t]$ above 
proves completeness of \emph{$E$-compatible} sets of reductions which are \emph{terminating} on classes of terms \cite[Theorem 9.3]{PetSti_CompleteSetsOfReductionsForSomeEquationalTheories_JACM81}.
Then, \cite[Theorem 9.5]{PetSti_CompleteSetsOfReductionsForSomeEquationalTheories_JACM81} provides a sufficient condition for $E$-compatibility provided that $E$ consists of \emph{linear equations} $s=t$ such that $\Var(s)=\Var(t)$.

\subsection{Jouannaud (1983). Jouannaud and Kirchner (1984, 1986)}
\label{SecRelatedWork_JouannaudAndKirchner}

Although Huet extensively used ``modulo'' to refer to properties of rewriting computations with some participation of equivalence relations, he did not use the term (or defined) ``rewriting modulo'' in \cite{Huet_ConfluentReductionsAbstractPropertiesAndApplicationsToTermRewritingSystems_FOCS77,
Huet_ConfluentReductionsAbstractPropertiesAndApplicationsToTermRewritingSystems_JACM80,HueOpp_EquationsAndRewriteRulesASurvey_FLT80}.
As mentioned in Remark \ref{RemTheNotationRewSim}, perhaps he was just borrowing the notion of quotient of relations from algebra to refer to the work of Lankford and Ballantine.
On the other hand, although Peterson and Stickel (improving on Lankford and Ballantine) introduced term rewriting modulo (Definition \ref{DefRewritingModulo_Definition8_2_PS81}), they do not use ``modulo'' or the notation $\rew{\cR/E}$ in their paper.
Furthermore, they do not provide an explicit definition of what is often called ``Peterson \& Stickel reduction'' and denoted $\rew{\cR,E}$.

It seems that the current definitions and notations for rewriting modulo come from 
\cite{Jouannaud_ConfluentAndCoherentEquationalTermRewritingSystemsApplicationToProofsInAbstractDataTypes_CAAP83,JouKir_CompletionOfASetOfRulesModuloASetOfEquations_POPL84}, where we find
\begin{itemize}
\item The use and definition of $\rew{\genrelation/\genequivalence}$ 
(or $\genrelation/\genequivalence$) to denote the relation $\equequ{\genequivalence}\circ\rew{\genrelation}\circ\equequ{\genequivalence}$, which ``simulates the induced relation in $E$-equivalence classes'' \cite[Definition 7]{Jouannaud_ConfluentAndCoherentEquationalTermRewritingSystemsApplicationToProofsInAbstractDataTypes_CAAP83}.
\item 
The use of $s\rew{\cR,E}t$ to denote a reduction step where $s|_p\equequ{E}\sigma(\lhsr)$ and $t=s[\sigma(\rhsr)]_p$ can be found in \cite[Definition 9]{Jouannaud_ConfluentAndCoherentEquationalTermRewritingSystemsApplicationToProofsInAbstractDataTypes_CAAP83} together with
the attribution of this notion to \cite{PetSti_CompleteSetsOfReductionsForSomeEquationalTheories_JACM81}.
\item The use of ``modulo'' to refer to properties (confluence, termination, convergence, etc.) of relations 
$\rew{\genrelation}$, 
$\rew{\genrelation,\genequivalence}$, 
$\rew{\genrelation/\genequivalence}$, etc., with respect to an equivalence $\genequivalence$ \cite[Definition 7]{JouKir_CompletionOfASetOfRulesModuloASetOfEquations_POPL84}
as already done by Huet.
\item The notion of $E$-critical pair \cite[Definition 10]{Jouannaud_ConfluentAndCoherentEquationalTermRewritingSystemsApplicationToProofsInAbstractDataTypes_CAAP83} and
\cite[Definition 12]{JouKir_CompletionOfASetOfRulesModuloASetOfEquations_SIAMJC86}.
\end{itemize}
However, no explicit definition of ``\emph{rewriting modulo}'' is found in 
\cite{Jouannaud_ConfluentAndCoherentEquationalTermRewritingSystemsApplicationToProofsInAbstractDataTypes_CAAP83,JouKir_CompletionOfASetOfRulesModuloASetOfEquations_POPL84,JouKir_CompletionOfASetOfRulesModuloASetOfEquations_SIAMJC86}.
In \cite[Example 4]{JouMun_TerminationOfASetOfRulesModuloASetOfEquations_CADE84}, ``rewriting modulo associativity'' is mentioned, but without any definition.
\begin{rem}[Different Church-Rosser properties]
Note that three different Church-Rosser properties are considered in the aforementioned works:
\begin{enumerate}
\item Hindley's (\ref{LblCRinHindley64}) involving a single relation 
$\rew{\genrelation}$ and its conversion $\conversion{\genrelation}$ (Section~\ref{SecRelatedWork_Hindley}).
\item Sethi's Figure \ref{FigFCRproperty_ASU72}(left) considers a relation $\rew{\genrelation}$ and an equivalence  $\equone{\genequivalence}^*$ (i.e., $\sim$) 
without mixing their use (Section~\ref{SecRelatedWork_Sethi}).
\item Jouannaud and Kirchner's Figure 
\ref{FigChurchRosserPropertyETerminatingR} considers several relations
$\rew{\genrelation}$,
$\rew{\genrelationUpE}$,
$\oneCRs$, and
 $\equone{\genequivalence}^*$. 
 In particular, $\oneCRs$ mixes $\rew{\genrelation}$, $\leftrew{\genrelation}$, and  $\equone{\genequivalence}$ steps (Section~\ref{SecAbstractFrameworkJK86}).
\end{enumerate}
\end{rem}
\subsection{Bachmair  and  Dershowitz (1986, 1987). Dershowitz and Jouannaud (1990)}
After 
\cite[page 6]{BacDer_CommutationTransformationAndTermination_CADE86}
(where $\rew{\cR/E}$ is also defined as $\equequ{E}\circ\rew{E}\circ\equequ{E}$) 
and
\cite[pages 193-194] {BacDer_CompletionForRewritingModuloACongruence_RTA87}
(where ``rewriting modulo'' is used for $\rew{\cR/E}$),
in \cite{DerJou_RewriteSystems_HTCS90}, 
the now standard definition, denomination, and notation 
of rewriting modulo is given:

\begin{defi}[Rewriting modulo \emph{\cite[Section 2.5, page 257]{DerJou_RewriteSystems_HTCS90}}]
Let $\cR$ be a \trs{} and $E$ be a set of equations.\footnote{In \cite{DerJou_RewriteSystems_HTCS90}, $S$ is used instead of $E$.} We say that 
\emph{$s$ rewrites to $t$ modulo E}, denoted $s\rew{\cR/E}t$,
if $s\equequ{E}u[\sigma(\lhsr)]$ and $u[\sigma(\lhsr)]\equequ{E}t$, 
for some term $u$, 
position $p\in\Pos(u)$, 
rule $\lhsr\to\rhsr\in\cR$ and substitution $\sigma$. 
\end{defi}

\subsection{Ohlebusch (1998)}
In \cite{Ohlebusch_ChurchRosserTheoremsForAbstractReductionModuloAnEquivalenceRelation_RTA98}, Ohlebusch considers additional abstract relations among relations and equivalences which can be used to prove $E$-confluence.
He also introduces short names to refer to such new (and existing) confluence (modulo) properties.
For instance, Property $\alpha$ is denoted as $\fS{LCON}\!\sim$.
An interesting aspect is that he provides some (abstract) 
$E$-confluence results which do \emph{not} require $E$-termination (e.g., \cite[Proposition 6 \& Corollary 10]{Ohlebusch_ChurchRosserTheoremsForAbstractReductionModuloAnEquivalenceRelation_RTA98}; however, some of the involved abstract properties 
(in particular, \emph{Strong Coherence with $\sim$, $\fS{SCOH}\!\sim$} \cite[Figure 1(v)]{Ohlebusch_ChurchRosserTheoremsForAbstractReductionModuloAnEquivalenceRelation_RTA98}) are \emph{not} local: peaks $t\leftrews{} s\equone{}^*t'$ must be considered with $\fS{SCOH}\!\sim$.
This makes the treatment of such peaks with conditional pairs, as discussed in this paper, unclear.

Another contribution of \cite{Ohlebusch_ChurchRosserTheoremsForAbstractReductionModuloAnEquivalenceRelation_RTA98} is the development of new techniques for proving $E$-confluence of \etrs{s} based on the use of van Oostrom's \emph{decreasing diagrams} \cite{Oostrom_ConfluenceByDecreasingDiagrams_TCS94}.
This technique has also been explored by other researchers for this purpose
\cite{Felgenhauer_ConfluenceOfTermRewritingTheoryAndAutomation_PhD15,%
JouLiu_FromDiagrammaticConfluenceToModularity_TCS12,%
LiuJouOga_ConfluenceOfLayeredRewriteSystems_CSL15}.
Ohlebusch's approach could also be easily used in our setting, as his main theorem \cite[Theorem 14]{Ohlebusch_ChurchRosserTheoremsForAbstractReductionModuloAnEquivalenceRelation_RTA98} identifies peaks
(\ref{GenericRPeakTerms}) which could be particularized for rules $\alpha_1,\alpha_2\in\cR^\RMabbr$ 
(for \cite[Figure 7(left)]{Ohlebusch_ChurchRosserTheoremsForAbstractReductionModuloAnEquivalenceRelation_RTA98}) 
and $\alpha_1\in\cR^\RMabbr$ and $\alpha_2\in\eqoriented{E}$ 
(for \cite[Figure 7(right)]{Ohlebusch_ChurchRosserTheoremsForAbstractReductionModuloAnEquivalenceRelation_RTA98}), respectively.
Then, the treatment of peaks would proceed by using appropriate CCPs and CVPs as explained in Section~\ref{SecConditionalPairsAlphaAndGamma}, although the joinability conditions should be adapted to those in \cite{Ohlebusch_ChurchRosserTheoremsForAbstractReductionModuloAnEquivalenceRelation_RTA98}. This is an interesting subject for future work.

\subsection{Meseguer (2017)}

Coherence of conditional rewriting modulo axioms is investigated in
\cite{Meseguer_StrictCoherenceOfConditionalRewritingModuloAxioms_TCS17}.
Working on Jouannaud and Kirchner's abstract approach, in \cite[Section 4.3]{Meseguer_StrictCoherenceOfConditionalRewritingModuloAxioms_TCS17}
the problem discussed in Example \ref{ExRModuloAndRrelativeE} is noticed.
That is: in conditional rewriting,
the \emph{natural choice} of 
interpreting 
$\genrelation$ as $\rew{\cR}$,
$\genequivalence$ as $\equequ{E}$, 
and $\genrelation/\genequivalence$ as $\rew{\cR/E}$ does \emph{not} fulfill the requirement (\ref{LblReductionModuloFromReductionAndEquivalence}) of Jouannaud and Kirchner's abstract framework, as  $\rew{\cR/E}$ and $\composeRel{\equequ{E}}{\composeRel{\rew{\cR}}{\equequ{E}}}$ does \emph{not} coincide.
Meseguer's solution is interpreting
\emph{both} $\genrelation$ and $\genrelationUpE$ as  
$\rew{\cR,E}$, see
\cite[paragraph above Proposition 1]{Meseguer_StrictCoherenceOfConditionalRewritingModuloAxioms_TCS17}.
This solution, though, is not appropriate for our purposes, as it implies that only peaks (\ref{GenericPSPeakTerms}), i.e.,
$\leftrew{\cR,E}\circ\rew{\cR,E}$ and $\leftrew{\cR,E}\circ\rew{\eqoriented{E}}$, can be considered,
thus making the use of $E$-unification and ECCPs mandatory.
In contrast, our approach also permits $E$-confluence proofs based on considering ``classical'' conditional critical pairs, not requiring $E$-unifiers,
as in Sections \ref{SecConditionalPairsAlphaAndGamma}
and \ref{SecEConfluenceWithRandCpeaks}.

\subsection{Jouannaud (2024)}
\label{SecJouannaud2024}

In \cite[Section 11.3]{Jouannaud_ConfluenceOfTerminatingRewritingComputations_TFSP24}, Jouannaud introduces a revision of the framework in \cite{JouKir_CompletionOfASetOfRulesModuloASetOfEquations_SIAMJC86}.
Some of these changes are not minor.
\begin{enumerate}
\item The \emph{fundamental assumption} (\ref{LblFundamentalAssumptionJK86}) is replaced by a 
\emph{soundness property} \cite[page 275]{Jouannaud_ConfluenceOfTerminatingRewritingComputations_TFSP24} which is, using the notation in this paper, as follows:
\begin{IEEEeqnarray}{r'C'l}
\genrelation\:\subseteq\:\genrelationUpE\:\subseteq\:(\genequivalence\circ\genrelation)^*\label{LblSoundnessJou24}
\end{IEEEeqnarray}
Note that this property and (\ref{LblFundamentalAssumptionJK86}) are \emph{not comparable}.
For instance, for an \etrs{} $\cR=(\Symbols,E,R)$, 
with $\genrelation$ equal to $\rew{\cR}$ and
$\genequivalence$ equal to $\equequ{E}$, in general, 
\begin{itemize}
\item If
$\genrelationUpE$ is $\rews{\cR,E}$, then (\ref{LblSoundnessJou24}) is satisfied, 
but (\ref{LblFundamentalAssumptionJK86}) is not.
\item If $\genrelationUpE$ is $\genrelation/\genequivalence$ (i.e., $\genequivalence\circ\genrelation\circ\genequivalence$), then (\ref{LblFundamentalAssumptionJK86}) is satisfied, but (\ref{LblSoundnessJou24}) is not.
\end{itemize}
\item Instead of $\fS{LCON}_\genequivalence(\genrelationUpE,\genrelation)$ in Definition \ref{DefLocalConfluenceAndCoherenceProperties_JK86}(\ref{DefLocalConfluenceAndCoherenceProperties_JK86_LocalConfluence}), see also Figure \ref{FigConfluenceAndCoherenceProperties} (left), 
the notion of \emph{local confluence modulo $\genequivalence$} in \cite[Definition 11.7(3)]{Jouannaud_ConfluenceOfTerminatingRewritingComputations_TFSP24} is used.
It coincides with \emph{local $\genequivalence$-confluence} of \cite[Definition 5(3)]{JouMun_TerminationOfASetOfRulesModuloASetOfEquations_CADE84}
and corresponds to $\fS{LCON}_\genequivalence(\genrelationUpE,\genrelationUpE)$, see Section~\ref{SecDiscussionJK86}.
Still, as discussed in Section~\ref{SecAboutLocalPeaksModulo}, if $\genrelationUpE$ is $\rew{\cR,E}$, as a consequence of Proposition \ref{PropLocalPeaksModulo}, both notions coincide in practice.
\item The notion of \emph{local coherence modulo $\genequivalence$} in \cite[Definition 11.7(4)]{Jouannaud_ConfluenceOfTerminatingRewritingComputations_TFSP24} uses $\ejoinability{\genrelationUpE}$-joinability instead of $\rsejoinability{\genrelationUpE}$-joinability as in $\fS{LCOH}_\genequivalence(\genrelationUpE)$ in Definition \ref{DefLocalConfluenceAndCoherenceProperties_JK86}(\ref{DefLocalConfluenceAndCoherenceProperties_JK86_LocalCoherence}), see also Figure \ref{FigConfluenceAndCoherenceProperties}(right).
However, for $\genequivalence$-terminating relations $\genrelation$, both notions coincide, see Proposition \ref{PropEquivalenceLocalCoherenceModuloEandProperty5prime_JK86}.
\item Regarding denominations and notations, we have the following:
\begin{enumerate}[align=left]
\item The relation $\equequ{E}\circ\rew{\cR}$ is called \emph{class rewriting} and denoted $\rew{\cR/E}$
(but $\rew{\cR/E}$ is usually defined as $\equequ{E}\circ\rew{\cR}\circ\equequ{E}$ \cite{Jouannaud_ConfluentAndCoherentEquationalTermRewritingSystemsApplicationToProofsInAbstractDataTypes_CAAP83,JouMun_TerminationOfASetOfRulesModuloASetOfEquations_CADE84,JouKir_CompletionOfASetOfRulesModuloASetOfEquations_POPL84,JouKir_CompletionOfASetOfRulesModuloASetOfEquations_SIAMJC86,DerJou_RewriteSystems_HTCS90,BaaNip_TermRewritingAndAllThat_1998}).
\item As in \cite{Jouannaud_ConfluentAndCoherentEquationalTermRewritingSystemsApplicationToProofsInAbstractDataTypes_CAAP83,JouKir_CompletionOfASetOfRulesModuloASetOfEquations_SIAMJC86}, 
in \cite[Definition 11.8]{Jouannaud_ConfluenceOfTerminatingRewritingComputations_TFSP24} the relation denoted $\rew{\cR,E}$ in this paper (following \cite{Jouannaud_ConfluentAndCoherentEquationalTermRewritingSystemsApplicationToProofsInAbstractDataTypes_CAAP83,JouMun_TerminationOfASetOfRulesModuloASetOfEquations_CADE84,JouKir_CompletionOfASetOfRulesModuloASetOfEquations_SIAMJC86,BaaNip_TermRewritingAndAllThat_1998}) is 
attributed to Peterson and Stickel \cite{PetSti_CompleteSetsOfReductionsForSomeEquationalTheories_JACM81}, but it is 
called \emph{rewriting modulo} and denoted $\rew{{\cR_E}}$, in contrast to
\cite{DerJou_RewriteSystems_HTCS90,BaaNip_TermRewritingAndAllThat_1998} which use $\rew{\cR/E}$.
\end{enumerate}
\end{enumerate}
The notion of \emph{extension} of a rule $\lhsr\to\rhsr\in R$ by an equation $s=t\in E$ \cite[Definition 11.10]{Jouannaud_ConfluenceOfTerminatingRewritingComputations_TFSP24} and \emph{closure under extensions} 
\cite[Definition 11.11]{Jouannaud_ConfluenceOfTerminatingRewritingComputations_TFSP24} are used 
to provide a new characterization of the Church-Rosser property as follows:
\begin{thmC}
[\label{Theorem11_5_Jou24}
{\cite[Theorem 11.5]{Jouannaud_ConfluenceOfTerminatingRewritingComputations_TFSP24}}]
Assume that $\cR$ is $E$-terminating and closed under $E$-extensions. Then, $\cR$ is Church-Rosser modulo $E$ iff all its $E$-critical pairs are joinable.
\end{thmC}
%
\noindent
However, it is unclear how to check  closedness under $E$-extensions, as
\begin{quote}
\ldots there could be infinitely many [$E$-extensions of a system $\cR$]\ldots 
We have seen this does not happen with associativity, or with associativity and commutativity\ldots 
But does sometimes happen. \cite[page 282, paragraph between Examples 11.5 and 11.6]{Jouannaud_ConfluenceOfTerminatingRewritingComputations_TFSP24}
\end{quote}
which was already observed by Peterson and Stickel, see Section~\ref{SecExtensionsOfRules_PS81}.
In \cite[Theorem 11.5]{Jouannaud_ConfluenceOfTerminatingRewritingComputations_TFSP24}, the requirement
of finite $E$-congruence classes in \cite[Theorem 16]{JouKir_CompletionOfASetOfRulesModuloASetOfEquations_SIAMJC86} has been dropped.\footnote{In
page 282, Jouannaud also mentions that the assumption of $E$ being variable preserving, i.e., every equation $s=t\in E$ satisfies $\Var(s)=\Var(t)$
(which is called \emph{nonerasingness} in \cite[page 1169]{{JouKir_CompletionOfASetOfRulesModuloASetOfEquations_SIAMJC86}}) has been dropped. 
However, as noticed in \cite{JouKir_CompletionOfASetOfRulesModuloASetOfEquations_SIAMJC86} (see Remark \ref{RemANecessaryConditionForETermination}), if $E$-termination of $\cR$ is required (as done in Theorem \ref{Theorem11_5_Jou24}), then every equation $s=t\in E$ must be variable preserving.}
However, in contrast to \cite[Theorem 16]{JouKir_CompletionOfASetOfRulesModuloASetOfEquations_SIAMJC86}, the decomposition $R=L\uplus N$, where $L$ consists of left-linear rules and $N$ is $R-L$ is missing.

\section{Conclusion and future work}
\label{SecConclusionAndFutureWork}

We have investigated how to prove confluence of conditional rewriting modulo  by using Huet and Jouannaud and Kirchner's abstract frameworks.
In particular, for abstract relations $\genequivalence$ (an equivalence),
$\genrelation$ (a reduction relation), and $\genrelationUpE$ (a reduction relation satisfying $\genrelation\subseteq\genrelationUpE\subseteq\genrelation/\genequivalence$), 
we relied on Jouannaud and Kirchner's notions of
(i) local confluence of $\genrelationUpE$ modulo $\genequivalence$  with $\genrelation$
and
(ii) local coherence of $\genrelationUpE$ modulo $\genequivalence$, which characterize 
the $\genrelationUpE$-Church-Rosser property modulo $\genequivalence$ of $\genrelation$, which is a sufficient condition for $\genequivalence$-confluence of $\genrelation$.
We have proved that (i) and (ii) are \emph{not} necessary for $\genequivalence$-confluence, though, which means that disproving them does \emph{not} disprove $\genequivalence$-confluence of $\genrelation$.
However, we have proved that the non-$\genrelation/\genequivalence$-joinability of  the peaks considered  in the first property (i) provides a method for \emph{disproving} $\genequivalence$-confluence of $\genrelation$.

After this analysis, we have introduced \emph{Equational Generalized Term Rewriting Systems} (\egtrs{s})
$\cR=(\Symbols,\SPredicates,\mu,E,H,R)$
over signatures $\Symbols$ and $\SPredicates$ of function and predicate symbols,
consisting of \emph{conditional equations} (in $E$)
and \emph{conditional rules} (in $R$)
whose conditions are sequences of \emph{atoms}, possibly defined by
additional Horn clauses (in $H$).
Replacement restrictions $\mu$ specifying which arguments of function symbols can be rewritten are also allowed.
Rewriting computations with \egtrs{s} $\cR$ 
are described by \emph{deduction} in 
appropriate First-Order theories (equivalently by proof in the corresponding Elementary Inference Systems).
We have shown that an appropriate definition of a rewrite relation 
$\rew{\cR^\RMabbr}$ based on a modification of the treatment of the conditional part of clauses in $H$ and rules in $R$ permits the use of Jouannaud and Kirchner's framework \cite{JouKir_CompletionOfASetOfRulesModuloASetOfEquations_SIAMJC86}
to prove $E$-confluence of \egtrs{s}. 

Then, we have shown how to prove 
$E$-confluence of $\cR$ as the joinability of different families of conditional critical pairs (CCPs)
and conditional variable pairs (CVPs) obtained from the equations $E$ and rules $R$ of $\cR$.
Although CCPs and CVPs were already used in proofs of confluence of \gtrs{s}, here CCPs are used not only to deal with overlaps between rules of $R$ but also to handle overlaps between rules from $R$ and $\eqoriented{E}$.
Furthermore, we consider \emph{parametric} CVPs, which are used in several novel ways specific for proving $E$-confluence, as in the definitions of 
$\CVPofR(E)$, $\CVPofEqOne(\cR)$, $\CVPofPStickel(\cR)$, and $\CVPofPStickel(E)$.
We have also introduced \emph{Logic-Based Conditional Critical Pairs} (LCCPs) 
as an alternative to $E$-critical CCPs to 
avoid the computation of $E$-unifiers in proofs of $E$-confluence, thus leading to proofs of $E$-confluence which always rely on \emph{finite} sets of conditional pairs.
Besides, \emph{Down Conditional Pairs} (DCPs) have been introduced here and showed useful to \emph{disprove} $E$-confluence.
As far as we know, CVPs, DCPs, and LCCPs 
have not been used yet  in proofs of $E$-confluence.

We have shown that $E$-confluence of $E$-terminating 
\egtrs{s} can be 
\emph{proved}
in two different ways, as:
\begin{enumerate}
\item $\ejoinability{\cR^\RMabbr}$-joinability of all conditional pairs in
\[
\CCP(\cR)\cup\CVPofR(\cR)\cup\CCP(\cR,E)\cup\CCP(E,\cR)\cup\CVPofR(E)\cup\CVPofEqOne(\cR)
\]
(Theorem \ref{TheoEConfluenceOfETerminatingEGTRSsHuet}(\ref{TheoEConfluenceOfETerminatingEGTRSsHuet_EConfluence})). Or else, as
\item $\ejoinability{\cR^\RMabbr,E}$-joinability of 
all conditional pairs in 
\[
\LCCP(\cR)\cup\CVPofPStickel(\cR)\cup\LCCP(E,\cR)\cup\CVPofPStickel(E)
\] 
(Theorem \ref{TheoEConfluenceWithoutDCPs}). 
\end{enumerate}
From a theoretical point of view, the first method is subsumed by the second. In practice, though, the first method is computationally simpler, as $E$-unifiers are not used to compute CCPs and 
$E$-matchers are not necessary for joinability tests.
Similarly, we provide two criteria for \emph{disproving} $E$-confluence of \egtrs{s} as 
\begin{enumerate}
\item non-$\ejoinability{\cR/E}$-joinability of a conditional pair in $\CCP(\cR)\cup\CVPofR(\cR)$ (Theorem \ref{TheoEConfluenceOfETerminatingEGTRSsHuet}(\ref{TheoEConfluenceOfETerminatingEGTRSsHuet_NonEConfluent})), or
\item non-$\ejoinability{\cR/E}$-joinability of a conditional pair in $\LCCP(\cR)\cup\CVPofPStickel(\cR)\cup\DCP(\cR)$ (Theorem \ref{TheoremNonEConfluenceOfEGRSs}).
\end{enumerate}
Again, the first method is subsumed by the second, but the first is simpler.
The discussed examples (see also Sections \ref{SecApplicationToETRSs} and \ref{SecApplicationToConditionalETRSs})
show that, beyond their use with \egtrs{s}, which is the main focus and contribution of this paper, the
new techniques are also useful with  \etrs{s} (Examples \ref{ExPeakNoCPs}, \ref{ExHuet80_RemarkPage818}, and \ref{ExHuet80_RemarkPage818variant}) 
or conditional rewrite theories (e.g., $\cR$ in Example \ref{ExRModuloAndRrelativeE}) as existing results cannot  appropriately handle them.

\medskip
\noindent
\textbf{Future work.} 
From a \emph{theoretical point of view}, 
\begin{enumerate}
\item 
Proposition \ref{PropLocalConfluenceAndLocalCoherenceNotNecessaryForEConfluence} shows that $E$-confluence is still \emph{not} appropriately captured by means of suitable local confluence/coherence properties (possibly requiring $E$-termination).
Relying on local properties is important as it opens the way to envisage appropriate conditional pairs which can be used to capture the corresponding peaks, hopefully leading to a \emph{characterization} of $E$-confluence as the joinability of such pairs using an appropriate joinability notion.
Such a characterization should be further explored.
\item Although sorts can be handled as predicates in the signature $\SPredicates$ of predicate symbols of an \egtrs{} and the subsort relation can be defined using Horn clauses in $H$, making the order-sorted discipline explicit in \egtrs{s} would be an interesting improvement, possibly leading to a more applicable framework, although with important theoretical challenges.
\item The use of our techniques for proving and disproving confluence modulo of \Maude{} programs should be further investigated.
The use of \Maude{} programs to simulate computations with \egtrs{s} would also be  interesting as a prototypic implementation of \egtrs{s}.
\item The decomposition of the set $R$ of rules of an \egtrs{} into a subset $L$ of left-linear rules and a subset $N$ containing the remaining ones, so that only rules in $N$ require the computation of $E$-critical pairs, as done in \cite[Theorem 16]{JouKir_CompletionOfASetOfRulesModuloASetOfEquations_SIAMJC86} is also an interesting subject of further research.
\item Exploring the impact of our techniques in \emph{first-order deduction modulo} see, e.g., \cite{Dowek_AutomatedTheoremProvingInFirstOrderLogicModuloOnTheDifferenceBetweenTypeTheoryAndSetTheory_arXiv23} and the references therein, 
is another interesting subject of future work.
\end{enumerate}
From a \emph{practical} point of view, much work remains to be done 
to use these new techniques.
\begin{enumerate}
\item $E$-termination of \egtrs{s} $\cR$, which is required in our main results, is underexplored to date.
As discussed in Section  \ref{SecETerminationOfEGTRSs}, the 
(automatic) computation of models of $\GLtheory_{\cR/E}$ 
making $\rewmodulo{}$ well-founded is possible by using existing tools (e.g., \AGES{} or \MaceFour{}) is an interesting starting point, but
the development of more specific techniques would be necessary.
\item 
Also, from a practical point of view,  
it is interesting to investigate the possibility 
of a \emph{mixed} use of ECCPs
together 
with LCCPs.
For instance, if there is an $E$-unification algorithm and the set of (complete) 
$E$-unifiers is finite, then the corresponding (computable and finite set of) 
ECCPs could be used instead 
of LCCPs.
\item Our results
heavily rely on checking 
$\ejoinability{\cR^\RMabbr}$-joinability,
$\ejoinability{\cR^\RMabbr,E}$-joinability and
non-$\ejoinability{\cR^\RMabbr/E}$-joinability of 
CCPs, (different kinds of) CVPs, DCPs, and  LCCPs to obtain proofs of (non)
$E$-confluence.
We plan to improve our tool \CONFident{} \cite{GutLucVit_ProvingConfluenceInTheConfluenceFrameworkWithCONFident_FI24}, 
which 
implements methods for the  analysis of similar conditional pairs. 
Such results heavily relies on the (in)feasibility results developed in \cite{GutLuc_AutomaticallyProvingAndDisprovingFeasibilityConditions_IJCAR20} and 
implemented in the tool \infChecker{}, which currently deals with \csctrs{s} and
should also be adapted to deal with more general theories as the ones used with \egtrs{s}.
\end{enumerate}

\medskip
\noindent
\textbf{Acknoledgements.} 
I thank the anonymous referees for their thorough analysis and helpful examples and remarks leading to many improvements on the first version of the paper.
I thank  Nachum Dershowitz and G\'erard Huet for kindly attending my requests about terminological and historical aspects of confluence modulo.
Thanks are also due to Pierre Lescanne for drawing my attention, many years ago, on Dallas Lankford's scientific contributions, many of them connected with the subject of this paper.
I thank Paco Dur\'an, Santiago Escobar, Jos\'e Meseguer, and Julia Sapi\~na for interesting discussions about the topics of this paper. 
Special thanks are due to Martin Zimmermann, for his patience during the development of this paper.

\bibliographystyle{alphaurl}
\bibliography{myBibliography}

\end{document}